\RequirePackage{fix-cm}
\documentclass[smallextended]{svjour3}       % onecolumn (second format)
\smartqed  % flush right qed marks, e.g. at end of proof
\usepackage{graphicx}
\usepackage{color}
\usepackage{rotating}
\usepackage{amssymb, amsmath, amsfonts}
\usepackage{tikz,cancel, enumerate, url, mathrsfs}
\usetikzlibrary{arrows,shapes,snakes,hobby}
\usepackage[hypertexnames=false]{hyperref}  
\usepackage{parskip} 
\usetikzlibrary{positioning}
\usepackage{graphicx}
\usepackage{subcaption}
\usepackage{mathtools}
\usepackage{tikz-3dplot}
\usepackage{multirow}
\usepackage[bbgreekl]{mathbbol}
\DeclareSymbolFontAlphabet{\mathbbm}{bbold}
\DeclareSymbolFontAlphabet{\mathbb}{AMSb}%

\newcommand{\bfA}{\mathbf{A}}

\newcommand{\bfC}{\mathbf{C}}

\newcommand{\bfI}{\mathbf{I}}

\newcommand{\bfQ}{\mathbf{Q}}

\newcommand{\bfT}{\mathbf{T}}

\newcommand{\bfb}{\mathbf{b}}

\newcommand{\bfe}{\mathbf{e}}
\newcommand{\bff}{\mathbf{f}}

\newcommand{\bfn}{\mathbf{n}}

\newcommand{\bfu}{\mathbf{u}}
\newcommand{\bfv}{\mathbf{v}}

\newcommand{\bfx}{\mathbf{x}}

\newcommand{\bftau}{{ \boldsymbol{\tau} }}

\newcommand{\bfxi}{{\boldsymbol{\xi}}}

\newcommand{\bfeta}{{\boldsymbol{\eta}}}

\newcommand{\bfzeta}{\boldsymbol{\zeta}}

\newcommand{\tr}{\operatorname{tr}}

\let\originalleft\left
\let\originalright\right
\renewcommand{\left}{\mathopen{}\mathclose\bgroup\originalleft}
\renewcommand{\right}{\aftergroup\egroup\originalright}

\journalname{}

\begin{document}

\title{Anisotropic Two-Dimensional, Plane Strain, and Plane Stress Models in Classical Linear Elasticity and Bond-Based Peridynamics 
}

\titlerunning{Two-Dimensional and Planar Classical Linear Elasticity and Peridynamics}        % if too long for running head

\author{Jeremy Trageser        \and
        Pablo Seleson %etc.
}


\institute{Jeremy Trageser \at
			Computer Science and Mathematics Division, \\
			Oak Ridge National Laboratory, \\
			One Bethel Valley Road, P.O.~Box 2008, MS-6211, Oak Ridge, TN 37831-6211 \\
              \email{trageserje@ornl.gov}           %  \\
           \and
           Pablo Seleson \at
           Computer Science and Mathematics Division, \\
			Oak Ridge National Laboratory, \\
			One Bethel Valley Road, P.O.~Box 2008, MS-6211, Oak Ridge, TN 37831-6211 \\
			\email{selesonpd@ornl.gov} 
}

\date{Received: date / Accepted: date}

\maketitle

\begin{abstract}
This paper concerns anisotropic two-dimensional and planar elasticity models within the frameworks of classical linear elasticity and the bond-based peridynamic theory of solid mechanics. We begin by reviewing corresponding models from the classical theory of linear elasticity. This review includes a new elementary and self-contained proof that there are exactly four material symmetry classes of the elasticity tensor in two dimensions. We also summarize classical plane strain and plane stress linear elastic models and explore their connections to the pure two-dimensional linear elastic model, relying on the definitions of the engineering constants. We then provide a novel formulation for pure two-dimensional anisotropic bond-based linear peridynamic models, which accommodates all four material symmetry classes. We further present innovative formulations for peridynamic plane strain and plane stress, which are obtained using direct analogies of the classical planar elasticity assumptions, and we specialize these formulations to a variety of material symmetry classes. The presented anisotropic peridynamic models are constrained by Cauchy's relations, which are an intrinsic property of bond-based peridynamic models. The uniqueness of the presented peridynamic plane strain and plane stress formulations in this work is that we directly reduce three-dimensional models to two-dimensional formulations, as opposed to matching two-dimensional peridynamic models to classical plane strain and plane stress formulations. This results in significant computational savings, while retaining the dynamics of the original three-dimensional bond-based peridynamic problems under suitable assumptions. 
\end{abstract}

\section{Introduction}

Modeling material failure and damage is an essential consideration in the engineering and materials science communities. In classical continuum mechanics, the governing equations are based on spatial derivatives, which poses difficulties when discontinuities such as cracks develop. As a remedy, the nonlocal peridynamic theory of solid mechanics was proposed in \cite{SILLING2000,Silling2007}. Rather than utilizing spatial derivatives, peridynamics employs spatial integrals, which model long-range interactions between material points. The idea of long-range interactions parallels the molecular dynamics theory and, in some sense, peridynamics could be considered a continuum version of molecular dynamics \cite{Seleson2014,Pablo2009}.

With the exception of various works on fiber-reinforced composites (e.g., \cite{askari2006,hu2011,kilic2009,oterkus2012peridynamic,oterkus2012}) and a few other works on materials with some degree of anisotropy (e.g., \cite{Azdoud2013,de2016,Ghajari2014,zhang2018supershear}), most peridynamic models in the literature describe isotropic materials; however, in practice, many materials are anisotropic \cite{Ting1996}. To facilitate analysis of a more diverse group of materials, a three-dimensional anisotropic peridynamic model was proposed in \cite{STG2019}. The model was shown to be capable of describing any of the eight material symmetry classes in three-dimensional classical linear elasticity. In two dimensions, there are exactly four material symmetry classes in classical linear elasticity \cite{He1996}. In this work, we provide an elementary self-contained proof of this fact and present a two-dimensional peridynamic model able to accommodate each of those four material symmetry classes.  

A common characteristic of nonlocal models is a penchant to be computationally expensive, particularly in higher dimensions. Even in classical (local) continuum mechanics, three-dimensional models are frequently simplified to two-dimensional formulations, e.g., plane strain and plane stress \cite{Ting1996}. Notably, peridynamic plane strain and plane stress models were considered in, e.g., \cite{Gerstle2005,Ghajari2014,Le2014,Sarego2016}. However, those works simply employ two-dimensional peridynamic models rather than placing assumptions on three-dimensional peridynamic models to derive two-dimensional formulations. In this work, we look in depth at two-dimensional simplifications of three-dimensional anisotropic bond-based linear peridynamic models. The simplifications are facilitated by the use of nonlocal analogues of the planar elasticity assumptions commonly appearing in classical linear elasticity. 

It is well established that bond-based peridynamic models can only describe a limited number of materials. In \cite{SILLING2000}, it was demonstrated that isotropic three-dimensional bond-based peridynamic models can only describe materials with a Poisson's ratio of $\frac{1}{4}$. This was accomplished by introducing the concept of peridynamic traction. A similar derivation leads to a Poisson's ratio of $\frac{1}{3}$ in two dimensions ({\em cf}. Appendix~\ref{appendix:poissonratio}). When anisotropy is considered, additional lesser-known restrictions are imposed by bond-based peridynamic models. These restrictions are due to the utilization of a pair potential in bond-based peridynamics. In molecular dynamics, it is well known that such a potential imposes Cauchy's relations on a corresponding linear elasticity model~\cite{Stakgold1950}, and in this work we show that Cauchy's relations are similarly imposed by bond-based linear peridynamic models as demonstrated in \cite{STG2019}. Moreover, we show these Cauchy's relations restrictions are an inherent property of the bond-based peridynamic model and independent of the definition of peridynamic traction. Furthermore, we show the Poisson's ratio restrictions of~$\frac{1}{4}$ in three dimensions and $\frac{1}{3}$ in two dimensions are a specific case of Cauchy's relations for isotropic materials. 

The organization of this paper is as follows. 
In Section \ref{sec:ClassicalLinearElasticity}, we review the classical theory of linear elasticity to provide a background for material anisotropy and planar elasticity (specificially plane strain and plane stress), in order to make connections with peridynamic models. Specifically, in Section~\ref{sec:ClassicalSymm}, we provide an elementary proof of the fact that there are exactly four material symmetry classes in two-dimensional classical linear elasticity. In Section~\ref{sec:puretwodimclass}, we review the equations of motion for classical pure two-dimensional linear elasticity. We continue in Section~\ref{sec:planemodelsclass} by deriving plane strain and plane stress formulations in classical linear elasticity, following \cite{Ting1996}, whose derivations we mimic in peridynamics.  We finalize the discussion of classical linear elasticity by reviewing engineering constants and Cauchy's relations in Sections~\ref{sec:techconsts} and \ref{sec:CauchyRelations}, respectively. In Section~\ref{sec:LinearElastPeri}, we delve into the bond-based peridynamic theory. In particular, in Section~\ref{sec:TwoDPeridynamicModels}, we develop a two-dimensional linear bond-based peridynamic model capable of describing materials in any of the four material symmetry classes appearing in two-dimensional classical linear elasticity. We continue in Section \ref{sec:planemodelsperi} by deriving anisotropic peridynamic formulations of plane strain and plane stress. Finally, conclusions are presented in Section~\ref{sec:Conclusions}.

%%%%%%%%%%%%%%%%%%%%%%%%%%%%%%%%%%%%%%%%%%%%%%%%%%%%%%%%%%%%%%%%%%%%%%%%%%%%%%%%%%%%%%%%%%%%%%%%%%%%%%%%%%%%%%

\section{Classical linear elasticity}\label{sec:ClassicalLinearElasticity}

In classical linear elasticity, the stress tensor $\boldsymbol{\sigma}$ and strain tensor $\boldsymbol{\varepsilon}$ are related via the generalized Hooke's law \cite{Ting1996}:
\begin{equation}\label{eqn:stressstrainrelation}
\sigma_{ij} = C_{ijkl} \varepsilon_{kl},
\end{equation}
where $C_{ijkl}$ are the components of the fourth-order elasticity tensor $\mathbb{C}$ and Einstein summation convention for repeated indices is employed. Due to the symmetries of the stress and strain tensors, $\mathbb{C}$ inherits the minor symmetries:
\begin{equation}\label{eqn:minorsymmetries}
C_{ijkl} = C_{jikl} = C_{ijlk}.
\end{equation}
Furthermore, the relation between the strain energy density $W = \frac{1}{2} \sigma_{ij} \varepsilon_{ij}$ and~$\mathbb{C}$,
\[
C_{ijkl} = \frac{\partial^2 W}{ \partial \varepsilon_{ij} \partial \varepsilon_{kl} } = \frac{\partial^2 W}{ \partial \varepsilon_{kl} \partial \varepsilon_{ij} } = C_{klij},
\]
guarantees the major symmetry of $\mathbb{C}$:
\begin{equation}\label{eqn:majorsymmetry}
C_{ijkl} = C_{klij}.
\end{equation}

\begin{remark}
Note that the relation $C_{ijkl} = C_{ikjl}$ does not hold in general. However, employing a molecular description of a material where interactions between particles are described by pairwise potentials will apply these restrictions to the elasticity tensor~\cite{Stakgold1950}. 
%{\color{red} However, a solid composed of lattice points that interact only through a central point potential will apply these restrictions to the elasticity tensor~\cite{poincare1892leccons}}. 
These relations are commonly referred to as Cauchy's relations~\cite{Love1892} and are discussed in further detail in Section \ref{sec:CauchyRelations}.
\end{remark}

Due to the major and minor symmetries, we may utilize Voigt notation to express the fourth-order tensor $\mathbb{C}$ as a symmetric second-order tensor $\mathbf{C}$ and, similarly, represent the second-order tensors $\boldsymbol{\sigma}$ and $\boldsymbol{\varepsilon}$ as vectors. In this formulation, we express \eqref{eqn:stressstrainrelation} in $\mathbb{R}^3$ as

\begin{equation}\label{eqn:3DElastTensor}
\left[
\begin{array}{c}
\sigma_{11} \\
\sigma_{22} \\
\sigma_{33} \\
\sigma_{23} \\
\sigma_{13} \\
\sigma_{12}
\end{array}
\right] = 
\left[
\begin{array}{cccccc}
C_{1111} & C_{1122} & C_{1133} & C_{1123} & C_{1113} & C_{1112} \\
\cdot & C_{2222} & C_{2233} & C_{2223} & C_{2213} & C_{2212} \\
\cdot & \cdot & C_{3333} & C_{3323} & C_{3313} & C_{3312} \\
\cdot&\cdot &\cdot & C_{2323} & C_{2313} & C_{2312} \\
\cdot& \cdot& \cdot& \cdot& C_{1313} & C_{1312} \\
\cdot& \cdot& \cdot& \cdot& \cdot & C_{1212}
\end{array}
\right]
\left[
\begin{array}{c}
\varepsilon_{11} \\
\varepsilon_{22} \\
\varepsilon_{33} \\
2 \varepsilon_{23} \\
2 \varepsilon_{13} \\
2 \varepsilon_{12}
\end{array}
\right]
\end{equation}
and in $\mathbb{R}^2$ as

\begin{equation}\label{eqn:2DElastTensor}
\left[
\begin{array}{c}
\sigma_{11} \\
\sigma_{22} \\
\sigma_{12}
\end{array}
\right] = 
\left[
\begin{array}{ccc}
C_{1111} & C_{1122} & C_{1112} \\
\cdot & C_{2222} &  C_{2212} \\
\cdot & \cdot & C_{1212}  
\end{array}
\right]
\left[
\begin{array}{c}
\varepsilon_{11} \\
\varepsilon_{22} \\
2 \varepsilon_{12}
\end{array}
\right].
\end{equation}

The inverse relation of \eqref{eqn:stressstrainrelation} can be expressed as
\begin{equation}\label{eqn:strainstressrelation}
\varepsilon_{ij} = S_{ijkl} \sigma_{ij},
\end{equation}
where $S_{ijkl}$ are the components of the fourth-order compliance tensor $\mathbb{S}$. Similarly to $\mathbb{C}$, we may employ Voigt notation to express the fourth-order tensor $\mathbb{S}$ as a symmetric second-order tensor $\mathbf{S}$ \cite{Ting1996}. The strain-stress relation~\eqref{eqn:strainstressrelation} may be expressed in $\mathbb{R}^3$ as
\begin{equation}\label{eqn:compliancetensoreqn}
\left[
\begin{array}{c}
\varepsilon_{11} \\
\varepsilon_{22} \\
\varepsilon_{33} \\
2 \varepsilon_{23} \\
2 \varepsilon_{13} \\
2 \varepsilon_{12}
\end{array}
\right]
 = 
\left[
\begin{array}{cccccc}
S_{1111} & S_{1122} & S_{1133} & 2 S_{1123} & 2S_{1113} & 2S_{1112} \\
\cdot & S_{2222} & S_{2233} & 2 S_{2223} & 2S_{2213} & 2S_{2212} \\
\cdot & \cdot & S_{3333} & 2 S_{3323} & 2S_{3313} & 2S_{3312} \\
\cdot&\cdot &\cdot & 4S_{2323} & 4S_{2313} & 4S_{2312} \\
\cdot& \cdot& \cdot& \cdot& 4S_{1313} & 4S_{1312} \\
\cdot& \cdot& \cdot& \cdot& \cdot & 4S_{1212}
\end{array}
\right]
\left[
\begin{array}{c}
\sigma_{11} \\
\sigma_{22} \\
\sigma_{33} \\
\sigma_{23} \\
\sigma_{13} \\
\sigma_{12}
\end{array}
\right]
\end{equation}
and in $\mathbb{R}^2$ as
\begin{equation*}
\left[
\begin{array}{c}
\varepsilon_{11} \\
\varepsilon_{22} \\
2 \varepsilon_{12}
\end{array}
\right]
 = 
\left[
\begin{array}{cccccc}
S_{1111} & S_{1122} & 2S_{1112} \\
\cdot & S_{2222}  & 2S_{2212} \\
\cdot & \cdot & 4S_{1212} 
\end{array}
\right]
\left[
\begin{array}{c}
\sigma_{11} \\
\sigma_{22} \\
\sigma_{12}
\end{array}
\right].
\end{equation*}

In classical linear elasticity, the strain tensor $\boldsymbol{\varepsilon}$ is related to the displacement field $\mathbf{u}$ through the relation
\begin{equation}\label{eqn:straindisplacementref}
\varepsilon_{ij} = \frac{1}{2} \left(\frac{\partial u_i}{\partial x_j} + \frac{\partial u_j}{\partial x_i} \right).
\end{equation}

We are now able to provide the equation of motion in classical linear elasticity~\cite{timoshenko1934theory}. Given a body $\mathcal{B} \subset \mathbb{R}^d$, where $d$ is the dimension, the equation of motion for a material point $\bfx \in \mathcal{B}$ at time $t \geqslant 0$ is given by
\begin{equation}\label{eqn:eqnofmotionclassicalvectorform}
    \rho(\bfx) \ddot{\bfu}(\bfx,t) = \nabla \cdot \boldsymbol{\sigma}(\bfx,t) + \bfb(\bfx,t),
\end{equation}
where $\rho$ is the mass density, $\ddot{\bfu}$ is the second derivative in time of the displacement field $\bfu$, and $\mathbf{b}$ is a prescribed body force density field. In component form, \eqref{eqn:eqnofmotionclassicalvectorform} may be written as ({\em cf.}~\eqref{eqn:stressstrainrelation} and~\eqref{eqn:straindisplacementref})
\begin{equation}\label{eqn:eqnofmotionclassical}
\begin{split}
\rho(\bfx) \ddot{u}_i(\bfx,t) =& \frac{\partial \sigma_{ij}}{\partial x_j}(\bfx,t) + b_i(\bfx,t) = C_{ijkl} \frac{\partial \varepsilon_{kl}}{\partial x_j}(\bfx,t) + b_i(\bfx,t) \\
=& \frac{C_{ijkl}}{2} \left( \frac{\partial^2 u_k}{\partial x_j \partial x_l}(\bfx,t) + \frac{\partial^2 u_l}{\partial x_j \partial x_k}(\bfx,t) \right) + b_i(\bfx,t) \\
=& C_{ijkl} \frac{\partial^2 u_k}{\partial x_j \partial x_l}(\bfx,t) + b_i(\bfx,t),
\end{split}  
\end{equation}
where we used the minor symmetries $C_{ijlk} = C_{ijkl}$ ({\em cf.}~\eqref{eqn:minorsymmetries}).
\begin{remark}
For the sake of brevity, we often omit the arguments $\bfx$ and $t$.
\end{remark}

\subsection{Material symmetry classes in two-dimensional classical linear elasticity}\label{sec:ClassicalSymm}

Suppose we have a fourth-order tensor $\mathbb{A}$ and a second-order tensor $\mathbf{T}$ with components $A_{ijkl}$ and $T_{pq}$, respectively, relative the basis $\left\{\bfe_{i} \right\}_{i=1,\ldots,d}$, where $d$ is the dimension. Then, we define $\mathbb{A}[\mathbf{T}]$ to be the second-order tensor with components given by $A_{ijkl}T_{kl}$ (this is sometimes called a double contraction and denoted by $\mathbb{A}{:}\mathbf{T}$ \cite{holzapfelnonlinear}). With this formulation, the components of $\mathbb{C}$ relative to the basis $\left\{\bfe_{i} \right\}_{i=1,\ldots,d}$ are

\begin{equation}\label{eqn:genCrelativetobasis}
C_{ijkl} = \tr \left\{ (\bfe_i \otimes \bfe_j ) \mathbb{C} [ \bfe_k \otimes \bfe_l] \right\},
\end{equation}
where $\tr$ denotes the trace. A transformation between orthonormal bases of $\mathbb{R}^d$, $\left\{\bfe_i \right\}_{i=1,\ldots,d}$ and $\left\{\bfe_i' \right\}_{i=1,\ldots,d}$ may be represented by an orthogonal matrix $\bfQ$, where the components are given by 
\begin{equation}\label{eqn:Qdef}
Q_{ij} = \bfe_{i}' \cdot \bfe_j.
\end{equation}
We call $\mathbf{Q}$ a symmetry transformation of $\mathbb{C}$ when the components of $\mathbb{C}$ are invariant under the transformation. In Definition \ref{def:symmtran} we formalize this concept.

\begin{definition}\label{def:symmtran}
An orthogonal transformation $\bfQ$ between orthonormal bases $\left\{ \mathbf{e}_i \right\}_{i=1,\ldots,d} $ and $\left\{ \mathbf{e}_{i}'\right\}_{i=1,\ldots,d}$ ({\em cf.}~\eqref{eqn:Qdef}) is a symmetry transformation of $\mathbb{C}$ if
\begin{equation}\label{eqn:symmtran}
\tr \left\{ (\bfe_i \otimes \bfe_j ) \mathbb{C} [\bfe_k \otimes \bfe_l] \right\} = \tr \left\{ (\bfe_i' \otimes \bfe_j') \mathbb{C} [\bfe_k' \otimes \bfe_l'] \right\}.
\end{equation}
Equivalently, one may write
\begin{equation}\label{eqn:symmtranbasic}
C_{ijkl} = Q_{ip} Q_{jq} Q_{kr} Q_{ls} C_{pqrs}.
\end{equation}
\end{definition}

In view of \eqref{eqn:symmtranbasic}, one may show that if $\bfQ_1$ and $\bfQ_2$ are symmetry transformations of $\mathbb{C}$ then $\bfQ_1^{-1}$ (as well as $\bfQ_2^{-1}$) and $\bfQ_1 \bfQ_2$ are also symmetry transformations of $\mathbb{C}$ \cite{Ting1996}. Clearly, the identity transformation, $\bfI$, is also a symmetry transformation of $\mathbb{C}$. Therefore, the set of symmetry transformations of $\mathbb{C}$ forms a group (see e.g.~\cite{dummit2004abstract}), which we call the symmetry group of $\mathbb{C}$ and denote it by $\mathcal{G}_{\mathbb{C}}$. We call the set of symmetry groups that are equivalent up to a change in orientation, the symmetry class of $\mathbb{C}$. Given a material described by $\mathbb{C}$, its material symmetry class is the corresponding symmetry class of $\mathbb{C}$.

In two dimensions, it is well known that every orthogonal transformation is either a reflection or a rotation. %\cite{cartan1966}
For convenience, we recall the corresponding transformation matrices. For a (counterclockwise) rotation by an angle $\theta$ about the origin, the corresponding transformation matrix is given by
\begin{align}\label{eqn:RotMat}
&\mathbf{Rot}(\theta) = \left[ 
\begin{array}{cc}
\cos(\theta) & -\sin(\theta) \\
\sin(\theta) & \cos(\theta)
\end{array}
\right].
\end{align}

For a reflection about the line through the origin making an angle of $\theta$ with the $x$-axis, the corresponding transformation matrix is given by
\begin{align}\label{eqn:RefMat}
&\mathbf{Ref}(\theta) = \left[
\begin{array}{cc}
\cos(2\theta) & \sin(2\theta) \\
\sin(2\theta) & -\cos(2\theta)
\end{array}
\right]. 
\end{align}

From the periodicity of sine and cosine, it is clear that all rotations may be represented by a rotation of $\theta \in [0,2\pi)$, while all reflections may be represented by a reflection through a line making an angle of $\theta \in [0,\pi)$ with the $x$-axis. We recall a useful identity for later use:
\begin{equation}\label{eqn:identityrefreftorot}
\mathbf{Ref}(\phi) \mathbf{Ref}(\theta) = \mathbf{Rot}(2[\phi - \theta]).
\end{equation}
Since reflections are their own inverses, we additionally have
\begin{equation}\label{eqn:identityrefrotref}
 \mathbf{Rot}(2[\phi - \theta]) \mathbf{Ref}(\theta) = \mathbf{Ref}(\phi).
\end{equation}

It is well known that there are exactly eight symmetry classes of $\mathbb{C}$ in three-dimensional classical linear elasticity~\cite{CHADWICK2001,Forte1996,Ting1996}. In the two-dimensional case, there are exactly four symmetry classes of $\mathbb{C}$ \cite{He1996}. Our first result provides an alternate proof that there are exactly four symmetry classes of $\mathbb{C}$ in two dimensions. The proof utilizes elementary methods with minimal machinery from abstract algebra.

\begin{theorem}\label{thm:symmetryclasses}
Up to a change in orientation, there are exactly four symmetry groups of the elasticity tensor $\mathbb{C}$ in two dimensions: oblique, rectangular, square, and isotropic. The corresponding elasticity tensors and group generators for each symmetry group are given by:

\begin{tabular}{lcc}
{\bf Symmetry Group} & {\bf Elasticity Tensor} & {\bf Group Generators} \\[0.1in]
Oblique & $\left[ \begin{array}{ccc}
C_{1111} & C_{1122} & C_{1112} \\
\cdot & C_{2222} & C_{2212} \\
\cdot & \cdot & C_{1212}
\end{array} \right]$ & $\left\{-\bfI \right\}$ \\[0.3in]
Rectangular & $\left[ \begin{array}{ccc}
C_{1111} & C_{1122} & 0 \\
\cdot & C_{2222} & 0 \\
\cdot & \cdot & C_{1212}
\end{array} \right]$ & $\left\{-\bfI,\mathbf{Ref}(0) \right\}$ \\[0.3in] 
Square & $\left[ \begin{array}{ccc}
C_{1111} & C_{1122} & 0 \\
\cdot & C_{1111} & 0 \\
\cdot & \cdot & C_{1212}
\end{array} \right]$ & $\left\{-\bfI,\mathbf{Ref}(0), \mathbf{Ref} (\frac{\pi}{4}) \right\}$ \\[0.3in]
Isotropic & $\left[ \begin{array}{ccc}
C_{1111} & C_{1122} & 0 \\
\cdot & C_{1111} & 0 \\
\cdot & \cdot & \frac{C_{1111}-C_{1122}}{2}
\end{array} \right]$ & $\left\{\mathbf{Ref}(\theta): \theta \in [0,\pi) \right\}$
\end{tabular}

\end{theorem}

\begin{proof}
Suppose $\mathcal{G}_{\mathbb{C}}$ is the symmetry group of $\mathbb{C}$. We first show that, up to a change in orientation, $\mathcal{G}_{\mathbb{C}}$ is one of the four symmetry groups from Lemma~\ref{lemma:reflections} below. By \eqref{eqn:symmtranbasic}, it is clear that $-\bfI \in \mathcal{G}_{\mathbb{C}}$. Consequently, at least one of the four symmetry groups described in Lemma \ref{lemma:reflections} is a subgroup of $\mathcal{G}_{\mathbb{C}}$. Let $\mathcal{N}$ be the most symmetric (largest quantity of reflection transformations) such subgroup of $\mathcal{G}_{\mathbb{C}}$. By definition, $\mathcal{N} \subseteq \mathcal{G}_{\mathbb{C}}$. We then show that $\mathcal{G}_{\mathbb{C}} \subseteq \mathcal{N}$, which implies $\mathcal{G}_{\mathbb{C}} = \mathcal{N}$. By Lemma \ref{lemma:reflections}, up to a change in orientation, every reflection in $\mathcal{G}_{\mathbb{C}}$ is contained in~$\mathcal{N}$. Furthermore, by Lemma \ref{lemma:rotations}, we know every rotation in $\mathcal{G}_{\mathbb{C}}$ is contained in $\mathcal{N}$. Since orthogonal transformations in two dimensions are either rotations or reflections, every element of $\mathcal{G}_{\mathbb{C}}$ is contained in $\mathcal{N}$, i.e. $\mathcal{G}_{\mathbb{C}} \subseteq \mathcal{N}$.  \qed 
\end{proof}

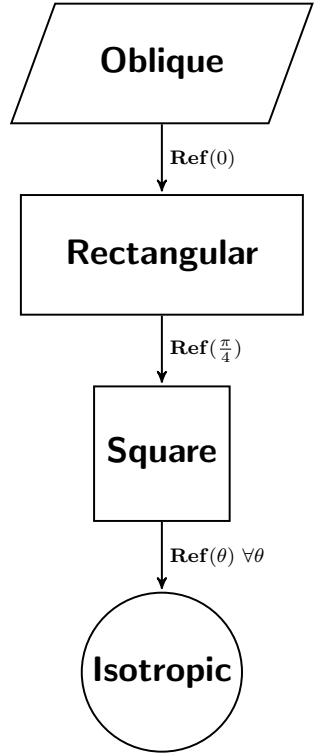
\begin{figure}
\begin{center}
\begin{tabular}{cc}
\resizebox{0.25\textwidth}{!}{
\begin{tikzpicture}[->,>=stealth',shorten >=1pt,auto,node distance=3cm,thick,main node/.style={circle,draw,font=\sffamily\Large\bfseries}, square/.style={regular polygon,regular polygon sides=4},oblique node/.style={font=\sffamily\Large\bfseries},use Hobby shortcut]
  \node[main node, trapezium,trapezium left angle=70,trapezium right angle=-70, inner sep = 18] (obl) {Oblique};
  \node[main node, rectangle, inner sep = 18] (rec) [below=1cm of obl] {Rectangular};
  \node[main node, square, inner sep = -3] (sqr) [below=1cm of rec] {Square};
  \node[main node] (iso) [below=1cm of sqr] {Isotropic};

  \path[every node/.style={font=\sffamily\small}]
    (obl) edge node {$\mathbf{Ref} (0)$} (rec)
    (rec) edge node {$\mathbf{Ref} ( \frac{\pi}{4} )$} (sqr)
    (sqr) edge node {$\mathbf{Ref}(\theta)$ $\forall \theta$} (iso);
\end{tikzpicture}}
\end{tabular}
\caption{The four symmetry groups (up to a change in orientation) of the elasticity tensor in two dimensions.}
\end{center}
\end{figure}

\begin{lemma}\label{lemma:reflections}
Up to a change in orientation, there are only four symmetry groups of the elasticity tensor $\mathbb{C}$ generated by reflections and $-\boldsymbol{I}$: \\[0.2in]
\begin{tabular}{llc}
{\bf Symmetry  Group} & {\bf Lines of Reflection Symmetry} & {\bf Group Generators} \\[0.1in]
Oblique & No lines of reflection symmetry & $ \left\{ -\bfI \right\}$ \\[0.1in]
Rectangular & Two lines of reflection symmetry & $\left\{-\bfI,\mathbf{Ref}(0) \right\}$ \\[0.1in]
Square & Four lines of reflection symmetry & $\left\{-\bfI,\mathbf{Ref}(0), \mathbf{Ref}\left(\frac{\pi}{4} \right) \right\}$ \\[0.1in]
Isotropic & All lines of reflection symmetry & $ \left\{\mathbf{Ref}(\theta) : \theta \in [0,\pi) \right\}$
\end{tabular}

\end{lemma}

\begin{proof} 

\textbf{No lines of reflection symmetry} \\
By \eqref{eqn:symmtranbasic}, the group generated by 
\begin{equation}\label{eqn:genoblique}
\left\{-\bfI \right\}
\end{equation} 
is a subgroup of any symmetry group of $\mathbb{C}$. This group, $\left\{\bfI, -\bfI \right\}$, imposes no restrictions on the elasticity tensor, i.e., \eqref{eqn:symmtranbasic} holds under any transformation in the group. We call the symmetry group generated by \eqref{eqn:genoblique} the \textit{oblique symmetry group} and the corresponding elasticity tensor is given by~({\em cf.}~\eqref{eqn:2DElastTensor})
\begin{equation}\label{eqn:tritensor}
\bfC = \left[ \begin{array}{ccc}
C_{1111} & C_{1122} & C_{1112} \\
\cdot & C_{2222} & C_{2212} \\
\cdot & \cdot & C_{1212}
\end{array} \right].
\end{equation}

\textbf{One line of reflection symmetry} \\
We now consider the implications of adding a reflection transformation to the subgroup $\left\{\bfI, -\bfI \right\}$. %Suppose the normal to the line of reflection makes an angle of $\phi$ relative to the positive $x$-axis. 
Without loss of generality, we choose an orthonormal basis $\left\{ \mathbf{e}_{i} \right\}_{i=1,2}$ so that the line of reflection symmetry coincides with the $x$-axis.  In this basis, we may write the corresponding reflection transformation as ({\em cf.}~\eqref{eqn:RefMat})
\begin{equation}\label{eqn:Ref0Mat}
\mathbf{Ref}(0) = \left[
\begin{array}{cc}
1 & 0 \\
0 & -1
\end{array}
\right].
\end{equation}
We are interested in the restrictions the group generated by 
\begin{equation}\label{eqn:generatorsforrectangular}
\left\{-\bfI, \mathbf{Ref}(0)\right\}
\end{equation} 
imposes on the elasticity tensor $\mathbb{C}$. 

In order for~\eqref{eqn:symmtranbasic} to be satisfied for $\mathbf{Q} = \mathbf{Ref}(0)$, any $C_{ijkl}$ where $2$ appears an odd number of times in the indices requires $C_{ijkl} = - C_{ijkl}$ and thus $C_{ijkl} = 0$. Consequently,
\begin{equation}\label{eqn:recelastrelations}
C_{1112} = C_{2212} = 0.
\end{equation}
The transformation in \eqref{eqn:generatorsforrectangular} impose no additional restrictions on $\mathbb{C}$. We call the symmetry group generated by \eqref{eqn:generatorsforrectangular} the \textit{rectangular symmetry group} and the corresponding elasticity tensor is given by
\begin{equation}\label{eqn:recttensor}
\bfC = \left[ \begin{array}{ccc}
C_{1111} & C_{1122} & 0 \\
\cdot & C_{2222} & 0 \\
\cdot & \cdot & C_{1212}
\end{array} \right].
\end{equation}

Notice $\mathbf{Rot}(\pi)$, a rotation by $\pi$, is equivalent to $-\bfI$ ({\em cf.}~\eqref{eqn:RotMat}). By \eqref{eqn:identityrefrotref}, we have
\begin{equation}\label{eqn:rotreftoref}
\mathbf{Rot}(\pi) \mathbf{Ref}(\theta) = \mathbf{Ref} \left( \theta + \frac{\pi}{2} \right).
\end{equation}
Since $-\bfI$ is in every symmetry group of $\mathbb{C}$, a line of reflection symmetry automatically induces a second line of reflection symmetry orthogonal to the first line of reflection symmetry by \eqref{eqn:rotreftoref}. In particular, this implies implies $\mathbf{Ref}\left( \frac{\pi}{2} \right)$ belongs to the rectangular symmetry group since the reflection transformation $\mathbf{Ref}(\bf0)$ belongs to the group. Thus, to consider a symmetry group distinct from the rectangular symmetry group, we next consider a group containing two non-orthogonal lines of reflection symmetry.

\textbf{Two non-orthogonal lines of reflection symmetry and isotropy} \\
Consider a group generated by $-\bfI$ and two non-orthogonal lines of reflection symmetry. Suppose the angle between the two lines of reflection symmetry is given by $\theta$. We further require $\theta \neq \frac{k \pi}{2}$ for $k \in \mathbb{Z}$, so that the lines of reflection symmetry are distinct and not orthogonal to each other. We may choose an orthonormal basis $\left\{ \bfe_i \right\}_{i=1,2}$ and the rotated basis $\left\{\bfe'_i \right\}_{i=1,2}$ such that
\begin{equation*}
\begin{split}
%&\bfe_1' = \sin\left(\phi + \frac{\pi}{2} \right) \bfe_1 + \cos\left(\theta+\frac{\pi}{2} \right) \bfe_2 = \cos(\theta) \bfe_1 - \sin(\theta) \bfe_2 \\
&\bfe_1' = \cos(\theta) \bfe_1 + \sin(\theta) \bfe_2, \\
&\bfe_2' = -\sin(\theta) \bfe_1 + \cos(\theta) \bfe_2, 
\end{split}
\end{equation*}

\begin{figure}
\begin{center}
\begin{tikzpicture}[scale=0.8]
\draw[dashed,black] (-5,0)--(5,0) node[above]{First line of reflection};
\draw[dashed,black] (-4.3301270,-5/2)--(4.3301270,5/2) node[above]{Second line of reflection};

\draw[->,thick] (-4,0)--(4,0) node[below]{$x$};
\draw[->,thick] (0,-4)--(0,4) node[left]{$y$};

\draw[->,ultra thick,blue] (0,0)--(2,0) node[below]{$\bfe_1$};
\draw[->,ultra thick,blue] (0,0)--(0,2) node[right]{$\bfe_2$};

\draw[->,ultra thick,red] (0,0)--(-3/2*2/3,2.598076*2/3) node[right]{$\, \bfe_2'$};
\draw[->,ultra thick,red] (0,0)--(2.598076*2/3,3/2*2/3) node[below]{$\bfe_1'$};

\node[color=black] at (0.93,0.27) {$\theta$};

\begin{scope}    
\path[clip] (0,0) -- (3,0) -- (2.598076,3/2);
\node[circle,draw=black,thick, minimum size=35pt] at (0,0) (circ) {};
\end{scope}

\end{tikzpicture}

\caption{Basis $\left\{\bfe_i \right\}_{i=1,2}$ and rotated basis $\left\{\bfe_i' \right\}_{i=1,2}$ for two non-orthogonal lines of reflection symmetry.} \label{fig:RotatedBasis}
\end{center}
\end{figure}
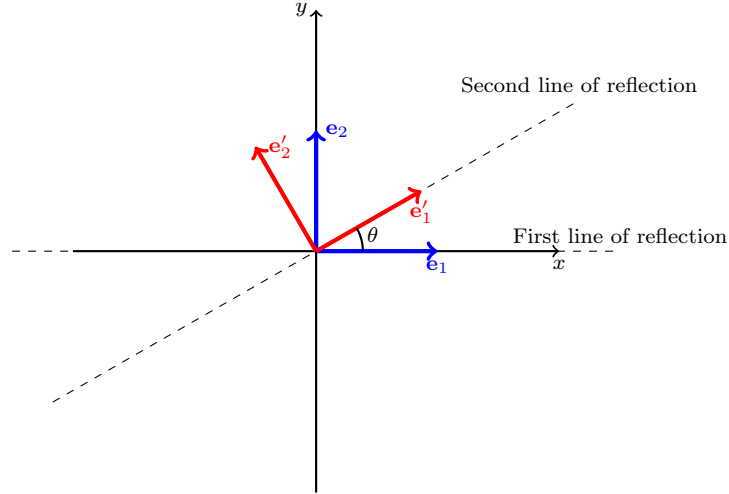

where $\bfe_1$ and $\bfe_1'$ coincide with the two lines of reflection symmetry, see Figure \ref{fig:RotatedBasis}. In the basis $\left\{\bfe_i \right\}_{i=1,2}$, the group is generated by
\begin{equation}\label{eqn:genrec}
\left\{-\bfI, \mathbf{Ref}(0), \mathbf{Ref}(\theta)\right\}.
\end{equation}

From the argument for one line of reflection, we know \eqref{eqn:recelastrelations} holds in both bases ({\em cf.} \cite{CHADWICK2001,Ting1996} for similar arguments in three dimensions). Consequently, by~\eqref{eqn:symmtran} we find

\begin{subequations}\label{eqn:twoplanerelations}
\begin{align}
\begin{split}
0 = C_{1112}' =& \tr \left\{ (\bfe_1' \otimes \bfe_1') \mathbb{C} [ \bfe_1' \otimes \bfe_2'] \right\} = \tr \left\{ (\bfe_1 \otimes \bfe_1) \mathbb{C} [ \bfe_1 \otimes \bfe_2] \right\} \\
=& \cos(\theta) \sin(\theta)[\cos^2(\theta) (C_{1111}-C_{1122} - 2 C_{1212} ) + \sin^2(\theta) ( C_{1122} - C_{2222} + 2C_{1212} ) ]
\end{split} \\
\begin{split}
0 = C_{2212}' =& \tr \left\{ (\bfe_2' \otimes \bfe_2') \mathbb{C} [ \bfe_1' \otimes \bfe_2'] \right\} = \tr \left\{ (\bfe_2 \otimes \bfe_2) \mathbb{C} [ \bfe_1 \otimes \bfe_2] \right\} \\
=& \cos(\theta) \sin(\theta)[ \sin^2(\theta) (C_{1111} - C_{1122} - 2 C_{1212}) + \cos^2(\theta) (C_{1122} - C_{2222}+ 2C_{1212} )].
\end{split}
\end{align}
\end{subequations}
We now consider all solutions of system \eqref{eqn:twoplanerelations}. First, recall $\sin(\theta)$ and $\cos(\theta)$ are nonzero since $\theta \neq \frac{k \pi}{2}$ with $k \in \mathbb{Z}$. Dividing both equations in \eqref{eqn:twoplanerelations} by $\cos(\theta) \sin(\theta)$, we are left with 
\begin{subequations}\label{eqn:fourrefrelation}
\begin{align}
&0 = a \cos^2(\theta) + b \sin^2(\theta), \label{eqn:fourrefrelationa} \\
&0 = a \sin^2(\theta) + b \cos^2(\theta), \label{eqn:fourrefrelationb}
\end{align}
\end{subequations}
where $a := C_{1111}-C_{1122} - 2 C_{1212}$ and $b := C_{1122} - C_{2222} + 2C_{1212}$. Summing~\eqref{eqn:fourrefrelationa} and~\eqref{eqn:fourrefrelationb}, and applying the Pythagorean identity, $\sin^2(\theta) + \cos^2(\theta)=~1$, we deduce $a+b = 0$ and therefore $a = -b$ or $a = b = 0$. 

If $a = -b \neq 0$, then by \eqref{eqn:fourrefrelation} we have $\cos^2(\theta) = \sin^2(\theta)$, which implies $\theta = \frac{\pi}{4}$ or $\frac{3 \pi}{4}$. Moreover, by the definitions of $a$ and $b$ we also have $C_{1111} = C_{2222}$. Recall from \eqref{eqn:rotreftoref} that a line of reflection symmetry induces a second line of reflection symmetry perpendicular to the first. Thus, the choice of $\theta = \frac{\pi}{4}$ or $\frac{3 \pi}{4}$ is irrelevant as either transformation generates the other (note $\mathbf{Ref} \left( \frac{5 \pi}{4} \right)$ is equivalent to $\mathbf{Ref} \left( \frac{\pi}{4} \right)$ by \eqref{eqn:RefMat}). The transformation $\mathbf{Ref}\left(\frac{\pi}{4} \right)$ imposes no additional restrictions on $\mathbb{C}$.  We call the symmetry group generated by 
\begin{equation}\label{eqn:gensqr}
\left\{-\bfI,\mathbf{Ref}(0), \mathbf{Ref}\left(\frac{\pi}{4} \right) \right\}
\end{equation}
the \textit{square symmetry group} and the corresponding elasticity tensor has the restrictions
\begin{equation}\label{eqn:sqrelastrelations}
C_{1112} = C_{2212} = 0 \text{ and } C_{1111} = C_{2222}.
\end{equation}
The elasticity tensor corresponding to the square symmetry group is given by
\begin{equation}\label{eqn:sqrtensor}
\bfC = \left[ \begin{array}{ccc}
C_{1111} & C_{1122} & 0 \\
\cdot & C_{1111} & 0 \\
\cdot & \cdot & C_{1212}
\end{array} \right].
\end{equation}

Alternatively, if $a=b=0$, we have from $a=0$, $C_{1212} = \frac{C_{1111}-C_{1122}}{2}$, and substituting this into $b=0$, we obtain $C_{1111} = C_{2222}$. The corresponding elasticity tensor has the restrictions
\begin{equation}\label{eqn:isoelastrelations}
C_{1112} = C_{2212} = 0, C_{1111} = C_{2222}, \text{ and } C_{1212} = \frac{C_{1111}-C_{1122}}{2},
\end{equation}
which produces the elasticity tensor
\begin{equation}\label{eqn:isotensor}
\bfC = \left[ \begin{array}{ccc}
C_{1111} & C_{1122} & 0 \\
\cdot & C_{1111} & 0 \\
\cdot & \cdot & \frac{C_{1111}-C_{1122}}{2}
\end{array} \right].
\end{equation}
Moreover, given \eqref{eqn:isotensor} one may show \eqref{eqn:symmtranbasic} holds for $\bfQ = \mathbf{Ref}(\theta)$ with any choice of $\theta$. Thus, this elasticity tensor remains invariant under any choice of line of reflection. We call this group the \textit{isotropic symmetry group}. The generators are given by
\begin{equation}\label{eqn:genisoref}
\left\{\mathbf{Ref}(\theta) : \theta \in [0,\pi) \right\}.
\end{equation}
The last piece to tie up the proof is to consider adding an additional line of reflection to the square symmetry group. In this case, we would have two lines of reflection which do not intersect at an angle of $\frac{\pi}{4}$ or $\frac{3 \pi}{4}$, and consequently~\eqref{eqn:fourrefrelation} would yield only the trivial solution $a=b=0$. In this case, we immediately obtain \eqref{eqn:genisoref} as a set of generators for the group. \qed
\end{proof}

The next lemma shows that adding rotations to the symmetry groups in Lemma \ref{lemma:reflections} does not produce new symmetry groups and therefore those symmetry groups actually describe all symmetry groups of $\mathbb{C}$.

\begin{lemma}\label{lemma:rotations}

The set of symmetry groups described in Lemma~\ref{lemma:reflections} is closed under the introduction of rotation symmetry transformations of $\mathbb{C}$. More specifically, if $\mathbb{C}$ is invariant under a rotation transformation $\mathbf{Q}$ and the members of a symmetry group $\mathcal{G}$ from Lemma~\ref{lemma:reflections}, then the symmetry group generated by $\left\{ \mathbf{Q} \right\} \cup \mathcal{G}$ is one of the four symmetry groups from Lemma~\ref{lemma:reflections}.
\end{lemma}

\begin{proof}
Let $\mathbb{C}$ be invariant under a rotation transformation $\mathbf{Q}$ and the members of a symmetry group $\mathcal{G}$ from Lemma \ref{lemma:reflections}. If $\mathbf{Q} \in \mathcal{G}$ then $\left\{ \mathbf{Q} \right\} \cup \mathcal{G} = \mathcal{G}$ and the conclusion of the lemma is immediate. Consequently, we suppose $\mathbf{Q} \notin \mathcal{G}$. We discuss each symmetry group $\mathcal{G}$: isotropic, square, rectangular, and oblique.

\textbf{Isotropic}: In this case $\mathcal{G}$ is generated by~\eqref{eqn:genisoref} and consequently contains all rotation transformations. This is easily seen through the closure property of groups and relation~\eqref{eqn:identityrefreftorot}. Ergo, $\mathbf{Q} \in \mathcal{G}$, which contradicts the assumption $\mathbf{Q} \notin \mathcal{G}$.

\textbf{Square}: In this case $\mathcal{G}$ is generated by~\eqref{eqn:gensqr}. Recall the composition of a reflection and a rotation is a reflection in two dimensions ({\em cf.}~\eqref{eqn:identityrefrotref}). Since $\mathbf{Q} \notin \mathcal{G}$, the closure property of groups implies the group generated by $\left\{\mathbf{Q} \right\} \cup \mathcal{G}$ contains a reflection transformation not in $\mathcal{G}$. Consequently, by Lemma~\ref{lemma:reflections}, the symmetry group generated by $\left\{ \mathbf{Q} \right\} \cup \mathcal{G}$ is the isotropic symmetry group. 

\textbf{Rectangular}: In this case $\mathcal{G}$ is generated by \eqref{eqn:genrec}. As in the square case, since $\mathbf{Q} \notin \mathcal{G}$, the group generated by $\left\{\mathbf{Q} \right\} \cup \mathcal{G}$ contains a reflection transformation not in $\mathcal{G}$. Therefore, by Lemma~\ref{lemma:reflections}, the isotropic symmetry group or square symmetry group is a subgroup of the group generated by $\left\{ \mathbf{Q} \right\} \cup \mathcal{G}$. Furthermore, the square and isotropic cases above guarantee the group generated by $\left\{ \mathbf{Q} \right\} \cup \mathcal{G}$ is either the square symmetry group or the isotropic symmetry group.

%%%%%%%%%%%%%%%%%%%%%%%%%%%%%%%%%%%%%%%%%%%%%%%%%%%%%%%%%%%%%%

\textbf{Oblique}: In this case $\mathcal{G}$ is generated by \eqref{eqn:genoblique}. We further suppose the rotation transformation $\mathbf{Q}$ corresponds to a counterclockwise rotation by $\theta \neq k \pi$ for $k \in \mathbb{Z}$ since rotations by $k \pi$ are already in the group generated by \eqref{eqn:genoblique} ({\em cf.}~\eqref{eqn:RotMat}). Since $\mathbb{C}$ is invariant under the rotation transformation $\mathbf{Q}$, we have by \eqref{eqn:symmtranbasic} that
\begin{subequations}\label{eqn:CijklRelationsObl}
\begin{align}
\begin{split}
C_{1111} ={}& C_{1111}\cos^4(\theta) -4C_{1112} \sin(\theta) \cos^3(\theta) + 2 C_{1122} \cos^2(\theta) \sin^2(\theta) + 4C_{1212} \cos^2(\theta) \sin^2(\theta) \\
&-4C_{2212} \sin^3(\theta) \cos(\theta) +  C_{2222}\sin^4(\theta),  
\end{split} \label{eqn:CijklRelationsObla} \\
%%%%%%%%%%%%%%%%%%%%%%%%%%%%%%%%%%%%%%%%%%%%%%%%%%%%%%%%%%%%%%
\begin{split}
C_{1112} ={}& C_{1111}\sin(\theta)\cos^3(\theta) + C_{1112}(4 \sin^4(\theta)-5\sin^2(\theta)+1) \\
&+C_{1122} (\sin^3(\theta)\cos(\theta) - \sin(\theta)\cos^3(\theta)) +2 C_{1212}(\sin^3(\theta)\cos(\theta) - \sin(\theta)\cos^3(\theta)) \\
&+ C_{2212}(3 \sin^2(\theta)-4 \sin^4(\theta)) - C_{2222} \sin^3(\theta), \cos(\theta) 
\end{split} \label{eqn:CijklRelationsOblb} \\
%%%%%%%%%%%%%%%%%%%%%%%%%%%%%%%%%%%%%%%%%%%%%%%%%%%%%%%%%%%%%%
\begin{split}
C_{1122} ={}& C_{1111} \cos^2(\theta) \sin^2(\theta) + 2C_{1112}(\sin(\theta) \cos^3(\theta) - \sin^3(\theta) \cos(\theta)) \\
&+ C_{1122} (\sin^4(\theta) + \cos^4(\theta)) - 4 C_{1212} \cos^2(\theta) \sin^2(\theta) \\
&+ 2C_{2212}(\sin^3(\theta)\cos(\theta)-\sin(\theta)\cos^3(\theta)) + C_{2222} \cos^2(\theta) \sin^2(\theta),  
\end{split} \\
%%%%%%%%%%%%%%%%%%%%%%%%%%%%%%%%%%%%%%%%%%%%%%%%%%%%%%%%%%%%%%
\begin{split}
C_{1212} ={}& C_{1111} \cos^2(\theta) \sin^2(\theta) +2C_{1112}(\sin(\theta)\cos^3(\theta)-\sin^3(\theta)\cos(\theta)) \\
&-2C_{1122} \cos^2(\theta) \sin^2(\theta) +C_{1212}(\sin^2(\theta)-\cos^2(\theta))^2 \\
&+2C_{2212}(\sin^3(\theta)\cos(\theta)-\sin(\theta)\cos^3(\theta))+C_{2222}\cos^2(\theta) \sin^2(\theta), 
\end{split} \\
%%%%%%%%%%%%%%%%%%%%%%%%%%%%%%%%%%%%%%%%%%%%%%%%%%%%%%%%%%%%%%
\begin{split}
C_{2212} ={}& C_{1111} \sin^3(\theta) \cos(\theta) + C_{1112}(3 \sin^2(\theta)-4\sin^4(\theta)) \\
&+ C_{1122}(\sin(\theta)\cos^3(\theta)-\sin^3(\theta)\cos(\theta))+2C_{1212}(\sin(\theta)\cos^3(\theta)-\sin^3(\theta)\cos(\theta))\\
&+C_{2212}(4\sin^4(\theta)-5\sin^2(\theta)+1)-C_{2222} \sin(\theta)\cos^3(\theta), 
\end{split} \\
%%%%%%%%%%%%%%%%%%%%%%%%%%%%%%%%%%%%%%%%%%%%%%%%%%%%%%%%%%%%%%
\begin{split}
C_{2222} ={}& C_{1111} \sin^4(\theta) +4C_{1112} \sin^3(\theta)\cos(\theta)  + 2C_{1122} \cos^2(\theta)\sin^2(\theta) + 4C_{1212} \cos^2(\theta) \sin^2(\theta) \\
&+4C_{2212} \sin(\theta)\cos^3(\theta) + C_{2222} \cos^4(\theta). \label{eqn:CijklRelationsOblf}
\end{split}
\end{align}
\end{subequations}

Take the difference of the equations for $C_{1111}$ and $C_{2222}$ in \eqref{eqn:CijklRelationsObl} and then simplify to find
\begin{equation}\label{eqn:c1111-c2222}
0 =  (C_{1111} - C_{2222}) \sin^2(\theta) + 2(C_{2212}+C_{1112})\sin(\theta)\cos(\theta). 
\end{equation}
Similarly, sum the equations for $C_{1112}$ and $C_{2212}$ and then simplify to find
\begin{equation}\label{eqn:c1112+c2212}
0 = (C_{1111}-C_{2222})\sin(\theta)\cos(\theta) - 2(C_{1112}+C_{2212})\sin^2(\theta).
\end{equation}
Since $\theta \neq k \pi$, we know $\sin(\theta) \neq 0$ and we may divide \eqref{eqn:c1111-c2222} and \eqref{eqn:c1112+c2212} by $\sin(\theta)$ to obtain
\begin{subequations}\label{eqn:simplifyc1111-c2222andC2212+C1112}
\begin{align}
&0 = (C_{1111}-C_{2222}) \sin(\theta) + 2(C_{1112}+C_{2212}) \cos(\theta), \label{eqn:simplifyc1111-c2222andC2212+C1112a} \\
&0 = (C_{1111}-C_{2222}) \cos(\theta) - 2(C_{1112}+C_{2212}) \sin(\theta). \label{eqn:simplifyc1111-c2222andC2212+C1112b} 
\end{align}
\end{subequations}
Multiplying \eqref{eqn:simplifyc1111-c2222andC2212+C1112a} by $\cos(\theta)$ and \eqref{eqn:simplifyc1111-c2222andC2212+C1112b} by $\sin(\theta)$, and taking the difference of the two equations, yields
\begin{equation}\label{eqn:c1112=c2212symmproof}
0 =  2(C_{1112}+C_{2212}) \Rightarrow C_{1112} = - C_{2212}.
\end{equation}
Taking the difference of \eqref{eqn:CijklRelationsObla} and \eqref{eqn:CijklRelationsOblf}, imposing \eqref{eqn:c1112=c2212symmproof}, and then simplifying produces
\begin{equation}\label{eqn:c1111=c2222symmproof}
(C_{1111} - C_{2222}) \sin^2(\theta) = 0 \Rightarrow C_{1111} = C_{2222}.
\end{equation}
The implication in \eqref{eqn:c1111=c2222symmproof} follows by recalling $\theta \neq k \pi$ for $k \in \mathbb{Z}$.  

Substituting \eqref{eqn:c1112=c2212symmproof} and \eqref{eqn:c1111=c2222symmproof} into \eqref{eqn:CijklRelationsObla} and \eqref{eqn:CijklRelationsOblb} and simplifying results in
\begin{subequations}\label{eqn:oblsimplifiedC1111C1112}
\begin{align}
&0 =  (C_{1111}-C_{1122}-2C_{1212})  \sin^2(2 \theta)+4C_{1112}\cos(2\theta)\sin(2\theta),  \label{eqn:oblsimplifiedC1111C1112a} \\
&0 =  -4 C_{1112}\sin^2(2\theta) + (C_{1111}-C_{1122}-2C_{1212})\cos(2\theta)\sin(2\theta). \label{eqn:oblsimplifiedC1111C1112b}
\end{align}
\end{subequations}
We will consider two cases: $\sin(2\theta) \neq 0$ and $\sin(2\theta) = 0$.

\textbf{Case 1}: Let us consider $\sin(2 \theta) \neq 0$.  Multiply~\eqref{eqn:oblsimplifiedC1111C1112a} by $\cot(2\theta)$ and then subtract the result by \eqref{eqn:oblsimplifiedC1111C1112b} to find ({\em cf.} \eqref{eqn:c1111=c2222symmproof})
\begin{equation}\label{eqn:C1112=C2212=0symm}
0 = 4C_{1112} \left( \sin^2(2\theta) + \cos^2(2\theta) \right) \Rightarrow 0 = C_{1112} = - C_{2212}.
\end{equation}
Imposing \eqref{eqn:C1112=C2212=0symm} on \eqref{eqn:oblsimplifiedC1111C1112a} results in
\begin{equation*}
0 = (C_{1111}-C_{1122}-2C_{1212})\sin^2(2 \theta).
\end{equation*}
Since $\sin(2\theta) \neq 0$, we are able to conclude 
\begin{equation}\label{eqn:c1212=(c1111-c1122)/2symm}
C_{1111}-C_{1122}-2C_{1212} = 0 \Rightarrow C_{1212} = \frac{C_{1111}-C_{1122}}{2}.
\end{equation}
Imposing \eqref{eqn:c1111=c2222symmproof}, \eqref{eqn:C1112=C2212=0symm}, and \eqref{eqn:c1212=(c1111-c1122)/2symm} on  the elasticity tensor \eqref{eqn:tritensor}, we obtain \eqref{eqn:isotensor} and thus the group generated by $\mathcal{G} \cup \left\{ \bfQ \right\}$ is the isotropic group.

\textbf{Case 2}: We suppose $\sin(2 \theta) = 0$, i.e., $\theta = \frac{(2k+1) \pi}{2}$ (recall $\theta \neq k \pi$ by assumption) for some $k \in \mathbb{Z}$. We further suppose $C_{1112} \neq 0$ as otherwise $\eqref{eqn:tritensor}$ reduces to the square tensor ({\em cf.} \eqref{eqn:c1112=c2212symmproof}) and we have already treated the case where a rotation was added to the square group. Note that 
\begin{equation*}
f(x): = 2(\cot(2x)-\tan(2 x) ) = 2 \left(\frac{\cos(2x)}{\sin(2x)} - \frac{\sin(2x)}{\cos(2x)} \right)
\end{equation*}
is continuous on $\left( 0, \frac{\pi}{4}\right)$. Moreover, 
\begin{equation*}
\lim_{x \rightarrow 0} f(x) = -\infty \quad \text{and} \quad \lim_{x \rightarrow \frac{\pi}{4}} f(x) = \infty,
\end{equation*}
and thus $f(x)$ has range $(-\infty, \infty)$ on the domain $\left(0, \frac{\pi}{4} \right)$ by the Intermediate Value Theorem. Consequently, we may find an $\alpha \in \left( 0, \frac{\pi}{4} \right)$ such that
\begin{equation}\label{assump:thetaref}
\frac{C_{1111} - C_{1122} - 2 C_{1212}}{C_{1112}} = 2\left( \cot(2\alpha) - \tan(2\alpha) \right).
\end{equation}
It turns out that $\mathbb{C}$ is invariant with respect to reflection transformation $\mathbf{Ref}(\alpha)$ ({\em cf.}~\eqref{eqn:RefMat}).
Thus, the group generated by $\mathcal{G} \cup \left\{ \bfQ \right\}$ contains the reflection transformation $\mathbf{Ref}(\alpha)$. Since $\mathbf{Rot}\left( \frac{\pi}{2} \right)$ is guaranteed to be in $\mathcal{G}$ because $\mathbf{Rot}\left( \frac{(2k+1)\pi}{2} \right)$ and $\mathbf{Rot}\left( \pi \right)$ are, using \eqref{eqn:identityrefreftorot} we see that
\begin{equation}\label{eqn:rotpi2isinsquaregroup}
\mathbf{Ref}\left( \alpha + \frac{\pi}{4} \right) = \mathbf{Rot} \left( \frac{\pi}{2} \right) \mathbf{Ref}( \alpha) 
\end{equation}
is also in the group generated by $\mathcal{G} \cup \left\{ \bfQ \right\}$. Thus, the group generated by $\mathcal{K} := \left\{ -\bfI, \mathbf{Ref}(\alpha), \mathbf{Ref}\left( \alpha + \frac{\pi}{4} \right) \right\}$ is a subgroup of the group generated by $\mathcal{G} \cup \left\{ \bfQ \right\}$. On the other hand, by noticing $\mathbf{Ref}( \alpha)$ is its own inverse, from \eqref{eqn:rotpi2isinsquaregroup} we deduce $\mathbf{Q}$ is in the group generated by $\mathcal{K}$. Consequently, the groups generated by $\mathcal{K}$ and $\mathcal{G} \cup \left\{ \bfQ \right\}$ are the same. By considering a change in orientation, specifically a clockwise rotation by $\alpha$, we see that the group generated by $\mathcal{K}$ is in the square symmetry class.
\qed

\end{proof}

\begin{remark}
Notice that adding a rotation to the oblique symmetry group produces either the isotropic or the square symmetry group, but not the rectangular symmetry group. In Table \ref{table:symtransin2d}, we summarize the symmetry transformations for each symmetry group. Unlike in the three-dimensional case, we cannot distinguish between the symmetry groups of classical linear elasticity solely based on the rotations contained in the group.
\end{remark}

\begin{table}[h!]
\begin{center}
\begin{tabular}{|l|c|c|l|}
\hline
Symmetry Group & Reflections & Rotations & \multicolumn{1}{c|}{$\mathbb{C}$ Restrictions } \\
\hline
Oblique & None  &  $\mathbf{Rot}( \pi )$ & \multicolumn{1}{c|}{None} \\
\hline
Rectangular  & $\mathbf{Ref}(0), \mathbf{Ref} ( \frac{\pi}{2})$ & $\mathbf{Rot}( \pi )$ & $C_{1112} = C_{2212} = 0$ \\
\hline
\multirow{2}{*}{Square} & \multirow{2}{*}{$\mathbf{Ref}(0), \mathbf{Ref}( \frac{\pi}{4} ), \mathbf{Ref} ( \frac{\pi}{2} ), \mathbf{Ref} ( \frac{3 \pi}{4} )$} & \multirow{2}{*}{$\mathbf{Rot}( \pi ), \mathbf{Rot}( \frac{\pi}{2} ), \mathbf{Rot} ( \frac{3\pi}{2} )$} & $C_{1112} = C_{2212} = 0$ \\
& & & $C_{1111} = C_{2222}$ \\
\hline
 &  &  & $C_{1112} = C_{2212} = 0$ \\
Isotropic & $\mathbf{Ref} ( \theta), \theta \in [0,\pi)$ & $\mathbf{Rot} ( \theta), \theta \in [0,2 \pi)$ & $C_{1111} = C_{2222}$ \\
& & & $C_{1212} = \frac{C_{1111} - C_{1122}}{2}$ \\
\hline
\end{tabular}
\caption{Symmetry transformations by symmetry class in two dimensions.}
\label{table:symtransin2d}
\end{center}
\end{table}

\subsection{Pure two-dimensional classical linear elasticity}\label{sec:puretwodimclass}
In this section, we review pure two-dimensional models in classical linear elasticity. We consider the equation of motion~\eqref{eqn:eqnofmotionclassical} for each of the four symmetry classes. The corresponding elasticity tensors are given in Theorem~\ref{thm:symmetryclasses}. 

\underline{Oblique}: There are no restrictions on $\mathbb{C}$ under oblique symmetry and thus the oblique equation of motion is given by \eqref{eqn:eqnofmotionclassical}, which we write out explicitly below:

\begin{subequations}\label{eqn:eqnmotionobl}
\begin{align}
\begin{split}
\rho \ddot{u}_{1} ={}& C_{1111} \frac{\partial^2 u_1}{\partial x^2} +2C_{1112} \frac{\partial^2 u_1}{\partial x \partial y}  + C_{1212} \frac{\partial^2 u_1}{\partial y^2}    \\
&+ C_{1112} \frac{\partial^2 u_2}{\partial x^2}  + \left( C_{1122} + C_{1212} \right)  \frac{\partial^2 u_2}{\partial x \partial y} + C_{2212} \frac{\partial^2 u_2}{\partial y^2}  + b_1,
\end{split} \label{eqn:eqnmotionobla} \\
%%%%%%%%%%%%%%%%%%%%%%%%%%%%%%%%%%%%%%%%%%%%%%%%%%%%%%%%%%%%
\begin{split}
\rho \ddot{u}_{2} ={}& C_{1112} \frac{\partial^2 u_1}{\partial x^2} + \left( C_{1122} + C_{1212} \right) \frac{\partial^2 u_1}{\partial x \partial y} +  C_{2212} \frac{\partial^2 u_1}{\partial y^2} \\
&+ C_{1212} \frac{\partial^2 u_2}{\partial x^2} + 2 C_{2212} \frac{\partial^2 u_2}{\partial x \partial y} + C_{2222} \frac{\partial^2 u_2}{\partial y^2} + b_2.
\end{split} \label{eqn:eqnmotionoblb}
\end{align}
\end{subequations}

\underline{Rectangular}: Imposing the rectangular symmetry restrictions \eqref{eqn:recelastrelations} on \eqref{eqn:eqnmotionobl} produces the rectangular equation of motion:

\begin{subequations}\label{eqn:eqnmotionrec}
\begin{align}
\rho \ddot{u}_{1} ={}& C_{1111} \frac{\partial^2 u_1}{\partial x^2} + C_{1212}  \frac{\partial^2 u_1}{\partial y^2}  + \left( C_{1122} + C_{1212} \right) \frac{\partial^2 u_2}{\partial x \partial y}  +b_1, \label{eqn:eqnmotionreca} \\
%%%%%%%%%%%%%%%%%%%%%%%%%%%%%%%%%%%%%%%%%%%%%%%%%%%%%%%%%%%%
\rho \ddot{u}_{2} ={}&  \left( C_{1122} + C_{1212} \right) \frac{\partial^2 u_1}{\partial x \partial y} + C_{1212} \frac{\partial^2 u_2}{\partial x^2} + C_{2222} \frac{\partial^2 u_2}{\partial y^2} + b_2. \label{eqn:eqnmotionrecb}
\end{align}
\end{subequations}

\underline{Square}: Imposing the square symmetry restrictions \eqref{eqn:sqrelastrelations} on \eqref{eqn:eqnmotionobl} produces the square equation of motion:

\begin{subequations}\label{eqn:eqnmotionsqr}
\begin{align}
\rho \ddot{u}_{1} ={}& C_{1111} \frac{\partial^2 u_1}{\partial x^2} + C_{1212} \frac{\partial^2 u_1}{\partial y^2} +  \left( C_{1122} + C_{1212} \right) \frac{\partial^2 u_2}{\partial x \partial y}  + b_1, \label{eqn:eqnmotionsqra} \\
%%%%%%%%%%%%%%%%%%%%%%%%%%%%%%%%%%%%%%%%%%%%%%%%%%%%%%%%%%%%
\rho \ddot{u}_{2} ={}&   \left( C_{1122} + C_{1212} \right) \frac{\partial^2 u_1}{\partial x \partial y}  + C_{1212} \frac{\partial^2 u_2}{\partial x^2} + C_{1111} \frac{\partial^2 u_2}{\partial y^2} + b_2. \label{eqn:eqnmotionsqrb}
\end{align}
\end{subequations}
\underline{Isotropic}: Imposing the isotropic symmetry restrictions \eqref{eqn:isoelastrelations} on \eqref{eqn:eqnmotionobl} produces the isotropic equation of motion:
\begin{subequations}\label{eqn:eqnmotioniso}
\begin{align}
\rho \ddot{u}_{1} ={}& C_{1111} \frac{\partial^2 u_1}{\partial x^2} + \frac{1}{2} \left( C_{1111} - C_{1122} \right) \frac{\partial^2 u_1}{\partial y^2} +  \frac{1}{2} \left( C_{1111} + C_{1122} \right) \frac{\partial^2 u_2}{\partial x \partial y}  + b_1, \label{eqn:eqnmotionisoa}  \\
%%%%%%%%%%%%%%%%%%%%%%%%%%%%%%%%%%%%%%%%%%%%%%%%%%%%%%%%%%%%
\rho \ddot{u}_{2} ={}&   \frac{1}{2} \left( C_{1111} + C_{1122} \right) \frac{\partial^2 u_1}{\partial x \partial y}  + \frac{1}{2} \left( C_{1111} - C_{1122} \right) \frac{\partial^2 u_2}{\partial x^2} + C_{1111} \frac{\partial^2 u_2}{\partial y^2} + b_2. \label{eqn:eqnmotionisob}
\end{align}
\end{subequations}
\subsection{Planar approximations in three-dimensional classical linear elasticity}\label{sec:planemodelsclass}
In some situations, a three-dimensional problem may be simplified to a two-dimensional formulation. This often greatly reduces the computational cost of numerically solving the problem, and sometimes allows solutions of intractable three-dimensional problems to be successfully approximated. In fact, some of the first successful applications of the finite element method were performed on two-dimensional elastic problems \cite{clough1960finite,mj1956stiffness}. In this section, we consider two-dimensional simplifications of the three-dimensional equations of motion of classical linear elasticity. In particular, we focus on planar approximations of anisotropic models, specifically plane strain and plane stress models, where the in-plane and out-of-plane deformations are decoupled. In Sections \ref{sec:PeridynamicPlaneStrain} and \ref{sec:PeridynamicPlaneStress} we will derive the peridynamic analogues of plane strain and plane stress, and the derivations will mirror those utilized in the classical theory. For this reason, we provide derivations of the classical planar elasticity models in this section, following~\cite{Ting1996}. Before we delve into the planar elasticity models, we briefly review symmetry classes of $\mathbb{C}$ in three-dimensional classical linear elasticity.

\subsubsection{Symmetry classes of the elasticity tensor in three-dimensional classical linear elasticity}\label{sec:threedimclassicalelasticity}
Recall there are exactly eight symmetry classes in classical linear elasticity \cite{CHADWICK2001,Forte1996,Ting1996}: triclinic, monoclinic, orthotropic, trigonal, tetragonal, transversely isotropic, cubic, and isotropic. Similarly to two-dimensional classical linear elasticity, the inversion transformation $-\bfI$ is a member of every symmetry group of $\mathbb{C}$. We now briefly consider a representative group from each symmetry class. For each symmetry group, we consider a generating set as well as the resulting elasticity tensor and its corresponding restrictions.

Let a given plane be defined by the unit normal $\bfn = \langle n_1, n_2, n_3 \rangle$. The transformation corresponding to reflection through the plane is given by
\begin{equation*}
\mathbf{Ref}(\bfn) = \bfI -2 \bfn \otimes \bfn.
\end{equation*}
It turns out that in classical linear elasticity, a symmetry group is entirely determined by the reflection transformations in the group \cite{CHADWICK2001}. Let a coordinate system be given by the basis $\left\{ \bfe_i \right\}_{i=1,\ldots,3}$. The following are our representative symmetry groups of the symmetry classes of $\mathbb{C}$ in this coordinate system:
 
\underline{Triclinic}:
The triclinic symmetry group may be generated by
\begin{equation*}
\left\{ - \bfI \right\}.
\end{equation*}
There are no restrictions on the elasticity tensor. The triclinic elasticity tensor is given by ({\em cf.}~\eqref{eqn:3DElastTensor})
\begin{equation}\label{eqn:3Delasttensortricinic}
\mathbf{C} = \left[
\begin{array}{cccccc}
C_{1111} & C_{1122} & C_{1133} & C_{1123} & C_{1113} & C_{1112} \\
\cdot & C_{2222} & C_{2233} & C_{2223} & C_{2213} & C_{2212} \\
\cdot & \cdot & C_{3333} & C_{3323} & C_{3313} & C_{3312} \\
\cdot&\cdot &\cdot & C_{2323} & C_{2313} & C_{2312} \\
\cdot& \cdot& \cdot& \cdot& C_{1313} & C_{1312} \\
\cdot& \cdot& \cdot& \cdot& \cdot & C_{1212}
\end{array}
\right]. 
%\xrightarrow[\text{Relations}]{\text{Cauchy's}} \left[
%\begin{array}{cccccc}
%C_{1111} & C_{1122} & C_{1133} & C_{1123} & C_{1113} & C_{1112} \\
%\cdot & C_{2222} & C_{2233} & C_{2223} & C_{2213} & C_{2212} \\
%\cdot & \cdot & C_{3333} & C_{3323} & C_{3313} & C_{3312} \\
%\cdot&\cdot &\cdot & C_{2233} & C_{3312} & C_{2213} \\
%\cdot& \cdot& \cdot& \cdot& C_{1133} & C_{1123} \\
%\cdot& \cdot& \cdot& \cdot& \cdot & C_{1122}
%\end{array}
%\right] . \\
\end{equation}

\underline{Monoclinic}:
A monoclinic symmetry group may be generated by
\begin{equation}\label{eqn:generatorsmonoclinic}
\left\{ - \bfI, \mathbf{Ref}(\mathbf{e}_3) \right\}.
\end{equation}
The corresponding monoclinic symmetry restrictions on the elasticity tensor are given by
\begin{equation}\label{eqn:monoelastrelations}
C_{1123} = C_{1113} = C_{2223} = C_{2213} = C_{3323} = C_{3313} = C_{2312} = C_{1312} = 0.
\end{equation}
The resulting monoclinic elasticity tensor is given by
\begin{equation*}
\mathbf{C} = \left[
\begin{array}{cccccc}
C_{1111} & C_{1122} & C_{1133} & 0 & 0 & C_{1112} \\
\cdot & C_{2222} & C_{2233} & 0 & 0 & C_{2212} \\
\cdot & \cdot & C_{3333} & 0 & 0 & C_{3312} \\
\cdot&\cdot &\cdot & C_{2323} & C_{2313} & 0 \\
\cdot& \cdot& \cdot& \cdot& C_{1313} & 0 \\
\cdot& \cdot& \cdot& \cdot& \cdot & C_{1212}
\end{array}
\right]. 
%\xrightarrow[\text{Relations}]{\text{Cauchy's}} \left[
%\begin{array}{cccccc}
%C_{1111} & C_{1122} & C_{1133} & 0 & 0 & C_{1112} \\
%\cdot & C_{2222} & C_{2233} & 0 & 0 & C_{2212} \\
%\cdot & \cdot & C_{3333} & 0 & 0 & C_{3312} \\
%\cdot&\cdot &\cdot & C_{2233} & C_{3312} & 0 \\
%\cdot& \cdot& \cdot& \cdot& C_{1133} & 0 \\
%\cdot& \cdot& \cdot& \cdot& \cdot & C_{1122}
%\end{array}
%\right].
\end{equation*}

\underline{Orthotropic}:
An orthotropic symmetry group may be generated by
\begin{equation*}
\left\{-\bfI, \mathbf{Ref}(\bfe_1), \mathbf{Ref}(\bfe_2), \mathbf{Ref}(\bfe_3) \right\}.
\end{equation*}
The corresponding orthotropic symmetry restrictions on the elasticity tensor are given by
\begin{subequations}\label{eqn:orthoelastrelations}
\begin{align}
&C_{1123} = C_{1113} = C_{2223} = C_{2213} = C_{3323} = C_{3313} = C_{2312} = C_{1312} = 0, \\
&C_{1112} = C_{2212} = C_{3312} = C_{2313} = 0.
\end{align}
\end{subequations}
The resulting orthotropic elasticity tensor is given by
\begin{equation*}
\mathbf{C} = \left[
\begin{array}{cccccc}
C_{1111} & C_{1122} & C_{1133} & 0 & 0 & 0 \\
\cdot & C_{2222} & C_{2233} & 0 & 0 & 0 \\
\cdot & \cdot & C_{3333} & 0 & 0 & 0 \\
\cdot&\cdot &\cdot & C_{2323} & 0 & 0 \\
\cdot& \cdot& \cdot& \cdot& C_{1313} & 0 \\
\cdot& \cdot& \cdot& \cdot& \cdot & C_{1212}
\end{array}
\right].
%\xrightarrow[\text{Relations}]{\text{Cauchy's}} \left[
%\begin{array}{cccccc}
%C_{1111} & C_{1122} & C_{1133} & 0 & 0 & 0 \\
%\cdot & C_{2222} & C_{2233} & 0 & 0 & 0 \\
%\cdot & \cdot & C_{3333} & 0 & 0 & 0 \\
%\cdot&\cdot &\cdot & C_{2323} & 0 & 0 \\
%\cdot& \cdot& \cdot& \cdot& C_{1313} & 0 \\
%\cdot& \cdot& \cdot& \cdot& \cdot & C_{1212}
%\end{array}
%\right].
\end{equation*}

\underline{Trigonal}:
A trigonal symmetry group\footnote{Typically, trigonal symmetry is presented with the normals in the plane $z=0$, but the present formulation has advantages that will become apparent when we discuss planar linear elasticity models in Sections \ref{sec:classplanestrain} and \ref{sec:ClassicalPlaneStress}.} may be generated by
\begin{equation*}
\left\{ - \bfI, \mathbf{Ref}(\bfe_3), \mathbf{Ref} \left( \frac{1}{2} \left\langle 0, \sqrt{3}, \pm 1 \right\rangle\right) \right\}.
\end{equation*}
The corresponding trigonal symmetry restrictions on the elasticity tensor are given by
\begin{subequations}\label{eqn:trigonalelastrelations}
\begin{align}
&C_{1123} = C_{1113} = C_{2223} = C_{2213} = C_{3323} = C_{3313} = C_{2312} = C_{1312} = 0, \\
&C_{1112} = 0, C_{3333} = C_{2222}, C_{1133} = C_{1122},  C_{1212} = C_{1313}, \\
&C_{2313} = C_{3312} = - C_{2212}, C_{2323} = \frac{C_{2222} - C_{2233}}{2} .
\end{align}
\end{subequations}
The resulting trigonal elasticity tensor is given by 
\begin{equation*}
\mathbf{C} = \left[
 \begin{array}{cccccc}
C_{1111} & C_{1122} & C_{1122} & 0 & 0 & 0 \\
\cdot & C_{2222} & C_{2233} & 0 & 0 & C_{2212} \\
\cdot & \cdot & C_{2222} & 0 & 0 & -C_{2212} \\
\cdot&\cdot &\cdot & \frac{C_{2222} - C_{2233}}{2} & -C_{2212} & 0 \\
\cdot& \cdot& \cdot& \cdot& C_{1313} & 0 \\
\cdot& \cdot& \cdot& \cdot& \cdot & C_{1313}
\end{array}
\right]. 
%\xrightarrow[\text{Relations}]{\text{Cauchy's}} \left[
%\begin{array}{cccccc}
%C_{1111} & C_{1122} & C_{1122} & 0 & 0 & 0 \\
%\cdot & C_{1111} & C_{2233} & 0 & 0 & C_{2212} \\
%\cdot & \cdot & C_{1111} & 0 & 0 & -C_{2212} \\
%\cdot&\cdot &\cdot & \frac{C_{1111} - C_{2233}}{2} & -C_{2212} & 0 \\
%\cdot& \cdot& \cdot& \cdot& C_{1133} & 0 \\
%\cdot& \cdot& \cdot& \cdot& \cdot & C_{1133}
%\end{array}
%\right].
\end{equation*}

\underline{Tetragonal}:
A tetragonal symmetry group may be generated by

\begin{equation*}
\left\{-\bfI, \mathbf{Ref}(\bfe_1), \mathbf{Ref}(\bfe_2), \mathbf{Ref}(\bfe_3), \mathbf{Ref} \left( \frac{1}{2} \left\langle  \sqrt{2} , \pm \sqrt{2}, 0 \right\rangle \right) \right\}.
\end{equation*}
The corresponding tetragonal symmetry restrictions on the elasticity tensor are given by
\begin{subequations}\label{eqn:tetraelastrelations}
\begin{align}
&C_{1123} = C_{1113} = C_{2223} = C_{2213} = C_{3323} = C_{3313} = C_{2312} = C_{1312} = 0, \\
&C_{1112} = C_{2212} = C_{3312} = C_{2313} = 0, \\
&C_{2233} = C_{1133}, C_{1313} = C_{2323}, C_{2222} = C_{1111}.
\end{align}
\end{subequations}
The resulting tetragonal elasticity tensor is given by
\begin{equation*}
\mathbf{C} = \left[
\begin{array}{cccccc}
C_{1111} & C_{1122} & C_{1133} & 0 & 0 & 0 \\
\cdot & C_{1111} & C_{1133} & 0 & 0 & 0 \\
\cdot & \cdot & C_{3333} & 0 & 0 & 0 \\
\cdot&\cdot &\cdot & C_{2323} & 0 & 0 \\
\cdot& \cdot& \cdot& \cdot& C_{2323} & 0 \\
\cdot& \cdot& \cdot& \cdot& \cdot & C_{1212}
\end{array}
\right]. 
%\xrightarrow[\text{Relations}]{\text{Cauchy's}} \left[
%\begin{array}{cccccc}
%C_{1111} & C_{1122} & C_{1133} & 0 & 0 & 0 \\
%\cdot & C_{1111} & C_{1133} & 0 & 0 & 0 \\
%\cdot & \cdot & C_{3333} & 0 & 0 & 0 \\
%\cdot&\cdot &\cdot & C_{1133} & 0 & 0 \\
%\cdot& \cdot& \cdot& \cdot& C_{1133} & 0 \\
%\cdot& \cdot& \cdot& \cdot& \cdot & C_{1122}
%\end{array}
%\right].
\end{equation*}

\underline{Transversely Isotropic}:
A transversely isotropic symmetry group may be generated by
\begin{equation*}
\left\{ - \bfI, \mathbf{Ref}(\bfe_3), \mathbf{Ref}\left( \langle \cos(\theta), \sin(\theta), 0\rangle \right): \theta \in [0,\pi) \right\}.
\end{equation*}

The corresponding transversely isotropic symmetry restrictions on the elasticity tensor are given by
\begin{subequations}\label{eqn:tisoelastrelations}
\begin{align}
&C_{1123} = C_{1113} = C_{2223} = C_{2213} = C_{3323} = C_{3313} = C_{2312} = C_{1312} = 0, \\
&C_{1112} = C_{2212} = C_{3312} = C_{2313} = 0, \\
&C_{2233} = C_{1133}, C_{1313} = C_{2323}, C_{2222} = C_{1111}, \\
&C_{1212} = \frac{C_{1111} - C_{1122}}{2}.
\end{align}
\end{subequations}
The resulting transversely isotropic elasticity tensor is given by
\begin{equation*}
\mathbf{C} = \left[
\begin{array}{cccccc}
C_{1111} & C_{1122} & C_{1133} & 0 & 0 & 0 \\
\cdot & C_{1111} & C_{1133} & 0 & 0 & 0 \\
\cdot & \cdot & C_{3333} & 0 & 0 & 0 \\
\cdot&\cdot &\cdot & C_{2323} & 0 & 0 \\
\cdot& \cdot& \cdot& \cdot& C_{2323} & 0 \\
\cdot& \cdot& \cdot& \cdot& \cdot & \frac{C_{1111} - C_{1122}}{2}
\end{array}
\right]. 
%\xrightarrow[\text{Relations}]{\text{Cauchy's}} \left[
%\begin{array}{cccccc}
%C_{1111} & \frac{C_{1111}}{3} & C_{1133} & 0 & 0 & 0 \\
%\cdot & C_{1111} & C_{1133} & 0 & 0 & 0 \\
%\cdot & \cdot & C_{3333} & 0 & 0 & 0 \\
%\cdot&\cdot &\cdot & C_{1133} & 0 & 0 \\
%\cdot& \cdot& \cdot& \cdot& C_{1133} & 0 \\
%\cdot& \cdot& \cdot& \cdot& \cdot & \frac{C_{1111}}{3}
%\end{array}
%\right].
\end{equation*}

\underline{Cubic}:
A cubic symmetry group may be generated by
\begin{equation*}
\left\{
- \bfI, \mathbf{Ref}(\bfe_1), \mathbf{Ref}(\bfe_2), \mathbf{Ref}(\bfe_3), \mathbf{Ref}\left(\frac{1}{2} \langle \sqrt{2}, \pm \sqrt{2}, 0 \rangle \right), \mathbf{Ref}\left(\frac{1}{2} \langle 0, \sqrt{2}, \pm \sqrt{2} \rangle \right), \mathbf{Ref}\left(\frac{1}{2} \langle \sqrt{2}, 0, \pm \sqrt{2} \rangle \right)
\right\}.
\end{equation*}
The corresponding cubic symmetry restrictions on the elasticity tensor are given by
\begin{subequations}\label{eqn:cubicelastrelations}
\begin{align}
&C_{1123} = C_{1113} = C_{2223} = C_{2213} = C_{3323} = C_{3313} = C_{2312} = C_{1312} = 0, \\
&C_{1112} = C_{2212} = C_{3312} = C_{2313} = 0, \\
&C_{3333} = C_{2222} =  C_{1111}, C_{1212} = C_{1313} = C_{2323}, \\
&C_{2233} = C_{1133} = C_{1122} .
\end{align}
\end{subequations}
The resulting cubic elasticity tensor is given by
\begin{equation*}
\mathbf{C} = \left[
\begin{array}{cccccc}
C_{1111} & C_{1122} & C_{1122} & 0 & 0 & 0 \\
\cdot & C_{1111} & C_{1122} & 0 & 0 & 0 \\
\cdot & \cdot & C_{1111} & 0 & 0 & 0 \\
\cdot&\cdot &\cdot & C_{2323} & 0 & 0 \\
\cdot& \cdot& \cdot& \cdot& C_{2323} & 0 \\
\cdot& \cdot& \cdot& \cdot& \cdot & C_{2323}
\end{array}
\right].
% \xrightarrow[\text{Relations}]{\text{Cauchy's}} \left[
%\begin{array}{cccccc}
%C_{1111} & C_{1122} & C_{1122} & 0 & 0 & 0 \\
%\cdot & C_{1111} & C_{1122} & 0 & 0 & 0 \\
%\cdot & \cdot & C_{1111} & 0 & 0 & 0 \\
%\cdot&\cdot &\cdot & C_{1122} & 0 & 0 \\
%\cdot& \cdot& \cdot& \cdot& C_{1122} & 0 \\
%\cdot& \cdot& \cdot& \cdot& \cdot & C_{1122}
%\end{array}
%\right].
\end{equation*}

\underline{Isotropic}: 
The isotropic symmetry group is $O(3)$, the set of orthogonal transformations in $\mathbb{R}^3$.

The isotropic symmetry restrictions on the elasticity tensor are given by
\begin{subequations}\label{eqn:iso3Delastrelations}
\begin{align}
&C_{1123} = C_{1113} = C_{2223} = C_{2213} = C_{3323} = C_{3313} = C_{2312} = C_{1312} = 0, \\
&C_{1112} = C_{2212} = C_{3312} = C_{2313} = 0, \\
&C_{3333} = C_{2222} = C_{1111} ,   C_{2233} = C_{1133} = C_{1122}, \\
&C_{1212} = C_{1313} = C_{2323} = \frac{C_{1111} - C_{1122}}{2}.
\end{align}
\end{subequations}
The isotropic elasticity tensor is given by
\begin{equation*}
 \mathbf{C} = \left[
\begin{array}{cccccc}
C_{1111} & C_{1122} & C_{1122} & 0 & 0 & 0 \\
\cdot & C_{1111} & C_{1122} & 0 & 0 & 0 \\
\cdot & \cdot & C_{1111} & 0 & 0 & 0 \\
\cdot&\cdot &\cdot & \frac{C_{1111} - C_{1122}}{2} & 0 & 0 \\
\cdot& \cdot& \cdot& \cdot& \frac{C_{1111} - C_{1122}}{2} & 0 \\
\cdot& \cdot& \cdot& \cdot& \cdot & \frac{C_{1111} - C_{1122}}{2}
\end{array}
\right]. 
%\xrightarrow[\text{Relations}]{\text{Cauchy's}} C_{1111} \left[
%\begin{array}{cccccc}
%1 & \frac{1}{3} & \frac{1}{3} & 0 & 0 & 0 \\
%\cdot & 1 & \frac{1}{3} & 0 & 0 & 0 \\
%\cdot & \cdot & 1 & 0 & 0 & 0 \\
%\cdot&\cdot &\cdot & \frac{1}{3} & 0 & 0 \\
%\cdot& \cdot& \cdot& \cdot& \frac{1}{3} & 0 \\
%\cdot& \cdot& \cdot& \cdot& \cdot & \frac{1}{3}
%\end{array}
%\right].
\end{equation*}

\subsubsection{Classical plane strain}\label{sec:classplanestrain}

Classical plane strain is often associated with thick structures \cite{Ting1996}. Due to the thickness of the structure, deformations in the thickness direction are often constrained. Provided certain assumptions are met, each cross-section of the material perpendicular to the thickness direction is in approximately the same deformed state. In this situation, a two-dimensional formulation of the dynamics in one cross-section is sufficient to provide information about the dynamics of the entire structure. 

The classical plane strain assumptions are given as follows:
 
\begin{enumerate}[(C$\varepsilon$1)]\label{assmp:classplanestrain}
\item The geometric form and mass density of the body, and the external loads exerted on it, do not change along some axis, which we take as the $z$-axis. In particular, \label{assmp:classplanestrain1}
\begin{equation*}
\rho = \rho(x,y) \text{ and } \bfb = \bfb(x,y,t).
\end{equation*}
\item The deformation of any arbitrary cross-section perpendicular to the $z$-axis is identical, i.e., the displacements are function of $x$, $y$, and $t$ only:
\begin{equation*}%\label{assump:genplanestrain}
\bfu = \bfu(x,y,t).
\end{equation*} \label{assmp:classplanestrain2}
\item \label{assmp:classplanestrain3} The material has at least monoclinic symmetry with a plane of reflection corresponding to the plane $z=0$, i.e., \eqref{eqn:monoelastrelations} holds\footnote{In \cite{Ting1996} it is shown the slightly weaker condition 
\begin{equation*}
C_{1123} = C_{1113} = C_{2223} = C_{2213} = C_{2312} = C_{1312} = 0
\end{equation*}
is sufficient for the derivation of generalized plane stress.}. 
\end{enumerate}
A system satisfying Assumptions (C$\varepsilon$\ref{assmp:classplanestrain1}) and (C$\varepsilon$\ref{assmp:classplanestrain2}) is said to be in a state of classical generalized plane strain. See Figure \ref{fig:planestrain} for an illustration of a body in a state of plane strain.
 
\begin{figure}
\begin{center}
\begin{tikzpicture}
\path [left color=black!50, right color=black!50, middle color=black!25] (-1.5,-5.1) arc (180:360:1.5-.01 and 1-.05*1) -- cycle; %this is the bottom
      \path [top color=black!25, bottom color=white] 
        (0,0) ellipse [x radius=1.5-.01, y radius=1-.01*2/3]; %this is the top
      \path [left color=black!25, right color=black!25, middle color=white]  (-1.5,0) -- (-1.5,-5) arc (180:360:1.5 and 1) -- (1.5,0) arc (360:180:1.5 and 1); %this is the middle portion
      \draw[->,thick] (0,0)--(0,2) node[right]{$z$};
       \draw[->, thick]  (-2.3,-3) -- (-1.5,-3) ;
       \draw[->, thick]  (-2.3,-4) -- (-1.5,-4) ;
       \draw[->, thick]  (-2.3,-5) -- (-1.5,-5) ;
       \draw[->, thick]  (-2.3,-2) -- (-1.5,-2) ;
       \draw[->, thick]  (-2.3,-1) -- (-1.5,-1) ;
       \draw[->, thick]  (-2.3,0) -- (-1.5,0);
       \draw[-, thick] (-2.3,-5)  -- (-2.3,0) node[midway,left]{$\bfb$};

       \draw[->, thick]  (2.3,-3) -- (1.5,-3) ;
       \draw[->, thick]  (2.3,-4) -- (1.5,-4) ;
       \draw[->, thick]  (2.3,-5) -- (1.5,-5) ;
       \draw[->, thick]  (2.3,-2) -- (1.5,-2) ;
       \draw[->, thick]  (2.3,-1) -- (1.5,-1) ;
       \draw[->, thick]  (2.3,0) -- (1.5,0);
       \draw[-, thick] (2.3,-5)  -- (2.3,0) node[midway,right]{$\bfb$};
\end{tikzpicture}
\caption{Illustration of plane strain.}
\label{fig:planestrain}
\end{center}
\end{figure}
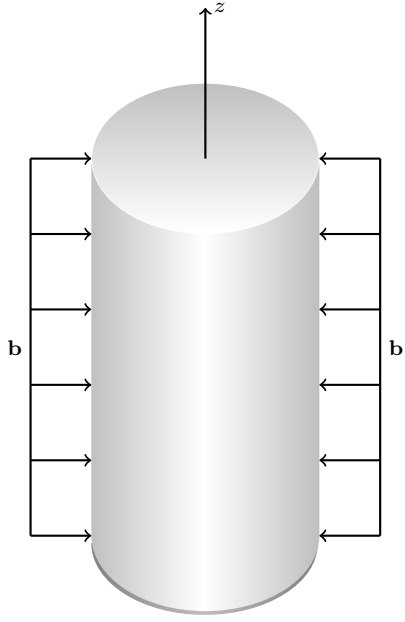

We start with the generalized Hooke's law~\eqref{eqn:stressstrainrelation}. By Assumption (C$\varepsilon$\ref{assmp:classplanestrain2}), the stress-strain relation reduces to
\begin{equation}\label{eqn:StrainDisplacementRelations}
\sigma_{ij} = C_{ij11} \frac{\partial u_1}{\partial x} + C_{ij12} \left(\frac{\partial u_1}{\partial y} + \frac{\partial u_2}{\partial x} \right) + C_{ij22} \frac{\partial u_2}{\partial y} + C_{ij13} \frac{\partial u_3}{\partial x} + C_{ij23} \frac{\partial u_3}{\partial y}. 
\end{equation}

\underline{Classical generalized plane strain}: Combining \eqref{eqn:StrainDisplacementRelations} with the three-dimensional equation of motion~\eqref{eqn:eqnofmotionclassical} and (C$\varepsilon$\ref{assmp:classplanestrain1}) results in 

\begin{subequations}\label{eqn:classeqnmot}
\begin{align}
\begin{split}
\rho \ddot{u}_1 ={}& C_{1111} \frac{\partial^2 u_1}{\partial x^2} + 2 C_{1112} \frac{\partial^2 u_1}{\partial x \partial y} + C_{1212} \frac{\partial^2 u_1}{\partial y^2} + C_{1112}\frac{\partial^2 u_2}{\partial x^2} + \left( C_{1122} + C_{1212} \right) \frac{\partial^2 u_2}{\partial x \partial y}   \\
&+ C_{2212} \frac{\partial^2 u_2}{\partial y^2}  + C_{1113} \frac{\partial^2 u_3}{\partial x^2}  + \left( C_{1123} + C_{1312} \right) \frac{\partial^2 u_3}{\partial x \partial y}   + C_{2312} \frac{\partial^2 u_3}{\partial y^2} + b_1, \label{eqn:classeqnmota}
\end{split}\\
%%%%%%%%%%%%%%%%%%%%%%%%%%%%%%%%%%%%%%%%%%%%%%%%%%%%%%%%%%%%%%
\begin{split}
\rho \ddot{u}_2 ={}& C_{1112} \frac{\partial^2 u_1}{\partial x^2}+ \left(C_{1122} + C_{1212} \right) \frac{\partial^2 u_1}{\partial x \partial y} + C_{2212} \frac{\partial^2 u_1}{\partial y^2}  + C_{1212} \frac{\partial^2 u_2}{\partial x^2} + 2 C_{2212} \frac{\partial^2 u_2}{\partial x \partial y}   \\
&+ C_{2222} \frac{\partial^2 u_2}{\partial y^2} + C_{1312} \frac{\partial^2 u_3}{\partial x^2} + \left(C_{2213} + C_{2312}\right) \frac{\partial^2 u_3}{\partial x \partial y}  + C_{2223} \frac{\partial^2 u_3}{\partial y^2} + b_2, \label{eqn:classeqnmotb}
\end{split}\\
%%%%%%%%%%%%%%%%%%%%%%%%%%%%%%%%%%%%%%%%%%%%%%%%%%%%%%%%%%%%%%
\begin{split}
\rho \ddot{u}_3 ={}& C_{1113} \frac{\partial^2 u_1}{\partial x^2}+  \left( C_{1123} + C_{1312}\right) \frac{\partial^2 u_1}{\partial x \partial y} +  C_{2312} \frac{\partial^2 u_1}{\partial y^2} + C_{1312} \frac{\partial^2 u_2}{\partial x^2} + \left( C_{2213} + C_{2312} \right) \frac{\partial^2 u_2}{\partial x \partial y} \\
&+ C_{2223} \frac{\partial^2 u_2}{\partial y^2} + C_{1313} \frac{\partial^2 u_3}{\partial x^2}  + 2 C_{2313} \frac{\partial^2 u_3}{\partial x \partial y} + C_{2323} \frac{\partial^2 u_3}{\partial y^2} + b_3. \label{eqn:classeqnmotc}
\end{split}
\end{align}
\end{subequations}

Equation \eqref{eqn:classeqnmot} is the \textit{classical generalized plane strain} equation of motion and reduces the degrees of freedom in the system significantly ($\bfu$ is a function of only $x$, $y$, and $t$, and thus the system may be solved on a two-dimensional spatial domain). However, the in-plane displacements, $u_1$ and $u_2$, are coupled with the out-of-plane displacement, $u_3$. Nevertheless, provided the material has sufficient symmetry, it is possible to further simplify the model. In fact, having a single plane of reflection symmetry (Assumption (C$\varepsilon$\ref{assmp:classplanestrain3})) is sufficient to decouple the in-plane and out-of-plane displacements provided the plane of reflection symmetry coincides with the plane $z= 0$, as demonstrated below. A system satisfying (C$\varepsilon$\ref{assmp:classplanestrain1})~--~(C$\varepsilon$\ref{assmp:classplanestrain3}) is said to be in a state of classical plane strain. 

\underline{Classical plane strain}: Imposing the monoclinic restrictions \eqref{eqn:monoelastrelations} on \eqref{eqn:classeqnmot}, we obtain the classical plane strain model. The in-plane equations of motion are given by
\begin{subequations}\label{eqn:classicalmonoeqnmotion}
\begin{align}
\begin{split}
\rho \ddot{u}_1 ={}& C_{1111} \frac{\partial^2 u_1}{\partial x^2}+ 2 C_{1112} \frac{\partial^2 u_1}{\partial x \partial y} + C_{1212} \frac{\partial^2 u_1}{\partial y^2}  \\
&+ C_{1112} \frac{\partial^2 u_2}{\partial x^2} + \left( C_{1122} + C_{1212} \right) \frac{\partial^2 u_2}{\partial x \partial y} + C_{2212} \frac{\partial^2 u_2}{\partial y^2} + b_1,
\end{split} \label{eqn:classicalmonoeqnmotion1} \\ 
%%%%%%%%%%%%%%%%%%%%%%%%%%%%%%%%%%%%%%%%%%%%%%%%%%%%%%%%%%%%%%
\begin{split}
\rho \ddot{u}_2 ={}& C_{1112} \frac{\partial^2 u_1}{\partial x^2}+ \left( C_{1122} + C_{1212}\right) \frac{\partial^2 u_1}{\partial x \partial y} + C_{2212} \frac{\partial^2 u_1}{\partial y^2}    \\
&+ C_{1212} \frac{\partial^2 u_2}{\partial x^2} + 2 C_{2212} \frac{\partial^2 u_2}{\partial x \partial y} + C_{2222} \frac{\partial^2 u_2}{\partial y^2} + b_2, 
\end{split} \label{eqn:classicalmonoeqnmotion2}
\end{align}
\end{subequations}
and the out-of-plane equation of motion is given by
\begin{align}\label{eqn:classicalmonoeqnmotion3}  
\rho \ddot{u}_3 ={}& C_{1313} \frac{\partial^2 u_3}{\partial x^2} + 2 C_{2313} \frac{\partial^2 u_3}{\partial x \partial y} + C_{2323} \frac{\partial^2 u_3}{\partial y^2} + b_3. 
\end{align}

The in-plane equations of motion \eqref{eqn:classicalmonoeqnmotion1} and \eqref{eqn:classicalmonoeqnmotion2} are equivalent mathematically to the pure two-dimensional oblique equations of motion \eqref{eqn:eqnmotionobla} and \eqref{eqn:eqnmotionoblb}, respectively. However, it should be noted the elasticity coefficients $C_{ijkl}$ do not have the same physical meaning between the two sets of equations, as we will see in Section \ref{sec:techconsts} where we discuss the engineering constants.

For monoclinic symmetry, there is only one choice of plane of reflection orientation which decouples the in-plane displacements, $u_1$ and $u_2$, and the out-of-plane displacement, $u_3$, i.e., when the plane of reflection symmetry coincides with the plane $z=0$. However, for the higher symmetry classes listed in Section \ref{sec:threedimclassicalelasticity}, there are multiple planes of reflection symmetry. Therefore, there are multiple orientations which decouple the in-plane and out-of-plane displacements, and consequently it is not possible to propose a unique plane strain model for those symmetry classes. For each symmetry class, we consider the equations of motion for the representative symmetry group presented in Section \ref{sec:threedimclassicalelasticity}.

\underline{Orthotropic}: Imposing \eqref{eqn:orthoelastrelations} on \eqref{eqn:classeqnmot}, we arrive at the following orthotropic plane strain model:
\begin{subequations}
\begin{align}
\rho \ddot{u}_1 ={}& C_{1111} \frac{\partial^2 u_1}{\partial x^2} + C_{1212} \frac{\partial^2 u_1}{\partial y^2} + \left( C_{1122} + C_{1212} \right) \frac{\partial^2 u_2}{\partial x \partial y}  + b_1, \label{eqn:orthoclassicalplanestrainu1} \\
%%%%%%%%%%%%%%%%%%%%%%%%%%%%%%%%%%%%%%%%%%%%%%%%%%%%%%%%%%%%%%
\rho \ddot{u}_2 ={}& \left(C_{1122} + C_{1212} \right) \frac{\partial^2 u_1}{\partial x \partial y} + C_{1212} \frac{\partial^2 u_2}{\partial x^2} + C_{2222} \frac{\partial^2 u_2}{\partial y^2} + b_2,  \label{eqn:orthoclassicalplanestrainu2} \\
%%%%%%%%%%%%%%%%%%%%%%%%%%%%%%%%%%%%%%%%%%%%%%%%%%%%%%%%%%%%%%
\rho \ddot{u}_3  ={}& C_{1313} \frac{\partial^2 u_3}{\partial x^2} + C_{2323} \frac{\partial^2 u_3}{\partial y^2} + b_3. \label{eqn:orthoclassicalplanestrainu3}
\end{align}
\end{subequations}

Notice the in-plane equations of motion \eqref{eqn:orthoclassicalplanestrainu1} and \eqref{eqn:orthoclassicalplanestrainu2} are equivalent mathematically to the two-dimensional rectangular equations of motion \eqref{eqn:eqnmotionreca} and \eqref{eqn:eqnmotionrecb}, respectively.

\underline{Trigonal}: Imposing \eqref{eqn:trigonalelastrelations} on \eqref{eqn:classeqnmot} we arrive at the following trigonal plane strain model:

\begin{subequations}\label{eqn:classicaltrigonaleqnmotion}
\begin{align}
\begin{split}
\rho \ddot{u}_1 ={}& C_{1111} \frac{\partial^2 u_1}{\partial x^2} + C_{1212} \frac{\partial^2 u_1}{\partial y^2} + \left( C_{1122} + C_{1212} \right) \frac{\partial^2 u_2}{\partial x \partial y} + C_{2212} \frac{\partial^2 u_2}{\partial y^2} + b_1, 
\end{split} \label{eqn:trigonalclassicalplanestrainu1} \\ 
%%%%%%%%%%%%%%%%%%%%%%%%%%%%%%%%%%%%%%%%%%%%%%%%%%%%%%%%%%%%%%
\begin{split}
\rho \ddot{u}_2 ={}& \left( C_{1122} + C_{1212}\right) \frac{\partial^2 u_1}{\partial x \partial y} + C_{2212} \frac{\partial^2 u_1}{\partial y^2} + C_{1212} \frac{\partial^2 u_2}{\partial x^2} + 2 C_{2212} \frac{\partial^2 u_2}{\partial x \partial y} + C_{2222} \frac{\partial^2 u_2}{\partial y^2} + b_2, 
\end{split} \label{eqn:trigonalclassicalplanestrainu2}\\
%%%%%%%%%%%%%%%%%%%%%%%%%%%%%%%%%%%%%%%%%%%%%%%%%%%%%%%%%%%%%%
\rho \ddot{u}_3 ={}& C_{1313} \frac{\partial^2 u_3}{\partial x^2} - 2 C_{2212} \frac{\partial^2 u_3}{\partial x \partial y} + \left( \frac{C_{2222} - C_{2233}}{2} \right) \frac{\partial^2 u_3}{\partial y^2} + b_3. \label{eqn:trigonalclassicalplanestrainu3} 
\end{align}
\end{subequations}
The in-plane equations of motion of motion \eqref{eqn:trigonalclassicalplanestrainu1} and \eqref{eqn:trigonalclassicalplanestrainu2} are mathematically equivalent to the two-dimensional oblique equations of motion \eqref{eqn:eqnmotionobla} and \eqref{eqn:eqnmotionoblb}, respectively (with $C_{1112} = 0$ imposed).

\underline{Tetragonal and Cubic}: 

Imposing \eqref{eqn:tetraelastrelations} or \eqref{eqn:cubicelastrelations} on \eqref{eqn:classeqnmot}, we arrive at the following in-plane equations of motion for both the tetragonal and cubic plane strain models:
\begin{subequations}
\begin{align}
\rho \ddot{u}_1 ={}& C_{1111} \frac{\partial^2 u_1}{\partial x^2} + C_{1212} \frac{\partial^2 u_1}{\partial y^2} + \left(C_{1122} + C_{1212} \right) \frac{\partial^2 u_2}{\partial x \partial y}  + b_1, \label{eqn:tetraclassicalplanestrainu1}  \\
%%%%%%%%%%%%%%%%%%%%%%%%%%%%%%%%%%%%%%%%%%%%%%%%%%%%%%%%%%%%%%
\rho \ddot{u}_2 ={}& \left(C_{1122} + C_{1212} \right) \frac{\partial^2 u_1}{\partial x \partial y} + C_{1212} \frac{\partial^2 u_2}{\partial x^2} + C_{1111} \frac{\partial^2 u_2}{\partial y^2} + b_2. \label{eqn:tetraclassicalplanestrainu2}  
%%%%%%%%%%%%%%%%%%%%%%%%%%%%%%%%%%%%%%%%%%%%%%%%%%%%%%%%%%%%%%
\end{align}
\end{subequations}

With \eqref{eqn:tetraelastrelations} imposed, the tetragonal out-of-plane deformation satisfies
\begin{align}
\rho \ddot{u}_3 = C_{2323} \left( \frac{\partial^2 u_3}{\partial x^2}  + \frac{\partial^2 u_3}{\partial y^2} \right) +b_3, \label{eqn:tetraclassicalplanestrainu3} 
\end{align}
whereas with \eqref{eqn:cubicelastrelations} imposed, the cubic out-of-plane deformation satisfies
\begin{align}
\rho \ddot{u}_3 = C_{1212} \left( \frac{\partial^2 u_3}{\partial x^2} +  \frac{\partial^2 u_3}{\partial y^2} \right) + b_3. \label{eqn:cubicclassicalplanestrainu3} 
\end{align}

Notice the in-plane equations of motion \eqref{eqn:tetraclassicalplanestrainu1} and \eqref{eqn:tetraclassicalplanestrainu2} are mathematically equivalent to the two-dimensional square equations of motion \eqref{eqn:eqnmotionsqra} and \eqref{eqn:eqnmotionsqrb}, respectively.

\underline{Transversely Isotropic and Isotropic}: 
Imposing \eqref{eqn:tisoelastrelations} or \eqref{eqn:iso3Delastrelations} on \eqref{eqn:classeqnmot}, we arrive at the following in-plane equations of motion for both the transversely isotropic and isotropic plane strain models:
\begin{subequations}\label{eqn:isoplanestrain}
\begin{align}
\rho \ddot{u}_1 ={}& C_{1111} \frac{\partial^2 u_1}{\partial x^2} + \frac{1}{2} \left( C_{1111}-C_{1122} \right) \frac{\partial^2 u_1}{\partial y^2} + \frac{1}{2} \left( C_{1111}+C_{1122} \right)  \frac{\partial^2 u_2}{\partial x \partial y}  + b_1, \label{eqn:isoclassicalplanestrainu1}  \\
%%%%%%%%%%%%%%%%%%%%%%%%%%%%%%%%%%%%%%%%%%%%%%%%%%%%%%%%%%%%%%
\rho \ddot{u}_2 ={}& \frac{1}{2} \left( C_{1111}+C_{1122} \right) \frac{\partial^2 u_1}{\partial x \partial y} + \frac{1}{2} \left( C_{1111}-C_{1122} \right) \frac{\partial^2 u_2}{\partial x^2} + C_{1111} \frac{\partial^2 u_2}{\partial y^2} + b_2. \label{eqn:isoclassicalplanestrainu2} 
%%%%%%%%%%%%%%%%%%%%%%%%%%%%%%%%%%%%%%%%%%%%%%%%%%%%%%%%%%%%%% 
\end{align}
\end{subequations}

With \eqref{eqn:tisoelastrelations} imposed, the transversely isotropic out-of-plane deformation satisfies
\begin{align}
\rho \ddot{u}_3= C_{2323} \left( \frac{\partial^2 u_3}{\partial x^2} + \frac{\partial^2 u_3}{\partial y^2} \right) + b_3, \label{eqn:tisoclassicalplanestrainu3} 
\end{align}
whereas with \eqref{eqn:iso3Delastrelations} imposed, the isotropic out-of-plane deformation satisfies
\begin{align}
\rho \ddot{u}_3 =  \left( \frac{C_{1111}-C_{1122}}{2} \right) \left( \frac{\partial^2 u_3}{\partial x^2} +  \frac{\partial^2 u_3}{\partial y^2}\right) + b_3. \label{eqn:isoclassicalplanestrainu3} 
\end{align}

Notice the in-plane equations of motion \eqref{eqn:isoclassicalplanestrainu1} and \eqref{eqn:isoclassicalplanestrainu2} are mathematically equivalent to the two-dimensional isotropic equations of motion \eqref{eqn:eqnmotionisoa} and \eqref{eqn:eqnmotionisob}, respectively.

\subsubsection{Classical plane stress}\label{sec:ClassicalPlaneStress}

Classical plane stress is often associated with thin plate-like structures subjected to edge forces which produce no normal displacement of the mid-plane of the structure \cite{reiss1961theory,Ting1996}. Due to the thinness of the structure, if certain assumptions are met, the stress tensor components $\sigma_{13}, \sigma_{23},$ and $\sigma_{33}$ are approximately constant over the thickness of the structure. A generalization of this, classical generalized plane stress, permits some variation in those components over the thickness of the structure. In this case, a two-dimensional formulation is obtained by considering averages of the displacement field $\bfu$ over the thickness of the structure. Derivations of classical generalized plane stress appear in a variety of sources. Derivations may be found in 
\cite{Filon1930,Love1892} for the isotropic case and in \cite{Ting1996} for the anisotropic case. We begin this section by considering classical generalized plane stress. Classical plane stress will be discussed later in Section \ref{sec:techconsts}. Since we are dealing with anisotropic material models, the following derivation mirrors the work presented in \cite{Ting1996}. 

The classical generalized plane stress assumptions are given as follows:

\begin{enumerate}[(C$\sigma$1)]
\item The body is a thin plate of thickness $2h$ occupying the region $-h~\leqslant~z~\leqslant~h$. \label{assump:PSs1}
\item The density is constant in the third dimension: $\rho = \rho(x,y)$. \label{assump:PSs5}
\item The body is subjected to a loading parallel and symmetric relative to the plane $z = 0$:  \label{assump:PSs2}
\begin{equation}
b_3(\bfx,t) = 0 \quad \text{ and } \quad \bfb(x,y,z,t) = \bfb(x,y,-z,t).
\end{equation}
\item The first and second components of the initial and boundary conditions are symmetric while their third component is antisymmetric relative to the plane $z=0$. \label{assump:PSsSymm} 
\item The surfaces of the plate are stress-free, i.e., $\sigma_{13} = \sigma_{23} = \sigma_{33} = 0$ for $z = \pm h$. \label{assump:PSs3}
\item The average stress $\overline{\sigma}_{33}$ ({\em cf.} \eqref{eqn:averagedefinition}) is zero\footnote{
It is possible to show that $\overline{\sigma}_{33}  = O(h^2)$. To see this,
perform a Taylor series about $z = 0$ to obtain (omitting $x$, $y$, and $t$ dependence for brevity)
\begin{equation*}
\sigma_{33} (z) = \sigma_{33} (0) + \frac{\partial \sigma_{33} }{\partial z}(0) z + O(z^2). 
\end{equation*}
Since the plate is traction-free on the top and bottom surfaces by (C$\sigma$\ref{assump:PSs3}), we obtain 
\begin{equation*}
0 = \sigma_{33}(h) + \sigma_{33}(-h) = 2\sigma_{33}(0) + O(h^2) \Rightarrow \sigma_{33}(0) = O(h^2). 
\end{equation*}
Thus,
\begin{equation}\label{eqn:sigma33small}
\overline{\sigma}_{33} = \frac{1}{2h} \int_{-h}^h \sigma_{33}(z) d z = \frac{1}{2h} \int_{-h}^h \sigma_{33}(0) + \frac{\partial \sigma_{33} }{\partial z}(0) z + O(z^2) d z = O(h^2).
\end{equation}
Consequently, if the plate is very thin, so that terms of order $O(h^2)$ may be neglected, we may suppose $\overline{\sigma}_{33} \approx 0$.
}
throughout the body. \label{assump:PSs4}
\item The material has at least monoclinic symmetry with a plane of reflection corresponding to the plane $z=0$, i.e., \eqref{eqn:monoelastrelations} holds\footnote{ If one were to consider a triclinic material, the in-plane stresses $\sigma_{11}, \sigma_{12},$ and $\sigma_{22}$ would generate shear strains $\varepsilon_{13}$ and $\varepsilon_{23}$ ({\em cf.} \eqref{eqn:compliancetensoreqn}). This, in turn, would cause the mid-plane of the plate, the plane $z = 0$, to no longer remain planar after the in-plane loading. Much like in the plane strain case, this can be alleviated by assuming monoclinic symmetry with the plane of reflection symmetry coinciding with the plane $z = 0$. }. \label{assump:PSs6}
\end{enumerate}  

A system satisfying Assumptions (C$\sigma$\ref{assump:PSs1})--(C$\sigma$\ref{assump:PSs6}) is said to be in a state of classical generalized plane stress. See Figure \ref{fig:planestress} for an illustration of a body in a state of plane stress. For a system in a state of classical generalized plane stress, the displacement field $\bfu$ has certain symmetries imposed on it, which we summarize in Lemma~\ref{lem:planestressymm}. 

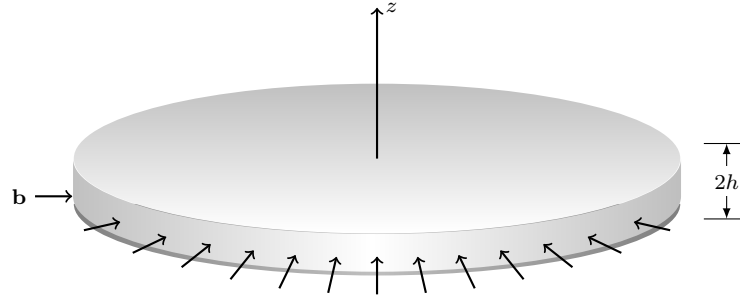
\begin{figure}
\begin{center}
\begin{tikzpicture}
\path [left color=black!50, right color=black!50, middle color=black!25] (-4+.01,-0.6) arc (180:360:4-.01 and 1-.05*1) -- cycle; %this is the bottom
      \path [top color=black!25, bottom color=white] 
        (0,0) ellipse [x radius=4-.01, y radius=1-.01*2/3]; %this is the top
      \path [left color=black!25, right color=black!25, middle color=white]  (-4,0) -- (-4,-0.5) arc (180:360:4 and 1) -- (4,0) arc (360:180:4 and 1); %this is the middle portion
      \draw[->,thick] (0,0)--(0,2) node[right]{$z$};
%      \draw[->] (0,0)--(1.4,0.8) node[right]{$\bfe_2$};
%      \draw[->] (0,0)--(2,0) node[right]{$\bfe_1$};
	   \draw (4.3,0.2) -- ++(0.3,0) coordinate (D1) -- +(5pt,0);
        \draw (4.3,-0.8) -- ++(0.3,0) coordinate (D2) -- +(5pt,0);
        \draw [<->,>=latex,thin,every rectangle node/.style={fill=white,midway,font=\sffamily}] (D1) -- (D2) node {$2h$};
        \draw[->,thick] (-4.5,-0.5) node[left]{$\bfb$} -- (-4.01,-0.5);
        %\draw[->,thick] (4.5,-0.5) node[left]{$\bfb$} -- (4.01,-0.5);       
       
        \draw[->,thick] (-3.857,-0.95) -- (-3.857+0.47771,-0.95+0.10904);
         \draw[->,thick] (3.857,-0.95) -- (3.857-0.47771,-0.95+0.10904);
         
        \draw[->,thick] (-3.2143,-1.25) -- (-3.2143+0.44147,-1.25+0.2126);
        \draw[->,thick] (3.2143,-1.25) -- (3.2143-0.44147,-1.25+0.2126);
        
        \draw[->,thick] (-2.5714,-1.44) -- (-2.5714+0.3831,-1.44+0.30551);
	    \draw[->,thick] (2.5714,-1.44) -- (2.5714-0.3831,-1.44+0.30551);        
        
        \draw[->,thick] (-1.92857,-1.6) -- (-1.92857+0.30551,-1.6+0.3831);
        \draw[->,thick] (1.92857,-1.6) -- (1.92857-0.30551,-1.6+0.3831);
       
        \draw[->,thick] (-1.28571,-1.72) -- (-1.2857+0.2126,-1.72+0.44147);
         \draw[->,thick] (1.28571,-1.72) -- (1.2857-0.2126,-1.72+0.44147);
        
        \draw[->,thick] (-0.6429,-1.78) -- (-0.6429+0.10904,-1.78+0.477771);
        \draw[->,thick] (0.6429,-1.78) -- (0.6429-0.10904,-1.78+0.477771);        
        
        \draw[->,thick](0,-1.8)--(0,-1.3);
\end{tikzpicture}
\caption{Illustration of plane stress.}
\label{fig:planestress}
\end{center}
\end{figure}

\begin{lemma}\label{lem:planestressymm} Under Assumptions (C$\sigma$\ref{assump:PSs1})--(C$\sigma$\ref{assump:PSsSymm}) and (C$\sigma$\ref{assump:PSs6}), the displacement field $\bfu$ of \eqref{eqn:eqnofmotionclassical} satisfies

\begin{align}\label{eqn:planestressmonoclinicdisplacement}
u_i(x,y,-z,t) = \left\{ \begin{array}{ll}
u_i(x,y,z,t),   & i = 1,2  \\
-u_i(x,y,z,t), & i =3
\end{array} \right. ,
\end{align}

i.e., the in-plane displacements, $u_1$ and $u_2$, are symmetric while the out-of-plane displacement, $u_3$, is antisymmetric relative to the plane $z = 0$. 
 \end{lemma}
\begin{proof}
See Appendix \ref{Appendix:classplanestresssymm}. \qed
\end{proof}
Before we derive the classical generalized plane stress equation of motion, let us first introduce a shorthand notation. Given a function $f(x,y,z,t)$, we define
\begin{equation}\label{eqn:averagedefinition}
\overline{f}(x,y,t) := \frac{1}{2h} \int_{-h}^h f(x,y,z,t) d z \quad \text{and} \quad 
[f](x,y,t) := \frac{f(x,y,h,t) - f(x,y,-h,t)}{2h},
\end{equation}

where $\overline{f}$ is the average of $f$ along the $z$-direction and $[f]$ is the average of the values of $f$ at the top and bottom surfaces of the plate. We observe that~\eqref{eqn:planestressmonoclinicdisplacement} implies
\[
\left[ u_1 \right] = \left[ u_2 \right] = \overline{u}_3 = 0.
\]
Consequently, we have
\begin{equation}\label{eqn:AvgDerivativePlainStress}
\frac{1}{2h} \int_{-h}^h \frac{\partial u_k}{\partial x_s} d z = \left\{
\begin{array}{ll}
\dfrac{\partial \overline{u}_k}{\partial x_s}, & k,s \neq 3 \\[8pt]
\left[u_3\right], & k=s=3 \\
0, & \text{otherwise}
\end{array} \quad \text{and} \quad 
\overline{\varepsilon}_{ks} = \left\{
\begin{array}{ll}
 \frac{1}{2} \left( \frac{\partial \overline{u}_k}{\partial x_s} + \frac{\partial \overline{u}_s}{\partial x_k} \right), & k,s \neq 3 \\
 \left[u_3\right], & k=s=3 \\
 0, & \text{otherwise}
\end{array}
\right.
\right. .
\end{equation}
Recalling the stress-strain relation \eqref{eqn:stressstrainrelation} and imposing Assumption (C$\sigma$\ref{assump:PSs6}), we obtain ({\em cf.}~\eqref{eqn:monoelastrelations})
\begin{subequations}\label{eqn:stressstrainplanestressmono}
\begin{align}
\sigma_{11} &= C_{1111} \varepsilon_{11} + C_{1122} \varepsilon_{22} + C_{1133} \varepsilon_{33} + 2 C_{1112} \varepsilon_{12}, \\
\sigma_{22} &= C_{1122} \varepsilon_{11}  + C_{2222} \varepsilon_{22} + C_{2233} \varepsilon_{33} + 2 C_{2212} \varepsilon_{12}, \\
\sigma_{33} &= C_{1133} \varepsilon_{11} + C_{2233} \varepsilon_{22} + C_{3333} \varepsilon_{33} + 2 C_{3312} \varepsilon_{12}, \\
\sigma_{23} &= 2 C_{2323} \varepsilon_{23} + 2 C_{2313} \varepsilon_{13}, \\
\sigma_{13} &= 2 C_{2313} \varepsilon_{23} + 2 C_{1313} \varepsilon_{13}, \\
\sigma_{12} &= C_{1112} \varepsilon_{11}  + C_{2212} \varepsilon_{22} + C_{3312} \varepsilon_{33} + 2 C_{1212} \varepsilon_{12}.
\end{align}
\end{subequations}

Taking the average along the $z$-direction for each of the stress-strain relations in \eqref{eqn:stressstrainplanestressmono}, we obtain ({\em cf.}~\eqref{eqn:AvgDerivativePlainStress})
\begin{subequations}\label{eqn:AvgStressRelation}
\begin{align}
\overline{\sigma}_{11} &= C_{1111} \frac{\partial \overline{u}_1}{\partial x} + C_{1112} \left(\frac{\partial \overline{u}_1}{\partial y} + \frac{\partial \overline{u}_2}{\partial x} \right) + C_{1122} \frac{\partial \overline{u}_2}{\partial y} + C_{1133} [u_3], \\
%%%%%%%%%%%%%%%%%%%%%%%%%%%%%%%%%%%%%%%%%%%%%%%%%%%%
\overline{\sigma}_{22} &= C_{1122} \frac{\partial \overline{u}_1}{\partial x} + C_{2212} \left(\frac{\partial \overline{u}_1}{\partial y} + \frac{\partial \overline{u}_2}{\partial x} \right) + C_{2222} \frac{\partial \overline{u}_2}{\partial y} + C_{2233} [u_3], \\
%%%%%%%%%%%%%%%%%%%%%%%%%%%%%%%%%%%%%%%%%%%%%%%%%%%%
\overline{\sigma}_{33} &= C_{1133} \frac{\partial \overline{u}_1}{\partial x} + C_{3312}\left(\frac{\partial \overline{u}_1}{\partial y} + \frac{\partial \overline{u}_2}{\partial x} \right) + C_{2233} \frac{\partial \overline{u}_2}{\partial y} + C_{3333} [u_3], \label{eqn:AvgStressRelationc} \\
%%%%%%%%%%%%%%%%%%%%%%%%%%%%%%%%%%%%%%%%%%%%%%%%%%%%
\overline{\sigma}_{23} &= 0, \\ %C_{2313} \left([u_1] + \frac{\partial \overline{u}_3}{\partial x} \right) + C_{2323} \left([u_2] + \frac{\partial \overline{u}_3}{\partial y} \right) = 0 \\
%%%%%%%%%%%%%%%%%%%%%%%%%%%%%%%%%%%%%%%%%%%%%%%%%%%%
\overline{\sigma}_{13} &= 0, \\ %C_{1313} \left([u_1] + \frac{\partial \overline{u}_3}{\partial x} \right) + C_{1323} \left([u_2] + \frac{\partial \overline{u}_3}{\partial y} \right) = 0 \\
%%%%%%%%%%%%%%%%%%%%%%%%%%%%%%%%%%%%%%%%%%%%%%%%%%%%
\overline{\sigma}_{12} &= C_{1112} \frac{\partial \overline{u}_1}{\partial x} + C_{1212} \left(\frac{\partial \overline{u}_1}{\partial y} + \frac{\partial \overline{u}_2}{\partial x} \right) + C_{2212} \frac{\partial \overline{u}_2}{\partial y} + C_{3312} [u_3]. 
\end{align}
\end{subequations}

By \eqref{eqn:AvgStressRelationc} and Assumption (C$\sigma$\ref{assump:PSs4}), we obtain
\[
0 = \overline{\sigma}_{33} = C_{1133} \frac{\partial \overline{u}_1}{\partial x} + C_{3312}\left(\frac{\partial \overline{u}_1}{\partial y} + \frac{\partial \overline{u}_2}{\partial x} \right) + C_{2233} \frac{\partial \overline{u}_2}{\partial y} + C_{3333} [u_3]. 
\]
Solving for $[u_3]$, we find\footnote{Note for material stability, $C_{3333} > 0$ is required \cite{born1954dynamical,cowley1976acoustic}.} 
\begin{equation}\label{eqn:[u3]solnplanestress}
[u_3] = -\left[ \frac{C_{1133}}{C_{3333}} \frac{\partial \overline{u}_1}{\partial x} + \frac{C_{3312}}{C_{3333}} \left(\frac{\partial \overline{u}_1}{\partial y} + \frac{\partial \overline{u}_2}{\partial x} \right) + \frac{C_{2233}}{C_{3333}} \frac{\partial \overline{u}_2}{\partial y} \right]. 
\end{equation}
Plugging \eqref{eqn:[u3]solnplanestress} into the equations of \eqref{eqn:AvgStressRelation}, we find 

\begin{subequations}\label{eqn:avgstressforclassicalplanestress}
\begin{align}
\overline{\sigma}_{11} &= \left(C_{1111} - \frac{C_{1133}^2}{C_{3333}} \right)\frac{\partial \overline{u}_1}{\partial x} + \left( C_{1112} - \frac{C_{1133} C_{3312}}{C_{3333}} \right)\left(\frac{\partial \overline{u}_1}{\partial y} + \frac{\partial \overline{u}_2}{\partial x} \right) + \left( C_{1122} - \frac{C_{1133} C_{2233}}{C_{3333}}\right) \frac{\partial \overline{u}_2}{\partial y}, \\
%%%%%%%%%%%%%%%%%%%%%%%%%%%%%%%%%%%%%%%%%%%%%%%%%%%%
\overline{\sigma}_{22} &= \left(C_{1122} - \frac{C_{2233} C_{1133}}{C_{3333}} \right) \frac{\partial \overline{u}_1}{\partial x} + \left( C_{2212} - \frac{C_{2233} C_{3312}}{C_{3333}}\right) \left(\frac{\partial \overline{u}_1}{\partial y} + \frac{\partial \overline{u}_2}{\partial x} \right) + \left(C_{2222}- \frac{C_{2233}^2}{C_{3333}}\right) \frac{\partial \overline{u}_2}{\partial y}, \\
%%%%%%%%%%%%%%%%%%%%%%%%%%%%%%%%%%%%%%%%%%%%%%%%%%%%
\overline{\sigma}_{33} &= 0, \\
%%%%%%%%%%%%%%%%%%%%%%%%%%%%%%%%%%%%%%%%%%%%%%%%%%%%
\overline{\sigma}_{23} &= 0, \\
%%%%%%%%%%%%%%%%%%%%%%%%%%%%%%%%%%%%%%%%%%%%%%%%%%%%
\overline{\sigma}_{13} &= 0, \\
%%%%%%%%%%%%%%%%%%%%%%%%%%%%%%%%%%%%%%%%%%%%%%%%%%%%
\overline{\sigma}_{12} &= \left(C_{1112} - \frac{C_{3312} C_{1133} }{C_{3333}}\right) \frac{\partial \overline{u}_1}{\partial x} + \left( C_{1212}- \frac{C_{3312}^2}{C_{3333}}\right) \left(\frac{\partial \overline{u}_1}{\partial y} + \frac{\partial \overline{u}_2}{\partial x} \right) + \left(C_{2212}- \frac{C_{3312} C_{2233}}{C_{3333}}\right) \frac{\partial \overline{u}_2}{\partial y}. 
\end{align}
\end{subequations}

\underline{Classical generalized plane stress}: 
 Averaging the equation of motion \eqref{eqn:eqnofmotionclassical} along the $z$-direction (recall $\rho$ is constant in $z$ by Assumption~(C$\sigma$\ref{assump:PSs5}) and $\left[ \sigma_{13} \right] = \left[ \sigma_{23} \right] = 0$ by Assumption~(C$\sigma$\ref{assump:PSs3})) and then employing \eqref{eqn:avgstressforclassicalplanestress}, results in
\begin{subequations}\label{eqn:planestressequationsofmotionclassical}
\begin{align}
\begin{split}
\rho \ddot{\overline{u}}_1 ={}& \frac{\partial \overline{\sigma}_{11}}{\partial x} + \frac{\partial \overline{\sigma}_{12}}{\partial y} + \left[ \sigma_{13} \right] + \overline{b}_1 \\
={}&\left(C_{1111} - \frac{C_{1133}^2}{C_{3333}} \right)\frac{\partial^2 \overline{u}_1}{\partial x^2} + 2 \left( C_{1112} - \frac{C_{1133}C_{3312}}{C_{3333}} \right) \frac{\partial^2 \overline{u}_1}{\partial x \partial y}  +  \left( C_{1212}- \frac{C_{3312}^2}{C_{3333}}\right) \frac{\partial^2 \overline{u}_1}{\partial y^2}  \\
&+ \left( C_{1112} - \frac{C_{1133}C_{3312}}{C_{3333}} \right) \frac{\partial^2 \overline{u}_2}{\partial x^2} + \left( C_{1122} - \frac{C_{1133}C_{2233}}{C_{3333}} + C_{1212}- \frac{C_{3312}^2}{C_{3333}} \right) \frac{\partial^2 \overline{u}_2}{\partial x \partial y} \\
&+\left(C_{2212}- \frac{C_{2233}C_{3312}}{C_{3333}}\right) \frac{\partial^2 \overline{u}_2}{\partial y^2} + \overline{b}_1, 
\end{split} \label{eqn:planestressequationsofmotionclassicala} \\
%%%%%%%%%%%%%%%%%%%%%%%%%%%%%%%%%%%%
\begin{split}
\rho \ddot{\overline{u}}_2 ={}& \frac{\partial \overline{\sigma}_{21}}{\partial x} + \frac{\partial \overline{\sigma}_{22}}{\partial y} + \left[ \sigma_{23} \right] +\overline{b}_2 \\
={}& \left(C_{1112} - \frac{C_{1133} C_{3312} }{C_{3333}}\right) \frac{\partial^2 \overline{u}_1}{\partial x^2} + \left( C_{1122} - \frac{C_{1133}C_{2233}}{C_{3333}} + C_{1212}- \frac{C_{3312}^2}{C_{3333}} \right) \frac{\partial^2 \overline{u}_1}{\partial x \partial y}    \\
&+ \left( C_{2212} - \frac{C_{2233}C_{3312}}{C_{3333}}\right) \frac{\partial^2 \overline{u}_1}{\partial y^2} + \left( C_{1212}- \frac{C_{3312}^2}{C_{3333}}\right) \frac{\partial^2 \overline{u}_2}{\partial x^2}    \\
&+ 2 \left( C_{2212} - \frac{C_{2233}C_{3312}}{C_{3333}}\right) \frac{\partial^2 \overline{u}_2}{\partial x \partial y} + \left(C_{2222}- \frac{C_{2233}^2}{C_{3333}}\right) \frac{\partial^2 \overline{u}_2}{\partial y^2} + \overline{b}_2. \label{eqn:planestressequationsofmotionclassicalb}
\end{split}
\end{align}
\end{subequations}
Equation \eqref{eqn:planestressequationsofmotionclassical} is the \textit{classical generalized plane stress} equation of motion. By performing the substitution
\begin{equation}\label{eqn:planestressto2Dcoefficientrelationship}
\tilde{C}_{ijkl} = C_{ijkl} -  \frac{C_{ij33} C_{33kl}}{C_{3333}}
\end{equation}
in \eqref{eqn:planestressequationsofmotionclassical}, we see that up to a change in coefficients, the pure two-dimensional oblique equation of motion \eqref{eqn:eqnmotionobl} and the generalized plane stress equation of motion \eqref{eqn:planestressequationsofmotionclassical} are equivalent mathematically (if we replace $\bfb$ and $\bfu$ in \eqref{eqn:eqnmotionobl} by $\overline{\bfb}$ and $\overline{\bfu}$, respectively). Moreover, in Section \ref{sec:techconsts} we demonstrate that, with respect to engineering constants, these two equations are equivalent.

\begin{remark}\label{rmk:planestressteps}
It is possible to perform the steps for the derivation of the classical generalized plane stress model \eqref{eqn:planestressequationsofmotionclassical} in an alternative order (given below), which facilitates the comparison with the derivation of the peridynamic generalized plane stress model in Section \ref{sec:PeridynamicPlaneStress}: 
\begin{enumerate}[Step 1:]
    \item Take the average of the equation of motion \eqref{eqn:eqnofmotionclassical} along the thickness of the plate.
    \item Utilize Assumption (C$\sigma$\ref{assump:PSs3}) to eliminate $[\sigma_{13}]$ and $[\sigma_{23}]$.
    \item Employ Assumption (C$\sigma$\ref{assump:PSs4}) to replace $[u_3]$ by expressions in $\overline{u}_1$ and $\overline{u}_2$.
\end{enumerate}
In the derivation, Lemma \ref{lem:planestressymm} is utilized to eliminate various terms.
\end{remark}

Similarly to the plane strain case, for monoclinic symmetry in generalized plane stress, there is only choice of plane of reflection orientation which decouples the in-plane and out-of-plane displacements (i.e., when the plane of reflection coincides with the plane $z=0$). However, for many of the other symmetry classes in Section \ref{sec:threedimclassicalelasticity} there are multiple planes of reflection symmetry. Therefore, there are multiple orientations that decouple the in-plane and out-of-plane displacements and, consequently, it is not possible to propose a unique generalized plane stress model for those symmetry classes. We consider the equation of motion for one possible orientation for each symmetry class. We suppose the elasticity tensors for each symmetry group are given by those found in Section \ref{sec:threedimclassicalelasticity}.

\underline{Orthotropic}: Imposing \eqref{eqn:orthoelastrelations} on \eqref{eqn:planestressequationsofmotionclassical}, we arrive at the following orthotropic generalized plane stress model:
\begin{subequations}\label{eqn:planestressequationsofmotionclassicalortho}
\begin{align}
\rho \ddot{\overline{u}}_1 =& \left(C_{1111} - \frac{C_{1133}^2}{C_{3333}} \right)\frac{\partial^2 \overline{u}_1}{\partial x^2} + C_{1212} \frac{\partial^2 \overline{u}_1}{\partial y^2} + \left(  C_{1122} + C_{1212} - \frac{C_{1133}C_{2233}}{C_{3333}} \right) \frac{\partial^2 \overline{u}_2}{\partial x \partial y} + \overline{b}_1, \\
%%%%%%%%%%%%%%%%%%%%%%%%%%%%%%%%%%%%
\rho \ddot{\overline{u}}_2 =&  \left( C_{1122} + C_{1212} - \frac{C_{1133}C_{2233}}{C_{3333}} \right) \frac{\partial^2 \overline{u}_1}{\partial x \partial y}  + C_{1212}  \frac{\partial^2 \overline{u}_2}{\partial x^2} + \left(C_{2222}- \frac{C_{2233}^2}{C_{3333}}\right) \frac{\partial^2 \overline{u}_2}{\partial y^2} + \overline{b}_2.
\end{align}
\end{subequations}

By performing the substitution \eqref{eqn:planestressto2Dcoefficientrelationship} in \eqref{eqn:planestressequationsofmotionclassicalortho}, we see that up to a change in coefficients, the orthotropic generalized plane stress equation of motion \eqref{eqn:planestressequationsofmotionclassicalortho} and the pure two-dimensional rectangular equation of motion \eqref{eqn:eqnmotionrec} are equivalent mathematically (if we replace $\bfb$ and $\bfu$ in \eqref{eqn:eqnmotionrec} by $\overline{\mathbf{b}}$ and $\overline{\mathbf{u}}$, respectively).

\underline{Trigonal}: Imposing \eqref{eqn:trigonalelastrelations} on \eqref{eqn:planestressequationsofmotionclassical}, we arrive at the following trigonal generalized plane stress model:
\begin{subequations}\label{eqn:planestressequationsofmotionclassicaltrigonal}
\begin{align}
\begin{split}
\rho \ddot{\overline{u}}_1 ={}&\left(C_{1111} - \frac{C_{1122}^2}{C_{2222}} \right)\frac{\partial^2 \overline{u}_1}{\partial x^2} + 2 \frac{C_{1122}C_{2212}}{C_{2222}} \frac{\partial^2 \overline{u}_1}{\partial x \partial y}  +  \left( C_{1212}- \frac{C_{2212}^2}{C_{2222}}\right) \frac{\partial^2 \overline{u}_1}{\partial y^2}  + \frac{C_{1122}C_{2212}}{C_{2222}} \frac{\partial^2 \overline{u}_2}{\partial x^2}  \\
&+ \left( C_{1122} - \frac{C_{1122}C_{2233}}{C_{2222}} + C_{1212}- \frac{C_{2212}^2}{C_{2222}} \right) \frac{\partial^2 \overline{u}_2}{\partial x \partial y}+\left(C_{2212} + \frac{C_{2233}C_{2212}}{C_{2222}}\right) \frac{\partial^2 \overline{u}_2}{\partial y^2} + \overline{b}_1, 
\end{split} \\
%%%%%%%%%%%%%%%%%%%%%%%%%%%%%%%%%%%%
\begin{split}
\rho \ddot{\overline{u}}_2 ={}& \frac{C_{1122} C_{2212} }{C_{2222}}  \frac{\partial^2 \overline{u}_1}{\partial x^2} + \left( C_{1122} - \frac{C_{1122}C_{2233}}{C_{2222}} + C_{1212} - \frac{C_{2212}^2}{C_{2222}} \right) \frac{\partial^2 \overline{u}_1}{\partial x \partial y}  + \left( C_{2212} + \frac{C_{2233}C_{2212}}{C_{2222}}\right) \frac{\partial^2 \overline{u}_1}{\partial y^2} \\
&+ \left( C_{1212}- \frac{C_{2212}^2}{C_{2222}}\right) \frac{\partial^2 \overline{u}_2}{\partial x^2} + 2 \left( C_{2212} + \frac{C_{2233}C_{2212}}{C_{2222}}\right) \frac{\partial^2 \overline{u}_2}{\partial x \partial y} + \left(C_{2222}- \frac{C_{2233}^2}{C_{2222}}\right) \frac{\partial^2 \overline{u}_2}{\partial y^2} + \overline{b}_2.
\end{split}
\end{align}
\end{subequations}

By performing the substitution \eqref{eqn:planestressto2Dcoefficientrelationship} in \eqref{eqn:planestressequationsofmotionclassicaltrigonal}, we see that up to a change in coefficients, the trigonal generalized plane stress equation of motion \eqref{eqn:planestressequationsofmotionclassicaltrigonal}  and the pure two-dimensional oblique equation of motion \eqref{eqn:eqnmotionobl} are equivalent mathematically (if we replace $\bfb$ and $\bfu$ in \eqref{eqn:eqnmotionobl} by $\overline{\mathbf{b}}$ and $\overline{\mathbf{u}}$, respectively).

\underline{Tetragonal}: Imposing \eqref{eqn:tetraelastrelations} on \eqref{eqn:planestressequationsofmotionclassical}, we arrive at the following tetragonal generalized plane stress model:
\begin{subequations}\label{eqn:planestressequationsofmotionclassicaltetra}
\begin{align}
\rho \ddot{\overline{u}}_1 =& \left(C_{1111} - \frac{C_{1133}^2}{C_{3333}} \right)\frac{\partial^2 \overline{u}_1}{\partial x^2} + C_{1212} \frac{\partial^2 \overline{u}_1}{\partial y^2} + \left( C_{1122} + C_{1212} - \frac{C_{1133}^2}{C_{3333}} \right) \frac{\partial^2 \overline{u}_2}{\partial x \partial y} + \overline{b}_1, \\
%%%%%%%%%%%%%%%%%%%%%%%%%%%%%%%%%%%%
\rho \ddot{\overline{u}}_2 =& \left( C_{1122} + C_{1212} - \frac{C_{1133}^2}{C_{3333}} \right) \frac{\partial^2 \overline{u}_1}{\partial x \partial y} +C_{1212}  \frac{\partial^2 \overline{u}_2}{\partial x^2} + \left(C_{1111}- \frac{C_{1133}^2}{C_{3333}}\right) \frac{\partial^2 \overline{u}_2}{\partial y^2}  + \overline{b}_2.
\end{align}
\end{subequations}

By performing the substitution \eqref{eqn:planestressto2Dcoefficientrelationship} in \eqref{eqn:planestressequationsofmotionclassicaltetra}, we see that up to a change in coefficients, the tetragonal generalized plane stress equation of motion \eqref{eqn:planestressequationsofmotionclassicaltetra} and the pure two-dimensional square equation of motion \eqref{eqn:eqnmotionsqr}  are equivalent mathematically (if we replace $\bfb$ and $\bfu$ in \eqref{eqn:eqnmotionsqr}  by $\overline{\mathbf{b}}$ and $\overline{\mathbf{u}}$, respectively).

\underline{Transversely Isotropic}: Imposing \eqref{eqn:tisoelastrelations} on \eqref{eqn:planestressequationsofmotionclassical}, we arrive at the following transversely isotropic generalized plane stress model:
\begin{subequations}\label{eqn:planestressequationsofmotionclassicaltiso}
\begin{align}
\rho \ddot{\overline{u}}_1 =& \left(C_{1111} - \frac{C_{1133}^2}{C_{3333}} \right)\frac{\partial^2 \overline{u}_1}{\partial x^2} + \frac{1}{2} \left(C_{1111}-C_{1122}\right) \frac{\partial^2 \overline{u}_1}{\partial y^2} + \frac{1}{2} \left( C_{1111} + C_{1122} - 2\frac{C_{1133}^2}{C_{3333}} \right) \frac{\partial^2 \overline{u}_2}{\partial x \partial y} + \overline{b}_1, \\
%%%%%%%%%%%%%%%%%%%%%%%%%%%%%%%%%%%%
\rho \ddot{\overline{u}}_2 =& \frac{1}{2} \left( C_{1111} + C_{1122} - 2 \frac{C_{1133}^2}{C_{3333}} \right) \frac{\partial^2 \overline{u}_1}{\partial x \partial y} +\frac{1}{2} \left(C_{1111}-C_{1122}\right)  \frac{\partial^2 \overline{u}_2}{\partial x^2} + \left(C_{1111}- \frac{C_{1133}^2}{C_{3333}}\right) \frac{\partial^2 \overline{u}_2}{\partial y^2}  + \overline{b}_2.
\end{align}
\end{subequations}
By performing the substitution \eqref{eqn:planestressto2Dcoefficientrelationship} in \eqref{eqn:planestressequationsofmotionclassicaltiso}, we see that up to a change in coefficients, the transversely isotropic generalized plane stress equation of motion \eqref{eqn:planestressequationsofmotionclassicaltiso} and the pure two-dimensional isotropic equation of motion \eqref{eqn:eqnmotioniso}  are equivalent mathematically (if we replace $\bfb$ and $\bfu$ in \eqref{eqn:eqnmotioniso}  by $\overline{\mathbf{b}}$ and $\overline{\mathbf{u}}$, respectively).

\underline{Cubic}: Imposing \eqref{eqn:cubicelastrelations} on \eqref{eqn:planestressequationsofmotionclassical}, we arrive at the following cubic generalized plane stress model:
\begin{subequations}\label{eqn:planestressequationsofmotionclassicalcubic}
\begin{align}
\rho \ddot{\overline{u}}_1 =& \left(C_{1111} - \frac{C_{1122}^2}{C_{1111}} \right)\frac{\partial^2 \overline{u}_1}{\partial x^2} + C_{1212} \frac{\partial^2 \overline{u}_1}{\partial y^2} + \left( C_{1122} + C_{1212} - \frac{C_{1122}^2}{C_{1111}} \right) \frac{\partial^2 \overline{u}_2}{\partial x \partial y} + \overline{b}_1, \\
%%%%%%%%%%%%%%%%%%%%%%%%%%%%%%%%%%%%
\rho \ddot{\overline{u}}_2 =& \left( C_{1122} + C_{1212} - \frac{C_{1122}^2}{C_{1111}} \right) \frac{\partial^2 \overline{u}_1}{\partial x \partial y} +C_{1212}  \frac{\partial^2 \overline{u}_2}{\partial x^2} + \left(C_{1111}- \frac{C_{1122}^2}{C_{1111}}\right) \frac{\partial^2 \overline{u}_2}{\partial y^2}  + \overline{b}_2.
\end{align}
\end{subequations}
By performing the substitution \eqref{eqn:planestressto2Dcoefficientrelationship} in \eqref{eqn:planestressequationsofmotionclassicalcubic}, we see that up to a change in coefficients, the cubic generalized plane stress equation of motion \eqref{eqn:planestressequationsofmotionclassicalcubic} and the pure two-dimensional square equation of motion \eqref{eqn:eqnmotionsqr}  are equivalent mathematically (if we replace $\bfb$ and $\bfu$ in \eqref{eqn:eqnmotionsqr}  by $\overline{\mathbf{b}}$ and $\overline{\mathbf{u}}$, respectively).

\underline{Isotropic}: Imposing \eqref{eqn:iso3Delastrelations} on \eqref{eqn:planestressequationsofmotionclassical}, we arrive at the following isotropic generalized plane stress model:
\begin{subequations}\label{eqn:planestressequationsofmotionclassicaliso}
\begin{align}
\rho \ddot{\overline{u}}_1 =& \left(C_{1111} - \frac{C_{1122}^2}{C_{1111}} \right)\frac{\partial^2 \overline{u}_1}{\partial x^2} + \frac{1}{2}\left(C_{1111} - C_{1122} \right) \frac{\partial^2 \overline{u}_1}{\partial y^2} + \frac{1}{2} \left( C_{1111} + C_{1122} - 2\frac{C_{1122}^2}{C_{1111}} \right) \frac{\partial^2 \overline{u}_2}{\partial x \partial y} + \overline{b}_1, \\
%%%%%%%%%%%%%%%%%%%%%%%%%%%%%%%%%%%%
\rho \ddot{\overline{u}}_2 =& \frac{1}{2} \left( C_{1111} + C_{1122} - 2 \frac{C_{1122}^2}{C_{1111}} \right) \frac{\partial^2 \overline{u}_1}{\partial x \partial y} +\frac{1}{2}\left(C_{1111} - C_{1122} \right)  \frac{\partial^2 \overline{u}_2}{\partial x^2} + \left(C_{1111}- \frac{C_{1122}^2}{C_{1111}}\right) \frac{\partial^2 \overline{u}_2}{\partial y^2}  + \overline{b}_2.
\end{align}
\end{subequations}
By performing the substitution \eqref{eqn:planestressto2Dcoefficientrelationship} in \eqref{eqn:planestressequationsofmotionclassicaliso}, we see that up to a change in coefficients, the isotropic generalized plane stress equation of motion \eqref{eqn:planestressequationsofmotionclassicaltetra} and the pure two-dimensional isotropic equation of motion \eqref{eqn:eqnmotioniso}  are equivalent mathematically (if we replace $\bfb$ and $\bfu$ in \eqref{eqn:eqnmotioniso}  by $\overline{\mathbf{b}}$ and $\overline{\mathbf{u}}$, respectively).

\begin{remark}
In Table \ref{tab:modelequivalenceplanestrainstressto2D}, we summarize the equivalence, up to a change in coefficients, between the in-plane equations of motion for the classical planar models and the corresponding classical pure two-dimensional models. Note that, unlike in classical plane strain ({\em cf.} Section \ref{sec:classplanestrain}), in classical generalized plane stress, the in-plane equations of motion for the tetragonal and cubic models as well as those for the transversely isotropic and isotropic models are not identical.  
\end{remark}

\begin{table}
\begin{center}
\begin{tabular}{|c|c|}
\hline
\textbf{Pure Two-Dimensional Classical Model} & \textbf{Classical Plane Strain or Stress Model} \\
\hline
Oblique & Monoclinic and Trigonal \\
\hline
Rectangular & Orthotropic \\
\hline
Square & Tetragonal and Cubic \\
\hline
Isotropic & Transversely Isotropic and Isotropic \\
\hline
\end{tabular}
\caption{Model equivalence (up to a change in constants) between classical pure two-dimensional models and classical planar models (provided the three-dimensional elasticity tensors have the form of those in Section~\ref{sec:threedimclassicalelasticity}).}\label{tab:modelequivalenceplanestrainstressto2D}
\end{center}
\end{table}

\begin{remark}
In classical plane stress, Assumption (C$\sigma$\ref{assump:PSs4}) is replaced by the assumption $\sigma_{13} = \sigma_{23} = \sigma_{33} = 0$ throughout the body.  Note in classical generalized plane stress, we analogously obtain $\overline{\sigma}_{13} = \overline{\sigma}_{23} = \overline{\sigma}_{33} = 0$; however, only the assumption $\overline{\sigma}_{33} = 0$ is directly imposed. It turns out the classical plane stress model is equivalent, in terms of engineering constants ({\em cf}. Section \ref{sec:techconsts}), to the classical generalized plane stress model \eqref{eqn:planestressequationsofmotionclassical} if one replaces the displacement field $\bfu$ with its average $\overline{\bfu}$ and the body force density field $\bfb$ with its average $\overline{\bfb}$. This is shown in more detail in Section \ref{sec:techconsts}. Consequently, we only consider classical generalized plane stress in this paper.
\end{remark}

\subsection{Engineering constants}\label{sec:techconsts}
An important consideration is the abuse of notation in the elasticity tensor between the two- and three-dimensional formulations of classical linear elasticity. In particular, even though the in-plane equations of motion for classical plane strain, \eqref{eqn:classicalmonoeqnmotion1} and \eqref{eqn:classicalmonoeqnmotion2}, look identical to the classical pure two-dimensional oblique equations of motion, \eqref{eqn:eqnmotionobla} and \eqref{eqn:eqnmotionoblb}, the coefficients are not equivalent.  For instance, the elasticity constant $C_{1111}$ in two dimensions does not have the same physical meaning as the one in three dimensions. In fact, it can be shown that the classical pure two-dimensional model~\eqref{eqn:eqnmotionobl} is, in some sense, equivalent to the classical generalized plane stress model~\eqref{eqn:planestressequationsofmotionclassical}. To facilitate these comparisons, it is useful to express the elasticity constants in terms of engineering constants \cite{jones2014mechanics,nemeth2011depth,Vannucci:2018:AE}, which we now briefly introduce. 

\textbf{Young's moduli}: 
For a given $i$, under a uniaxial stress $\sigma_{ii}$, we define the Young's modulus
\[
E_i := \frac{\sigma_{ii}}{\varepsilon_{ii}}. 
\]

\textbf{Shear moduli}: Given $i$ and $j$, under a shear stress $\sigma_{ij}$, we define the shear modulus
\[
G_{ij} := \frac{\sigma_{ij}}{2 \varepsilon_{ij}}. 
\]

\textbf{Poisson's ratios}: 
Given $i$ and $j$, under a uniaxial stress $\sigma_{ii}$, we define the Poisson's ratio
\[
\nu_{ij} := -\frac{\varepsilon_{jj}}{\varepsilon_{ii}}.
\] 

\textbf{Chentsov's coefficients}\footnote{As we only have one shear strain component in two dimensions, the Chentsov's coefficients do not play a role in the classical pure two-dimensional model. However, we include them here for the comparisons between the classical pure two-dimensional model and the classical plane strain and classical generalized plane stress models. }: 
Given $i,j,k,$ and $l$, under a shear stress $\sigma_{ij}$, we define the Chentsov's coefficient
\[
\mu_{kl,ij} := \frac{\varepsilon_{kl}}{\varepsilon_{ij}}, \quad \left\{ij\right\} \neq \left\{kl \right\}.
\]

\textbf{Coefficients of mutual influence of the first type}: Given $i,j,$ and $k$, under a shear stress $\sigma_{ij}$, we define the coefficient of mutual influence of the first type
\[
\eta_{kk,ij} := \frac{\varepsilon_{kk}}{2 \varepsilon_{ij}}.
\]

\textbf{Coefficients of mutual influence of the second type}: Given $i,j,$ and $k$, under a uniaxial stress $\sigma_{ii}$, we define the coefficient of mutual influence of the second type
\[
\eta_{jk,ii} := \frac{2 \varepsilon_{jk}}{\varepsilon_{ii}}.
\]

\begin{remark}
Note these coefficients can be divided into two groups, those which are measured under a uniaxial stress: $E_i$, $\nu_{ij}$, and $\eta_{jk,ii}$, and those which are measured under a shear stress: $G_{ij}$, $\mu_{kl,ij}, \eta_{kk,ij}$. 
\end{remark}

To express the elasticity constants in terms of the engineering constants, one typically relates the compliance tensor $\mathbb{S}$ ({\em cf.} \eqref{eqn:strainstressrelation}) to the engineering  constants and then inverts $\mathbb{S}$ to obtain $\mathbb{C}$. In order to relate a component of the compliance tensor $S_{ijkl}$ to the engineering constants, one imposes a single nonzero component of the stress in \eqref{eqn:compliancetensoreqn}, either a uniaxial stress $\sigma_{ii}$ for a given $i$ to obtain
\begin{align}
S_{iiii} = \frac{1}{E_i}, S_{jjii} = - \frac{\nu_{ij}}{E_i}, \text{ and } 2 S_{jkii} = \frac{\eta_{jk,ii}}{E_i},
\end{align}
or a shear stress $\sigma_{ij}$ for a given $i$ and $j$ to obtain
\begin{align}
4 S_{ijij} = \frac{1}{G_{ij}}, 4 S_{klij} = \frac{\mu_{kl,ij}}{G_{ij}}, \text{ and } 2 S_{kkij} = \frac{\eta_{kk,ij}}{G_{ij}}.
\end{align}

Additionally, the major symmetry ($S_{ijkl} = S_{klij}$) of the compliance tensor \eqref{eqn:compliancetensoreqn} implies that not all of these constants are independent. In fact, 
\begin{equation}\label{eqn:techconstrelations2}
\frac{\nu_{ij}}{E_i} = \frac{\nu_{ji}}{E_j}, \frac{\eta_{ij,kk}}{E_{k}} = \frac{\eta_{kk,ij}}{G_{ij}}, \text{ and } \frac{\mu_{ij,kl}}{G_{kl}} = \frac{\mu_{kl,ij}}{G_{ij}}.
\end{equation}

In three dimensions, the strain-stress relation in \eqref{eqn:compliancetensoreqn} can be expressed in terms of engineering constants as
\begin{equation}\label{eqn:compliancetensor3dtechnicalconstants}
\left[
\begin{array}{c}
\varepsilon_{11} \\
\varepsilon_{22} \\
\varepsilon_{33} \\
2\varepsilon_{23} \\
2\varepsilon_{13} \\
2\varepsilon_{12}
\end{array}
\right] =
\left[
\begin{array}{cccccc}
\frac{1}{E_1} & - \frac{\nu_{12}}{E_1} & - \frac{\nu_{13}}{E_1} & \frac{\eta_{23,11}}{E_1} & \frac{\eta_{13,11}}{E_1} & \frac{\eta_{12,11}}{E_1} \\
\cdot & \frac{1}{E_2} & - \frac{\nu_{23}}{E_2} & \frac{\eta_{23,22}}{E_2} & \frac{\eta_{13,22}}{E_2} & \frac{\eta_{12,22}}{E_2} \\
\cdot & \cdot & \frac{1}{E_3} & \frac{\eta_{23,33}}{E_3} & \frac{\eta_{13,33}}{E_3} & \frac{\eta_{12,33}}{E_3} \\
\cdot & \cdot & \cdot & \frac{1}{G_{23}} & \frac{\mu_{13,23}}{G_{23}} & \frac{\mu_{12,23}}{G_{23}} \\
\cdot & \cdot & \cdot & \cdot & \frac{1}{G_{13}} & \frac{\mu_{12,13}}{G_{13}} \\
\cdot & \cdot & \cdot & \cdot & \cdot & \frac{1}{G_{12}}
\end{array}
\right]
\left[ 
\begin{array}{c}
\sigma_{11} \\
\sigma_{22} \\
\sigma_{33} \\
\sigma_{23} \\
\sigma_{13} \\
\sigma_{12}
\end{array}
\right] .
\end{equation}

In classical plane strain and classical generalized plane stress monoclinic symmetry is assumed. Under monoclinic symmetry, we may simplify \eqref{eqn:compliancetensor3dtechnicalconstants} to
\begin{equation}\label{eqn:compliancetensor3dtechnicalconstantsmono}
\left[
\begin{array}{c}
\varepsilon_{11} \\ 
\varepsilon_{22} \\
\varepsilon_{33} \\
2\varepsilon_{12}
\end{array}
\right] =
\left[
\begin{array}{cccc}
\frac{1}{E_1} & - \frac{\nu_{12}}{E_1} & - \frac{\nu_{13}}{E_1} & \frac{\eta_{12,11}}{E_1} \\
\cdot & \frac{1}{E_2} & - \frac{\nu_{23}}{E_2} & \frac{\eta_{12,22}}{E_2} \\
\cdot & \cdot & \frac{1}{E_3}  & \frac{\eta_{12,33}}{E_3} \\
\cdot & \cdot & \cdot  & \frac{1}{G_{12}}
\end{array}
\right]
\left[ 
\begin{array}{c}
\sigma_{11} \\
\sigma_{22} \\
\sigma_{33} \\
\sigma_{12}
\end{array}
\right] \text{ and } 
\left[
\begin{array}{c} 2 \varepsilon_{23} \\ 2 \varepsilon_{13} \end{array}
\right] = \left[ \begin{array}{cc} \frac{1}{G_{23}} & \frac{\mu_{13,23}}{G_{23}} \\ \cdot & \frac{1}{G_{13}} 
\end{array}
\right]
\left[ \begin{array}{c} \sigma_{23} \\ \sigma_{13} \end{array} \right].
\end{equation}
\begin{remark}
In classical plane strain and classical generalized plane stress, $\varepsilon_{23} = \varepsilon_{13} = 0$ and $\overline{\sigma}_{23} = \overline{\sigma}_{13} = 0$, respectively. Immediately from the second system in~\eqref{eqn:compliancetensor3dtechnicalconstantsmono}, we obtain $\varepsilon_{23} = \varepsilon_{13} = \sigma_{23} = \sigma_{13} = 0$ for classical plane strain and $\overline{\varepsilon}_{23} = \overline{\varepsilon}_{13} = \overline{\sigma}_{23} = \overline{\sigma}_{13} = 0$ for classical generalized plane stress.
\end{remark}

By inverting the systems in \eqref{eqn:compliancetensor3dtechnicalconstantsmono}, we obtain expressions for the elasticity constants $C_{ijkl}$ for monoclinic symmetry in terms of the engineering constants:
\begin{equation}\label{eqn:monoclinicInTermsofTechnicalConstants}
\begin{split}
&C_{1111} = -\frac{2 \, E_{3} G_{12} \eta_{12,22} \eta_{12,33} \nu_{23} + E_{3} G_{12} \eta_{12,22}^{2} + E_{2} G_{12} \eta_{12,33}^{2} + E_{3}^{2} \nu_{23}^{2} - E_{2} E_{3}}{E_{2}^{2} E_{3}^{2} G_{12}|\mathbf{A}|}, \\
&C_{1122} = \frac{E_{3} G_{12} \eta_{12,22} \eta_{12,33} \nu_{13} + E_{3} G_{12} \eta_{12,11} \eta_{12,22} - {\left(E_{2} G_{12} \eta_{12,33}^{2} - E_{2} E_{3}\right)} \nu_{12} + {\left(E_{3} G_{12} \eta_{12,11} \eta_{12,33} + E_{3}^{2} \nu_{13}\right)} \nu_{23}}{E_{1} E_{2} E_{3}^{2} G_{12}|\mathbf{A}|}, \\
&C_{1133} = \frac{E_{2} G_{12} \eta_{12,22} \eta_{12,33} \nu_{12} + E_{2} G_{12} \eta_{12,11} \eta_{12,33} - {\left(E_{3} G_{12} \eta_{12,22}^{2} - E_{2} E_{3}\right)} \nu_{13} + {\left(E_{3} G_{12} \eta_{12,11} \eta_{12,22} + E_{2} E_{3} \nu_{12}\right)} \nu_{23}}{E_{1} E_{2}^{2} E_{3} G_{12}|\mathbf{A}|}, \\
&C_{1112} = \frac{E_{3} \eta_{12,11} \nu_{23}^{2} - E_{2} \eta_{12,22} \nu_{12} - E_{2} \eta_{12,33} \nu_{13} - E_{2} \eta_{12,11} - {\left(E_{2} \eta_{12,33} \nu_{12} + E_{3} \eta_{12,22} \nu_{13}\right)} \nu_{23}}{E_{1} E_{2}^{2} E_{3}|\mathbf{A}|},\\
&C_{2222} = -\frac{2 \, E_{3} G_{12} \eta_{12,11} \eta_{12,33} \nu_{13} + E_{3} G_{12} \eta_{12,11}^{2} + E_{1} G_{12} \eta_{12,33}^{2} + E_{3}^{2} \nu_{13}^{2} - E_{1} E_{3}}{E_{1}^{2} E_{3}^{2} G_{12}|\mathbf{A}|}, \\
&C_{2233} = \frac{E_{2} G_{12} \eta_{12,11} \eta_{12,33} \nu_{12} + E_{1} G_{12} \eta_{12,22} \eta_{12,33} + {\left(E_{3} G_{12} \eta_{12,11} \eta_{12,22} + E_{2} E_{3} \nu_{12}\right)} \nu_{13} - {\left(E_{3} G_{12} \eta_{12,11}^{2} - E_{1} E_{3}\right)} \nu_{23}}{E_{1}^{2} E_{2} E_{3} G_{12}|\mathbf{A}|}, \\
&C_{2212} = -\frac{E_{2} \eta_{12,33} \nu_{12} \nu_{13} - E_{3} \eta_{12,22} \nu_{13}^{2} + E_{2} \eta_{12,11} \nu_{12} + E_{1} \eta_{12,22} + {\left(E_{3} \eta_{12,11} \nu_{13} + E_{1} \eta_{12,33}\right)} \nu_{23}}{E_{1}^{2} E_{2} E_{3}|\mathbf{A}| }, \\
&C_{3333} = -\frac{2 \, E_{2} G_{12} \eta_{12,11} \eta_{12,22} \nu_{12} + E_{2} G_{12} \eta_{12,11}^{2} + E_{1} G_{12} \eta_{12,22}^{2} + E_{2}^{2} \nu_{12}^{2} - E_{1} E_{2}}{E_{1}^{2} E_{2}^{2} G_{12}|\mathbf{A}|}, \\
&C_{3312} = \frac{E_{2}^{2} \eta_{12,33} \nu_{12}^{2} - E_{1} E_{2} \eta_{12,33} - {\left(E_{2} E_{3} \eta_{12,22} \nu_{12} + E_{2} E_{3} \eta_{12,11}\right)} \nu_{13} - {\left(E_{2} E_{3} \eta_{12,11} \nu_{12} + E_{1} E_{3} \eta_{12,22}\right)} \nu_{23}}{E_{1}^{2} E_{2}^{2} E_{3}|\mathbf{A}|}, \\
&C_{2323} = \frac{G_{23}^{2}}{G_{23} - G_{13} \mu_{13,23}^{2}}, \\
&C_{2313} = -\frac{G_{13} G_{23} \mu_{1323}}{G_{23} - G_{13} \mu_{13,23}^{2}}, \\
&C_{1313} = \frac{G_{13} G_{23}}{G_{23} - G_{13} \mu_{13,23}^{2}}, \\
&C_{1212} =-\frac{2 \, E_{2} E_{3} \nu_{12} \nu_{13} \nu_{23} + E_{2}^{2} \nu_{12}^{2} + E_{2} E_{3} \nu_{13}^{2} + E_{1} E_{3} \nu_{23}^{2} - E_{1} E_{2}}{E_{1}^{2} E_{2}^{2} E_{3}|\mathbf{A}|},
\end{split} 
\end{equation}
where 
\begin{equation*}
|\mathbf{A}| := \det \left( \left[
\begin{array}{cccc}
\frac{1}{E_1} & - \frac{\nu_{12}}{E_1} & - \frac{\nu_{13}}{E_1} & \frac{\eta_{12,11}}{E_1} \\
\cdot & \frac{1}{E_2} & - \frac{\nu_{23}}{E_2} & \frac{\eta_{12,22}}{E_2} \\
\cdot & \cdot & \frac{1}{E_3}  & \frac{\eta_{12,33}}{E_3} \\
\cdot & \cdot & \cdot  & \frac{1}{G_{12}}
\end{array}
\right] \right).
\end{equation*}
In classical plane stress, one typically supposes $\sigma_{13} = \sigma_{23} = \sigma_{33} = 0$. Imposing this on the systems in \eqref{eqn:compliancetensor3dtechnicalconstantsmono}, we may relate the strains $\varepsilon_{11}, \varepsilon_{22},$ and $\varepsilon_{12}$ to the stresses $\sigma_{11}, \sigma_{22}$, and $\sigma_{12}$ as follows:
\begin{equation}\label{eqn:compliancetensorpure2dtechnicalconstants}
\left[
\begin{array}{c}
\varepsilon_{11} \\
\varepsilon_{22} \\
2 \varepsilon_{12}
\end{array}
\right] =
\left[
\begin{array}{ccc}
\frac{1}{E_1} & - \frac{\nu_{12}}{E_1} & \frac{\eta_{12,11}}{E_1} \\
\cdot & \frac{1}{E_2} & \frac{\eta_{12,22}}{E_2} \\
\cdot & \cdot & \frac{1}{G_{12}}
\end{array}
\right]
\left[ 
\begin{array}{c}
\sigma_{11} \\
\sigma_{22} \\
\sigma_{12}
\end{array}
\right].
\end{equation}
Note that~\eqref{eqn:compliancetensorpure2dtechnicalconstants} is exactly the strain-stress relationship one obtains for classical pure two-dimensional linear  elasticity, i.e., the two-dimensional analogue of~\eqref{eqn:compliancetensor3dtechnicalconstants}. Moreover, since $\sigma_{13}$ and $\sigma_{23}$ are null in classical plane stress, the resulting in-plane equations of motion are identical in terms of engineering constants to the classical pure two-dimensional equation of motion \eqref{eqn:eqnmotionobl}. 

By inverting the system \eqref{eqn:compliancetensorpure2dtechnicalconstants}, we obtain expressions for the elasticity constants $C_{ijkl}$ in terms of the engineering constants for the classical pure two-dimensional model (and the classical plane stress model):
\begin{equation}\label{eqn:2DCijklInTermsOfTechConstants}
\begin{split}
&C_{1111} = \frac{E_2 - G_{12} \eta_{12,22}^{2} }{E_{2}^{2} G_{12}| \mathbf{S}|}, \\
&C_{1122} = \frac{E_{2} \nu_{12} + G_{12} \eta_{12,11} \eta_{12,22}}{E_{1} E_{2} G_{12}|\mathbf{S}|}, \\
&C_{1112} = -\frac{\eta_{12,22} \nu_{12} + \eta_{12,11}}{E_{1} E_{2}|\mathbf{S}|}, \\
&C_{2222} = \frac{E_1 - G_{12} \eta_{12,11}^{2} }{E_{1}^{2} G_{12}|\mathbf{S}|}, \\
&C_{2212} = -\frac{E_{1} \eta_{12,22} + E_{2} \eta_{12,11} \nu_{12}}{E_{1}^{2} E_{2}|\mathbf{S}|}, \\
&C_{1212} = \frac{E_1 - E_{2} \nu_{12}^{2}}{E_{1}^{2} E_{2}|\mathbf{S}|},
\end{split}
\end{equation}
where
\begin{equation*}
|\mathbf{S}|:= \det \left( \left[
\begin{array}{ccc}
\frac{1}{E_1} & - \frac{\nu_{12}}{E_1} & \frac{\eta_{12,11}}{E_1} \\
\cdot & \frac{1}{E_2} & \frac{\eta_{12,22}}{E_2} \\
\cdot & \cdot & \frac{1}{G_{12}}
\end{array}
\right] \right).
\end{equation*}
 
From this analysis, it is clear that the elasticity constants $C_{ijkl}$ appearing in the classical plane strain and classical generalized plane stress models are not equivalent to those in classical pure two-dimensional linear elasticity. However, it is interesting to note that substituting \eqref{eqn:monoclinicInTermsofTechnicalConstants} into the classical generalized plane stress equation of motion \eqref{eqn:planestressequationsofmotionclassical} and substituting \eqref{eqn:2DCijklInTermsOfTechConstants} into the classical pure two-dimensional equation of motion \eqref{eqn:eqnmotionobl} results in the exact same equation in terms of engineering constants (if we replace $\bfb$ and $\bfu$ in \eqref{eqn:eqnmotionobl} by $\overline{\mathbf{b}}$ and $\overline{\mathbf{u}}$, respectively). Further details on engineering constants in three dimensions can be found in \cite{STG2019}.

\subsection{Cauchy's Relations}\label{sec:CauchyRelations}
In the nineteenth century, there was some contention about the quantity of independent constants in the equations of classical linear elasticity. Proponents of the multi-constant theory, such as Green, Stokes, and Thomson, supported a model with twenty-one elasticity constants. Alternatively, proponents of the rari-constant theory, such as Navier, Poisson, and Saint-Venant, supported a model with fifteen elasticity constants. The fifteen elasticity constant formulation, first derived by Cauchy, is obtained by assuming a molecular description of materials based on central forces between pairs of molecules. This fifteen constant formulation places six additional restrictions on the elasticity tensor $\mathbb{C}$:

\begin{equation}\label{eqn:cauchyrelationgeneral}
    C_{ijkl} = C_{ikjl}.
\end{equation}
In \cite{Love1892}, the term Cauchy's relations was coined to refer to these additional symmetries.  A far more detailed and complete historical account of the controversy can be found in \cite{Love1892}.  

Eventually the fifteen constant formulation from the rari-constant theory was shown to be invalid for many materials; however, as we will see later, there are materials whose properties satisfy \eqref{eqn:cauchyrelationgeneral}. Cauchy's relations make appearances in a variety of theories of mechanics. For example, it is well known that a molecular theory of elasticity based on pairwise potentials imposes Cauchy's relations \eqref{eqn:cauchyrelationgeneral} on the corresponding elasticity tensor \cite{Stakgold1950}. Since bond-based peridynamics is similarly based on a pairwise potential formulation, it is perhaps unsurprising that bond-based peridynamics suffers from the same limitations, i.e., it should only be utilized to describe materials with properties satisfying Cauchy's relations. This is most evident when considering isotropic materials where it is well known that Poisson's ratio is restricted to $\frac{1}{4}$ in three dimensions \cite{SILLING2000} and $\frac{1}{3}$ in two dimensions ({\em cf.} Appendix \ref{appendix:poissonratio}). Since this work deals with bond-based peridynamics, a discussion of Cauchy's relations is essential.

In three dimensions, there are six Cauchy's relations:
\begin{equation}\label{eqn:3DCauchyRelations}
\begin{split}
&C_{1212} = C_{1122}, \; C_{1313} = C_{1133}, \;C_{2323} = C_{2233},\\
&C_{1312} = C_{1123}, \; C_{2312} = C_{2213}, \; C_{2313} = C_{3312}.
\end{split}
\end{equation}
With the addition of Cauchy's relations \eqref{eqn:3DCauchyRelations} to the major symmetry~\eqref{eqn:majorsymmetry}  and the minor symmetries~\eqref{eqn:minorsymmetries}, the number of independent elasticity constants reduces from twenty-one to fifteen, and $\mathbb{C}$ becomes a completely symmetric tensor. See \cite{STG2019} for an in-depth discussion of the implications of Cauchy's relations in three dimensions.

In two dimensions, there is a single Cauchy's relation: 
\begin{equation}\label{eqn:cauchyrelation2d}
C_{1122} = C_{1212}. 
\end{equation}
With the addition of Cauchy's relation \eqref{eqn:cauchyrelation2d} to the major symmetry~\eqref{eqn:majorsymmetry}  and the minor symmetries~\eqref{eqn:minorsymmetries}, the number of independent elasticity constants reduces from six to five ({\em cf.} Table \ref{tab:NumberOfElastConstants}), and $\mathbb{C}$ becomes a completely symmetric tensor. In terms of the engineering constants, \eqref{eqn:cauchyrelation2d} can be expressed as ({\em cf.}~\eqref{eqn:2DCijklInTermsOfTechConstants} and~\eqref{eqn:techconstrelations2})
\begin{equation}\label{eqn:CauchyRelationTechConstant}
G_{12} = \displaystyle\frac{E_2 \nu_{12}}{1- \nu_{12}\nu_{21} - \eta_{12,11} \eta_{12,22}}.
\end{equation}
Imposing \eqref{eqn:cauchyrelation2d}  on the two-dimensional elasticity tensor for each of the symmetry classes ({\em cf.} Theorem \ref{thm:symmetryclasses}) results in the following forms for the elasticity tensor:

\begin{subequations}
\begin{align}
&\textbf{Symmetry Group} & \textbf{Elasticity Tensor} \hspace*{0.4in} \nonumber \\[-0.05in]
& & \textbf{(Cauchy's relation imposed)} \nonumber \\[0.05in]
&\text{Oblique} & \left[ \begin{array}{ccc}
C_{1111} & C_{1122} & C_{1112} \\
\cdot & C_{2222} & C_{2212} \\
\cdot & \cdot & C_{1122}
\end{array} \right] \hspace*{0.29in} \label{eqn:ObliqueTensorCauchy} \\[0.08in]
& \text{Rectangular} & \left[ \begin{array}{ccc}
C_{1111} & C_{1122} & 0 \\
\cdot & C_{2222} & 0 \\
\cdot & \cdot & C_{1122}
\end{array} \right] \hspace*{0.29in} \label{eqn:RectangularTensorCauchy} \\[0.08in]
& \text{Square} & \left[ \begin{array}{ccc}
C_{1111} & C_{1122} & 0 \\
\cdot & C_{1111} & 0 \\
\cdot &  \cdot& C_{1122}
\end{array} \right] \hspace*{0.29in} \label{eqn:SquareTensorCauchy} \\[0.08in]
& \text{Isotropic}\footnotemark & \left[ \begin{array}{ccc}
C_{1111} & \frac{1}{3} C_{1111} & 0 \\
\cdot & C_{1111} & 0 \\
\cdot & \cdot & \frac{1}{3} C_{1111}
\end{array} \right]. \hspace*{0.15in}  \label{eqn:IsoTensorCauchy}
\end{align}
\end{subequations}
\footnotetext{In the isotropic case, Cauchy's relation reduces to $C_{1122} = \frac{C_{1111} - C_{1122}}{2}$ and thus $C_{1122} = C_{1212} = \frac{1}{3} C_{1111}$.\label{footnote:CauchyRelation}}

The restrictions each symmetry class imposes on the engineering constants as well as the corresponding Cauchy's relation is summarized in Table \ref{tab:CauchyRelationsTechnicalConstants}. To obtain the engineering constant symmetry restrictions, simply invert each elasticity tensor in Theorem \ref{thm:symmetryclasses} and equate it to the compliance tensor in \eqref{eqn:compliancetensorpure2dtechnicalconstants}. Then, to get the corresponding Cauchy's relation in terms of technical constants, impose the engineering constant symmetry restrictions on \eqref{eqn:CauchyRelationTechConstant}. Additionally, due to Cauchy's relation, we lose one degree of freedom in each symmetry class. The number of independent constants in the elasticity tensor with and without the Cauchy's relation imposed is summarized for each symmetry class in Table \ref{tab:NumberOfElastConstants}.  

\begin{table}[h!]
\begin{center}
\caption{Restrictions on engineering constants by symmetry class in two-dimensional classical linear elasticity.}\label{tab:CauchyRelationsTechnicalConstants}
\begin{tabular}{|l|c|c|}
\hline 
Symmetry Class & Symmetry Restrictions & Cauchy's Relation \\
\hline
\rule{0pt}{4ex}  Oblique &  &  $G_{12} = \displaystyle\frac{E_2 \nu_{12}}{1- \nu_{12}\nu_{21} - \eta_{12,11} \eta_{12,22}}$ \\[0.1in]
\hline
\rule{0pt}{4ex} Rectangular & $\eta_{12,11}=\eta_{12,22} = 0$ & $G_{12} = \displaystyle\frac{E_2 \nu_{12} }{1-\nu_{12} \nu_{21}}$ \\[0.1in]
\hline
\rule{0pt}{4ex} Square & $\eta_{12,11}=\eta_{12,22} = 0, \nu_{12} = \nu_{21}, E_1 = E_2$ & $G_{12} = \displaystyle\frac{E_1 \nu_{12}}{1-\nu_{12}^2}$ \\[0.1in]
\hline
\rule{0pt}{4ex} Isotropic & $\eta_{12,11}=\eta_{12,22} = 0, \nu_{12} = \nu_{21}, E_1 = E_2, G_{12} = \displaystyle\frac{E_1}{2(1+\nu_{12})}$ & $G_{12} = \displaystyle\frac{3 E_1}{8} \text{ or } \nu_{12} = \displaystyle\frac{1}{3}$ \\[0.1in]
\hline
\end{tabular}
\end{center}
\end{table}

\begin{table}[h!]
\begin{center}
\caption{Number of independent constants by symmetry class in two-dimensional classical linear elasticity.}\label{tab:NumberOfElastConstants}
\begin{tabular}{|l|c|c|}
\hline
Symmetry Class & Number of Constants & \begin{tabular}{c} Number of Constants \\ (Cauchy's relation imposed) \end{tabular} \\
\hline
Oblique & 6 & 5 \\
\hline
Rectangular & 4 & 3 \\
\hline
Square & 3 & 2 \\
\hline
Isotropic & 2 & 1 \\
\hline
\end{tabular}
\end{center}
\end{table}

%%%%%%%%%%%%%%%%%%%%%%%%%%%%%%%%%%%%%%%%%%%%%%%%%%%%%%%%%%%%%%%%%%%%%%%%%%%%%%%%%%%%%%%%%%%%%%%%%%%%%%%%%%%%%%

\section{Bond-based peridynamics}\label{sec:LinearElastPeri}

A goal of this paper is to develop bond-based peridynamic analogues of the classical two-dimensional and planar linear elastic anisotropic models. In the bond-based peridynamic theory \cite{SILLING2000}, given a body $\mathcal{B} \subset \mathbb{R}^d$, where $d$ is the dimension, the equation of motion for a material point $\bfx \in \mathcal{B}$ at time $t \geqslant 0$ is given by ({\em cf}. \cite{emmrich2007peridynamic})
\begin{equation}\label{eqn:perieqnmotiongeneral}
\rho(\bfx) \ddot{\bfu}(\bfx,t) = \int_{\mathcal{H}_{\bfx}} \bff(\bfu(\bfx',t),\bfu(\bfx,t),\bfx',\bfx) d V_{\bfx'} + \bfb(\bfx,t),
\end{equation}
where $\rho$ is the mass density, $\ddot{\bfu}$ is the second derivative in time of the displacement field $\bfu$, $\bff$ is the pairwise force function representing the nonlocal interaction a material point $\bfx'$ exerts on the material point $\bfx$, and $\bfb$ is a prescribed body force density field. The neighborhood of $\bfx$, $\mathcal{H}_{\bfx}$, represents the domain over which the material point $\bfx$ is able to directly interact. A material point $\bfx$ is typically assumed to only directly interact with material points in the body $\mathcal{B}$ within some prescribed distance $\delta$, called the peridynamic horizon. In this case\footnote{When interactions are limited to the body, boundary conditions are effectively imposed on the model.\label{footnote:imposingboundary}},
\begin{equation}\label{def:Hx}
\mathcal{H}_{\bfx} = \mathcal{B} \cap B_{\delta}(\bfx),
\end{equation}
where $B_{\delta}(\bfx):= \left\{ \bfx' \in \mathbb{R}^d: \|\bfx' - 
\bfx \| < \delta \right\}$ is the ball in $\mathbb{R}^d$ of radius $\delta$ centered at $\bfx$. In particular, when a material point $\bfx$ is in the bulk of the body, i.e., further than $\delta$ from the boundary of the body~$\partial \mathcal{B}$, $\mathcal{H}_{\bfx} = B_{\delta}(\bfx)$.

For convenience, we introduce the usual shorthand notation $\bfxi := \bfx' - \bfx$ and $\bfeta := \bfu(\bfx',t) - \bfu(\bfx,t)$, which represent the relative position in the undeformed configuration and the relative displacement, respectively, of the material points $\bfx$ and $\bfx'$; we refer to $\bfxi$ as a peridynamic bond.

There are several essential conditions which must be placed upon the pairwise force function $\bff$ in \eqref{eqn:perieqnmotiongeneral}. To ensure 
invariance under rigid body motion, we require
\begin{equation}\label{eqn:periinvunderrigid}
\begin{split}
\bff(\bfu(\bfx',t),\bfu(\bfx,t),\bfx',\bfx) = \bff(\bfeta,\bfx',\bfx), \quad \forall \bfeta,\bfx',\bfx \in \mathbb{R}^d.
\end{split}
\end{equation}
To enforce conservation of linear momentum, we require
\begin{equation}\label{eqn:periinvunderlinearmom}
\begin{split}
\bff(\bfeta,\bfx',\bfx) = - \bff(-\bfeta,\bfx,\bfx'), \quad \forall \bfeta,\bfx',\bfx \in \mathbb{R}^d.
\end{split}
\end{equation}

To guarantee balance of angular momentum, we require
\begin{equation}\label{eqn:peribalanceangmomentum}
( \bfxi + \bfeta) \times \bff(\bfeta,\bfx',\bfx) = \mathbf{0}, \quad \forall \bfeta,\bfx',\bfx \in \mathbb{R}^d,%, \quad \forall \bfx \in \mathbb{R}^d,
\end{equation}
i.e., the pairwise force function acts along the direction of the deformed bond $\bfxi + \bfeta$. We further suppose the material is microelastic, i.e., the pairwise force function $\bff$ derives from a scalar-valued pairwise potential function $w$:
\begin{equation}\label{eqn:microelastic}
\bff(\bfeta,\bfx',\bfx) = \frac{\partial w}{\partial \bfeta}(\bfeta,\bfx',\bfx), \quad \forall \bfeta,\bfx',\bfx \in \mathbb{R}^d.
\end{equation}

In order to obtain a linear bond-based peridynamic model, we suppose a small deformation, in particular $\| \bfeta \| \ll \delta$ \cite{silling2010linearized}, and linearize the pairwise force function $\bff(\bfeta, \cdot, \cdot)$ while holding $\bfx'$ and $\bfx$ fixed to obtain
\begin{equation}\label{eqn:genpairwiseforcefunction}
\bff(\bfeta,\bfx',\bfx) = \bff(\bf0,\bfx',\bfx) + \bfC(\bfx',\bfx) \bfeta,
\end{equation}
where $\bfC(\bfx',\bfx)$ is the second-order micromodulus tensor given by
\begin{equation}\label{eqn:initialCform}
\bfC(\bfx',\bfx):= \frac{\partial \bff}{\partial \bfeta} (\bf0,\bfx',\bfx)
\end{equation}
and terms of order $O(\| \bfeta \|^2)$ have been omitted. Next, we look at the implications of enforcing conditions \eqref{eqn:periinvunderrigid}, \eqref{eqn:periinvunderlinearmom}, \eqref{eqn:peribalanceangmomentum}, and \eqref{eqn:microelastic}  on the linearized pairwise force function \eqref{eqn:genpairwiseforcefunction}. Combining  \eqref{eqn:periinvunderlinearmom} and \eqref{eqn:genpairwiseforcefunction} we obtain
\begin{equation}\label{eqn:Csymmetricfinal}
\bfC(\bfx,\bfx') = \bfC(\bfx',\bfx), \quad \forall \bfx',\bfx \in \mathbb{R}^d.
\end{equation}
Additionally, by considering \eqref{eqn:microelastic} and \eqref{eqn:initialCform}, we deduce
\begin{equation}\label{eqn:CperiSymmetric}
\bfC^\text{T}(\bfx',\bfx) = \bfC(\bfx',\bfx), \quad \forall \bfx',\bfx \in \mathbb{R}^d.
\end{equation}
An immediate consequence of \eqref{eqn:peribalanceangmomentum}, is the existence of a scalar-valued function $F(\bfeta,\bfx',\bfx)$ such that
\begin{equation}\label{eqn:altformfwithscalarF}
\bff(\bfeta,\bfx',\bfx) = ( \bfxi + \bfeta) F(\bfeta,\bfx',\bfx), \quad \forall \bfeta,\bfx',\bfx \in \mathbb{R}^d.
\end{equation} 
Differentiating \eqref{eqn:altformfwithscalarF} with respect to $\bfeta$, we obtain from \eqref{eqn:initialCform} that 
\begin{equation}\label{eqn:Csecondtolastform}
\bfC(\bfx',\bfx) = \bfxi \otimes \frac{\partial F}{\partial \bfeta} (\mathbf{0} ,\bfx',\bfx) + F(\mathbf{0},\bfx',\bfx) \bfI, \quad \forall \bfx',\bfx \in \mathbb{R}^d.
\end{equation}
A necessary and sufficient condition\footnote{Sufficiency is trivial. To see necessary, set $\bfA(\bfx',\bfx) :=\frac{\partial F}{\partial \bfeta} (\mathbf{0} ,\bfx',\bfx)$. Then $\bfxi \otimes \bfA(\bfx',\bfx)$ is symmetric if and only if 
\begin{equation}\label{eqn:symmneccsuffcond}
\xi_i A_j(\bfx',\bfx) = \xi_j A_i(\bfxi,\bfxi), \quad \forall i,j.
\end{equation} 
If $\xi_i = 0$ and $\xi_j \neq 0$ then $A_i(\bfx',\bfx) = 0$ by \eqref{eqn:symmneccsuffcond}. Thus we may suppose there exists a function $a_i(\bfx',\bfx)$ such that $A_i(\bfx',\bfx) = \xi_i a_i(\bfx',\bfx)$ (no summation implied by repeated indices). From \eqref{eqn:symmneccsuffcond} we find that $a_i(\bfx',\bfx) = a_j(\bfx',\bfx)$ for all $i,j$. Consequently there is a scalar-valued function $a(\bfx',\bfx)$ such that $\bfA(\bfx',\bfx) = a(\bfx',\bfx) \bfxi$.} for \eqref{eqn:CperiSymmetric} to be imposed on \eqref{eqn:Csecondtolastform} is the existence of a scalar-valued function $\lambda(\bfx',\bfx)$ such that
\begin{equation}
\bfxi \otimes \frac{\partial F}{\partial \bfeta}(\mathbf{0},\bfx',\bfx) = \lambda(\bfx',\bfx) \bfxi \otimes \bfxi, \quad \forall \bfx',\bfx \in \mathbb{R}^d.
\end{equation}
The general form of the linearized pairwise force function \eqref{eqn:genpairwiseforcefunction} satisfying conditions \eqref{eqn:periinvunderrigid}, \eqref{eqn:periinvunderlinearmom}, \eqref{eqn:peribalanceangmomentum}, and \eqref{eqn:microelastic} is therefore given by
\begin{equation}\label{eqn:generalBBpairwiseforceform}
\bff (\bfeta,\bfx',\bfx) = \bff(\mathbf{0},\bfx',\bfx) + \left[ \lambda(\bfx',\bfx) \bfxi \otimes \bfxi  +F(\bf0,\bfx',\bfx) \bfI \right] \bfeta.
\end{equation}
Henceforth, we refer to $\lambda(\bfx',\bfx)$ as the micromodulus function. Due to \eqref{eqn:Csymmetricfinal}, the micromodulus function has the symmetry
\begin{equation}\label{eqn:lambdasymm}
\lambda(\bfx',\bfx) = \lambda(\bfx,\bfx'), \quad \forall \bfx',\bfx \in \mathbb{R}^d. 
\end{equation}

Commonly in bond-based peridynamics, one supposes the system is pairwise equilibrated in the reference configuration, i.e., $\bff(\mathbf{0},\bfx',\bfx) = \mathbf{0}, \; \forall \bfx',\bfx \in \mathbb{R}^d$. In this case~\eqref{eqn:generalBBpairwiseforceform} becomes ({\em cf.}~\eqref{eqn:altformfwithscalarF})
\begin{equation}\label{eqn:generalBBpairwiseforceformfinal}
\bff(\bfeta,\bfx',\bfx) =  \lambda(\bfx',\bfx) \bfxi \otimes \bfxi  (\bfu(\bfx',t)-\bfu(\bfx,t)).
\end{equation}
Substituting \eqref{eqn:generalBBpairwiseforceformfinal} into the peridynamic equation of motion \eqref{eqn:perieqnmotiongeneral} ({\em cf.} \eqref{eqn:periinvunderrigid}), we obtain the linear bond-based peridynamic equation of motion:
\begin{equation}\label{eqn:linearperieqn}
\rho(\bfx) \ddot{\bfu}(\bfx,t) = \int_{\mathcal{H}_\bfx} \lambda(\bfx',\bfx) \bfxi \otimes \bfxi ( \bfu(\bfx',t) - \bfu(\bfx,t)) d \bfx' + \bfb(\bfx,t).
\end{equation}
Written in component form, \eqref{eqn:linearperieqn} is expressed as (for $i \in \left\{1,\ldots,d \right\}$)
\begin{equation}\label{eqn:linearperieqncomponentform}
\rho(\bfx) \ddot{u}_i(\bfx,t) = \int_{\mathcal{H}_\bfx} \lambda(\bfx',\bfx) \xi_i \xi_k ( u_k(\bfx',t) - u_k(\bfx,t)) d \bfx' + b_i(\bfx,t).
\end{equation}

In order to consider material symmetry classes and facilitate the utilization of~\eqref{eqn:linearperieqncomponentform}, we relate the classical elasticity tensor $\mathbb{C}$ from \eqref{eqn:stressstrainrelation}  to the micromodulus function $\lambda$ appearing in \eqref{eqn:linearperieqncomponentform}. For this purpose, we suppose $\bfu$ is smooth, perform a Taylor expansion, and equate the coefficients of the derivatives of~$\bfu$ in the resulting peridynamic equation of motion with the coefficients of the corresponding derivatives of~$\bfu$ in the classical equation of motion. More formally, by expanding $\bfu(\bfx',t)$ about $\bfx$, we obtain
\begin{equation}\label{eqn:TaylorExpansionuprimeminusu}
u_k(\bfx',t) - u_k(\bfx,t) = \frac{\partial u_k}{\partial x_j}(\bfx,t) \xi_j + \frac{1}{2} \frac{\partial^2 u_k}{\partial x_j \partial x_l}(\bfx,t) \xi_j \xi_l + \cdots, \quad j,l \in \left\{1,\ldots, d \right\}.
\end{equation}
Substituting \eqref{eqn:TaylorExpansionuprimeminusu} into \eqref{eqn:linearperieqncomponentform}, we obtain
\begin{equation}\label{eqn:perieqnmotionTaylor}
\begin{split}
\rho(\bfx) \ddot{u}_i(\bfx,t) =& \int_{\mathcal{H}_\bfx} \lambda(\bfx',\bfx) \xi_i \xi_k \left( \frac{\partial u_k}{\partial x_j}(\bfx,t) \xi_j + \frac{1}{2} \frac{\partial^2 u_k}{\partial x_j \partial x_l}(\bfx,t) \xi_j \xi_l + \cdots \right) d \bfx' + b_i(\bfx,t).
\end{split}
\end{equation}
Equating \eqref{eqn:perieqnmotionTaylor} with \eqref{eqn:eqnofmotionclassical} and comparing coefficients for terms up to second order, we find the following conditions imposed on $\lambda(\bfx',\bfx)$:
\begin{subequations}\label{eqn:lambdacondfullgen}
\begin{align}
0 ={}& \int_{\mathcal{H}_{\bfx}} \lambda(\bfx',\bfx) \xi_i \xi_j \xi_k d \bfx', \label{eqn:lambdacondfullgena} \\
C_{ijkl} ={}& \frac{1}{2} \int_{\mathcal{H}_{\bfx}} \lambda(\bfx',\bfx) \xi_i \xi_j \xi_k \xi_l d \bfx'. \label{eqn:lambdacondfullgenb}
\end{align}
\end{subequations}
Notice that \eqref{eqn:lambdacondfullgenb} can only hold if $\mathbb{C}$ is a completely symmetric tensor,  i.e., in addition to the minor symmetries \eqref{eqn:minorsymmetries} and major symmetry \eqref{eqn:majorsymmetry} of $\mathbb{C}$, Cauchy's relations \eqref{eqn:cauchyrelationgeneral} hold. Consequently, \textit{the linear bond-based peridynamic model~\eqref{eqn:linearperieqn} agrees with the classical linear elasticity model~\eqref{eqn:eqnofmotionclassical}  up to second-order terms only when describing a material satisfying Cauchy's relations}. Throughout the remainder of the paper, we assume Cauchy's relations are imposed on $\mathbb{C}$.

We now define the peridynamic analogue of Definition \ref{def:symmtran}.
\begin{definition}\label{def:generalsymm}
An orthogonal transformation $\mathbf{Q}$ is a symmetry transformation of a micromodulus function $\lambda$ if
\begin{equation}
\lambda(\mathbf{Q} \bfx',\mathbf{Q} \bfx) = \lambda(\bfx', \bfx), \quad \forall \bfx,\bfx' \in \mathbb{R}^d.
\end{equation}
\end{definition} 
Note Definition \ref{def:generalsymm} is a generalization of the symmetry definition introduced in~\cite{SILLING2000} for the case when $\lambda$ is not strictly a function of the bond $\bfxi$. Up to this point, no assumptions have been made on the homogeneity of the material response. With the assumption of homogeneity, further simplifications are possible. In particular, in the bulk of the body, the micromodulus function may be assumed to be a function solely of the bond ({\em cf}. \cite{SILLING2000}):
\begin{equation}\label{eqn:lambhomogeneous}
\lambda(\bfx',\bfx) = \lambda(\bfx'-\bfx) = \lambda(\bfxi), \quad \forall \bfx,\bfx' \in \mathbb{R}^d.
\end{equation}
If we assume a homogeneous material response, for a given material point $\bfx$ in the bulk of the body ($\mathcal{H}_{\bfx} = B_\delta(\bfx)$), the peridynamic equation of motion~\eqref{eqn:linearperieqn} becomes (after the change of variables $\bfx' \rightarrow \bfx + \bfxi$):
\begin{equation}\label{eqn:linearperieqnhomogeneous}
\rho(\bfx) \ddot{\bfu}(\bfx,t) = \int_{B_\delta(\mathbf{0})} \lambda(\bfxi) \bfxi \otimes \bfxi ( \bfu(\bfx + \bfxi,t) - \bfu(\bfx,t)) d \bfxi + \bfb(\bfx,t).
\end{equation}
Written in component form, \eqref{eqn:linearperieqnhomogeneous} becomes 
\begin{equation}\label{eqn:linearperieqncomponentformhomogeneous}
\rho(\bfx) \ddot{u}_i(\bfx,t) = \int_{B_\delta(\mathbf{0})} \lambda(\bfxi) \xi_i \xi_k ( u_k(\bfx+\bfxi,t) - u_k(\bfx,t)) d \bfxi + b_i(\bfx,t).
\end{equation}
Moreover, for a material point in the bulk of a body, \eqref{eqn:lambdacondfullgena} holds trivially by antisymmetry since $\lambda(\bfxi) = \lambda(-\bfxi)$ by \eqref{eqn:lambdasymm}. Additionally, \eqref{eqn:lambdacondfullgenb} reduces to (see \cite{STG2019} and related expressions in \cite{Azdoud2013,seleson2015concurrent})
\begin{equation}\label{eqn:Cperiexpression}
C_{ijkl} = \frac{1}{2} \int_{B_\delta(\bf0)}  \lambda(\bfxi) \xi_i \xi_j \xi_k \xi_l d \bfxi. 
\end{equation}
The following proposition shows that if \eqref{eqn:Cperiexpression} holds, $\mathbb{C}$ inherits the symmetries of $\lambda(\bfxi)$. 

\begin{proposition}[{\em cf}. \cite{STG2019}]\label{prop:symminheritC}
Let $\mathbf{Q}$ be an orthogonal transformation. Assume the micromodulus function satisfies the symmetry property $\lambda(\bfQ \bfxi) = \lambda(\bfxi)$ for all $\bfxi  \in \mathbb{R}^d$. Then, a fourth-order tensor  $\mathbb{C}$ defined by \eqref{eqn:Cperiexpression} is symmetric with respect to $\mathbf{Q}$, i.e., \eqref{eqn:symmtranbasic} holds.
\end{proposition}
Note the converse of Proposition \ref{prop:symminheritC} is not necessarily true for an arbitrarily chosen micromodulus function $\lambda(\bfxi)$. 

In order to develop anisotropic peridynamic models, we provide a specific form for $\lambda$. In \cite{STG2019}, the following micromodulus was proposed:
\begin{equation}\label{def:lambdagenform}
\lambda(\bfxi) =  \frac{1}{m} \frac{ (\bfxi \otimes \bfxi) \mathbbm{\Lambda} (\bfxi \otimes \bfxi) }{ \| \bfxi \|^4} \frac{\omega( \| \bfxi \|)}{ \| \bfxi \|^2},
\end{equation}
where $\mathbbm{\Lambda}$ is a fully symmetric fourth-order tensor, $\omega$ is an influence function, and $m$ is the weighted volume \cite{Silling2007}
\begin{equation}\label{def:weightvolumem}
m := \int_{B_\delta(\bf0)} \omega(\|\bfxi\|) \| \bfxi \|^2 d \bfxi.
\end{equation}
Influence functions in peridynamics are commonly utilized to control the radial dependence of the interaction between material points \cite{seleson2011role,Silling2007}. For convenience, we call $\mathbbm{\Lambda}$ the peridynamic tensor. 

Note we have not selected a dimension for the peridynamic model. For a suitable choice of $\omega$, the theory defined up to this point is valid in $\mathbb{R}^d$ for $d \in \mathbb{N}$. In particular, in $\mathbb{R}^3$, the micromodulus \eqref{def:lambdagenform} is given by \footnote{It is shown in \cite{STG2019} that under assumption \eqref{eqn:Cperiexpression}, the micromodulus \eqref{eqn:lambdatriclinic} in spherical coordinates may equivalently be formulated by replacing the anisotropic portion of \eqref{eqn:lambdatriclinic}, with a fourth-order spherical harmonics expansion. }
\begin{equation}\label{eqn:lambdatriclinic}
\begin{split}
\lambda(\bfxi) = &\frac{1}{m} \frac{\omega(\|\bfxi \|)}{\| \bfxi \|^2} \left[ \frac{\Lambda_{1111} \xi_1^4 + 4 \Lambda_{1112}\xi_1^3\xi_2 + 4 \Lambda_{1113}\xi_1^3\xi_3 + 6 \Lambda_{1122}\xi_1^2\xi_2^2 + 12 \Lambda_{1123}\xi_1^2\xi_2\xi_3 }{\| \bfxi \|^4}  \right. \\
&+ \frac{ 6 \Lambda_{1133}\xi_1^2\xi_3^2 + 4\Lambda_{2212}\xi_1\xi_2^3 + 12\Lambda_{2213}\xi_1\xi_2^2\xi_3 + 12\Lambda_{3312}\xi_1\xi_2\xi_3^2 + 4\Lambda_{3313}\xi_1\xi_3^3      }{\| \bfxi \|^4} \\
&\left. + \frac{\Lambda_{2222}\xi_2^4 + 4\Lambda_{2223}\xi_2^3\xi_3 +6\Lambda_{2233}\xi_2^2\xi_3^2 + 12\Lambda_{3323}\xi_2\xi_3^3  + \Lambda_{3333}\xi_3^4     }{\| \bfxi \|^4} \right],
\end{split}
\end{equation}
in $\mathbb{R}^2$, the micromodulus \eqref{def:lambdagenform} is given by \footnote{One can produce an equivalent formulation by replacing the anisotropic portion of \eqref{eqn:lambdaobl} with a fourth-order Fourier series formulation. In particular, changing to polar coordinates, $(\xi_1,\xi_2) = (r \cos(\theta), r \sin(\theta))$, in \eqref{eqn:lambdaobl} and setting
\begin{equation*}
\begin{split}
\Lambda_{1111} =& a_0 + a_2 + a_4, \\
\Lambda_{1122} =& 2a_0 - 6a_4, \\
\Lambda_{1112} =& 2 b_2 + 4 b_4, \\
\Lambda_{2222} =& a_0 - a_2 + a_4, \\
\Lambda_{2212} =& 2b_2-4b_4,
\end{split}
\end{equation*}
we immediately obtain
\begin{equation}\label{eqn:polarexpansion2}
\lambda(r,\theta) = \frac{1}{m} \left( \sum_{n=0}^2 a_{2n} \cos(2n \theta) + b_{2n} \sin(2n \theta)\right) \frac{\omega \left( r \right)}{r^2}.
\end{equation}
Equation \eqref{eqn:polarexpansion2} is equivalent to
\begin{equation}\label{eqn:polarexpansion1}
\lambda(r,\theta) = \frac{1}{m} \left( \sum_{n=0}^4 a_n \cos(n \theta) + b_n \sin(n \theta)\right) \frac{\omega \left( r \right)}{r^2},
\end{equation}
when invariance with respect to the inversion symmetry $\lambda(r,\theta+\pi) = \lambda(r,\theta)$ is assumed. To see this note that $\sin(n\theta)$ and $\cos(n \theta)$ do not have inversion symmetry when $n$ is odd and thus we must have $a_1=a_3=b_1=b_3=0$.

}

\begin{equation}\label{eqn:lambdaobl}
\lambda(\bfxi) = \frac{1}{m} \frac{\Lambda_{1111} \xi_1^4 + 4\Lambda_{1112} \xi_1^3 \xi_2 + 6 \Lambda_{1122} \xi_1^2 \xi_2^2 + 4 \Lambda_{2212} \xi_1 \xi_2^3 + \Lambda_{2222} \xi_2^4}{\|\bfxi\|^4} \frac{\omega(\| \bfxi \|)}{\| \bfxi \|^2},
\end{equation}
and in $\mathbb{R}^1$, the micromodulus \eqref{def:lambdagenform} is given by
\begin{equation*}
\lambda(\xi) = \frac{\Lambda_{1111}}{m} \frac{\omega(\| \xi \|)}{\| \xi \|^2},
\end{equation*}
% \eqref{def:lambdagenform} reduces to $\lambda(\xi) = \frac{\Lambda}{m} \frac{\omega(\| \xi \|)}{\| \xi \|^2}$, where $\Lambda$ is a scalar. We may express the micromodulus \eqref{def:lambdagenform} in $\mathbb{R}^3$ as
where $\Lambda_{ijkl}$ are the components of the fourth-order peridynamic tensor $\mathbbm{\Lambda}$.

Analogously to classical linear elasticity, where the fourth-order elasticity tensor $\mathbb{C}$ can be represented by a second-order tensor $\mathbf{C}$ ({\em cf.} \eqref{eqn:3DElastTensor} and \eqref{eqn:2DElastTensor}), we may utilize Voigt notation %\cite{voigt1928lehrbuch} 
to express the fourth-order peridynamic tensor $\mathbbm{\Lambda}$ as a second-order symmetric tensor $\mathbf{\Lambda}$. In this case, the anisotropic part of \eqref{eqn:lambdatriclinic} can be written as
\begin{equation}
 (\bfxi \otimes \bfxi) \mathbbm{\Lambda} (\bfxi \otimes \bfxi)  = \left[
\begin{array}{c}
\xi_1^2 \\
\xi_2^2 \\
\xi_3^2 \\
2 \xi_2 \xi_3 \\
2 \xi_1 \xi_3 \\
2 \xi_1 \xi_2
\end{array}
\right]^T \left[
\begin{array}{cccccc}
\Lambda_{1111} & \Lambda_{1122} & \Lambda_{1133} & \Lambda_{1123} & \Lambda_{1113} & \Lambda_{1112} \\
\cdot & \Lambda_{2222} & \Lambda_{2233} & \Lambda_{2223} & \Lambda_{2213} & \Lambda_{2212} \\
\cdot & \cdot & \Lambda_{3333} & \Lambda_{3323} & \Lambda_{3313} & \Lambda_{3312} \\
\cdot&\cdot &\cdot & \Lambda_{2233} & \Lambda_{3312} & \Lambda_{2213} \\
\cdot& \cdot& \cdot& \cdot& \Lambda_{1133} & \Lambda_{1123} \\
\cdot& \cdot& \cdot& \cdot& \cdot & \Lambda_{1122}
\end{array}
\right]
\left[
\begin{array}{c}
\xi_1^2 \\
\xi_2^2 \\
\xi_3^2 \\
2 \xi_2 \xi_3 \\
2 \xi_1 \xi_3 \\
2 \xi_1 \xi_2
\end{array}
\right]
\end{equation}
and the anisotropic part of \eqref{eqn:lambdaobl} can be written as
\begin{equation}\label{eqn:gen2Dmicromodulustensordescription}
(\bfxi \otimes \bfxi) \mathbbm{\Lambda} (\bfxi \otimes \bfxi) = \left[
\begin{array}{c}
\xi_1^2 \\
\xi_2^2 \\
2 \xi_1 \xi_2
\end{array}
\right]^T \left[
\begin{array}{ccc}
\Lambda_{1111} & \Lambda_{1122} & \Lambda_{1112} \\
\cdot & \Lambda_{2222} &  \Lambda_{2212} \\
\cdot & \cdot & \Lambda_{1122}  
\end{array}
\right] 
\left[
\begin{array}{c}
\xi_1^2 \\
\xi_2^2 \\
2 \xi_1 \xi_2
\end{array}
\right].
\end{equation}

The micromodulus function \eqref{def:lambdagenform} has many desirable properties such as:
\begin{itemize}
\item It can be informed by the classical elasticity tensor through~\eqref{eqn:Cperiexpression}.
\item Its peridynamic tensor $\mathbbm{\Lambda}$ and the elasticity tensor $\mathbb{C}$ (with Cauchy's relations imposed) have the same number of degrees of freedom.
\item When it is related to the classical elasticity tensor $\mathbb{C}$ through~\eqref{eqn:Cperiexpression}, it has the same symmetries as $\mathbb{C}$ with respect to  Definitions~\ref{def:symmtran} and~\ref{def:generalsymm} ({\em cf.}~ Proposition \ref{prop:lambinheritsymm}  for two dimensions and \cite{STG2019} for three dimensions). 
\end{itemize}

We now have sufficient background to develop pure two-dimensional, plane strain, and plane stress anisotropic linear bond-based peridynamic models.

%%%%%%%%%%%%%%%%%%%%%%%%%%%%%%%%%%%%%%%%%%%%%%%%%%%%%%%%%%%%%%%%%%%%%%%%%%%%%%%%%%%%%%%%%%%%%%%%%%%%%%%%%%%%%%

\subsection{Pure two-dimensional bond-based peridynamics}\label{sec:TwoDPeridynamicModels}
In this section, we look at bond-based peridynamic counterparts of the pure two-dimensional classical linear elasticity models presented in Section \ref{sec:puretwodimclass}. We consider the bulk of a body and assume a homogeneous linear elastic material response, so that the two-dimensional peridynamic model is given by~\eqref{eqn:linearperieqnhomogeneous}, which in component form is:
\begin{subequations}\label{eqn:puretwocomponentmodel}
\begin{align}
\rho(\bfx) \ddot{u}_1(\bfx,t) = \int_{B_\delta(\mathbf{0})} \lambda(\bfxi) \xi_1 \left[ \xi_1 ( u_1(\bfx+\bfxi,t) - u_1(\bfx,t)) + \xi_2 ( u_2(\bfx+\bfxi,t) - u_2(\bfx,t)) \right] d \bfxi + b_1(\bfx,t), \\
\rho(\bfx) \ddot{u}_2(\bfx,t) = \int_{B_\delta(\mathbf{0})} \lambda(\bfxi) \xi_2 \left[ \xi_1 ( u_1(\bfx+\bfxi,t) - u_1(\bfx,t)) + \xi_2 ( u_2(\bfx+\bfxi,t) - u_2(\bfx,t)) \right] d \bfxi + b_2(\bfx,t). 
\end{align}
\end{subequations}

We further suppose the micromodulus function $\lambda(\bfxi)$ is described by \eqref{eqn:lambdaobl}. We relate $\Lambda_{ijkl}$ to $C_{ijkl}$ through relation \eqref{eqn:Cperiexpression} to obtain
\begin{subequations}\label{eqn:SijkltoCijkl}
\begin{align}
\Lambda_{{1111}} &= 10 C_{{1111}}-20C_{{1122}} +2 C_{{2222}}, \\
\Lambda_{{1112}} &= 20C_{{1112}} - 12 C_{{2212}}, \label{eqn:SijkltoCijklS1112} \\
\Lambda_{{1122}} &= \frac{1}{3} \left( -10C_{{1111}}+76 C_{{1122}}-10C_{{2222}} \right), \\
\Lambda_{{2212}} &= -12 C_{{1112}}+20 C_{{2212}}, \label{eqn:SijkltoCijklS2212} \\
\Lambda_{{2222}} &= 2 C_{{1111}}-20C_{{1122}}+10C_{{2222}}.
\end{align}
\end{subequations} 
%Analogously to classical linear elasticity, where the fourth-order elasticity tensor $\mathbb{C}$ can be represented by a second-order tensor $\mathbf{C}$ ({\em cf.} \eqref{eqn:tritensor}), we may express the fourth-order 
We can relate the peridynamic tensor $\mathbbm{\Lambda}$ to the engineering constants by substituting \eqref{eqn:2DCijklInTermsOfTechConstants} into \eqref{eqn:SijkltoCijkl}.
\begin{remark}
It is interesting to observe the system \eqref{eqn:SijkltoCijkl} can be expressed as two decoupled subsystems of equations:
\begin{equation}\label{eqn:partitionlambdatensor}
\left[ 
\begin{array}{c}
\Lambda_{1111} \\
\Lambda_{1122} \\
\Lambda_{2222}
\end{array}
\right] =
\left[
\begin{array}{ccc}
10 & -20 & 2 \\
-\frac{10}{3} & \frac{76}{3} & -\frac{10}{3} \\
2 & -20 & 10
\end{array}
\right]
\left[ 
\begin{array}{c}
C_{1111} \\
C_{1122} \\
C_{2222}
\end{array}
\right] \quad \text{and} \quad 
\left[ 
\begin{array}{c}
\Lambda_{1112} \\
\Lambda_{2212}
\end{array}
\right] = \left[
\begin{array}{cc}
20 & -12 \\
-12 & 20
\end{array}
\right]
\left[ 
\begin{array}{c}
C_{1112} \\
C_{2212}
\end{array}
\right].
\end{equation}
\end{remark}

The most general peridynamic tensor $\mathbbm{\Lambda}$, the oblique peridynamic tensor, may be represented as ({\em cf.} \eqref{eqn:gen2Dmicromodulustensordescription}): 
\begin{align}\label{eqn:Stensorobl}
\mathbf{\Lambda} =  \left[
\begin{array}{ccc}
\Lambda_{1111} & \Lambda_{1122} & \Lambda_{1112} \\
\cdot & \Lambda_{2222} & \Lambda_{2212} \\
\cdot & \cdot & \Lambda_{1122}
\end{array}
\right].
\end{align}

In Proposition \ref{prop:lambinheritsymm} we prove the converse of Proposition \ref{prop:symminheritC} when the micromodulus function $\lambda(\bfxi)$ is given by \eqref{def:lambdagenform}. For an analogous proof in three dimensions see \cite{STG2019}.
\begin{proposition}\label{prop:lambinheritsymm} 
Let $\lambda(\bfxi)$ be given by \eqref{def:lambdagenform} and suppose \eqref{eqn:Cperiexpression} holds. If $\mathbb{C}$ is invariant with respect to one of the four symmetry groups in two-dimensional classical linear elasticity ({\em cf.} Theorem \ref{thm:symmetryclasses}), then $\lambda(\bfxi)$ is also invariant under the symmetry group transformations ({\em cf.} Definition \ref{def:generalsymm}). 
\end{proposition}

\begin{proof}
Let $\lambda(\bfxi)$ be given by \eqref{def:lambdagenform} and suppose \eqref{eqn:Cperiexpression} holds.
Recall that orthogonal transformations preserve length. Hence $\|\bfxi\|^6$ and $\omega(\| \bfxi \|)$ are always invariant under orthogonal transformations. Therefore, we only need to consider the anisotropic part of $\lambda(\bfxi)$ ({\em cf.} \eqref{eqn:gen2Dmicromodulustensordescription}), when showing $\lambda(\bfxi)$ is invariant with respect to a given orthogonal transformation. 

\textbf{Oblique}: Suppose $\mathbb{C}$ is given by \eqref{eqn:ObliqueTensorCauchy}. The corresponding symmetry group is generated by $\left\{ -\bfI \right\}$ ({\em cf.}~Lemma~\ref{lemma:reflections}), i.e., $\xi_1 \rightarrow -\xi_1$ and $\xi_2 \rightarrow -\xi_2$. Since the anisotropic part of $\lambda(\bfxi)$ is a sum of fourth-order monomials of components of~$\bfxi$ ({\em cf.}~\eqref{eqn:lambdaobl}), we clearly have $\lambda(\bfxi)$ is symmetric with respect to $-\mathbf{I}$.

\textbf{Rectangular}: Suppose $\mathbb{C}$ is given by \eqref{eqn:RectangularTensorCauchy}. Then, by \eqref{eqn:SijkltoCijkl} we have ({\em cf.}~\eqref{eqn:recelastrelations})

\begin{subequations}\label{eqn:recsymmSijklrelations}
\begin{align}
&\Lambda_{{1111}}= 10 C_{{1111}}-20C_{{1122}} +2 C_{{2222}}, \\
&\Lambda_{{1112}}= 0, \\
&\Lambda_{{1122}}= \frac{1}{3} \left( -10C_{{1111}}+76 C_{{1122}}-10C_{{2222}} \right), \\
&\Lambda_{{2212}}= 0, \\
&\Lambda_{{2222}}= 2 C_{{1111}}-20C_{{1122}}+10C_{{2222}},
\end{align}
\end{subequations}
and
\begin{equation}\label{eqn:lambdarec}
\lambda(\bfxi) = \frac{1}{m} \frac{\Lambda_{1111} \xi_1^4 + 6 \Lambda_{1122} \xi_1^2 \xi_2^2 + \Lambda_{2222} \xi_2^4}{\|\bfxi\|^4} \frac{\omega(\| \bfxi \|)}{\| \bfxi \|^2}.
\end{equation}

The rectangular peridynamic tensor $\mathbbm{\Lambda}$ may be represented as ({\em cf.} \eqref{eqn:Stensorobl})
\begin{align}\label{eqn:Stensorrec}
\mathbf{\Lambda} =  \left[
\begin{array}{ccc}
\Lambda_{1111} & \Lambda_{1122} & 0 \\
\cdot & \Lambda_{2222} & 0 \\
\cdot & \cdot & \Lambda_{1122}
\end{array}
\right].
\end{align}

The rectangular symmetry group is generated by $\left\{ -\mathbf{I}, \mathbf{Ref}\left( 0 \right) \right\}$ ({\em cf.}~Lemma~\ref{lemma:reflections}). From the oblique portion of the proof, we already know $\lambda(\bfxi)$ is symmetric with respect to $-\mathbf{I}$. The transformation $\mathbf{Ref}\left( 0 \right)$ implies $\xi_2 \rightarrow -\xi_2$ ({\em cf.} \eqref{eqn:Ref0Mat}). Since the anisotropic part of \eqref{eqn:lambdarec} is even in $\xi_2$, the invariance with respect to $\mathbf{Ref}\left( 0 \right)$ follows. Note that the functional form of $\lambda(\bfxi)$ is also invariant under the transformation $\xi_1 \rightarrow -\xi_1$, which is expected since this transformation is given by $\mathbf{Ref} ( \frac{\pi}{2} ) = -\bfI  \, \mathbf{Ref}(0)$.

\textbf{Square}: Suppose $\mathbb{C}$ is given by \eqref{eqn:SquareTensorCauchy}. Then, by \eqref{eqn:SijkltoCijkl} we have ({\em cf.}~\eqref{eqn:sqrelastrelations})
\begin{subequations}\label{eqn:sqrsymmSijklrelations}
\begin{align}
&\Lambda_{{1111}}= 12 C_{{1111}}-20C_{{1122}}, \label{eqn:sqrsymmSijklrelationsS1111} \\
&\Lambda_{{1112}}= 0, \\
&\Lambda_{{1122}}= \frac{1}{3} \left( -20C_{{1111}}+76 C_{{1122}} \right), \label{eqn:sqrsymmSijklrelationsS1122} \\
&\Lambda_{{2212}}= 0, \\
&\Lambda_{{2222}}=\Lambda_{1111}, 
\end{align}
\end{subequations}
and 
\begin{equation}\label{eqn:lambdasqr}
\lambda(\bfxi) = \frac{1}{m} \frac{\Lambda_{1111} (\xi_1^4+\xi_2^4) + 6 \Lambda_{1122} \xi_1^2 \xi_2^2}{\|\bfxi\|^4} \frac{\omega(\| \bfxi \|)}{\| \bfxi \|^2}.
\end{equation}
The square peridynamic tensor $\mathbbm{\Lambda}$ may be represented as ({\em cf.} \eqref{eqn:Stensorobl})
\begin{align}\label{eqn:Stensorsqr}
\mathbf{\Lambda} = \left[
\begin{array}{ccc}
\Lambda_{1111} & \Lambda_{1122} & 0 \\
\cdot & \Lambda_{1111} & 0 \\
\cdot & \cdot & \Lambda_{1122}
\end{array}
\right].
\end{align}

The square symmetry group is generated by $\left\{ -\mathbf{I}, \mathbf{Ref}\left( 0 \right), \mathbf{Ref}\left( \frac{\pi}{4} \right) \right\}$ ({\em cf.}~Lemma~\ref{lemma:reflections}). From the rectangular portion of the proof, we already know $\lambda(\bfxi)$ is symmetric with respect to $\left\{ -\mathbf{I}, \mathbf{Ref}\left( 0 \right) \right\}$. The transformation $\mathbf{Ref}\left( \frac{\pi}{4} \right)$ implies $\xi_1 \rightarrow \xi_2$ and $\xi_2 \rightarrow \xi_1$ ({\em cf.} \eqref{eqn:RefMat}). We note \eqref{eqn:lambdasqr} is clearly invariant when interchanging $\xi_1$ and $\xi_2$. 

\textbf{Isotropic}: Suppose that $\mathbb{C}$ is given by \eqref{eqn:IsoTensorCauchy}. Then, by \eqref{eqn:SijkltoCijkl} we have ({\em cf.}~\eqref{eqn:isoelastrelations})

\begin{subequations}\label{eqn:isosymmSijklrelations}
\begin{align}
&\Lambda_{{1111}}=\frac{16}{3} C_{1111},  \label{eqn:isosymmSijklrelationsS1111} \\
&\Lambda_{{1112}}= 0, \\
&\Lambda_{{1122}}= \frac{1}{3} \Lambda_{1111}, \\
&\Lambda_{{2212}}= 0, \\
&\Lambda_{{2222}}=\Lambda_{1111}, 
\end{align}
\end{subequations}
and
\begin{equation}\label{eqn:lambdaiso}
\lambda(\bfxi) = \frac{1}{m} \frac{\Lambda_{1111} (\xi_1^4 + 2 \xi_1^2 \xi_2^2 +\xi_2^4)}{\|\bfxi\|^4} \frac{\omega(\| \bfxi \|)}{\| \bfxi \|^2} = \frac{\Lambda_{1111}}{m} \frac{\omega(\| \bfxi \|)}{\| \bfxi \|^2}. 
\end{equation}
The isotropic peridynamic tensor $\mathbbm{\Lambda}$ may be represented as ({\em cf.} \eqref{eqn:Stensorobl})
\begin{align}\label{eqn:Stensoriso}
\mathbf{\Lambda} =  \left[
\begin{array}{ccc}
\Lambda_{1111} & \frac{1}{3}\Lambda_{1111} & 0 \\
\cdot & \Lambda_{1111} & 0 \\
\cdot & \cdot & \frac{1}{3}\Lambda_{1111} 
\end{array}
\right].
\end{align}
To obtain \eqref{eqn:isosymmSijklrelations}, note Cauchy's relation in the isotropic setting implies $C_{1122} = \frac{1}{3}C_{1111}$ ({\em cf.} Footnote \ref{footnote:CauchyRelation}). Notice~\eqref{eqn:lambdaiso} is invariant under any orthogonal transformation of $\bfxi$ as $\lambda(\bfxi)$ is only dependent on $\| \bfxi \|$, which is preserved by orthogonal transformations. \qed
\end{proof}

\begin{remark}
It is interesting to note that the peridynamic tensor $\mathbbm{\Lambda}$ has the same form as the elasticity tensor $\mathbb{C}$ for each of the symmetry classes when Cauchy's relation $C_{1122} = C_{1212}$ is imposed. This can be readily seen by comparing \eqref{eqn:Stensorobl}, \eqref{eqn:Stensorrec}, \eqref{eqn:Stensorsqr}, and \eqref{eqn:Stensoriso} with \eqref{eqn:ObliqueTensorCauchy}, \eqref{eqn:RectangularTensorCauchy}, \eqref{eqn:SquareTensorCauchy},  and \eqref{eqn:IsoTensorCauchy}, respectively. 
\end{remark}

\subsubsection{Visualization of kernels}
In this section, we study the angular dependence of the two-dimensional  micromodulus function in each of the four symmetry classes. We do this by plotting the angular portion of $\lambda(\bfxi)$, which we define as
\begin{equation}\label{eqn:AngularPartLambda}
\hat{\gamma}(\bfxi) := \frac{(\bfxi \otimes \bfxi) \mathbbm{\Lambda} (\bfxi \otimes \bfxi)}{\| \bfxi \|^4}.
\end{equation}

If we change to polar coordinates, $\xi_1 = r \cos(\theta)$ and $\xi_2 = r \sin(\theta)$, notice~\eqref{eqn:AngularPartLambda} is independent of the radial component $r = \| \bfxi \|$, and so one may write $\gamma(\theta):=\hat{\gamma}(\bfxi)$, where $\theta$ is the angle that $\bfxi$ makes with the positive $x$-axis. To provide a visualization of the variation in $\gamma$ with respect to this angle~$\theta$, we present plots in polar coordinates $(\theta,s)$ with $s = \gamma(\theta)$. Throughout this section, we normalize the elasticity constant $C_{1111}$ to $1$, and we consider multiples of $C_{1111}$ for the other elasticity constants.

\textbf{Isotropic}: As one would expect, for isotropic symmetry we obtain a circle (i.e., $\gamma \equiv \frac{16}{3} C_{1111}$ is independent of the bond orientation ({\em cf.} \eqref{eqn:isosymmSijklrelations})) as can be seen in Figure \ref{fig:isokernel}.

\begin{figure}
\begin{center}
\includegraphics[scale=0.25]{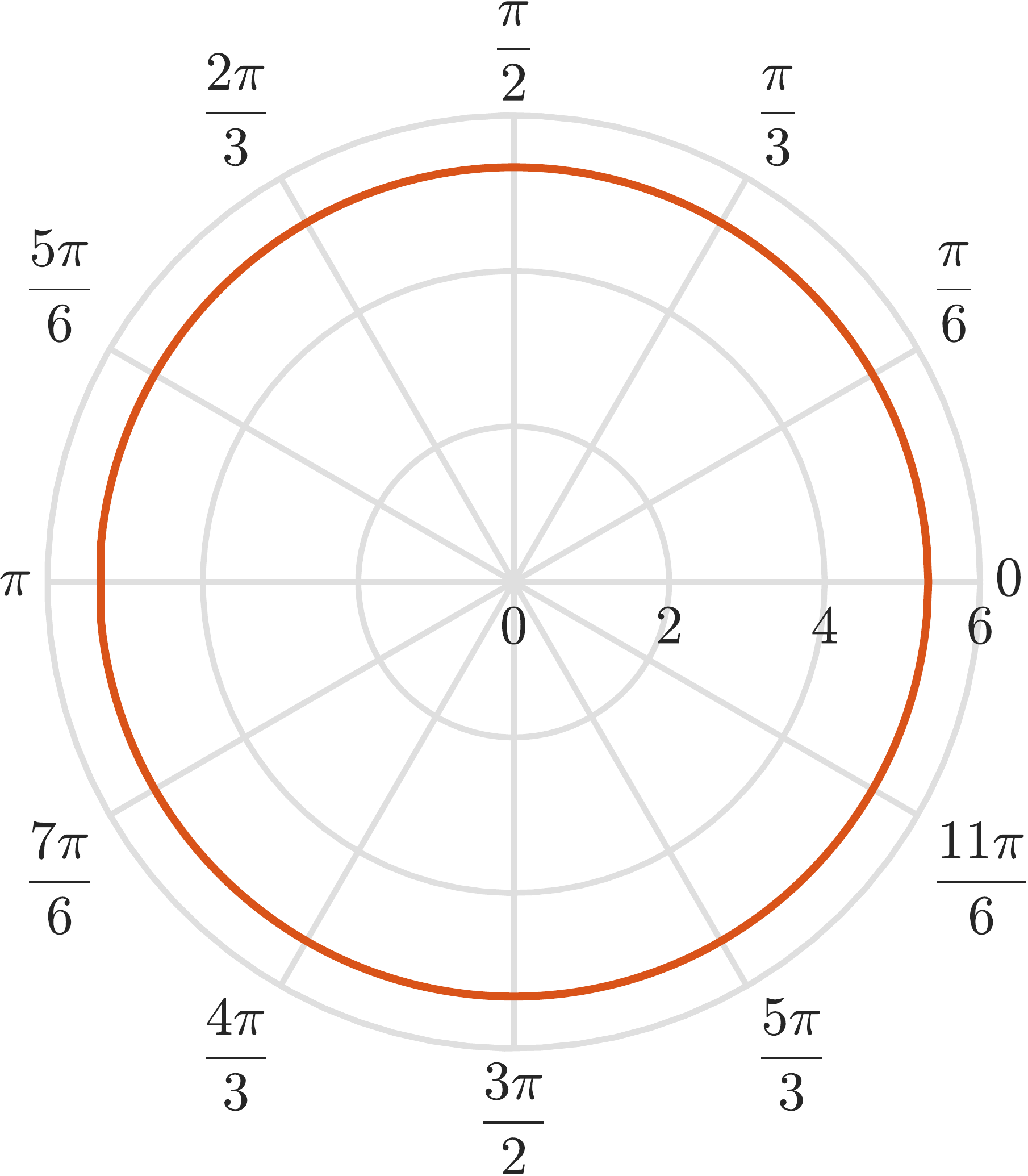}
\caption{Isotropic $\gamma(\theta)$ with $C_{1111} = 1$.}
\label{fig:isokernel}
\end{center}
\end{figure}

\textbf{Square}: For square symmetry, we have two independent constants $C_{1111}$ and $C_{1122}$ ({\em cf.}~\eqref{eqn:sqrsymmSijklrelations}). By varying the ratio between $C_{1122}$ and $C_{1111}$, one is able to increase or decrease the influence of $\gamma$ in the $\theta = \frac{k \pi}{2}$ directions (where $\xi_1 = 0$ or $\xi_2 = 0$ and $\gamma(\theta) = 12 C_{1111} - 20 C_{1122}$) while inversely influencing in the $\theta = \frac{(2k+1)\pi}{4}$ directions (where $\xi_1 = \xi_2$ and $\gamma(\theta) = -4C_{1111} + 28 C_{1122}$), where $k \in \mathbb{Z}$. In Figure \ref{fig:sqrkernel} we vary $C_{1122}$ to observe this dependence. When $C_{1122} = \frac{7}{27} < \frac{1}{3}$ (in solid blue), we see $\gamma$ favors the $\theta = \frac{k \pi}{2}$ directions over the $\theta = \frac{(2k+1) \pi}{4}$ directions, where $k \in \mathbb{Z}$. When $C_{1122} = \frac{1}{3}$ (in dash-dotted red), we have the isotropic case ($C_{1122} = \frac{1}{3} C_{1111}$). When $C_{1122} = \frac{11}{27} > \frac{1}{3}$ (in dotted yellow), we see $\gamma$ favors the $\theta = \frac{(2k+1) \pi}{4}$ directions over the $\theta = \frac{k \pi}{2}$ directions, where $k \in \mathbb{Z}$. 

\begin{figure}
\begin{center}
\begin{tabular}{c}
\includegraphics[scale=0.25,clip]{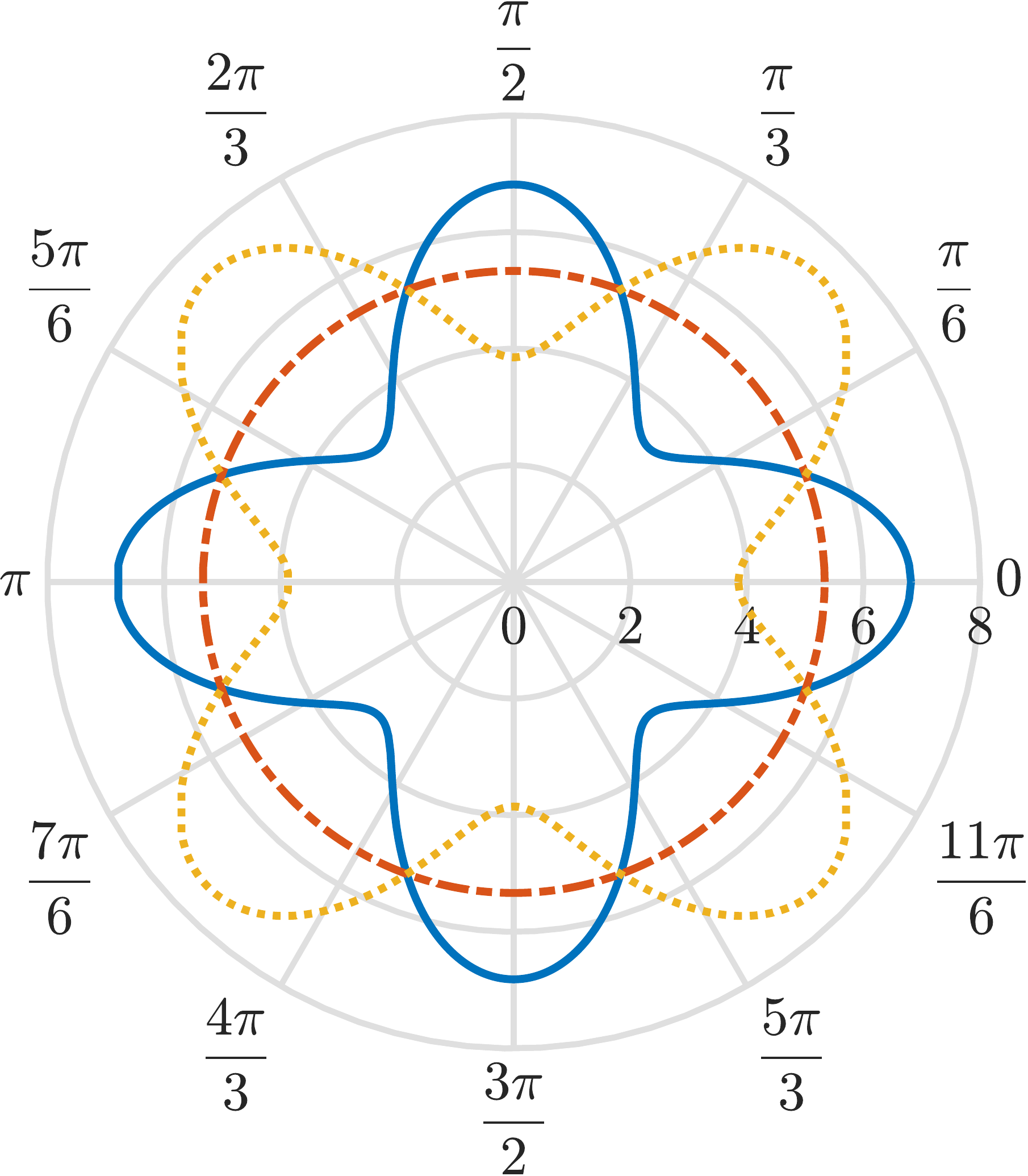}% \\
%$C_{1122} = \frac{7}{27}, \frac{1}{3},\frac{11}{27}$ 
\end{tabular}
\caption{Square $\gamma(\theta)$ with $C_{1111} = 1$ for $C_{1122} = \frac{7}{27}$ (solid blue), $\frac{1}{3}$ (dash-dotted red), $\frac{11}{27}$ (dotted yellow).}
\label{fig:sqrkernel}
\end{center}
\end{figure}

\textbf{Rectangular}: For rectangular symmetry, we have, in addition to $C_{1111}$ and $C_{1122}$, a third independent constant, $C_{2222}$ ({\em cf.} \eqref{eqn:recsymmSijklrelations}). By varying the ratio of $C_{1111}$ and $C_{2222}$ we are able to increase or decrease the influence of $\gamma$ in the $\theta = \frac{(2k+1)\pi}{2}$ directions (where $\xi_1 = 0$ and $\gamma(\theta) = 2C_{1111} - 20 C_{1122} + 10C_{2222}$) relative to the influence of $\gamma$ in the $\theta = k\pi$ directions (where $\xi_2 = 0$ and $\gamma(\theta) = 10 C_{1111} - 20 C_{1122} + 2 C_{2222}$), where $k \in \mathbb{Z}$. In Figure \ref{fig:reckernel}, we have three plots corresponding to $C_{1122} = \frac{7}{27}$ (left plot), $\frac{1}{3}$ (middle plot), and $\frac{11}{27}$ (right plot), which are the three cases we considered in the square case. For each of these cases of $C_{1122}$, we consider $C_{2222} = \frac{3}{5}$ (in solid blue), $1$ (in dash-dotted red), and $\frac{7}{5}$ (in dotted yellow). In each image in Figure \ref{fig:reckernel}, we see that when $C_{2222} = \frac{3}{5} < C_{1111}$, $\gamma$ favors the $\theta = k \pi$ directions over the $\theta = \frac{(2k+1)\pi}{2}$ directions. When $C_{2222} = 1 = C_{1111}$, there is no preference between the axes of the plot. Lastly, when $C_{2222} = \frac{7}{5} > C_{1111}$, we see the $\theta = \frac{(2k+1)\pi}{2}$ directions are favored over the $\theta = k \pi$ directions.   

\begin{figure}
\begin{center}
\begin{tabular}{ccc}
\includegraphics[scale=0.25,clip]{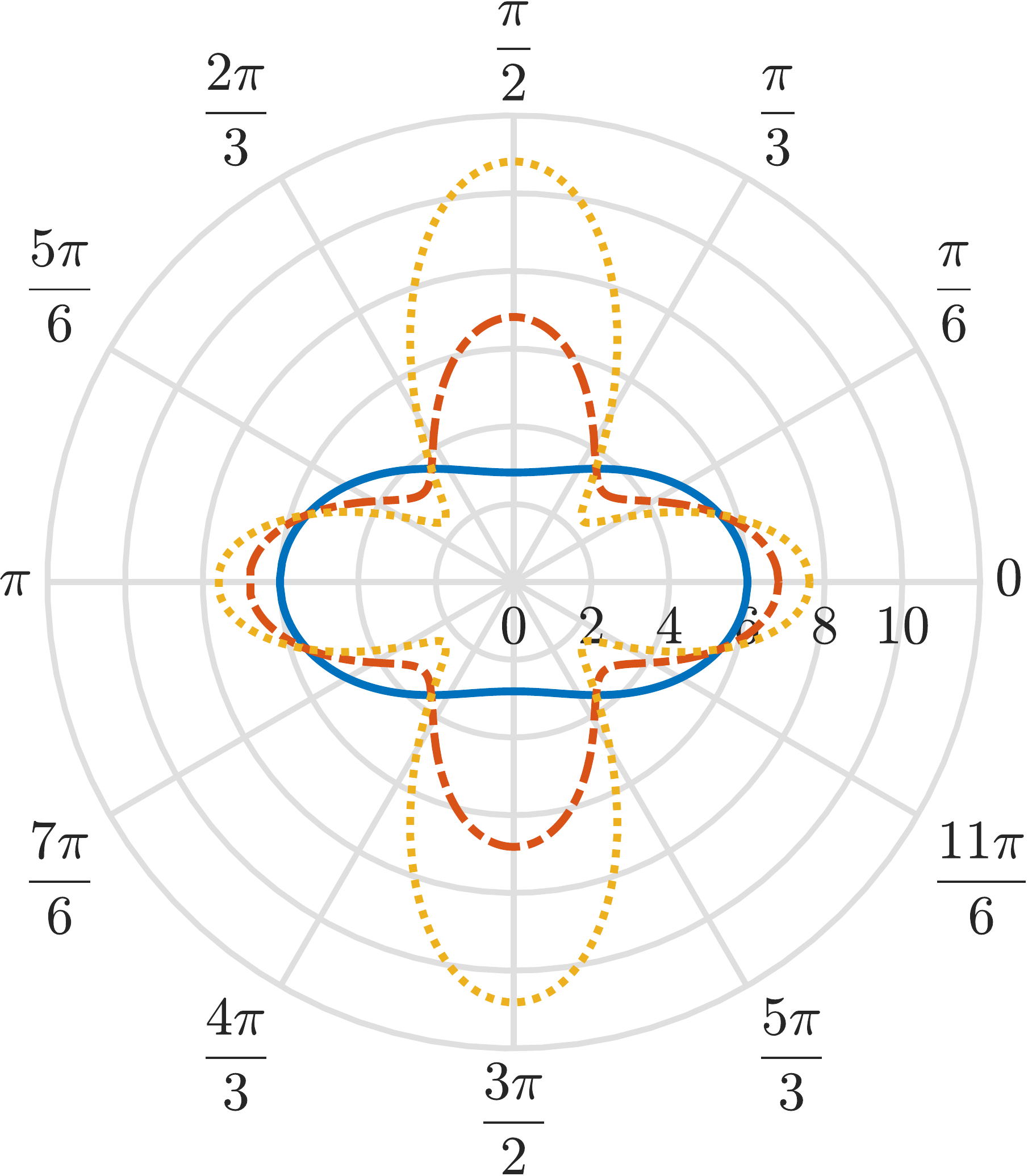}
&
\includegraphics[scale=0.25,clip]{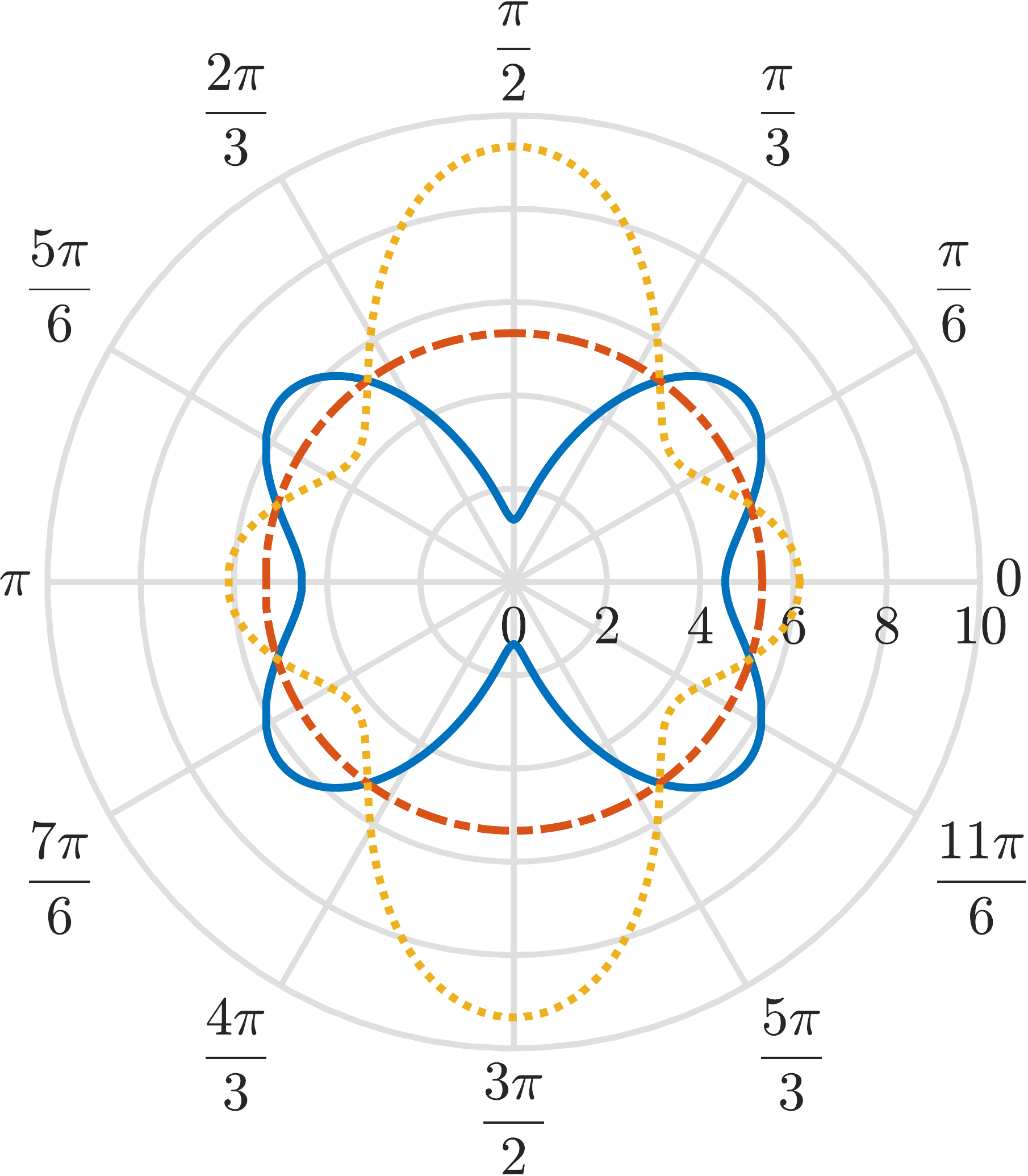}
&
\includegraphics[scale=0.25,clip]{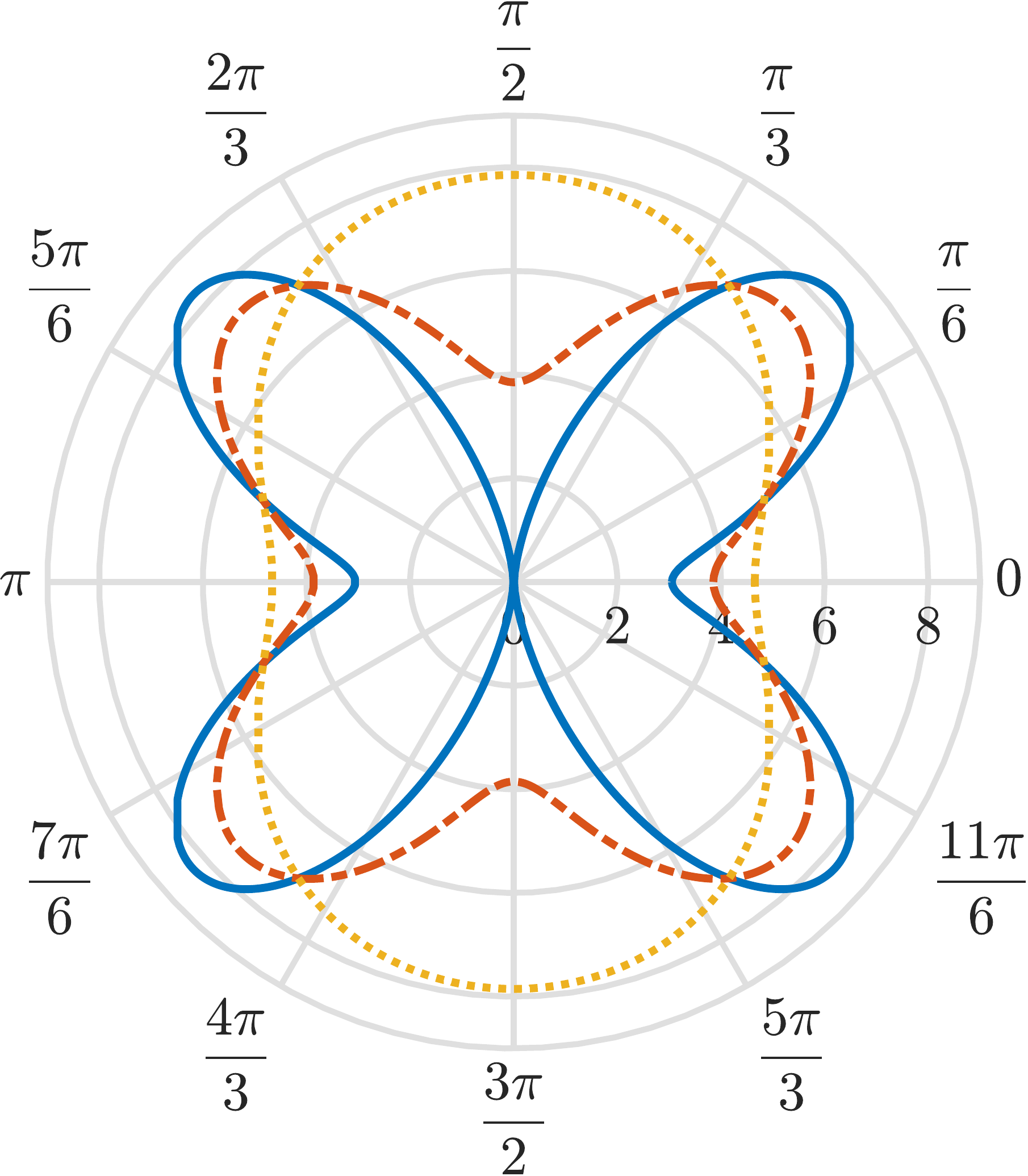}
\\
%$C_{1111}=1$ & $C_{1111}=1$ & $C_{1111}=1$ \\
$C_{1122} = \frac{7}{27}$ & $C_{1122} = \frac{1}{3}$ & $C_{1122} = \frac{11}{27}$
\end{tabular}
\caption{Rectangular $\gamma(\theta)$ with $C_{1111} = 1$ for $C_{1122} = \frac{7}{27}$ (left), $\frac{1}{3}$ (middle), $\frac{11}{27}$ (right) and $C_{2222} = \frac{3}{5}$~(solid blue), $1$~(dash-dotted red), $\frac{7}{5}$~(dotted yellow).}
\label{fig:reckernel}
\end{center}
\end{figure}

\textbf{Oblique}: For oblique symmetry, in addition to $C_{1111}, C_{1122},$ and $C_{2222}$, there are two additional constants, $C_{1112}$ and $C_{2212}$. These two additional constants only contribute to the peridynamic tensor components $\Lambda_{1112}$ and $\Lambda_{2212}$ ({\em cf.}~\eqref{eqn:partitionlambdatensor}). The contribution to the angular portion of the micromodulus function from $C_{1112}$ and $C_{2212}$ can be entirely described by ({\em cf.}~\eqref{eqn:lambdaobl})
\begin{equation}\label{eqn:angularpartofobliquekernel}
\frac{4\xi_1 \xi_2 \left( \Lambda_{1112} \xi_1^2 + \Lambda_{2212} \xi_2^2 \right) }{\|\bfxi \|^4}.% = \frac{4 \xi_1 \xi_2 \left( C_{1112}(20 \xi_1^2 - 12 \xi_2^2) + C_{2212}(-12 \xi_1^2 + 20 \xi_2^2) \right) }{\|\bfxi \|^4}.
\end{equation}

First, note the magnitude of expression \eqref{eqn:angularpartofobliquekernel} favors $\bfxi$ away from the main axes as $\xi_1 \xi_2$ approaches zero in those cases. Second, the term $\Lambda_{1112} \xi_1^2 + \Lambda_{2212} \xi_2^2$ favors either the $\xi_1$ ($x$-direction) or $\xi_2$ ($y$-direction) depending on the values of $C_{1112}$ and $C_{2212}$. There is a multitude of possible behaviors that may be represented by \eqref{eqn:lambdaobl}. In an attempt to consider the relative influences of $C_{1112}$ and $C_{2212}$ on the micromodulus function, we present a small subset of these behaviors in Figure \ref{fig:obl1kernel}. We consider the cases where $C_{1112} = -C_{2212} = -\frac{1}{8}$ (solid blue), $0$ (dash-dotted red), and $\frac{1}{8}$ (dotted yellow). 

\begin{figure}
\begin{center}
\begin{tabular}{ccc}
\includegraphics[scale=0.25,clip]{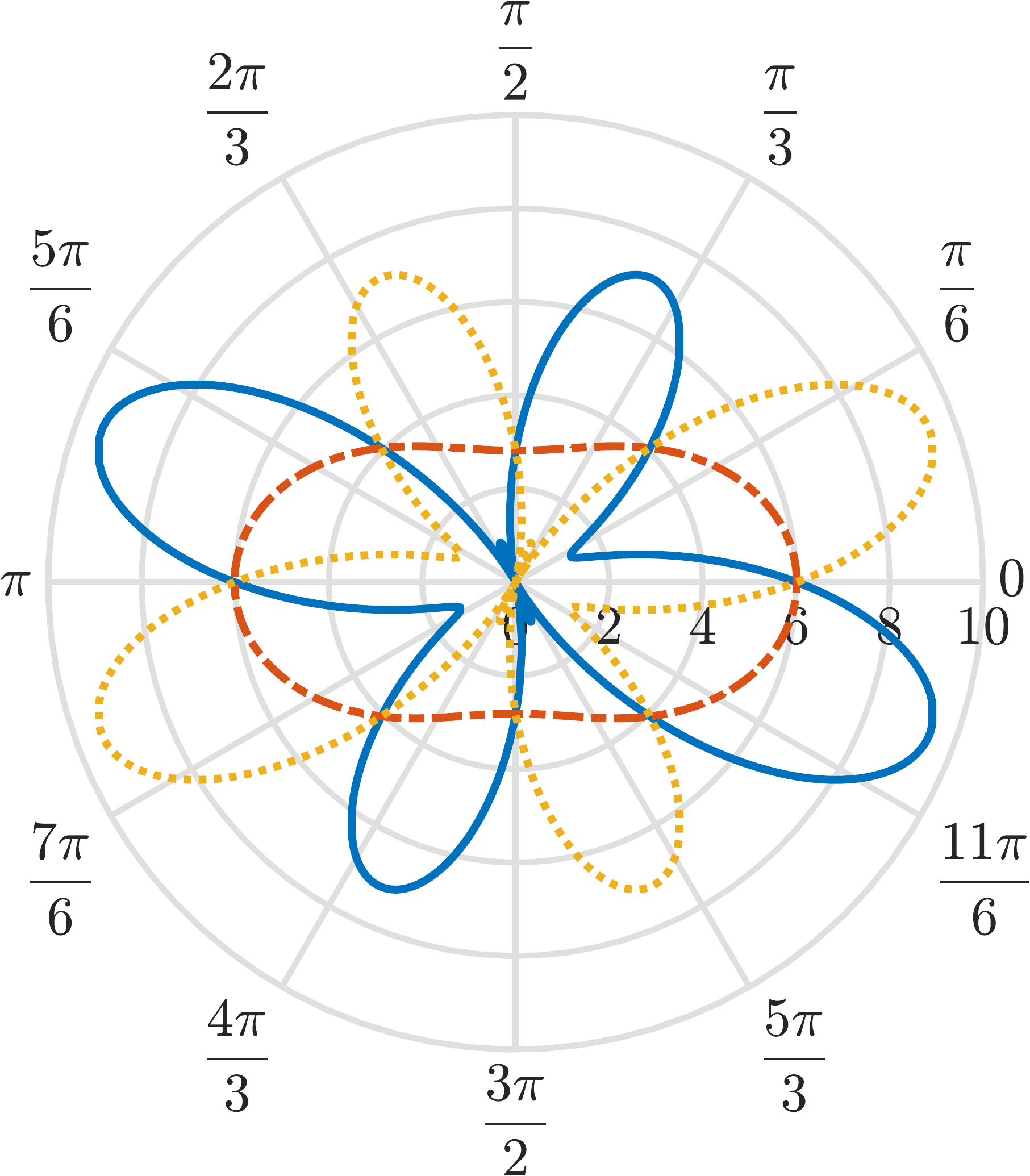}
&
\includegraphics[scale=0.25,clip]{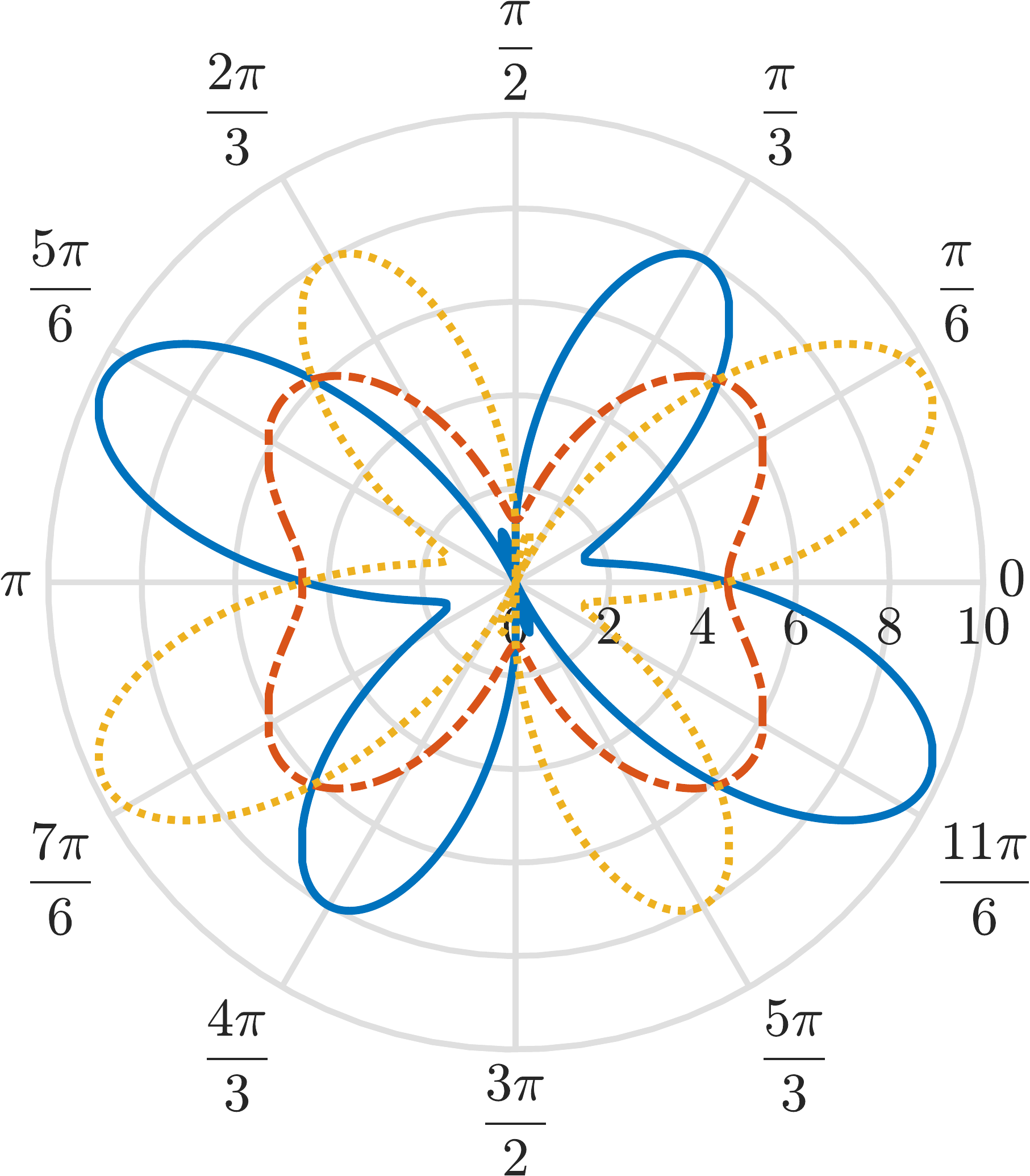}
&
\includegraphics[scale=0.25,clip]{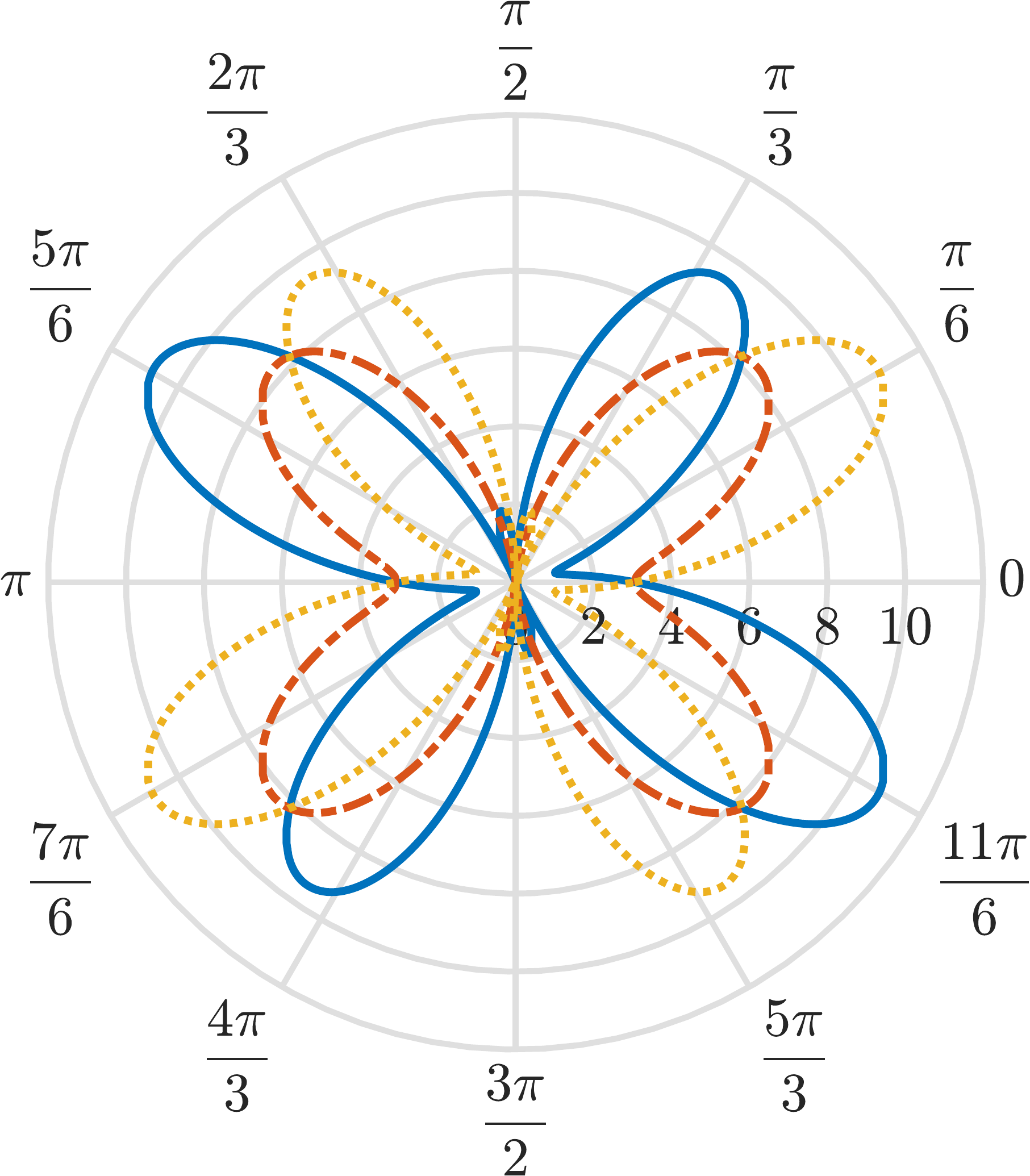}
\\
$C_{1122} = \frac{7}{27}$ & $C_{1122} = \frac{1}{3}$ & $C_{1122} = \frac{11}{27}$ \\[0.05in]
$C_{2222} = \frac{3}{5}$ & $C_{2222} = \frac{3}{5}$ &
$C_{2222} = \frac{3}{5}$ \\
%%%%%%%%%%%%%%%%%%%%%%%%%%%%%%%%%%%%%%%%%%%%%%%%%%%%%%
\includegraphics[scale=0.25,clip]{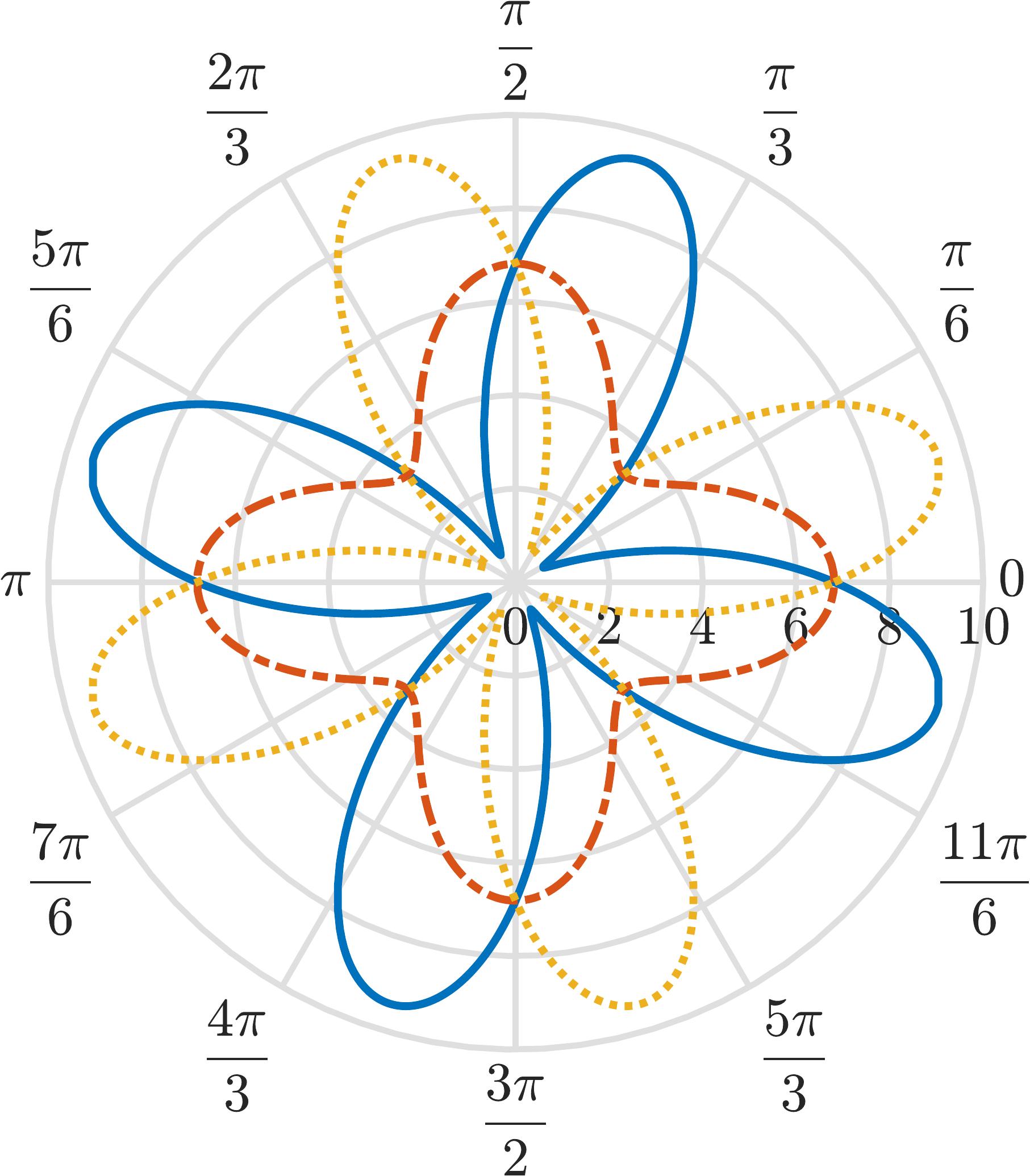}
&
\includegraphics[scale=0.25,clip]{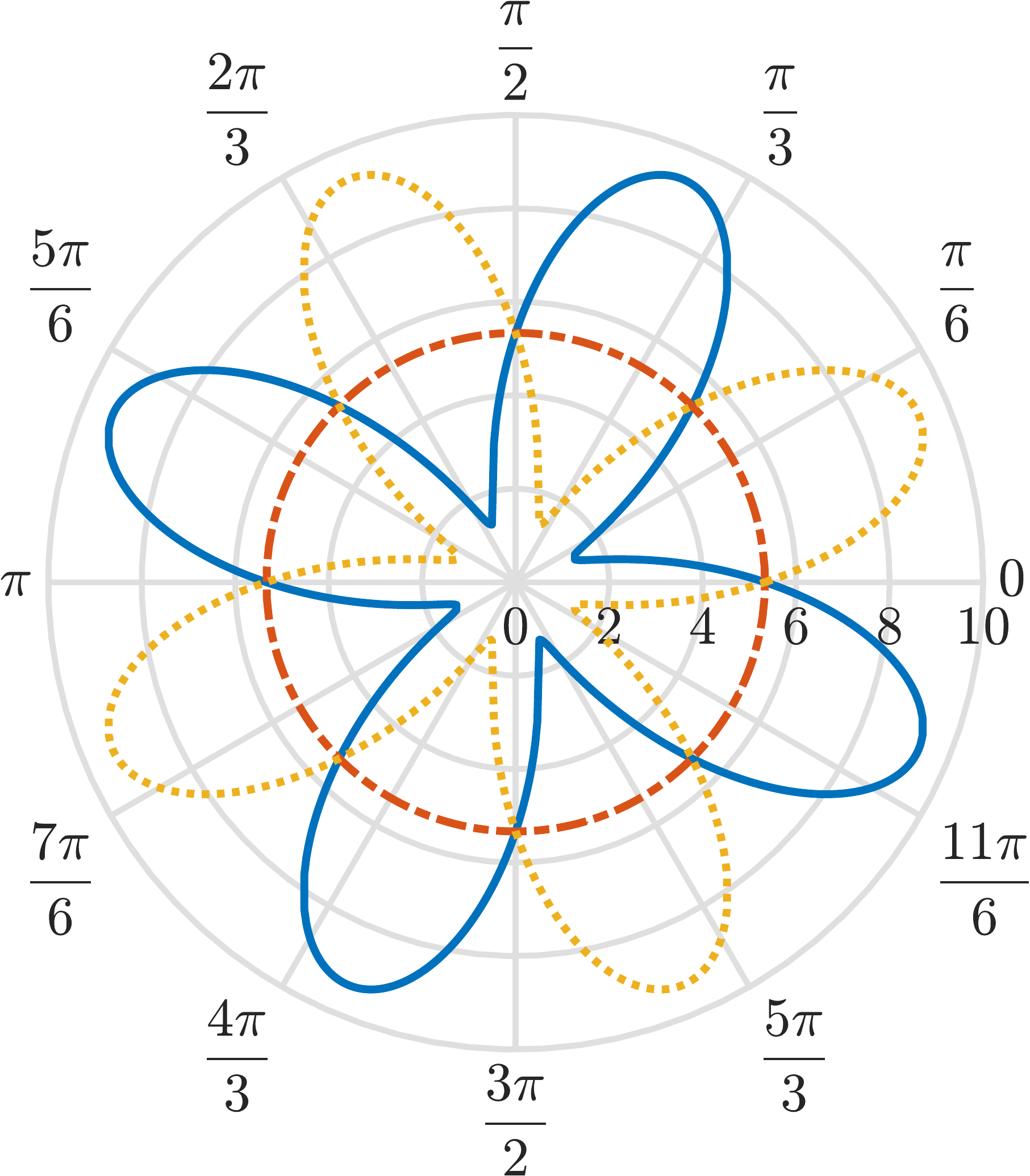}
&
\includegraphics[scale=0.25,clip]{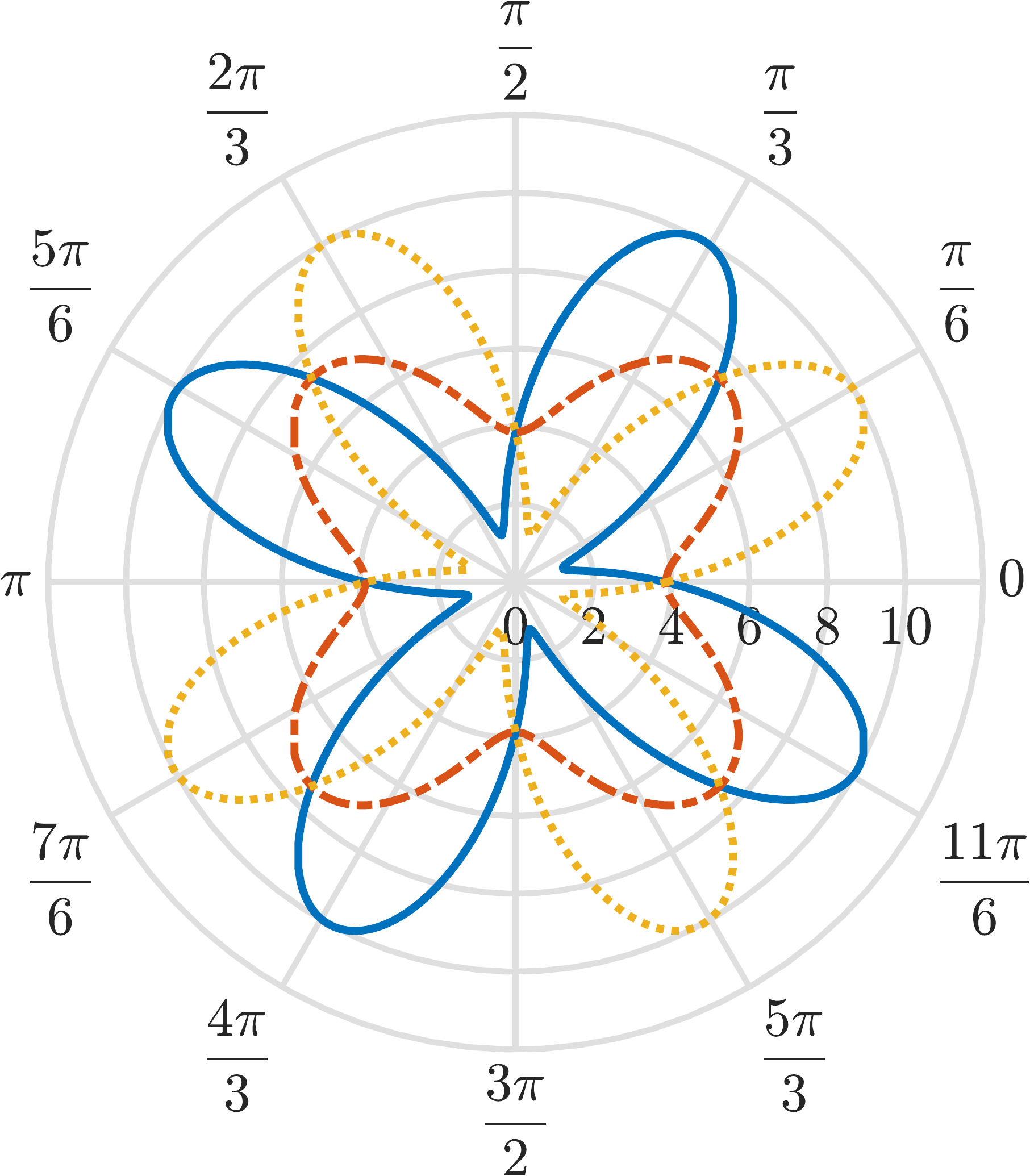}
\\
$C_{1122} = \frac{7}{27}$ & $C_{1122} = \frac{1}{3}$ & $C_{1122} = \frac{11}{27}$ \\[0.05in]
$C_{2222} = 1$ & $C_{2222} = 1$ &
$C_{2222} = 1$ \\
%%%%%%%%%%%%%%%%%%%%%%%%%%%%%%%%%%%%%%%%%%%%%%%%%%%%%%
\includegraphics[scale=0.25,clip]{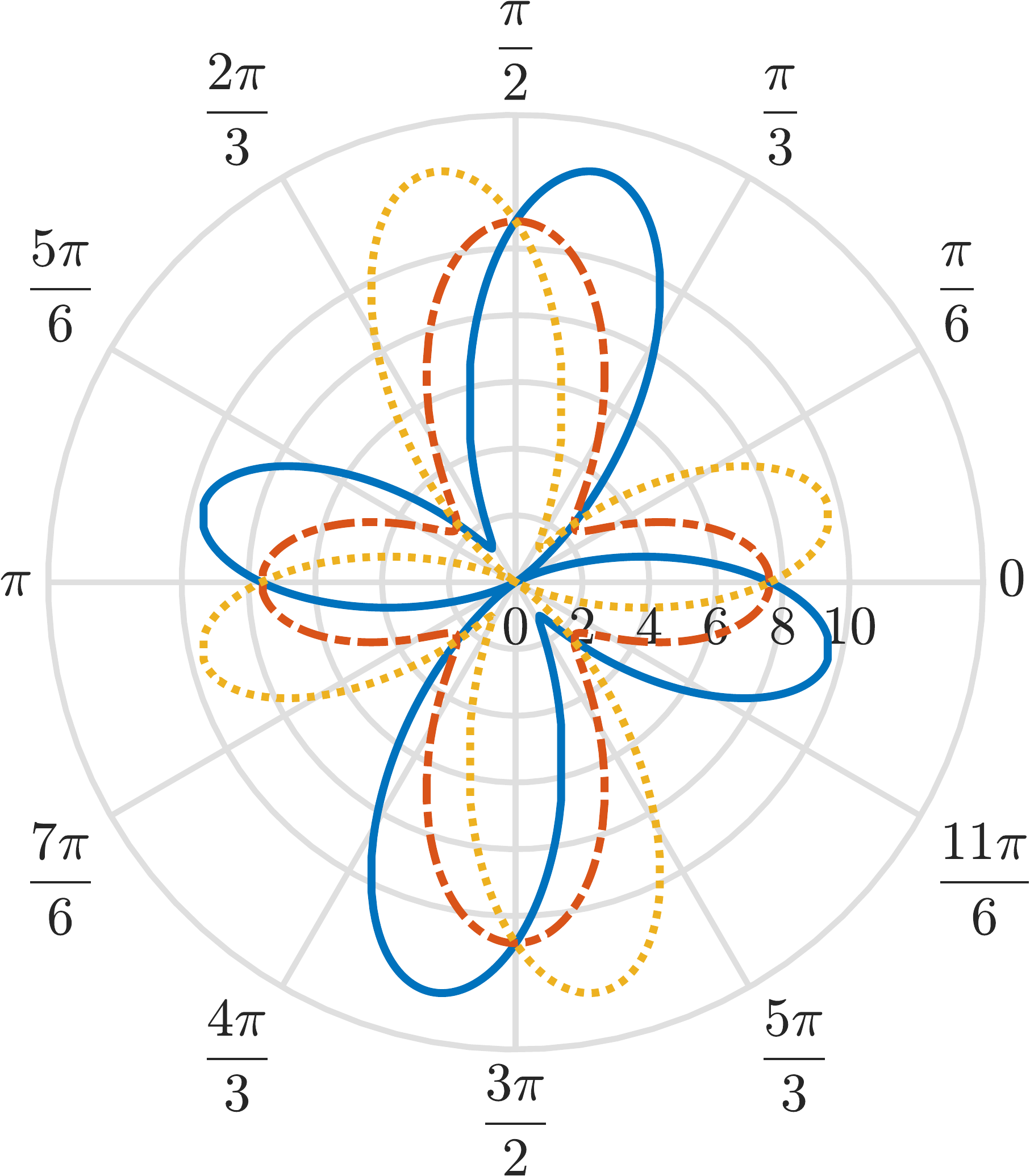}
&
\includegraphics[scale=0.25,clip]{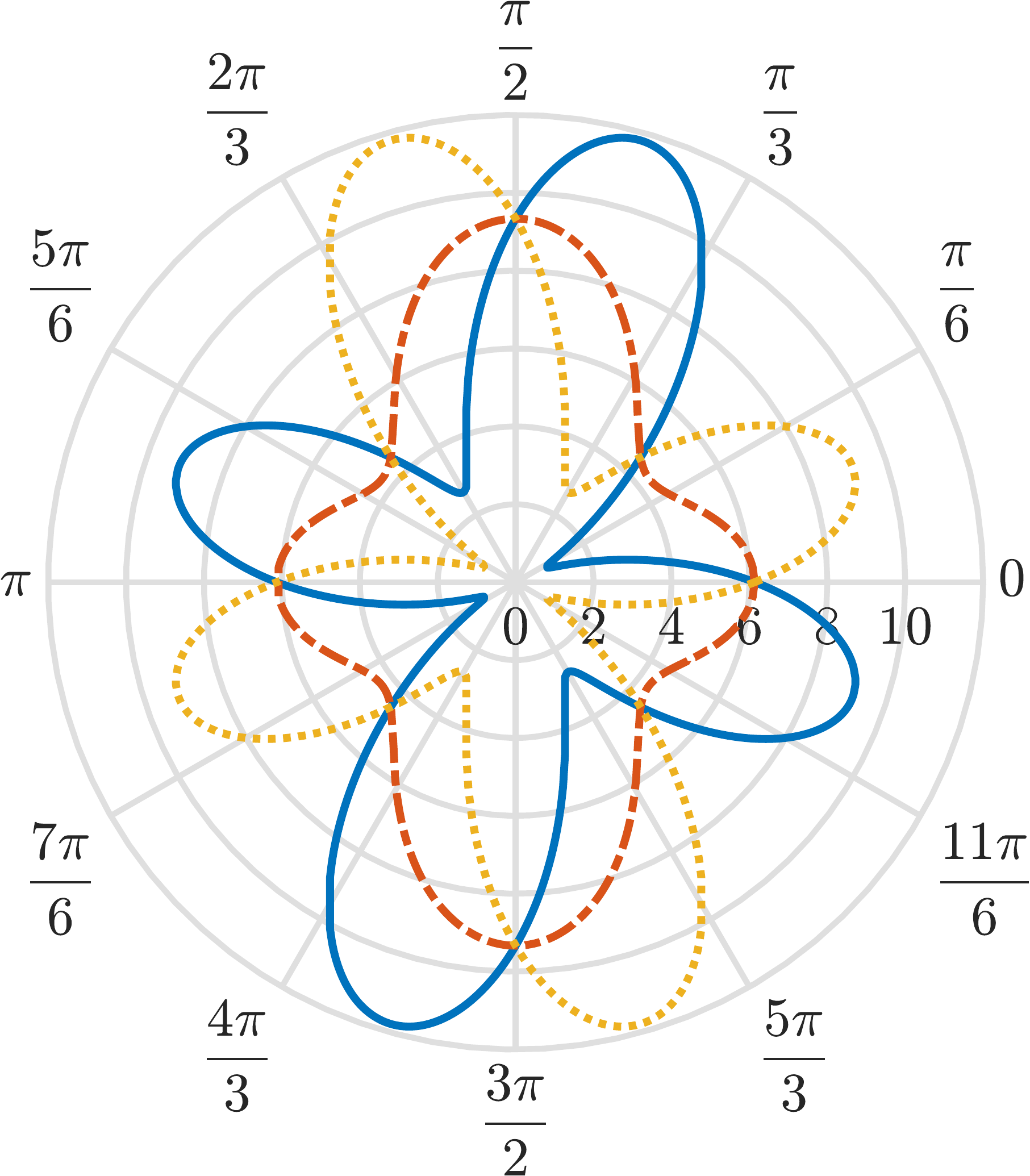}
&
\includegraphics[scale=0.25,clip]{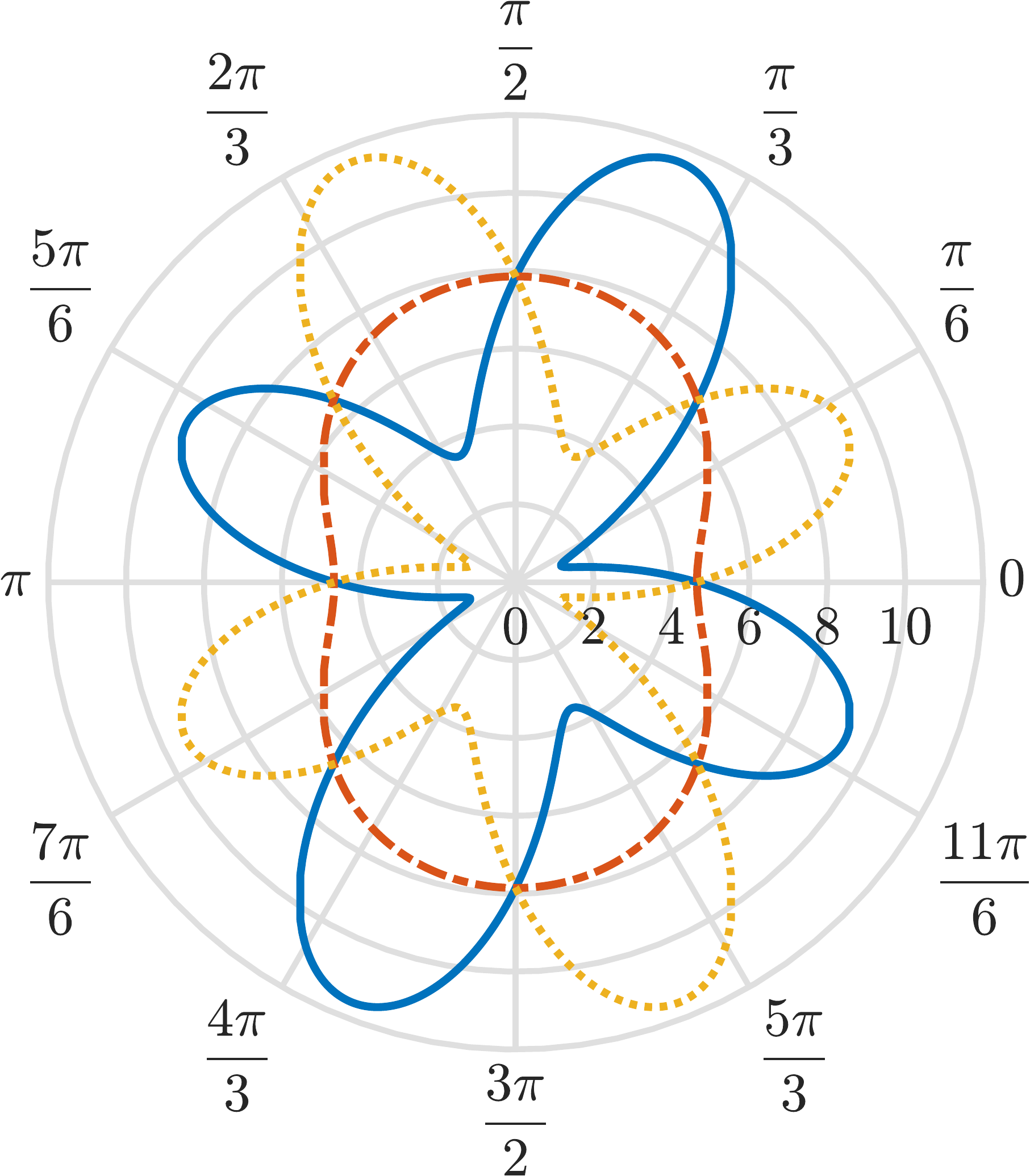}
\\
$C_{1122} = \frac{7}{27}$ & $C_{1122} = \frac{1}{3}$ & $C_{1122} = \frac{11}{27}$ \\[0.05in]
$C_{2222} = \frac{7}{5}$ & $C_{2222} =\frac{7}{5} $ &
$C_{2222} = \frac{7}{5}$
\end{tabular}
\caption{Oblique $\gamma(\theta)$ for $C_{1111} = 1$; $C_{1122} = \frac{7}{27}$ (left), $\frac{1}{3}$ (middle), $\frac{11}{27}$ (right); $C_{2222} = \frac{3}{5}$ (top), $1$ (middle), $\frac{7}{5}$ (bottom); and $C_{1112} = - C_{2212} = -\frac{1}{8}$(solid blue), $0$ (dash-dotted red), $\frac{1}{8}$ (dotted yellow).}
\label{fig:obl1kernel}
\end{center}
\end{figure}

\subsection{Planar approximations in three-dimensional bond-based peridynamics}\label{sec:planemodelsperi}

In this section, we consider plane strain and plane stress formulations of linear elastic bond-based peridynamic models. Plane strain and plane stress peridynamic models appear in, e.g., \cite{de2016,Gerstle2005,Ghajari2014,Le2014,Sarego2016}. The plane strain and plane stress peridynamic models in these papers were all derived from two-dimensional peridynamic models by simply matching peridynamic constants to constants appearing in classical plane strain and plane stress models. With the exception of \cite{de2016,Ghajari2014}, where a cubic model and a transversely isotropic model, respectively, were considered, these works only dealt with isotropic models.

As opposed to previous works, here we begin with a three-dimensional peridynamic model and directly impose peridynamic analogues of plane strain and plane stress conditions utilized in classical linear elasticity. In addition, our treatment considers all symmetry classes of classical linear elasticity ({\em cf}. Section \ref{sec:threedimclassicalelasticity}). The resulting models are two-dimensional, anisotropic, and reduce to the classical planar elasticity models when the displacements are smooth and higher-order terms are negligible. 

\subsubsection{Peridynamic plane strain}\label{sec:PeridynamicPlaneStrain}
To derive the peridynamic plane strain model, we begin with the three-dimensional bond-based linear peridynamic model \eqref{eqn:linearperieqn} and impose peridynamic analogues of the classical plane strain assumptions. We suppose a homogeneous material response and only consider material points within the bulk of the body. In this case, the three-dimensional bond-based linear peridynamic equation of motion is given in component form by ({\em cf}. \eqref{eqn:linearperieqncomponentformhomogeneous}):
\begin{subequations}\label{eqn:perieqnmot}
\begin{align}
\begin{split}
\rho(\bfx) \ddot{u}_1(\bfx,t) = \int_{B_{\delta}^{3D}(\mathbf{0}) } \lambda(\bfxi) &\left[ \xi_1^2 (u_1(\bfx + \bfxi,t) - u_1(\bfx,t)) + \xi_1 \xi_2 (u_2(\bfx + \bfxi,t) - u_2(\bfx,t) ) \right. \\
&\left.+ \xi_1 \xi_3(u_3(\bfx + \bfxi,t) - u_3(\bfx,t))\right] d\bfxi + b_1(\bfx,t), 
\end{split} \\
\begin{split}
\rho(\bfx) \ddot{u}_2(\bfx,t) = \int_{B_{\delta}^{3D}(\mathbf{0}) } \lambda(\bfxi) &\left[ \xi_1 \xi_2 (u_1(\bfx + \bfxi,t) - u_1(\bfx,t)) + \xi_2^2 (u_2(\bfx + \bfxi,t) - u_2(\bfx,t)) \right. \\
&\left.+ \xi_2 \xi_3(u_3(\bfx + \bfxi,t) - u_3(\bfx,t)) \right] d\bfxi +  b_2(\bfx,t),
\end{split} \\
\begin{split}
\rho(\bfx) \ddot{u}_3(\bfx,t)  =  \int_{B_{\delta}^{3D}(\mathbf{0}) } \lambda(\bfxi) &\left[ \xi_1 \xi_3 (u_1(\bfx + \bfxi,t) - u_1(\bfx,t)) + \xi_2 \xi_3 (u_2(\bfx + \bfxi,t) - u_2(\bfx,t)) \right. \\
&\left.+ \xi_3^2 (u_3(\bfx + \bfxi,t) - u_3(\bfx,t))\right] d\bfxi + b_3(\bfx,t),
\end{split}
\end{align}
\end{subequations}
where $B^{3D}_{\delta}(\bf0)$ is the ball in $\mathbb{R}^3$ of radius $\delta$ about the origin. To avoid confusion, the ball in $\mathbb{R}^2$ (i.e. a disk) of radius $\delta$ about the origin is denoted by $B^{2D}_{\delta}(\bf0)$. We assume almost identical assumptions to those in classical plane strain ({\em cf.}~Section \ref{assmp:classplanestrain}). 

The peridynamic plane strain assumptions are given as follows:
\begin{enumerate}[(P$\varepsilon$1)]\label{assmp:periplanestrain}
\item The geometric form and mass density of the body, and the external loads exerted on it, do not change along some axis, which we take as the $z$-axis. In particular, \label{assmp:periplanestrain1}
\begin{equation*}
\rho = \rho(x,y) \text{ and } \bfb = \bfb(x,y,t).
\end{equation*}
\item The deformation of any arbitrary cross-section perpendicular to the $z$-axis is identical, i.e., the displacements are function of $x$, $y$, and $t$ only:
\begin{equation*}%\label{assump:genplanestrain}
\bfu = \bfu(x,y,t).
\end{equation*} \label{assmp:periplanestrain2}
\item The material has at least monoclinic symmetry with a plane of reflection corresponding to the plane $z=0$, i.e., the micromodulus function is symmetric in its third component ({\em cf.} \cite{STG2019}): \label{assmp:periplanestrain3full}
\begin{equation}\label{assmp:periplanestrain3}
\lambda(\xi_1, \xi_2, \xi_3) = \lambda( \xi_1, \xi_2, - \xi_3).
\end{equation}
\end{enumerate}

A system satisfying Assumptions (P$\varepsilon$\ref{assmp:periplanestrain1}) and (P$\varepsilon$\ref{assmp:periplanestrain2}) is said to be in a state of peridynamic generalized plane strain. See Figure \ref{fig:planestrain} for an illustration of a body in a state of plane strain. %We call the corresponding system the \textit{peridynamic generalized plane strain} model.

\underline{Peridynamic generalized plane strain}: Imposing Assumptions (P$\varepsilon$\ref{assmp:periplanestrain1}) and (P$\varepsilon$\ref{assmp:periplanestrain2}) on \eqref{eqn:perieqnmot} and then integrating over the third component, $\xi_3$, we obtain the peridynamic analogue of \eqref{eqn:classeqnmot}:
\begin{subequations}\label{eqn:genplanestrainmodel}
\begin{align}
\begin{split}
\rho(x,y) \ddot{u}_1(x,y,t)  =& \int_{B^{2D}_{\delta}(\bf0)} \lambda_0(\xi_1,\xi_2) \left[ \xi_1^2 (u_1(x+\xi_1,y+\xi_2,t) - u_1(x,y,t)) + \xi_1 \xi_2 (u_2(x+\xi_1,y+\xi_2,t) - u_2(x,y,t)) \right] d \xi_1 d \xi_2 \\
&+ \int_{B^{2D}_{\delta}(\bf0)} \lambda_1(\xi_1,\xi_2,t) \xi_1 (u_3(x+\xi_1,y+\xi_2,t) - u_3(x,y,t)) d \xi_1 d \xi_2 + b_1(x,y,t), 
\end{split} \\
%%%%%%%%%%%%%%%%%%%%%%%%%%%%%%%%%%%%%%%%%%%
\begin{split}
\rho(x,y) \ddot{u}_2(x,y,t) =& \int_{B^{2D}_{\delta}(\bf0)}  \lambda_0(\xi_1,\xi_2) \left[ \xi_1 \xi_2 (u_1(x+\xi_1,y+\xi_2,t) - u_1(x,y,t)) + \xi_2^2 (u_2(x+\xi_1,y+\xi_2,t) - u_2(x,y,t)) \right] d \xi_1 d \xi_2 \\
&+ \int_{B^{2D}_{\delta}(\bf0)} \lambda_1(\xi_1,\xi_2) \xi_2 (u_3(x+\xi_1,y+\xi_2,t) - u_3(x,y,t)) d \xi_1 d \xi_2 + b_2(x,y,t), 
\end{split} \\
%%%%%%%%%%%%%%%%%%%%%%%%%%%%%%%%%%%%%%%%%%%%
\begin{split}
\rho(x,y) \ddot{u}_3(x,y,t)  =&  \int_{B^{2D}_{\delta}(\bf0)}  \lambda_1(\xi_1,\xi_2) \left[ \xi_1  (u_1(x+\xi_1,y+\xi_2,t) - u_1(x,y,t)) + \xi_2 (u_2(x+\xi_1,y+\xi_2,t) - u_2(x,y,t)) \right] d \xi_1 d \xi_2 \\
&+ \int_{B^{2D}_{\delta}(\bf0)} \lambda_2(\xi_1,\xi_2) (u_3(x+\xi_1,y+\xi_2,t) - u_3(x,y,t)) d \xi_1 d \xi_2 + b_3(x,y,t),
\end{split}
\end{align}
\end{subequations}
where the micromodulus functions $\left\{\lambda_n(\xi_1,\xi_2) \right\}$ are defined by
\begin{equation}\label{eqn:xi3integrationterms}
\begin{split}
\lambda_n(\xi_1,\xi_2) :=& \int_{-\sqrt{\delta^2-\xi_1^2-\xi_2^2}}^{\sqrt{\delta^2-\xi_1^2-\xi_2^2}}  \lambda(\bfxi) \xi_3^n d \xi_3, \quad n \in \left\{0,1,2 \right\}.  \\
\end{split}
\end{equation}

Equation \eqref{eqn:genplanestrainmodel} is the \textit{peridynamic generalized plane strain} equation of motion and reduces the degrees of freedom in the system significantly ($\bfu$ is a function of only $x$, $y$, and $t$, and thus the simulation may be run on a two-dimensional spatial domain). However, the in-plane displacements, $u_1$ and $u_2$, are coupled with the out-of-plane displacement, $u_3$. Nevertheless, provided the material has sufficient symmetry, it is possible to further simplify the model. In fact, having a single plane of reflection symmetry (Assumption (P$\varepsilon$\ref{assmp:periplanestrain3full})) is sufficient to decouple the in-plane and out-of-plane displacements provided the plane of reflection symmetry coincides with the plane $z= 0$, as demonstrated below. A system satisfying (P$\varepsilon$\ref{assmp:periplanestrain1})~--~(P$\varepsilon$\ref{assmp:periplanestrain3full}) is said to be in a state of peridynamic plane strain.

\underline{Peridynamic plane strain}:
Emulating the classical theory by imposing monoclinic symmetry from Assumption (P$\varepsilon$\ref{assmp:periplanestrain3full}) on \eqref{eqn:genplanestrainmodel}, the in-plane displacements, $u_1$ and $u_2$, and the out-of-plane displacement, $u_3$, decouple. A system satisfying Assumptions (P$\varepsilon$\ref{assmp:periplanestrain1})~--~(P$\varepsilon$\ref{assmp:periplanestrain3full}) is said to be in a state of peridynamic plane strain.
In this case, the micromodulus function $\lambda(\bfxi)$ is even in its third component and therefore $\lambda_1(\xi_1,\xi_2) \equiv 0$ by the antisymmetry of the integrand in \eqref{eqn:xi3integrationterms}. The resulting in-plane equations of motion (analogues of \eqref{eqn:classicalmonoeqnmotion1} and \eqref{eqn:classicalmonoeqnmotion2}) are given by

\begin{subequations}\label{eqn:planestrainmodelinplane}
\begin{align}
\begin{split}
\rho(x,y) \ddot{u}_1(x,y,t)  =& \int_{B^{2D}_{\delta}(\bf0)} \lambda_0(\xi_1,\xi_2) \left[ \xi_1^2 (u_1(x+\xi_1,y+\xi_2,t) - u_1(x,y,t)) + \xi_1 \xi_2 (u_2(x+\xi_1,y+\xi_2,t) - u_2(x,y,t)) \right] d \xi_1 d \xi_2 \\
&+ b_1(x,y,t), \label{eqn:planestrainmodelu1}
\end{split} \\
%%%%%%%%%%%%%%%%%%%%%%%%%%%%%%%%%%%%%%%%%%%
\begin{split}
\rho(x,y) \ddot{u}_2(x,y,t) =& \int_{B^{2D}_{\delta}(\bf0)}  \lambda_0(\xi_1,\xi_2) \left[ \xi_1 \xi_2 (u_1(x+\xi_1,y+\xi_2,t) - u_1(x,y,t)) + \xi_2^2 (u_2(x+\xi_1,y+\xi_2,t) - u_2(x,y,t)) \right] d \xi_1 d \xi_2 \\
&+ b_2(x,y,t), \label{eqn:planestrainmodelu2}
\end{split} 
\end{align}
\end{subequations}
and the resulting out-of-plane equation of motion (analogue of \eqref{eqn:classicalmonoeqnmotion3}) is given by
\begin{align}\label{eqn:planestrainmodeloutofplane}
%%%%%%%%%%%%%%%%%%%%%%%%%%%%%%%%%%%%%%%%%%%%
\begin{split}
\rho(x,y) \ddot{u}_3(x,y,t) =&  \int_{B^{2D}_{\delta}(\bf0)} \lambda_2(\xi_1,\xi_2) (u_3(x+\xi_1,y+\xi_2,t) - u_3(x,y,t)) d \xi_1 d \xi_2 + b_3(x,y,t).
\end{split}
\end{align}

\begin{remark}\label{rmk:planestrainvector}
One may express the in-plane equations of motion \eqref{eqn:planestrainmodelinplane} in vector form by letting $\bfx=(x,y),\bfxi = (\xi_1,\xi_2),\bfu = (u_1,u_2), \mathcal{H} = B_\delta^{2D}(\bf0)$, and $\bfb=(b_1,b_2)$. In this case, the in-plane equations of motion are formulated as
\begin{equation*}
\rho(\bfx) \ddot{\bfu}(\bfx,t) = \int_{\mathcal{H}} \lambda_0(\bfxi) \bfxi \otimes \bfxi (\bfu(\bfx+\bfxi,t) - \bfu(\bfx,t)) d \bfxi + \bfb(\bfx,t).
\end{equation*}
\end{remark}

\underline{Peridynamic plane strain micromodulus functions}: In the peridynamic generalized plane strain model presented in this work, the requirements placed on the micromodulus function $\lambda(\bfxi)$ have been fairly minimal up to this point.
%Up to this point we have allowed the form of the plane strain micromodulus function to be quite general.
%The requirements placed on the micromodulus function are fairly minimal up to this point. 
In order to investigate the differences between the pure two-dimensional peridynamic equation of motion \eqref{eqn:puretwocomponentmodel} and the in-plane equations of motion \eqref{eqn:planestrainmodelinplane} for the peridynamic plane strain model, we consider the micromodulus function \eqref{def:lambdagenform} related to the elasticity tensor through \eqref{eqn:Cperiexpression}. In the pure two-dimensional peridynamic equation of motion,  indices of $\mathbbm{\Lambda}$ are in $\left\{1,2 \right\}$, \eqref{def:lambdagenform} is given by \eqref{eqn:lambdaobl}, and $\mathbbm{\Lambda}$ is related to $\mathbb{C}$ through \eqref{eqn:SijkltoCijkl}. Alternatively, in the in-plane equations of motion for the peridynamic plane strain model, the indices of $\mathbbm{\Lambda}$ are in $\left\{1,2,3 \right\}$, \eqref{def:lambdagenform} is given by \eqref{eqn:lambdatriclinic}, and $\mathbbm{\Lambda}$ is related to $\mathbb{C}$ through ({\em cf.} \cite{STG2019}):

\begin{subequations}\label{eqn:SijklToCijklRelations3D}
\begin{align}
\Lambda_{1111} &= 30 C_{1111} - 45 C_{1122} - 45 C_{1133} + \frac{15}{4} C_{2222} + \frac{15}{2} C_{2233} + \frac{15}{4} C_{3333},\\
\Lambda_{1122} &= - \frac{15}{2}C_{1111} + \frac{255}{4} C_{1122} - \frac{25}{4} C_{1133} - \frac{15}{2} C_{2222} - \frac{25}{4} C_{2233} + \frac{5}{4} C_{3333},\\
\Lambda_{1133} &= - \frac{15}{2}C_{1111} - \frac{25}{4} C_{1122} + \frac{255}{4} C_{1133} + \frac{5}{4} C_{2222} - \frac{25}{4} C_{2233} - \frac{15}{2} C_{3333},\\
\Lambda_{1123} &= 70 C_{1123} - \frac{35}{4} C_{2223} - \frac{35}{4} C_{3323},\\
\Lambda_{1113}  &=  \frac{105}{2} C_{1113} - \frac{105}{4} C_{2213} - \frac{105}{4} C_{3313},\\
\Lambda_{1112} &= \frac{105}{2} C_{1112} - \frac{105}{4} C_{2212} - \frac{105}{4} C_{3312},\\
\Lambda_{2222} &= \frac{15}{4} C_{1111} - 45 C_{1122} + \frac{15}{2} C_{1133} + 30 C_{2222} - 45 C_{2233} + \frac{15}{4} C_{3333},\\
\Lambda_{2233} &=  \frac{5}{4} C_{1111} - \frac{25}{4} C_{1122} - \frac{25}{4} C_{1133} - \frac{15}{2}  C_{2222} + \frac{255}{4} C_{2233} - \frac{15}{2} C_{3333},\\
\Lambda_{2223} &=  -\frac{105}{4}  C_{1123} + \frac{105}{2} C_{2223} - \frac{105}{4} C_{3323},\\
\Lambda_{2213} &= - \frac{35}{4} C_{1113} + 70 C_{2213} - \frac{35}{4} C_{3313},\\
\Lambda_{2212} &= -\frac{105}{4} C_{1112} + \frac{105}{2} C_{2212} - \frac{105}{4} C_{3312},\\
\Lambda_{3333} &=  \frac{15}{4} C_{1111} + \frac{15}{2} C_{1122} - 45 C_{1133} + \frac{15}{4} C_{2222} - 45 C_{2233} + 30 C_{3333},\\
\Lambda_{3323} &= - \frac{105}{4}  C_{1123} - \frac{105}{4} C_{2223} + \frac{105}{2} C_{3323},\\
\Lambda_{3313} &= -\frac{105}{4} C_{1113} - \frac{105}{4} C_{2213} + \frac{105}{2} C_{3313},\\
\Lambda_{3312} &= -\frac{35}{4} C_{1112} - \frac{35}{4} C_{2212} + 70 C_{3312}. 
\end{align}
\end{subequations} 

In order to obtain expressions for $\lambda_n(\xi_1,\xi_2)$ in \eqref{eqn:xi3integrationterms}, it is convenient to introduce a shorthand notation:

\begin{subequations}\label{eqn:PeriPlaneStrainAList}
\begin{align}
&A_0(\xi_1,\xi_2) := \Lambda_{1111} \xi_1^4 + 4\Lambda_{1112}\xi_1^3 \xi_2 +6 \Lambda_{1122} \xi_1^2 \xi_2^2 + 4 \Lambda_{2212} \xi_1 \xi_2^3 + \Lambda_{2222} \xi_2^4, \\
&A_1(\xi_1,\xi_2) := 4\left( \Lambda_{1113}\xi_1^3 + \Lambda_{2223}\xi_2^3 + 3 \Lambda_{1123} \xi_1^2 \xi_2 + 3 \Lambda_{2213} \xi_1\xi_2^2 \right), \\
&A_2(\xi_1,\xi_2) := 6 \left( \Lambda_{1133}\xi_1^2 + \Lambda_{2233}\xi_2^2 + 2 \Lambda_{3312} \xi_1 \xi_2  \right), \\
&A_3(\xi_1,\xi_2) := 4 \left( \Lambda_{3313}\xi_1 + \Lambda_{3323}\xi_2 \right), \\
&A_4(\xi_1,\xi_2) := \Lambda_{3333}.
\end{align}
\end{subequations}

Then, the micromodulus function \eqref{eqn:lambdatriclinic} can be expressed as
\begin{equation}\label{eqn:lambformwithA}
\begin{split}
\lambda(\bfxi) ={}& \frac{1}{m} \frac{\omega(\|\bfxi\|)}{\|\bfxi \|^2} \sum_{i=0}^4 \frac{A_{i}(\xi_1,\xi_2) \xi_3^i}{\|\bfxi \|^4} \\
={}& \frac{1}{m} \frac{\omega(\|\bfxi\|)}{\|\bfxi \|^2} \frac{A_0(\xi_1,\xi_2) + A_1(\xi_1,\xi_2)\xi_3 + A_2(\xi_1,\xi_2)\xi_3^2 + A_3(\xi_1,\xi_2)\xi_3^3 + A_4(\xi_1,\xi_2)\xi_3^4 }{\|\bfxi \|^4}.
\end{split} 
\end{equation}

Depending on the choice of influence function $\omega$, it may be possible to provide closed-form expressions for the integrals in \eqref{eqn:xi3integrationterms}. Two commonly utilized influence functions in peridynamics are $\omega(\| \bfxi \|) = \frac{1}{\| \bfxi \|}$ and $\omega(\|\bfxi \|) = 1$. For convenience, we additionally introduce the shorthand notation $r = \sqrt{\xi_1^2 + \xi_2^2}$, which represents the magnitude of the projection of $\bfxi$ on the $xy$-plane. This is in contrast to $\| \bfxi \| = \sqrt{ \xi_1^2 + \xi_2^2 + \xi_3^2}$, which is the magnitude of $\bfxi$ in $\mathbb{R}^3$. Given $\omega(\| \bfxi \|) = 1$ or $\omega(\| \bfxi \|) = \frac{1}{\|\bfxi\|}$, the micromodulus functions \eqref{eqn:xi3integrationterms} are given by 

\begin{subequations}\label{eqn:PlaneStrainLambda_i}
\begin{align}
\begin{split}
\lambda_0(\xi_1,\xi_2) ={}& \frac{\omega(r)}{m r}  \left[ \frac{A_0(\xi_1,\xi_2)}{ r^4 } M_0\left( \frac{r}{\delta} \right)   +  \frac{A_2(\xi_1,\xi_2)}{r^2 } M_1 \left( \frac{r}{\delta} \right) + A_4(\xi_1,\xi_2)  M_2\left( \frac{r}{\delta} \right)  \right], \label{eqn:PlaneStrainLambda_iForw=1overxia}
\end{split} \\ 
%%%%%%%%%%%%%%%%%%%%%%%%%%%%%%%%%%%%%%%%%%%%%%%%%%%%%
\begin{split}
\lambda_1(\xi_1,\xi_2) ={}& \frac{\omega(r)}{ m}  \left[  \frac{A_1(\xi_1,\xi_2)}{r^3  } M_1\left( \frac{r}{\delta} \right)  + \frac{A_3(\xi_1,\xi_2)}{r}  M_2\left( \frac{r}{\delta} \right)  \right] ,
\end{split} \\
%%%%%%%%%%%%%%%%%%%%%%%%%%%%%%%%%%%%%%%%%%%%%%%%%%%%
\begin{split}
\lambda_2(\xi_2,\xi_2) ={}& \frac{r \omega(r)}{m } \left[  \frac{A_0(\xi_1,\xi_2)}{r^4  }  M_1\left( \frac{r}{\delta} \right)  + \frac{A_2(\xi_1,\xi_2)}{r^2}  M_2\left( \frac{r}{\delta} \right) + A_4(\xi_1,\xi_2) M_3\left( \frac{r}{\delta} \right) \right], \label{eqn:PlaneStrainLambda_iForw=1overxic}
\end{split}
\end{align}
\end{subequations}

where for $\omega(\| \bfxi \|) = 1$ we have:
\begin{equation*}
\begin{split}
    &M_0(x) := \frac{3}{4}  \arctan\left(\sqrt{ x^{-2} - 1}\right) + \frac{1}{4} \left(3  x + 2 x^3 \right) \sqrt{1 - x^{2}}, \\
    &M_1(x) := \frac{1}{4}\arctan\left(\sqrt{ x^{-2} - 1}\right) + \frac{1}{4}\left(x - 2 x^{3}\right) \sqrt{1 - x^{2}}, \\
    &M_2(x) := \frac{3}{4}  \arctan\left(\sqrt{ x^{-2}  - 1}\right) - \frac{1}{4}\left(5 x - 2 x^{3}\right) \sqrt{1 - x^{2}}, \\
    &M_3(x) := -\frac{15}{4} \arctan\left(\sqrt{ x^{-2} - 1}\right) + \frac{1}{4} \left(8 x^{-1}  + 9 x - 2 x^3 \right) \sqrt{1 - x^{2}},\\
\end{split}
\end{equation*}
and for $\omega(\| \bfxi \|) = \frac{1}{\|\bfxi\|}$ we have\footnote{The equations for $M_3(x)$ below is only valid for $x = \frac{r}{\delta} > 0$. As the set $\left\{\bfxi \in \mathbb{R}^3: \xi_1^2 + \xi_2^2 = 0 \right\}$ is a set of zero measure in $B_{\delta}(\bf0)$, \eqref{eqn:PlaneStrainLambda_iForw=1overxic} remains valid with this formulation for $M_3$.}:

\begin{equation*}
\begin{split}
    &M_0(x) := \frac{2}{15} \sqrt{1-x^2} \left( 8 + 4x^2 + 3x^4 \right), \\
    &M_1(x) := \frac{2}{15} \sqrt{1-x^2} \left( 2 + x^2 - 3x^4 \right), \\
    &M_2(x) := \frac{2}{5} \sqrt{1-x^2} \left(1-2x^2 + x^4\right), \\
    &M_3(x) := -\frac{2}{15} \sqrt{1-x^2} \left(23-11x^2+3x^4 \right) + 2 \, \text{arsinh} \left(\sqrt{x^{-2}-1 }\right). \\
\end{split}
\end{equation*}

\begin{remark}\label{rem:angulardpendenceplanestrain}
For $i \in \left\{0,\ldots,4 \right\}$, note that $\frac{A_i(\xi_1,\xi_2)}{r^{4-i}}$ is radially independent ({\em cf.}~\eqref{eqn:PeriPlaneStrainAList}) and, consequently, only contributes to the angular portion of \eqref{eqn:PlaneStrainLambda_i}. Moreover, as the other terms in \eqref{eqn:PlaneStrainLambda_i} are radial functions, the $\frac{A_i(\xi_1,\xi_2)}{r^{4-i}}$ terms make up the entirety of the angular dependence of the micromodulus functions.
\end{remark}

Next, we take a closer look at the micromodulus functions $\left\{ \lambda_n(\xi_1,\xi_2) \right\}$ ({\em cf.}~\eqref{eqn:xi3integrationterms}) of the peridynamic plane strain model \eqref{eqn:planestrainmodelinplane} and~\eqref{eqn:planestrainmodeloutofplane} for various symmetry classes. Following Remark \ref{rem:angulardpendenceplanestrain}, the choice of symmetry class only has an effect on the $\left\{ A_i(\xi_1,\xi_2) \right\}$, while the general form of the micromodulus functions~\eqref{eqn:PlaneStrainLambda_i} remains unchanged. Therefore, we only present the $\left\{ A_i(\xi_1,\xi_2) \right\}$ for each symmetry class. %appearing in the peridynamic plane strain model \eqref{eqn:planestrainmodelinplane} and \eqref{eqn:planestrainmodeloutofplane}. 
As explained earlier, for many of the symmetry classes, there are multiple planes of reflection symmetry to choose from in order to satisfy Assumption~(P$\varepsilon$\ref{assmp:periplanestrain3full}). However, here we only consider a specific example from each symmetry class. More specifically, we choose the orientations corresponding to the elasticity tensors presented in Section \ref{sec:threedimclassicalelasticity}. 

\underline{Monoclinic}: We substitute \eqref{eqn:monoelastrelations}, with Cauchy's relations imposed, into \eqref{eqn:SijklToCijklRelations3D} to find ({\em cf.}~\eqref{eqn:PeriPlaneStrainAList}):

\begin{subequations}\label{eqn:AtermsMono}
\begin{align}
&A_0(\xi_1,\xi_2) = \Lambda_{1111} \xi_1^4 + 4\Lambda_{1112}\xi_1^3 \xi_2 +6 \Lambda_{1122} \xi_1^2 \xi_2^2 + 4 \Lambda_{2212} \xi_1 \xi_2^3 + \Lambda_{2222} \xi_2^4, \\
&A_1(\xi_1,\xi_2) = 0, \\
&A_2(\xi_1,\xi_2) = 6 \left( \Lambda_{1133}\xi_1^2 + \Lambda_{2233}\xi_2^2 + 2 \Lambda_{3312} \xi_1 \xi_2  \right), \\
&A_3(\xi_1,\xi_2) = 0, \\
&A_4(\xi_1,\xi_2) = \Lambda_{3333},
\end{align}
\end{subequations}
where
\begin{equation*}
    \begin{split}
        \Lambda_{1111} ={}& 30 \, C_{1111} - 45 \, C_{1122} - 45 \, C_{1133} + \frac{15}{4} \, C_{2222} + \frac{15}{2} \, C_{2233} + \frac{15}{4} \, C_{3333}, \\
        \Lambda_{1122} ={}& -\frac{15}{2} \, C_{1111} + \frac{255}{4} \, C_{1122} - \frac{25}{4} \, C_{1133} - \frac{15}{2} \, C_{2222} - \frac{25}{4} \, C_{2233} + \frac{5}{4} \, C_{3333}, \\
        \Lambda_{1133} ={}& -\frac{15}{2} \, C_{1111} - \frac{25}{4} \, C_{1122} + \frac{255}{4} \, C_{1133} + \frac{5}{4} \, C_{2222} - \frac{25}{4} \, C_{2233} - \frac{15}{2} \, C_{3333}, \\
        \Lambda_{1112} ={}& \frac{105}{2} \, C_{1112} - \frac{105}{4} \, C_{2212} - \frac{105}{4} \, C_{3312}, \\
        \Lambda_{2222} ={}& \frac{15}{4} \, C_{1111} - 45 \, C_{1122} + \frac{15}{2} \, C_{1133} + 30 \, C_{2222} - 45 \, C_{2233} + \frac{15}{4} \, C_{3333}, \\
        \Lambda_{2233} ={}& \frac{5}{4} \, C_{1111} - \frac{25}{4} \, C_{1122} - \frac{25}{4} \, C_{1133} - \frac{15}{2} \, C_{2222} + \frac{255}{4} \, C_{2233} - \frac{15}{2} \, C_{3333}, \\
        \Lambda_{2212} ={}& -\frac{105}{4} \, C_{1112} + \frac{105}{2} \, C_{2212} - \frac{105}{4} \, C_{3312}, \\
        \Lambda_{3333} ={}& \frac{15}{4} \, C_{1111} + \frac{15}{2} \, C_{1122} - 45 \, C_{1133} + \frac{15}{4} \, C_{2222} - 45 \, C_{2233} + 30 \, C_{3333}, \\
        \Lambda_{3312} ={}& -\frac{35}{4} \, C_{1112} - \frac{35}{4} \, C_{2212} + 70 \, C_{3312}.
    \end{split}
\end{equation*}
Since $A_1(\xi_1,\xi_2) = A_3(\xi_1,\xi_2) = 0$, we immediately have $\lambda_1(\xi_1,\xi_2) = 0$ as expected for peridynamic plane strain ({\em cf.} \eqref{eqn:planestrainmodelinplane} and \eqref{eqn:planestrainmodeloutofplane}). 

\underline{Orthotropic}: We substitute \eqref{eqn:orthoelastrelations}, with Cauchy's relations imposed, into \eqref{eqn:SijklToCijklRelations3D} to find ({\em cf.}~\eqref{eqn:PeriPlaneStrainAList})

\begin{subequations}\label{eqn:AtermsOrtho}
\begin{align}
&A_0(\xi_1,\xi_2) = \Lambda_{1111} \xi_1^4 +6 \Lambda_{1122} \xi_1^2 \xi_2^2 + \Lambda_{2222} \xi_2^4, \\
&A_1(\xi_1,\xi_2) = 0, \\
&A_2(\xi_1,\xi_2) = 6 \left( \Lambda_{1133}\xi_1^2 + \Lambda_{2233}\xi_2^2  \right), \\
&A_3(\xi_1,\xi_2) = 0, \\
&A_4(\xi_1,\xi_2) = \Lambda_{3333},
\end{align}
\end{subequations}
where
\begin{equation*}
    \begin{split}
        \Lambda_{1111} ={}& 30  C_{1111} - 45  C_{1122} - 45  C_{1133} + \frac{15}{4}  C_{2222} + \frac{15}{2}  C_{2233} + \frac{15}{4}  C_{3333}, \\
        \Lambda_{1122} ={}& -\frac{15}{2}  C_{1111} + \frac{255}{4}  C_{1122} - \frac{25}{4}  C_{1133} - \frac{15}{2}  C_{2222} - \frac{25}{4}  C_{2233} + \frac{5}{4}  C_{3333}, \\
        \Lambda_{1133} ={}& -\frac{15}{2}  C_{1111} - \frac{25}{4}  C_{1122} + \frac{255}{4}  C_{1133} + \frac{5}{4}  C_{2222} - \frac{25}{4}  C_{2233} - \frac{15}{2}  C_{3333}, \\
        \Lambda_{2222} ={}& \frac{15}{4}  C_{1111} - 45  C_{1122} + \frac{15}{2}  C_{1133} + 30  C_{2222} - 45  C_{2233} + \frac{15}{4}  C_{3333}, \\
        \Lambda_{2233} ={}& \frac{5}{4}  C_{1111} - \frac{25}{4}  C_{1122} - \frac{25}{4}  C_{1133} - \frac{15}{2}  C_{2222} + \frac{255}{4}  C_{2233} - \frac{15}{2}  C_{3333}, \\
        \Lambda_{3333} ={}& \frac{15}{4}  C_{1111} + \frac{15}{2}  C_{1122} - 45  C_{1133} + \frac{15}{4}  C_{2222} - 45  C_{2233} + 30 C_{3333}. \\
    \end{split}
\end{equation*}

\underline{Trigonal}: We substitute \eqref{eqn:trigonalelastrelations}, with Cauchy's relations imposed, into \eqref{eqn:SijklToCijklRelations3D} to find ({\em cf.}~\eqref{eqn:PeriPlaneStrainAList})

\begin{subequations}\label{eqn:AtermsTrigonal}
\begin{align}
&A_0(\xi_1,\xi_2) = \Lambda_{1111} \xi_1^4 + 6 \Lambda_{1122} \xi_1^2 \xi_2^2 + 4 \Lambda_{2212} \xi_1 \xi_2^3 + \Lambda_{2222} \xi_2^4, \\
&A_1(\xi_1,\xi_2) = 0, \\
&A_2(\xi_1,\xi_2) = 2 \left( 3 \Lambda_{1122}\xi_1^2 + \Lambda_{2222}\xi_2^2 - 6 \Lambda_{2212} \xi_1 \xi_2  \right), \\
&A_3(\xi_1,\xi_2) = 0, \\
&A_4(\xi_1,\xi_2) = \Lambda_{2222},
\end{align}
\end{subequations}
where
\begin{equation*}
    \begin{split}
        \Lambda_{1111} ={}& 30  C_{1111} - 90  C_{1122} + 10  C_{2222}, \\
        \Lambda_{1122} ={}& -\frac{15}{2}  C_{1111} + \frac{115}{2}  C_{1122} - \frac{25}{3}  C_{2222}, \\
        \Lambda_{2222} ={}& \frac{15}{4}  C_{1111} - \frac{75}{2}  C_{1122} + \frac{75}{4}  C_{2222}, \\
        \Lambda_{2212} ={}& \frac{315}{4}  C_{2212}. \\
    \end{split}
\end{equation*}
%where the $\Lambda_{ijkl}$ are given by \eqref{eqn:SijklToCijklRelations3D}. 

\underline{Tetragonal}: We substitute \eqref{eqn:tetraelastrelations}, with Cauchy's relations imposed, into \eqref{eqn:SijklToCijklRelations3D} to find ({\em cf.}~\eqref{eqn:PeriPlaneStrainAList})

\begin{subequations}\label{eqn:AtermsTetra}
\begin{align}
&A_0(\xi_1,\xi_2) = \Lambda_{1111} (\xi_1^4 + \xi_2^4) + 6 \Lambda_{1122} \xi_1^2 \xi_2^2, \\
&A_1(\xi_1,\xi_2) = 0, \\
&A_2(\xi_1,\xi_2) = 6 \Lambda_{1133}r^2, \\
&A_3(\xi_1,\xi_2) = 0, \\
&A_4(\xi_1,\xi_2) = \Lambda_{3333},
\end{align}
\end{subequations}
where
\begin{equation*}
    \begin{split}
        \Lambda_{1111} ={}& \frac{135}{4}  C_{1111} - 45  C_{1122} - \frac{75}{2}  C_{1133} + \frac{15}{4}  C_{3333}, \\
        \Lambda_{1122} ={}& -15  C_{1111} + \frac{255}{4}  C_{1122} - \frac{25}{2}  C_{1133} + \frac{5}{4}  C_{3333}, \\
        \Lambda_{1133} ={}& -\frac{25}{4}  C_{1111} - \frac{25}{4}  C_{1122} + \frac{115}{2}  C_{1133} - \frac{15}{2}  C_{3333}, \\
        \Lambda_{3333} ={}& \frac{15}{2}  C_{1111} + \frac{15}{2}  C_{1122} - 90  C_{1133} + 30 C_{3333}.
    \end{split}
\end{equation*}

\underline{Transversely Isotropic}: We substitute \eqref{eqn:tisoelastrelations}, with Cauchy's relations imposed, into \eqref{eqn:SijklToCijklRelations3D} to find ({\em cf.}~\eqref{eqn:PeriPlaneStrainAList})

\begin{subequations}\label{eqn:AtermsTiso}
\begin{align}
&A_0(\xi_1,\xi_2) = \Lambda_{1111} r^4, \\
&A_1(\xi_1,\xi_2) = 0, \\
&A_2(\xi_1,\xi_2) = 6 \Lambda_{1133}r^2, \\
&A_3(\xi_1,\xi_2) = 0, \\
&A_4(\xi_1,\xi_2) = \Lambda_{3333},
\end{align}
\end{subequations}
where
\begin{equation*}
    \begin{split}
        \Lambda_{1111} ={}& \frac{75}{4}  C_{1111} - \frac{75}{2}  C_{1133} + \frac{15}{4}  C_{3333}, \\
        \Lambda_{1133} ={}& -\frac{25}{3}  C_{1111} + \frac{115}{2}  C_{1133} - \frac{15}{2}  C_{3333}, \\
        \Lambda_{3333} ={}& 10  C_{1111} - 90  C_{1133} + 30  C_{3333}.
    \end{split}
\end{equation*}

\underline{Cubic}: We substitute \eqref{eqn:cubicelastrelations}, with Cauchy's relations imposed, into \eqref{eqn:SijklToCijklRelations3D} to find ({\em cf.}~\eqref{eqn:PeriPlaneStrainAList})

\begin{subequations}\label{eqn:AtermsCubic}
\begin{align}
&A_0(\xi_1,\xi_2) = \Lambda_{1111} (\xi_1^4 + \xi_2^4) + 6 \Lambda_{1122} \xi_1^2 \xi_2^2, \\
&A_1(\xi_1,\xi_2) = 0, \\
&A_2(\xi_1,\xi_2) = 6 \Lambda_{1122}r^2, \\
&A_3(\xi_1,\xi_2) = 0, \\
&A_4(\xi_1,\xi_2) = \Lambda_{1111},
\end{align}
\end{subequations}
where
\begin{equation*}
    \begin{split}
        \Lambda_{1111} ={}& \frac{75}{2} C_{1111} - \frac{165}{2} C_{1122}, \\
        \Lambda_{1122} ={}& -\frac{55}{4} C_{1111} + \frac{205}{4} C_{1122}.
    \end{split}
\end{equation*}
%where the $\Lambda_{ijkl}$ are given by \eqref{eqn:SijklToCijklRelations3D}. 

\underline{Isotropic}: We substitute \eqref{eqn:iso3Delastrelations}, with Cauchy's relations imposed, into \eqref{eqn:SijklToCijklRelations3D} to find ({\em cf.}~\eqref{eqn:PeriPlaneStrainAList})

\begin{subequations}\label{eqn:AtermsIso}
\begin{align}
&A_0(\xi_1,\xi_2) = \Lambda_{1111} r^4, \\
&A_1(\xi_1,\xi_2) = 0, \\
&A_2(\xi_1,\xi_2) = 2 \Lambda_{1111} r^2, \\
&A_3(\xi_1,\xi_2) = 0, \\
&A_4(\xi_1,\xi_2) = \Lambda_{1111},
\end{align}
\end{subequations}
where
\begin{equation*}
    \Lambda_{1111} = 10 C_{1111}.
\end{equation*}

\begin{remark}
In classical linear elasticity, the plane strain model for each symmetry class reduces identically to a corresponding two-dimensional model ({\em cf}.~Table~\ref{tab:modelequivalenceplanestrainstressto2D}). A similar situation occurs with the peridynamic plane strain model. While the plane strain micromodulus function $\lambda_0$, given by~\eqref{eqn:PlaneStrainLambda_iForw=1overxia}, is not equivalent to the two-dimensional micromodulus function \eqref{eqn:lambdaobl}, it does possess one of the four symmetries of two-dimensional classical linear elasticity ({\em cf}. Theorem \ref{thm:symmetryclasses}) for each symmetry class ({\em cf}.~Table~\ref{tab:micromodulusequivalenceplanestrainstressto2D}). This can be observed by considering the symmetries of $\left\{ A_i(\xi_1,\xi_2) \right\}$ for each symmetry class. The reason why the plane strain micromodulus functions are not identical to their two-dimensional counterparts is that they incorporate out-of-plane information. It is also interesting to note that the same correspondence between each three-dimensional symmetry class and the corresponding two-dimensional symmetry class in Table~\ref{tab:modelequivalenceplanestrainstressto2D} occurs for the plane strain micromodulus function, which is summarized in Table \ref{tab:micromodulusequivalenceplanestrainstressto2D}. However, while different plane strain micromodulus functions may possess the same two-dimensional symmetry, e.g., tetragonal and cubic micromodulus functions both have square symmetry, the resulting plane strain micromodulus functions are unique for each three-dimensional symmetry class. Due to this fact, the resulting in-plane plane strain equations of motion are unique for each three-dimensional symmetry class. This is in contrast to classical linear elasticity, where the in-plane plane strain equations of motion may be identical for two three-dimensional symmetry classes, specifically tetragonal and cubic models as well as transversely isotropic and isotropic models.
\end{remark}

\begin{table}
\begin{center}
\begin{tabular}{|c|c|}
\hline
\textbf{Pure Two-Dimensional Peridynamic Model} & \textbf{Peridynamic Plane Strain Model} \\
\hline
Oblique & Monoclinic and Trigonal \\
\hline
Rectangular & Orthotropic \\
\hline
Square & Tetragonal and Cubic \\
\hline
Isotropic & Transversely Isotropic and Isotropic \\
\hline
\end{tabular}
\caption{Symmetry equivalence between the pure two-dimensional peridynamic models and the peridynamic plane strain models (provided the micromodulus functions $\lambda_0$ and $\lambda_2$ is informed by three-dimensional elasticity tensors having the form of those in Section~\ref{sec:threedimclassicalelasticity}).}\label{tab:micromodulusequivalenceplanestrainstressto2D}
\end{center}
\end{table}

It is of interest to compare our resulting peridynamic plane strain micromodulus function \eqref{eqn:PlaneStrainLambda_iForw=1overxia} with plane strain micromodulus functions commonly utilized in the peridynamic literature. This can be directly done in the case of isotropy. The isotropic micromodulus functions are actually quite simple. When the influence function $\omega(\| \bfxi \|) = 1$, $m = \frac{4}{5} \pi \delta^5 $ ({\em cf.} \eqref{def:weightvolumem}) and \eqref{eqn:PlaneStrainLambda_i} simplifies to
\begin{subequations}\label{eqn:PlaneStrainLambda_iForw=1iso}
\begin{align}
\lambda_0 ={}& \frac{25 C_{1111}}{\pi \delta^5 r} \arctan\left(\sqrt{\frac{\delta^2}{r^2} - 1}\right), %\frac{20 C_{1111}}{mr} \arctan\left(\sqrt{\frac{\delta^2}{r^2} - 1}\right), % = \frac{20 C_{1111}}{mr} \arctan\left(\sqrt{\frac{\delta^2}{r^2} - 1}\right) 
\label{eqn:PlaneStrainLambda_iForw=1isolamb0} 
\\
%%%%%%%%%%%%%%%%%%%%%%%%%%%%%%%%%%%%%%%%%%%%%%%%%%%%%%
\lambda_1 ={}& 0,  \\
%%%%%%%%%%%%%%%%%%%%%%%%%%%%%%%%%%%%%%%%%%%%%%%%%%%%%%
\lambda_2 ={}& \frac{25 C_{1111}}{\pi \delta^4 } \left[  \sqrt{1- \frac{r^2}{\delta^2}} - \frac{r}{\delta} \arctan\left(\sqrt{\frac{\delta^2}{r^2} - 1}\right) \right]. %\frac{20 C_{1111} r}{m} \left[ \sqrt{\frac{\delta^2}{r^2} - 1} - \arctan\left(\sqrt{\frac{\delta^2}{r^2} - 1}\right) \right].
%=& \frac{20 C_{1111}r }{m} \left( \sqrt{\frac{\delta^2}{r^2} - 1} - \arctan\left(\sqrt{\frac{\delta^2}{r^2} - 1}\right) \right).
\end{align}
\end{subequations}

Alternatively, when the influence function $\omega(\| \bfxi \|) = \frac{1}{\| \bfxi \|}$, $m = \pi \delta^4$ ({\em cf.} \eqref{def:weightvolumem}) and \eqref{eqn:PlaneStrainLambda_i} simplifies to

\begin{subequations}\label{eqn:PlaneStrainLambda_iForw=1overxiiso}
\begin{align}
\lambda_0 ={}&  \frac{20 C_{1111} }{ \pi \delta^4 r^2} \sqrt{1-\frac{r^2}{\delta^2}} %\frac{20 C_{1111} }{m r^2} \sqrt{1-\frac{r^2}{\delta^2}} %=  \frac{20 C_{1111} }{m r^2} \sqrt{1-\frac{r^2}{\delta^2}} 
\label{eqn:PlaneStrainLambda_iForw=1overxiisolamb0}, \\
\lambda_1 ={}& 0, \\
\lambda_2 ={}& \frac{20 C_{1111}}{\pi \delta^4} \left[ \text{arsinh} \left( \sqrt{\frac{\delta^2}{r^2} - 1}\right) - \sqrt{1-\frac{r^2}{\delta^2}} \, \right]. %\frac{20 C_{1111}}{m} \left[ \text{arsinh} \left( \sqrt{\frac{\delta^2}{r^2} - 1}\right) - \sqrt{1-\frac{r^2}{\delta^2}} \, \right]. %\\
%=& \frac{20 C_{1111}}{m} \left( \text{arsinh} \left( \sqrt{\frac{\delta^2}{r^2} - 1}\right) - \sqrt{1-\frac{r^2}{\delta^2}} \right)
\end{align}
\end{subequations}

In terms of the engineering constants, \eqref{eqn:PlaneStrainLambda_iForw=1isolamb0}  and \eqref{eqn:PlaneStrainLambda_iForw=1overxiisolamb0} are given, respectively, by:
\begin{equation}\label{eqn:isoplanestrainkernelomega13D}
    \lambda_0 = \frac{30 E }{\pi \delta^5 r} \arctan\left(\sqrt{\frac{\delta^2}{r^2} - 1}\right)
\end{equation}
and
\begin{equation}\label{eqn:isoplanestrainkernelomega1overxi3D}
    \lambda_0 = \frac{24 E }{\pi \delta^4 r^2} \sqrt{1-\frac{r^2}{\delta^2}},
\end{equation}
where $E$ is the Young's modulus. Here, we used the fact that for an isotropic material with Cauchy's relations satisfied, Poisson's ratio $\nu = \frac{1}{4}$ and thus~({\em cf.}~\eqref{eqn:monoclinicInTermsofTechnicalConstants})
\[
C_{1111}~=~\frac{(1-\nu) E}{(1+\nu)(1-2\nu)} = \frac{6E}{5}.
\] 

A common approach for modeling plane strain in the peridynamic literature simply employs a two-dimensional peridynamic model and matches the model constants to those in the in-plane equations of the corresponding classical plane strain model (see, e.g., \cite{Gerstle2005,Ghajari2014,Le2014,Sarego2016}). To emulate this approach, we consider our two-dimensional micromodulus function \eqref{eqn:lambdaobl}, where $m$ is the two-dimensional weighted volume and $\Lambda_{ijkl}$ are given by \eqref{eqn:SijkltoCijkl}, except that the $C_{ijkl}$ are the three-dimensional elasticity constants appearing in the classical in-plane equations of motion \eqref{eqn:classicalmonoeqnmotion}. In the isotropic case, we combine \eqref{eqn:lambdaiso} with \eqref{eqn:isosymmSijklrelations} and recall for isotropic symmetry that $C_{1111} = \frac{6E}{5}$ when Cauchy's relations are imposed. When the influence function $\omega(r) = 1$, $m = \frac{1}{2} \pi \delta^4$ and the micromodulus function is given by
\begin{equation}\label{eqn:isoplanestrainkernelomega12D}
    \lambda(r) = \frac{64 E}{5\pi \delta^4} \frac{1}{r^2}.
\end{equation}
Alternatively, when the influence function $\omega(r) = \frac{1}{r}$, $m = \frac{2}{3} \pi \delta^3$ and the micromodulus function is given by
\begin{equation}\label{eqn:isoplanestrainkernelomega1overxi2D}
    \lambda(r) = \frac{48 E}{5\pi \delta^3} \frac{1}{r^3}.
\end{equation}
The micromodulus function in \eqref{eqn:isoplanestrainkernelomega1overxi2D} coincides with the micromodulus function obtained by linearizing the peridynamic plane strain model derived in~\cite{Gerstle2005}. We therefore refer to the two-dimensional peridynamic micromodulus function~\eqref{eqn:lambdaobl} with constants matched to those in the in-plane equations of motion~\eqref{eqn:classicalmonoeqnmotion} of the corresponding classical plane strain model as the \textit{traditional plane strain micromodulus function}. We observe that comparing~\eqref{eqn:isoplanestrainkernelomega13D} with~\eqref{eqn:isoplanestrainkernelomega12D} and~\eqref{eqn:isoplanestrainkernelomega1overxi3D} with~\eqref{eqn:isoplanestrainkernelomega1overxi2D}, our plane strain micromodulus functions possess a weaker singularity than the traditional plane strain micromodulus functions. Moreover, our plane strain micromodulus functions continuously transition to zero when $r$ approaches $\delta$. Next, we further explore the differences between our plane strain micromodulus functions and the traditional plane strain micromodulus functions for various symmetry classes.

\textbf{Micromodulus function visualization for peridynamic plane strain}

In order to more easily compare our plane strain micromodulus functions \eqref{eqn:PlaneStrainLambda_iForw=1overxia} and the traditional plane strain micromodulus functions based on \eqref{eqn:lambdaobl}, we create visualizations of their behavior for various symmetry classes. We observe that in the peridynamic equation of motion \eqref{eqn:linearperieqncomponentform}, the micromodulus function is multiplied by~$\xi_i \xi_j$. Consequently, the factor of $\frac{1}{\| \bfxi \|^2}$  appearing in the micromodulus function \eqref{eqn:lambdaobl} represents a removable (or ``artificial") singularity of the model. To better visualize the nonremovable singularities, when present, in the micromodulus function, we plot $r^2 \lambda(\xi_1,\xi_2)$ for the traditional plane strain micromodulus function (recall $\| \bfxi \|^2 = r^2$ in two dimensions). For a proper comparison to the plane strain micromodulus function \eqref{eqn:PlaneStrainLambda_iForw=1overxia}, we accordingly plot $r^2\lambda_0(\xi_1,\xi_2)$.   

Our comparison of the micromodulus functions covers materials in every symmetry class except triclinic, because the plain strain Assumption (P$\varepsilon$\ref{assmp:periplanestrain3full}) requires a plane of reflection symmetry. Since we are dealing with bond-based peridynamic models, we must consider materials satisfying (at least approximately) Cauchy's relations. The elasticity tensors for all anisotropic materials were obtained from \cite{de2015charting}. For isotropic symmetry, the elasticity tensor was obtained from \cite{Stookey1959}. In Table~\ref{tab:MaterialProp}, we present the chosen materials and the corresponding elasticity tensors for each symmetry class. All of the chosen anisotropic materials approximately satisfy Cauchy's relations. In order to produce elasticity tensors satisfying Cauchy's relations exactly, given an elasticity tensor $\bfC$, we produce a new elasticity tensor $\tilde{\bfC}$ in the following manner. For monoclinic, orthotropic, tetragonal, and cubic symmetries, we take
\begin{equation*}
    \tilde{C}_{ijkl} := \frac{C_{ijkl} + C_{ikjl}}{2}.
\end{equation*}
For trigonal symmetry, Cauchy's relations additionally impose $C_{2233} = \frac{1}{3} C_{2222}$, and thus we take
\begin{equation*}
\tilde{C}_{ijkl} = \left\{    \begin{array}{ll}
         \frac{1}{3} \left( C_{2233} + C_{2323} + \frac{1}{3} C_{2222} \right), & \left\{i,j,k,l \right\} = \left\{2,2,3,3 \right\} \\[5pt]
         \left( C_{2233} + C_{2323} + \frac{1}{3} C_{2222} \right), & i = j = k = l = 2 \\[5pt]
        \displaystyle\frac{C_{ijkl} + C_{ikjl}}{2}, & \text{else}
    \end{array} \right. .
\end{equation*}
For transversely isotropic symmetry, Cauchy's relations additionally impose $C_{1122} = \frac{1}{3} C_{1111}$, and thus we take
\begin{equation*}
\tilde{C}_{ijkl} = \left\{    \begin{array}{ll}
        \frac{1}{3} \left( C_{1122} + C_{1212} + \frac{1}{3} C_{1111} \right), & \left\{i,j,k,l \right\} = \left\{1,1,2,2 \right\}  \\[5pt]
        \left( C_{1122} + C_{1212} + \frac{1}{3} C_{1111} \right), & i = j = k = l = 1 \\[5pt]
        \displaystyle\frac{C_{ijkl} + C_{ikjl}}{2}, & \text{else}
    \end{array} \right. .
\end{equation*}
For isotropic symmetry, we selected the material Pyroceram 9608 which already satisfies Cauchy's relations; in this case, we simply take $\tilde{\bfC} = \bfC$. The resulting elasticity tensors $\tilde{\bfC}$ for each symmetry class are summarized in the fourth column of Table \ref{tab:MaterialProp}.

In Figure \ref{fig:planestrainmicromodulioneoverxi}, we present plots for $r^2 \lambda_0(\xi_1,\xi_2)$ with $\omega(\|\bfxi\|) = \frac{1}{\| \bfxi \|}$ (in green) and $r^2 \lambda(\xi_1,\xi_2)$ with $\omega(r) = \frac{1}{r}$ (in red). Similarly, in Figure \ref{fig:planestrainmicromodulione}, we present plots for $r^2 \lambda_0(\xi_1,\xi_2)$  with $\omega(\|\bfxi\|) = 1$ (in blue) and $r^2 \lambda(\xi_1,\xi_2)$ with $\omega(r) = 1$ (in orange). For each symmetry class, these functions are found by substituting the corresponding elasticity tensor $\tilde{\bfC}$ from Table \ref{tab:MaterialProp} into \eqref{eqn:PlaneStrainLambda_iForw=1overxia} (with~\eqref{eqn:PeriPlaneStrainAList} and~\eqref{eqn:SijklToCijklRelations3D}) and \eqref{eqn:lambdaobl} (with~\eqref{eqn:SijkltoCijkl}), respectively. For all of the plots, we took $\delta = 1$. One of the most obvious differences between the plots of $r^2 \lambda_0(\xi_1,\xi_2)$ and $r^2 \lambda(\xi_1,\xi_2)$ is their behavior near the origin. As remarked earlier for the isotropic case, our plane strain micromodulus functions have weaker singularities compared to the traditional plane strain micromodulus functions. This generalizes to all the symmetry classes, which is clearly observed in Figure \ref{fig:planestrainmicromodulioneoverxi}. Another significant difference between the two micromodulus functions is that our plane strain micromodulus functions effectively incorporate out-of-plane information. This is clearly seen in the isotropic and transversely isotropic cases in Figure \ref{fig:planestrainmicromodulione}.

\begin{table}
\begin{tabular}{|c|c|c|c|}
\hline
\rule{0pt}{4ex} \begin{tabular}{c} \textbf{Symmetry} \\ \textbf{Class} \end{tabular} & \textbf{Material} & \textbf{Elasticity Tensor $\bfC$} (in GPa) & \begin{tabular}{c} \textbf{Elasticity Tensor $\tilde{\bfC}$} (in GPa) \\ \textbf{(Cauchy's relations imposed)} \end{tabular}\\
\hline
Monoclinic & $\text{CoTeO}_4$ & {\small \arraycolsep=2pt $\left[
\begin{array}{cccccc}
135 & 19 & 54 & 0 & 0 & 42 \\
\cdot & 13 & 15 & 0 & 0 & 6 \\
\cdot & \cdot & 269 & 0 & 0 & 18 \\
\cdot & \cdot & \cdot & 14 & 25 & 0 \\
\cdot & \cdot & \cdot & \cdot & 66 & 0 \\
\cdot & \cdot & \cdot & \cdot & \cdot & 18
\end{array}
\right]$} & {\small \arraycolsep=2pt $\left[
\begin{array}{cccccc}
135 & 18.5 & 60 & 0 & 0 & 42 \\
\cdot & 13 & 14.5 & 0 & 0 & 6 \\
\cdot & \cdot & 269 & 0 & 0 & 21.5 \\
\cdot & \cdot & \cdot & 14.5 & 21.5 & 0 \\
\cdot & \cdot & \cdot & \cdot & 60 & 0 \\
\cdot & \cdot & \cdot & \cdot & \cdot & 18.5
\end{array}
\right]$ } \\
\hline
%%%%%%%%%%%%%%%%%%%%%%%%%%%%%%%%%%%%%%%%
%%%%%%%%%%%%%%%%%%%%%%%%%%%%%%%%%%%%%%%%
Orthotropic & $\text{Te}_2 \text{W}$ &
{\small \arraycolsep=2pt $\left[
\begin{array}{cccccc}
143 & 1 & 37 & 0 & 0 & 0 \\
\cdot & 3 & 3 & 0 & 0 & 0 \\
\cdot & \cdot & 102 & 0 & 0 & 0 \\
\cdot & \cdot & \cdot & 2 & 0 & 0 \\
\cdot & \cdot & \cdot & \cdot & 46 & 0 \\
\cdot & \cdot & \cdot & \cdot & \cdot & 1
\end{array}
\right]$} & {\small \arraycolsep=2pt $\left[
\begin{array}{cccccc}
143 & 1 & 41.5 & 0 & 0 & 0 \\
\cdot & 3 & 2.5 & 0 & 0 & 0 \\
\cdot & \cdot & 102 & 0 & 0 & 0 \\
\cdot & \cdot & \cdot & 2.5 & 0 & 0 \\
\cdot & \cdot & \cdot & \cdot & 41.5 & 0 \\
\cdot & \cdot & \cdot & \cdot & \cdot & 1
\end{array}
\right]$ } \\
\hline
%%%%%%%%%%%%%%%%%%%%%%%%%%%%%%%%%%%%%%%%
%%%%%%%%%%%%%%%%%%%%%%%%%%%%%%%%%%%%%%%%
Trigonal & $\text{Ta}_2\text{C}$ & {\small \arraycolsep=2pt $\left[
\begin{array}{cccccc}
493 & 141 & 141 & 0 & 0 & 0 \\
\cdot & 464 & 159 & 0 & 0 & 45 \\
\cdot & \cdot & 464 & 0 & 0 & -45 \\
\cdot & \cdot & \cdot & 153 & -45 & 0 \\
\cdot & \cdot & \cdot & \cdot & 125 & 0 \\
\cdot & \cdot & \cdot & \cdot & \cdot & 125
\end{array}
\right]$ } & {\small \arraycolsep=2pt $\left[
\begin{array}{cccccc}
493 & 133 & 133 & 0 & 0 & 0 \\
\cdot & 467 & 156 & 0 & 0 & 45 \\
\cdot & \cdot & 467 & 0 & 0 & -45 \\
\cdot & \cdot & \cdot & 156 & -45 & 0 \\
\cdot & \cdot & \cdot & \cdot & 133 & 0 \\
\cdot & \cdot & \cdot & \cdot & \cdot & 133
\end{array}
\right]$ }\\
\hline
%%%%%%%%%%%%%%%%%%%%%%%%%%%%%%%%%%%%%%%%
%%%%%%%%%%%%%%%%%%%%%%%%%%%%%%%%%%%%%%%%
Tetragonal & Si & {\small \arraycolsep=2pt $\left[
\begin{array}{cccccc}
212 & 70 & 58 & 0 & 0 & 0 \\
\cdot & 212 & 58 & 0 & 0 & 0 \\
\cdot & \cdot & 179 & 0 & 0 & 0 \\
\cdot & \cdot & \cdot & 58 & 0 & 0 \\
\cdot & \cdot & \cdot & \cdot & 58 & 0 \\
\cdot & \cdot & \cdot & \cdot & \cdot & 85
\end{array}
\right]$ } & {\small \arraycolsep=2pt $\left[
\begin{array}{cccccc}
212 & 77.5 & 58 & 0 & 0 & 0 \\
\cdot & 212 & 58 & 0 & 0 & 0 \\
\cdot & \cdot & 179 & 0 & 0 & 0 \\
\cdot & \cdot & \cdot & 58 & 0 & 0 \\
\cdot & \cdot & \cdot & \cdot & 58 & 0 \\
\cdot & \cdot & \cdot & \cdot & \cdot & 77.5
\end{array}
\right]$ } \\
\hline
%%%%%%%%%%%%%%%%%%%%%%%%%%%%%%%%%%%%%%%%
%%%%%%%%%%%%%%%%%%%%%%%%%%%%%%%%%%%%%%%%
\begin{tabular}{c}Transversely \\ Isotropic \end{tabular} & $\text{MoN}$ & {\small \arraycolsep=2pt $\left[
\begin{array}{cccccc}
499 & 177 & 235 & 0 & 0 & 0 \\
\cdot & 499 & 235 & 0 & 0 & 0 \\
\cdot & \cdot & 714 & 0 & 0 & 0 \\
\cdot & \cdot & \cdot & 241 & 0 & 0 \\
\cdot & \cdot & \cdot & \cdot & 241 & 0 \\
\cdot & \cdot & \cdot & \cdot & \cdot & 161
\end{array}
\right]$} & {\small \arraycolsep=2pt $\left[
\begin{array}{cccccc}
504 & 168 & 238 & 0 & 0 & 0 \\
\cdot & 504 & 238 & 0 & 0 & 0 \\
\cdot & \cdot & 714 & 0 & 0 & 0 \\
\cdot & \cdot & \cdot & 238 & 0 & 0 \\
\cdot & \cdot & \cdot & \cdot & 238 & 0 \\
\cdot & \cdot & \cdot & \cdot & \cdot & 168
\end{array}
\right]$ }\\
\hline
%%%%%%%%%%%%%%%%%%%%%%%%%%%%%%%%%%%%%%%%
%%%%%%%%%%%%%%%%%%%%%%%%%%%%%%%%%%%%%%%%
Cubic & \begin{tabular}{c} $\text{MgAl}_2\text{O}_4$ \\ (Spinel) \end{tabular} 
& {\small \arraycolsep=2pt $\left[
\begin{array}{cccccc}
252 & 145 & 145 & 0 & 0 & 0 \\
\cdot & 252 & 145 & 0 & 0 & 0 \\
\cdot & \cdot & 252 & 0 & 0 & 0 \\
\cdot & \cdot & \cdot & 142 & 0 & 0 \\
\cdot & \cdot & \cdot & \cdot & 142 & 0 \\
\cdot & \cdot & \cdot & \cdot & \cdot & 142
\end{array}
\right]$ } & {\small \arraycolsep=2pt $\left[ \begin{array}{cccccc}
252 & 143.5 & 143.5 & 0 & 0 & 0 \\
\cdot & 252 & 143.5 & 0 & 0 & 0 \\
\cdot & \cdot & 252 & 0 & 0 & 0 \\
\cdot & \cdot & \cdot & 143.5 & 0 & 0 \\
\cdot & \cdot & \cdot & \cdot & 143.5 & 0 \\
\cdot & \cdot & \cdot & \cdot & \cdot & 143.5
\end{array}
\right]$ }\\
\hline
%%%%%%%%%%%%%%%%%%%%%%%%%%%%%%%%%%%%%%%%
%%%%%%%%%%%%%%%%%%%%%%%%%%%%%%%%%%%%%%%%
Isotropic & \begin{tabular}{c}Pyroceram \\9608 \end{tabular} & {\small \arraycolsep=2pt $\left[
\begin{array}{cccccc}
103.2 & 34.4 & 34.4 & 0 & 0 & 0 \\
\cdot & 103.2 & 34.4 & 0 & 0 & 0 \\
\cdot & \cdot & 103.2 & 0 & 0 & 0 \\
\cdot & \cdot & \cdot & 34.4 & 0 & 0 \\
\cdot & \cdot & \cdot & \cdot & 34.4 & 0 \\
\cdot & \cdot & \cdot & \cdot & \cdot & 34.4
\end{array}
\right]$ } & {\small \arraycolsep=2pt $\left[
\begin{array}{cccccc}
103.2 & 34.4 & 34.4 & 0 & 0 & 0 \\
\cdot & 103.2 & 34.4 & 0 & 0 & 0 \\
\cdot & \cdot & 103.2 & 0 & 0 & 0 \\
\cdot & \cdot & \cdot & 34.4 & 0 & 0 \\
\cdot & \cdot & \cdot & \cdot & 34.4 & 0 \\
\cdot & \cdot & \cdot & \cdot & \cdot & 34.4
\end{array}
\right]$ } \\
\hline
\end{tabular}
\caption{Materials and their corresponding symmetry classes and elasticity tensors, $\bfC$ and $\tilde{\bfC}$.} 
\label{tab:MaterialProp}
\end{table}

\begin{figure}
\begin{tabular}{cc}
\begin{tabular}{cc}
\includegraphics[scale=0.45,trim={3cm 0 3cm 0.5cm},clip]{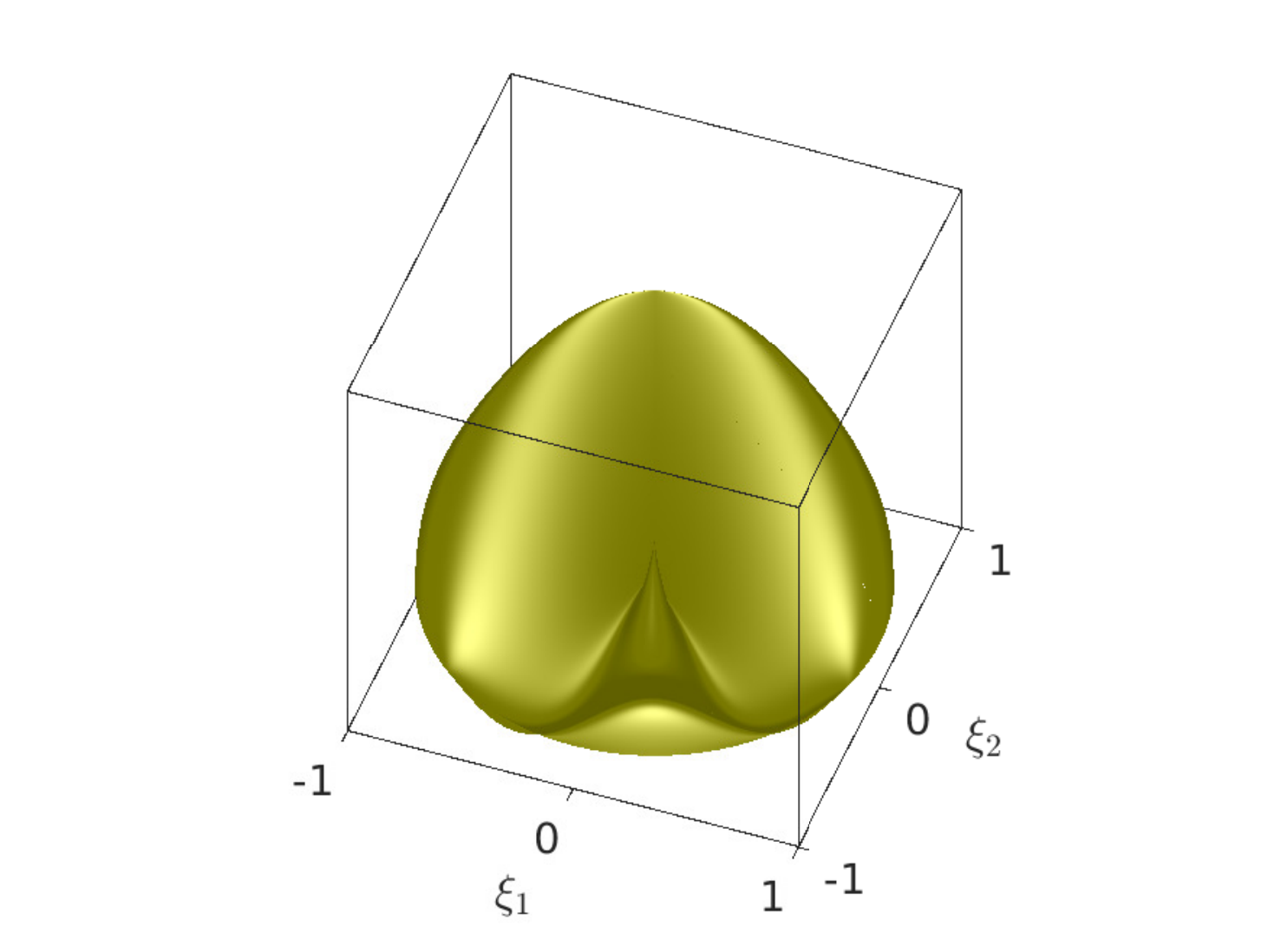} &
\includegraphics[scale=0.45,trim={3cm 0 3cm 0.5cm},clip]{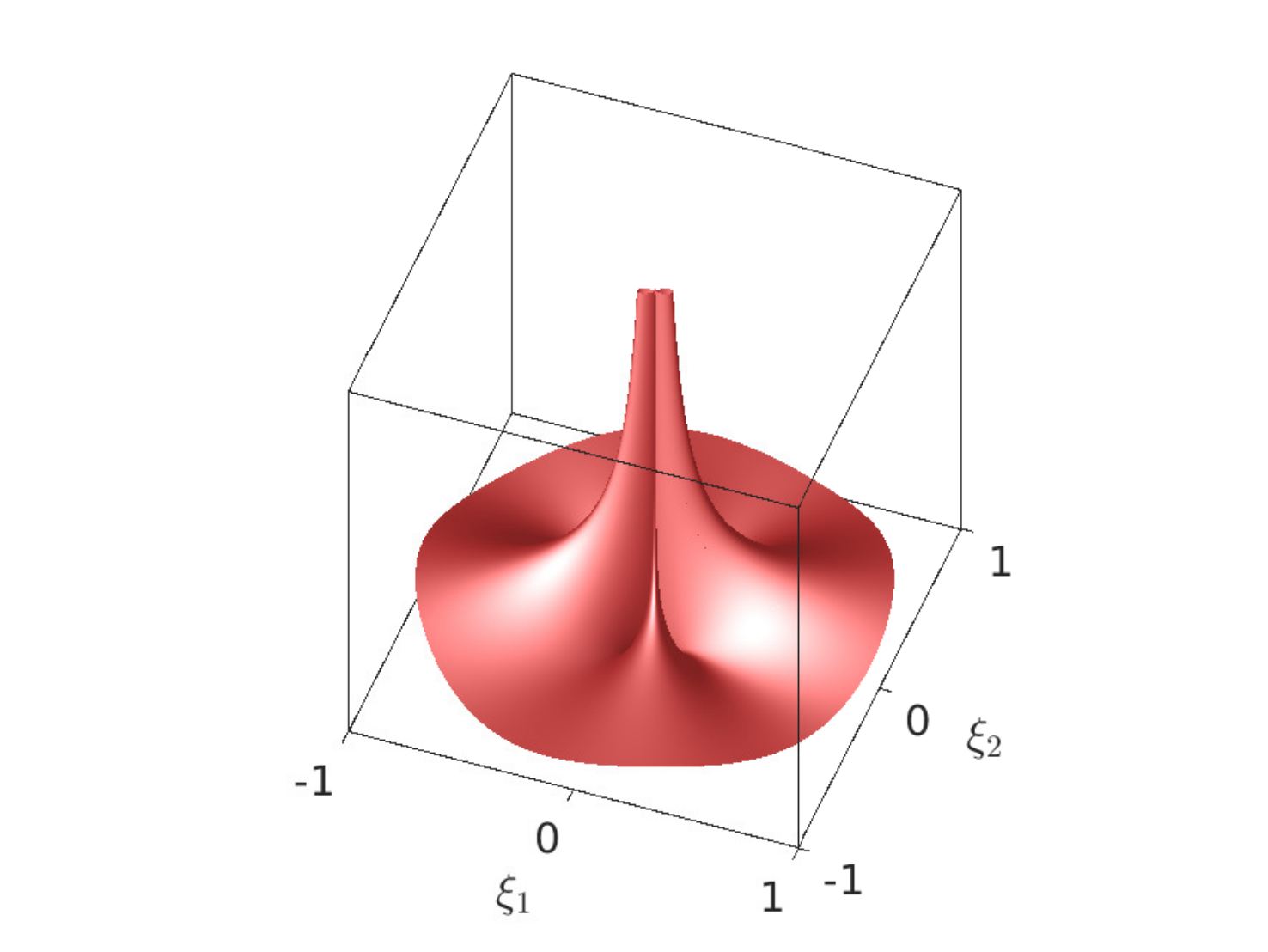}
\end{tabular} & \begin{tabular}{cc}
\includegraphics[scale=0.45,trim={3cm 0 3cm 0.5cm},clip]{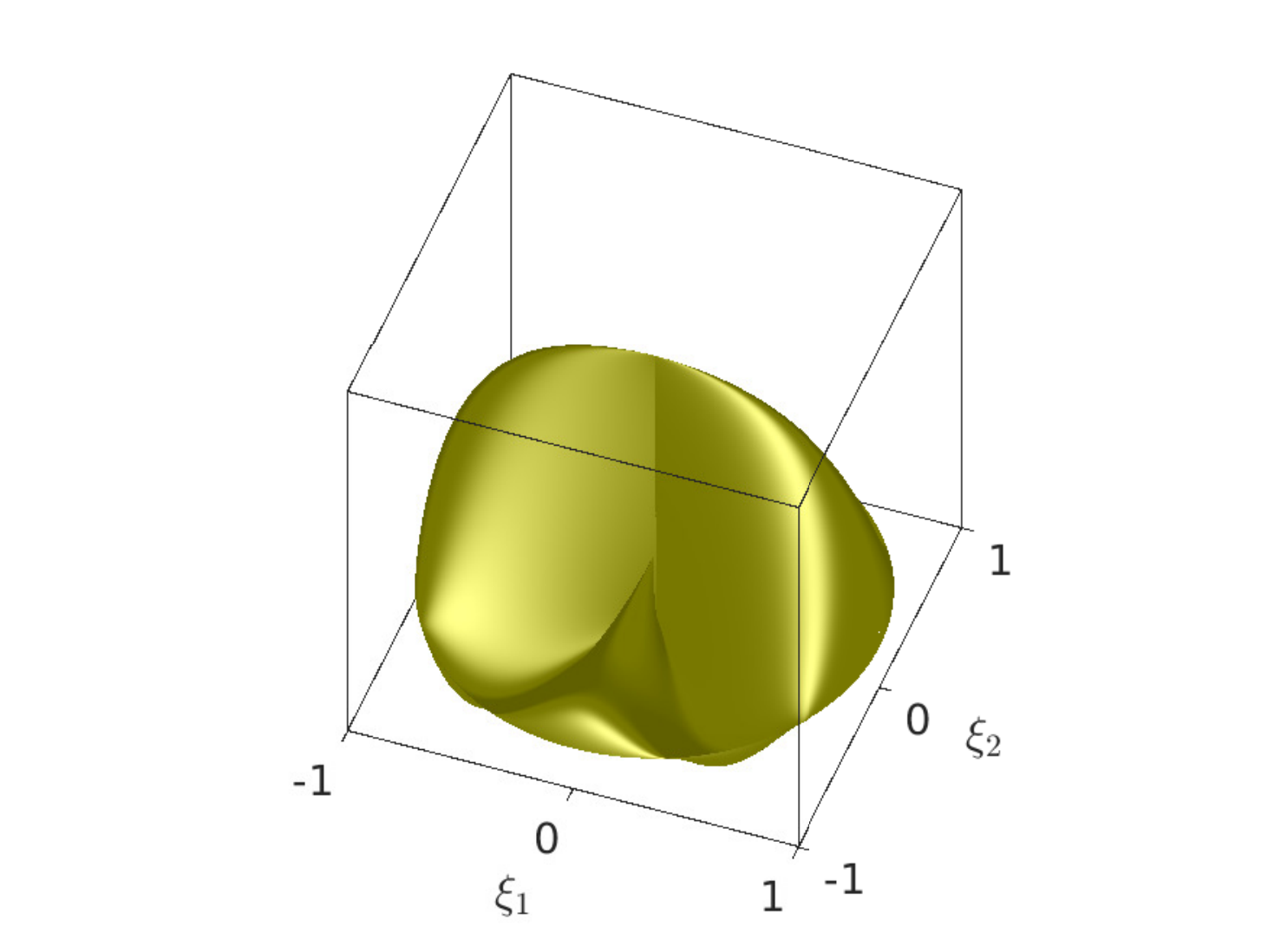} &
\includegraphics[scale=0.45,trim={3cm 0 3cm 0.5cm},clip]{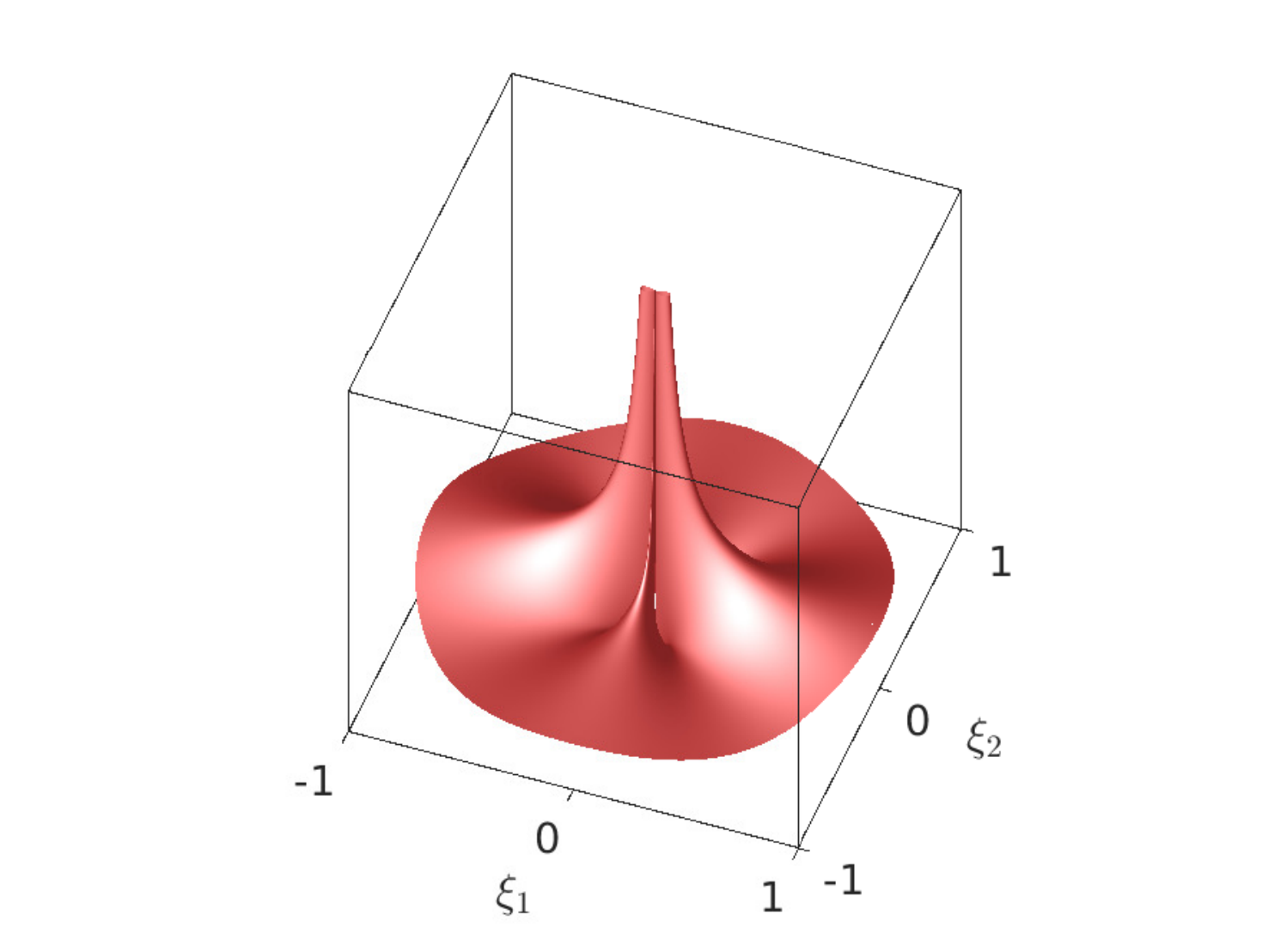}
\end{tabular} \\
Monoclinic ($\text{CoTeO}_4$) & Orthotropic ($\text{Te}_2\text{W}$) \\
%%%%%%%%%%%%%%%%%%%%%%%%%%%%%%%%%%%%%%%%%%%%%%%%%%%%%
\begin{tabular}{cc}
\includegraphics[scale=0.45,trim={3cm 0 3cm 0.5cm},clip]{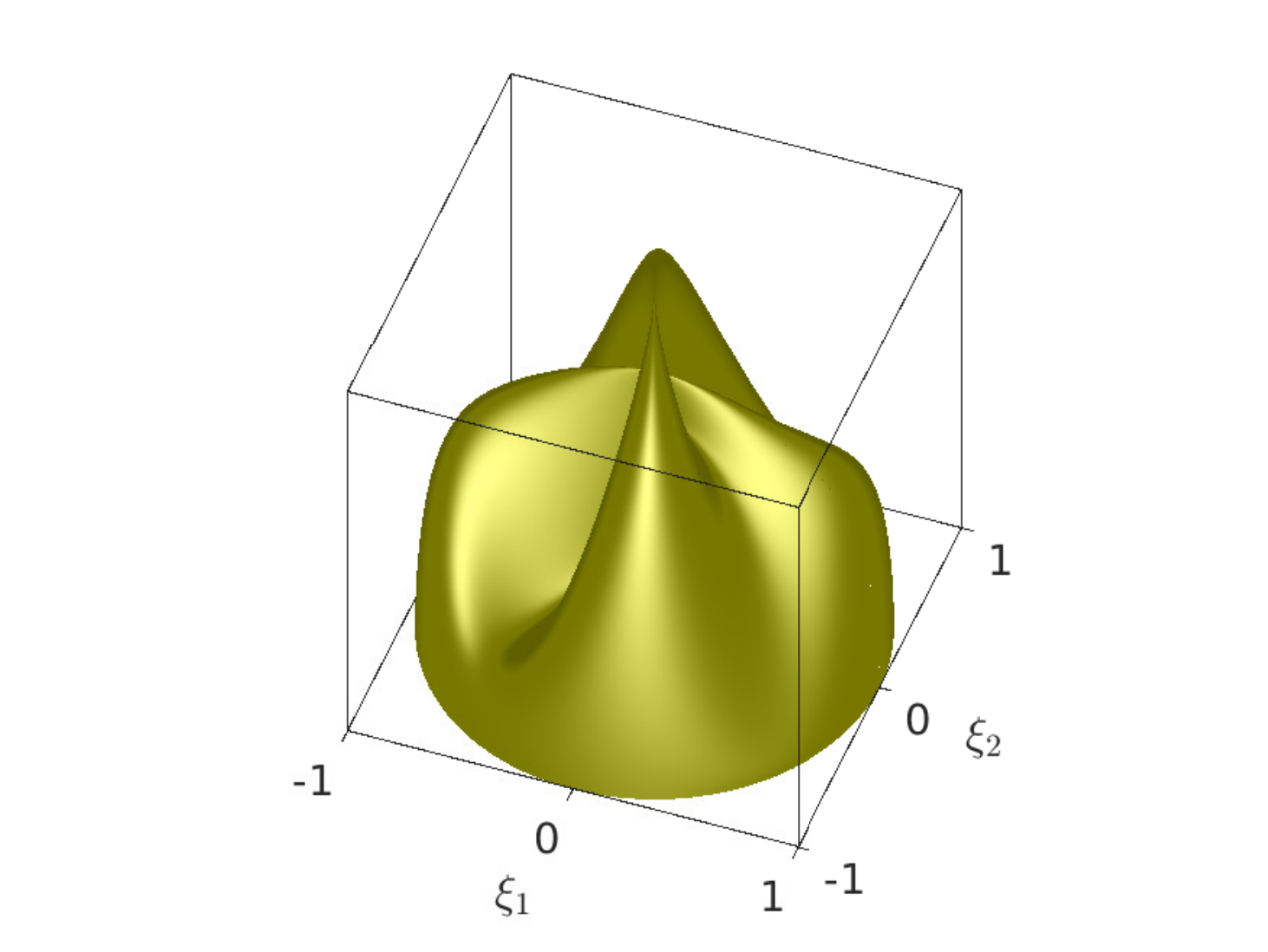} &
\includegraphics[scale=0.45,trim={3cm 0 3cm 0.5cm},clip]{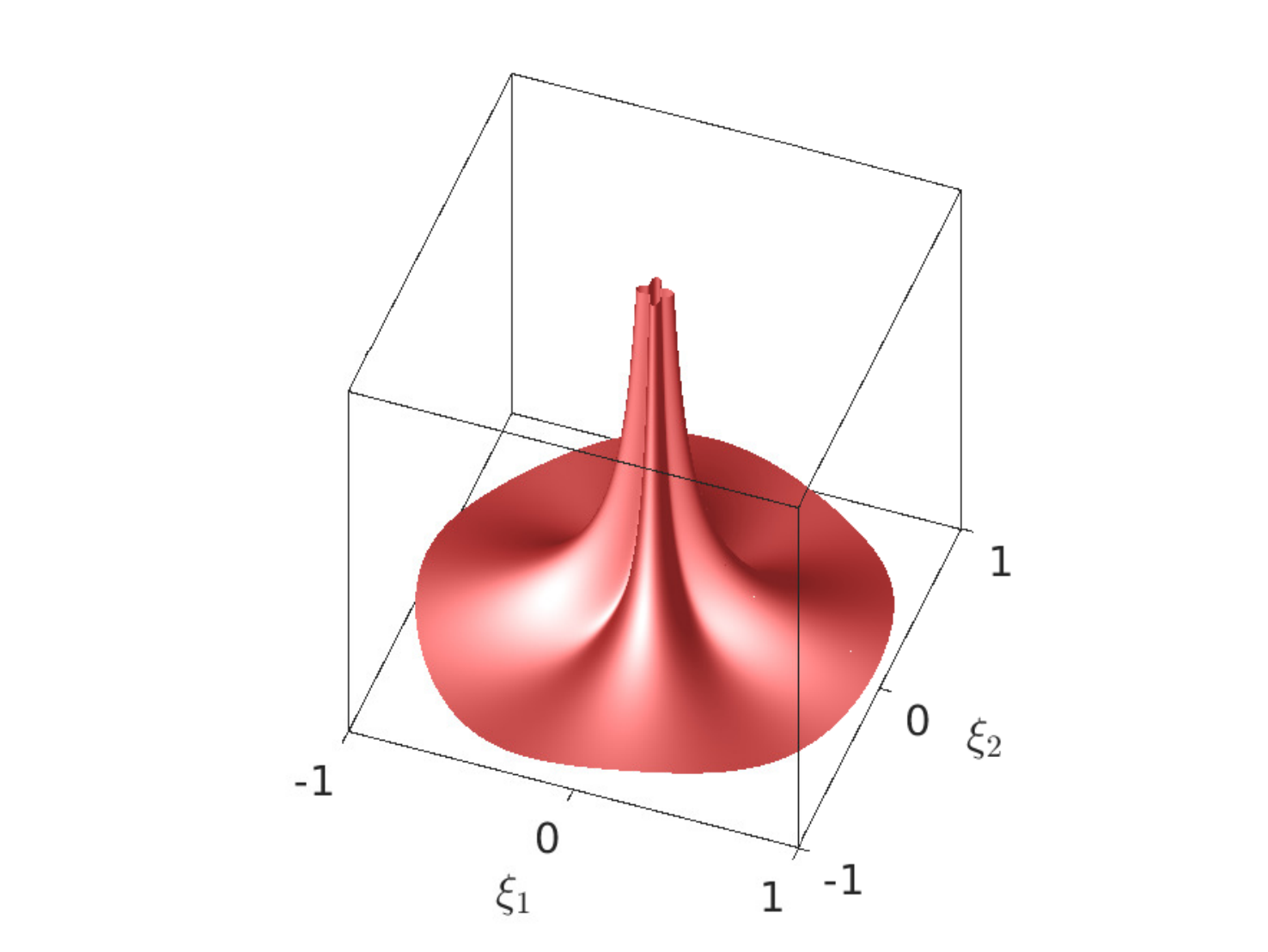}
\end{tabular} & \begin{tabular}{cc}
\includegraphics[scale=0.45,trim={3cm 0 3cm 0.5cm},clip]{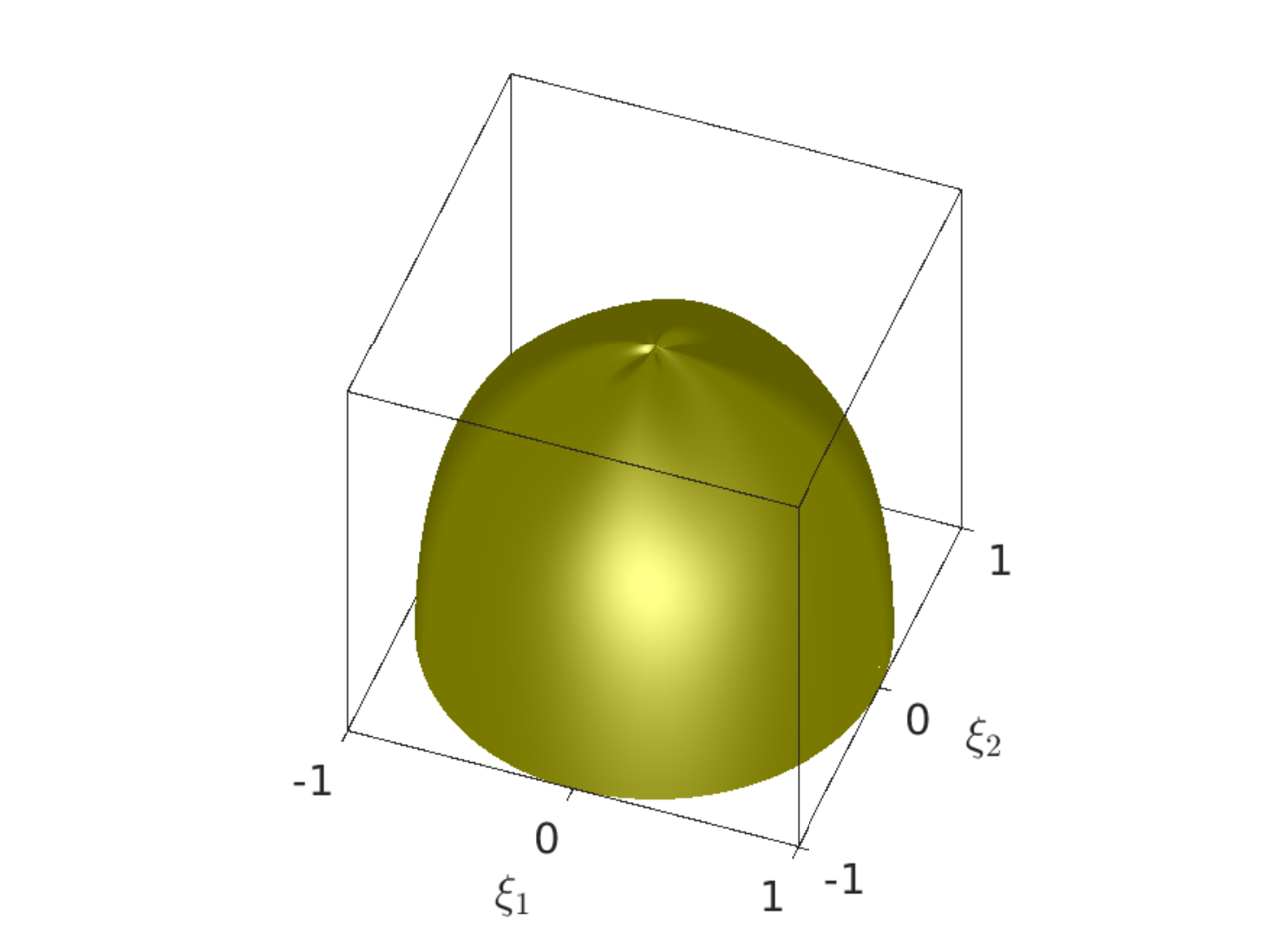} &
\includegraphics[scale=0.45,trim={3cm 0 3cm 0.5cm},clip]{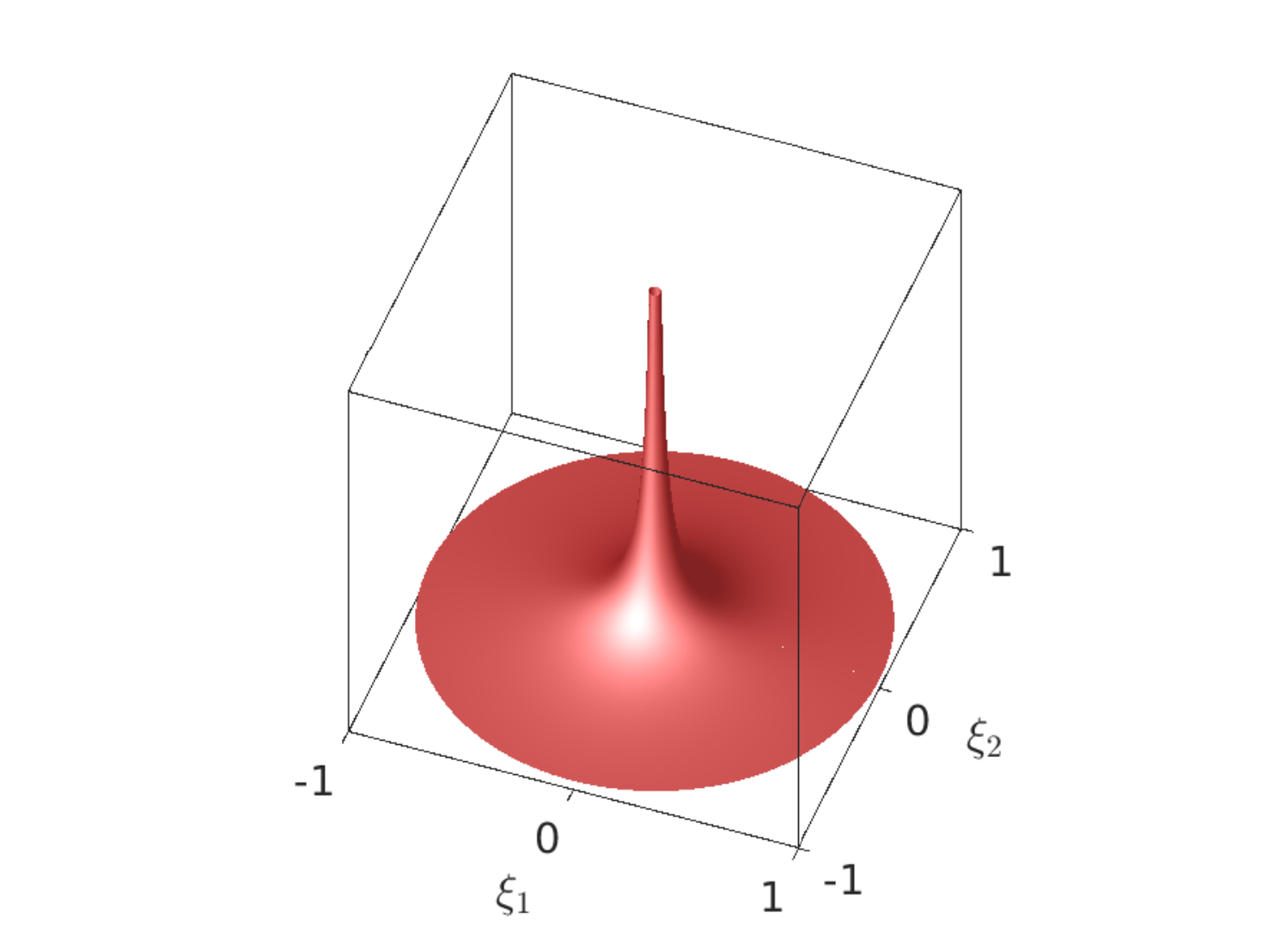}
\end{tabular} \\
Trigonal ($\text{Ta}_2\text{C}$) & Tetragonal (Si) \\
%%%%%%%%%%%%%%%%%%%%%%%%%%%%%%%%%%%%%%%%%%%%%%%%%%%%%
\begin{tabular}{cc}
\includegraphics[scale=0.45,trim={3cm 0 3cm 0.5cm},clip]{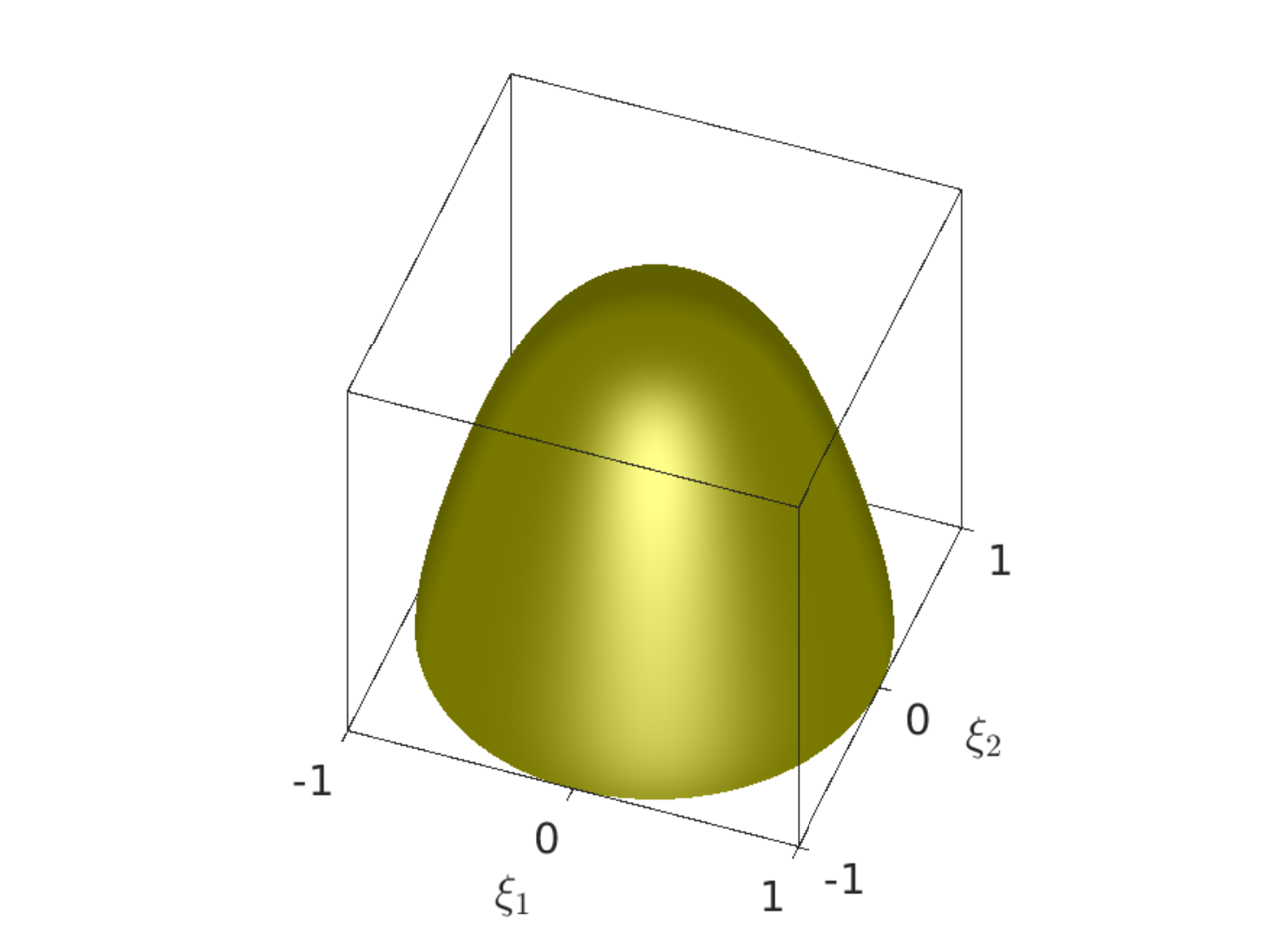} &
\includegraphics[scale=0.45,trim={3cm 0 3cm 0.5cm},clip]{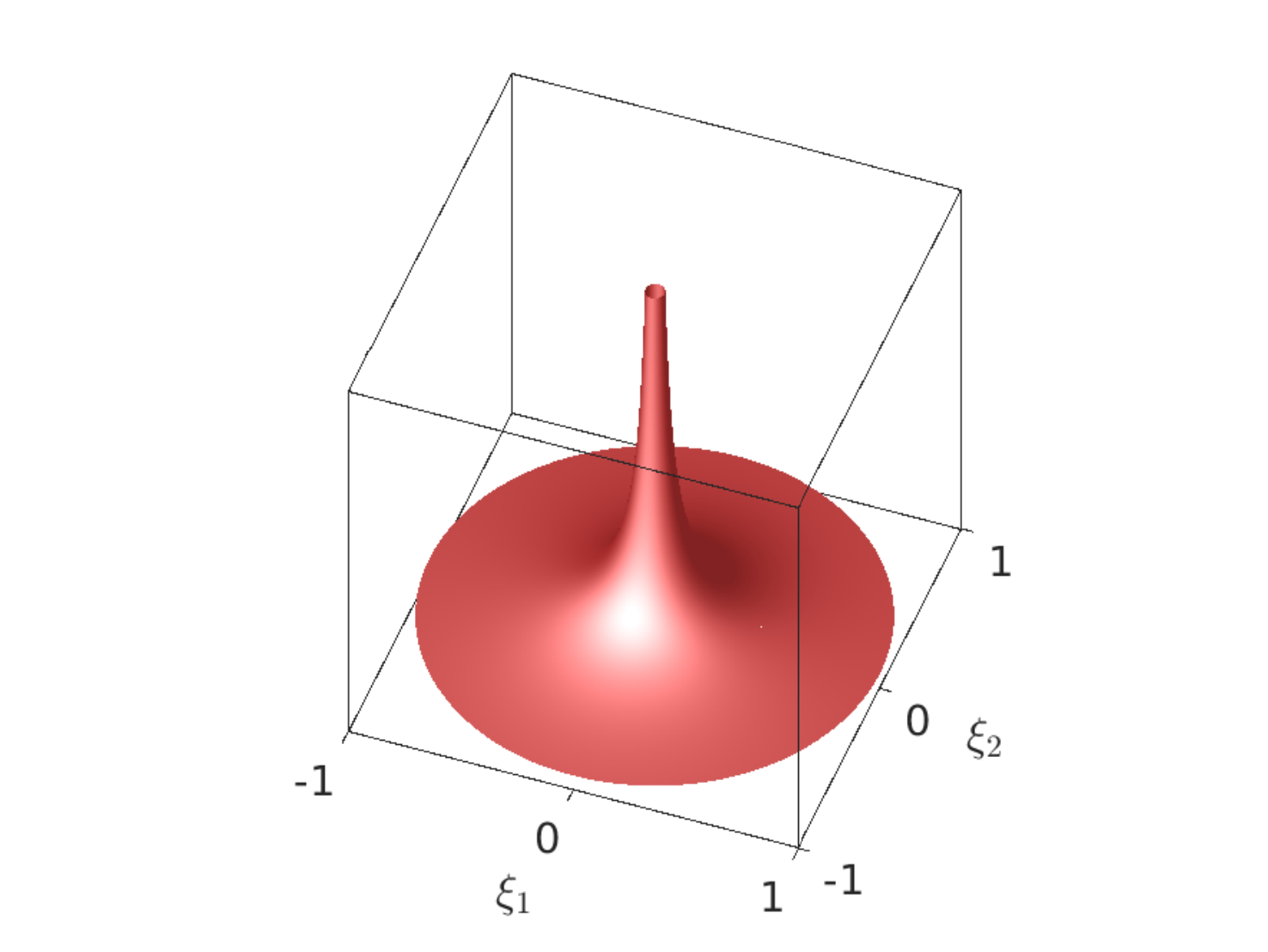}
\end{tabular} & \begin{tabular}{cc}
\includegraphics[scale=0.45,trim={3cm 0 3cm 0.5cm},clip]{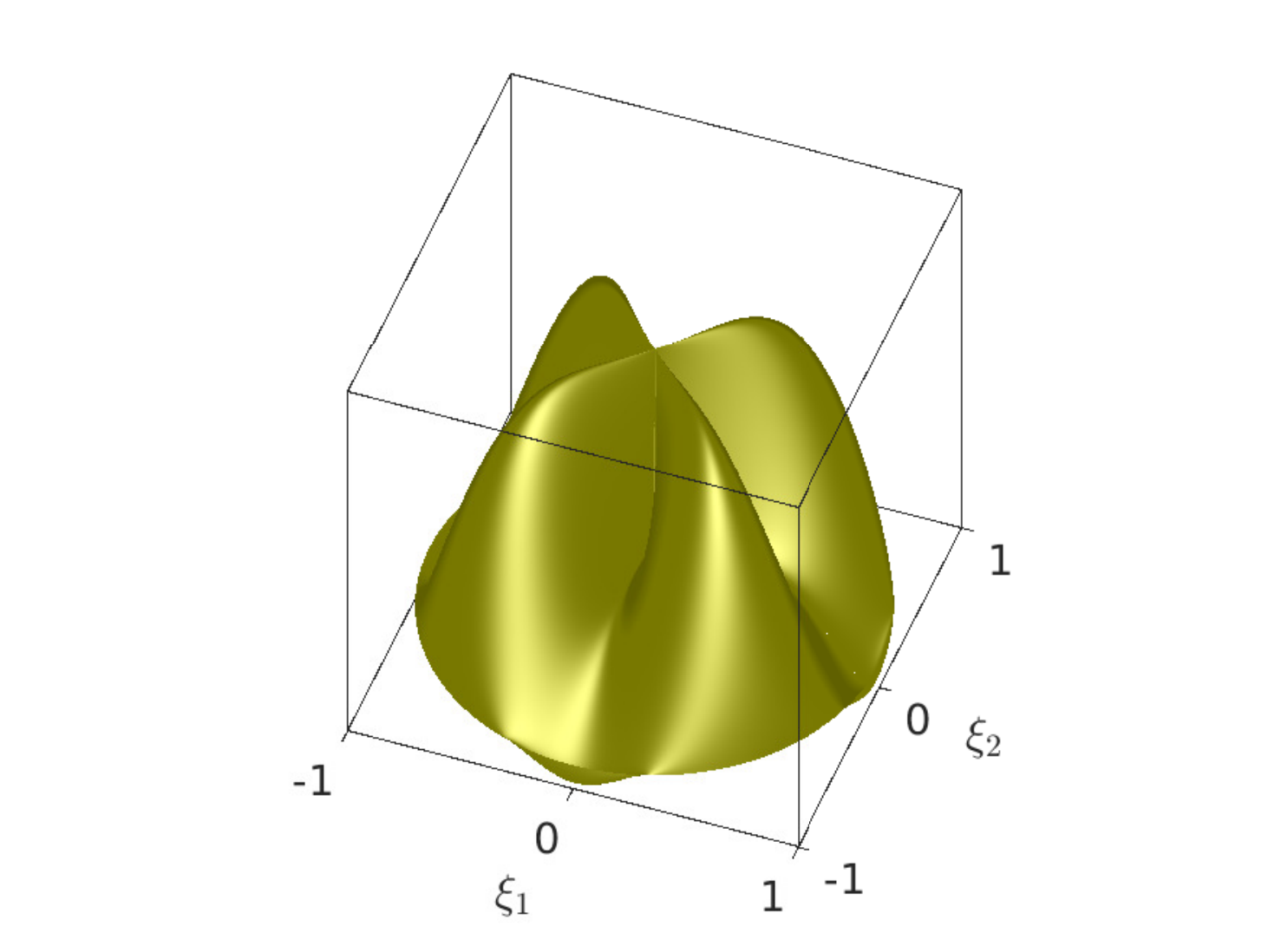} &
\includegraphics[scale=0.45,trim={3cm 0 3cm 0.5cm},clip]{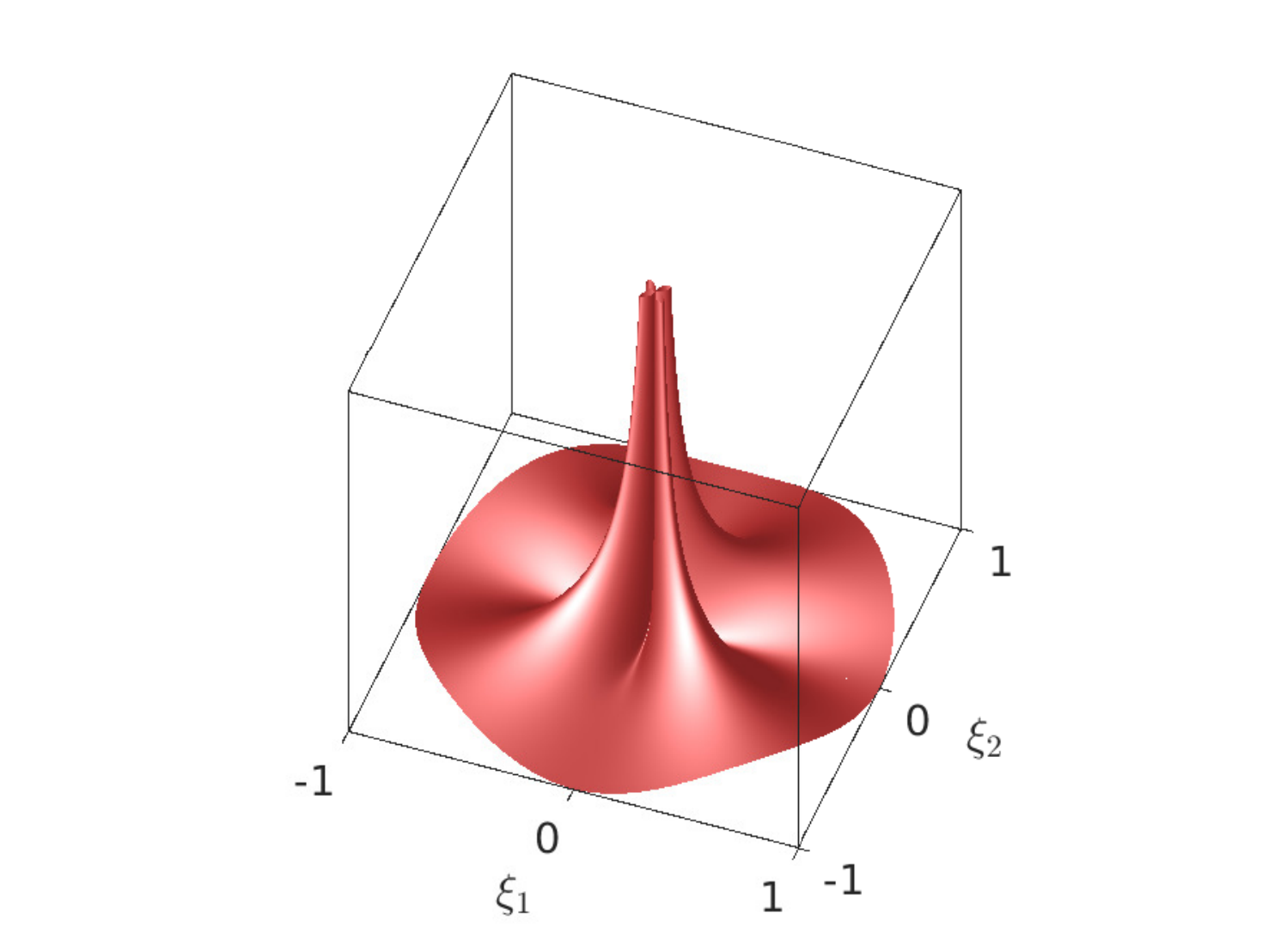}
\end{tabular} \\
Transversely Isotropic (MoN) & Cubic ($\text{MgAl}_2\text{O}_4$) \\
%%%%%%%%%%%%%%%%%%%%%%%%%%%%%%%%%%%%%%%%%%%%%%%%%%%%%
\multicolumn{2}{c}{
\begin{tabular}{cc}
\includegraphics[scale=0.45,trim={3cm 0 3cm 0.5cm},clip]{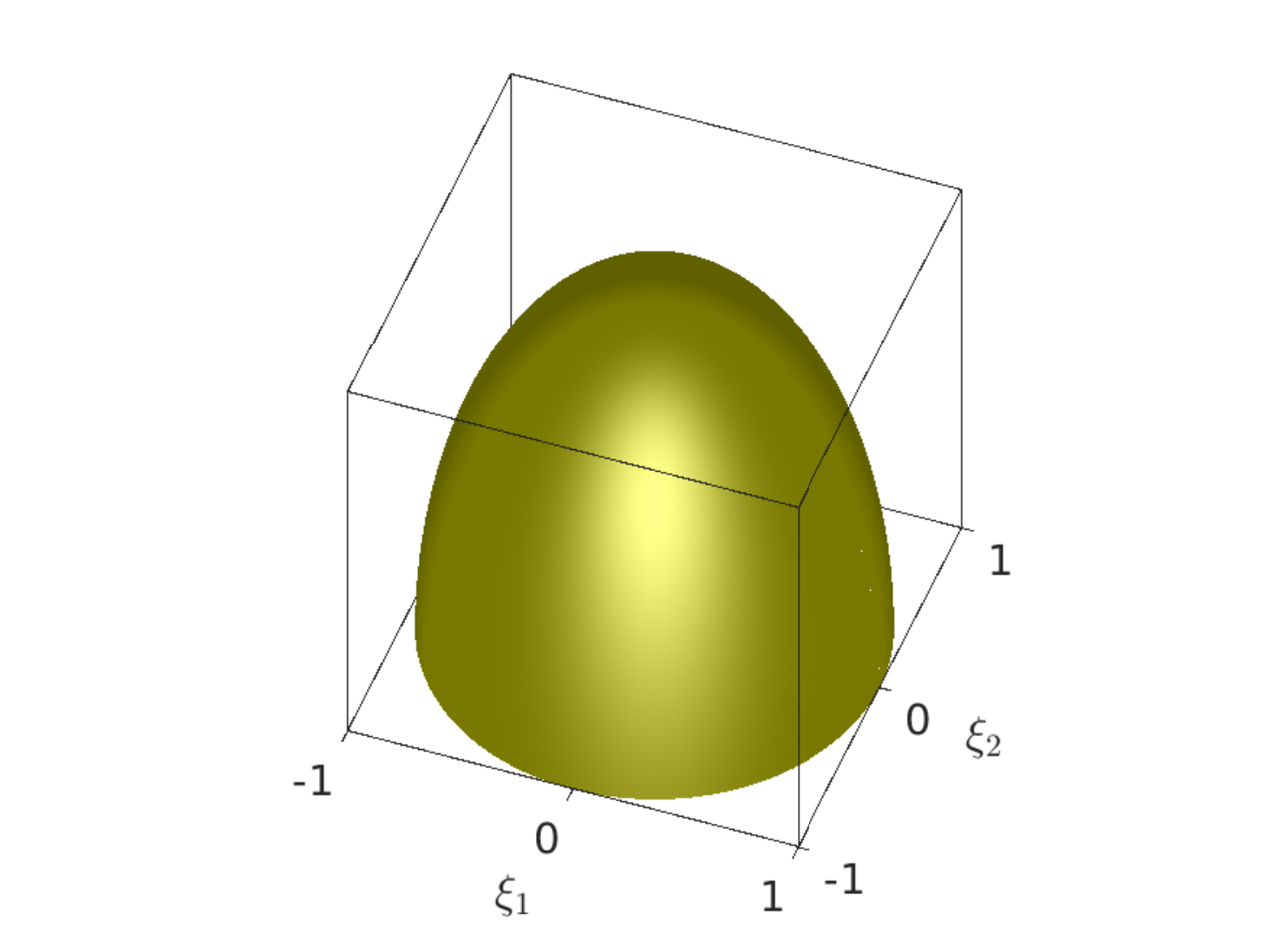} &
\includegraphics[scale=0.45,trim={3cm 0 3cm 0.5cm},clip]{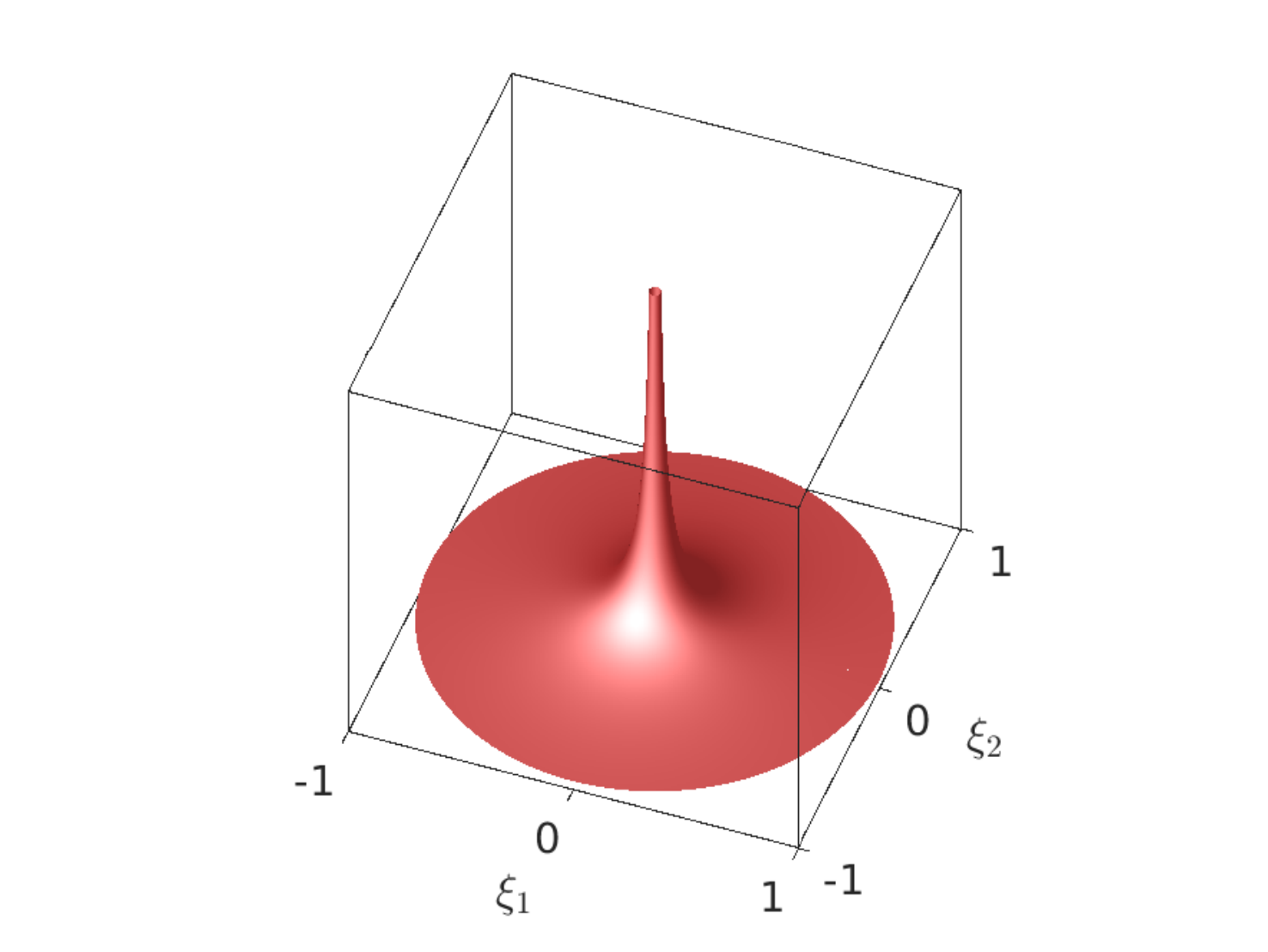}
\end{tabular} } \\
\multicolumn{2}{c}{Isotropic (Pyroceram 9608)}
\end{tabular}
\caption{Plots of $r^2 \lambda_0(\xi_1,\xi_2)$ ({\em cf.}~\eqref{eqn:PlaneStrainLambda_iForw=1overxia}) with $\omega( \|\bfxi \|) = \frac{1}{\| \bfxi \|}$ (in green) and $r^2 \lambda(\xi_1,\xi_2)$ ({\em cf.}~\eqref{eqn:lambdaobl}) with $\omega(r) = \frac{1}{r}$ (in red).}
\label{fig:planestrainmicromodulioneoverxi}
\end{figure}

\begin{figure}
\begin{tabular}{cc}
\begin{tabular}{cc}
\includegraphics[scale=0.45,trim={3cm 0 3cm 0.5cm},clip]{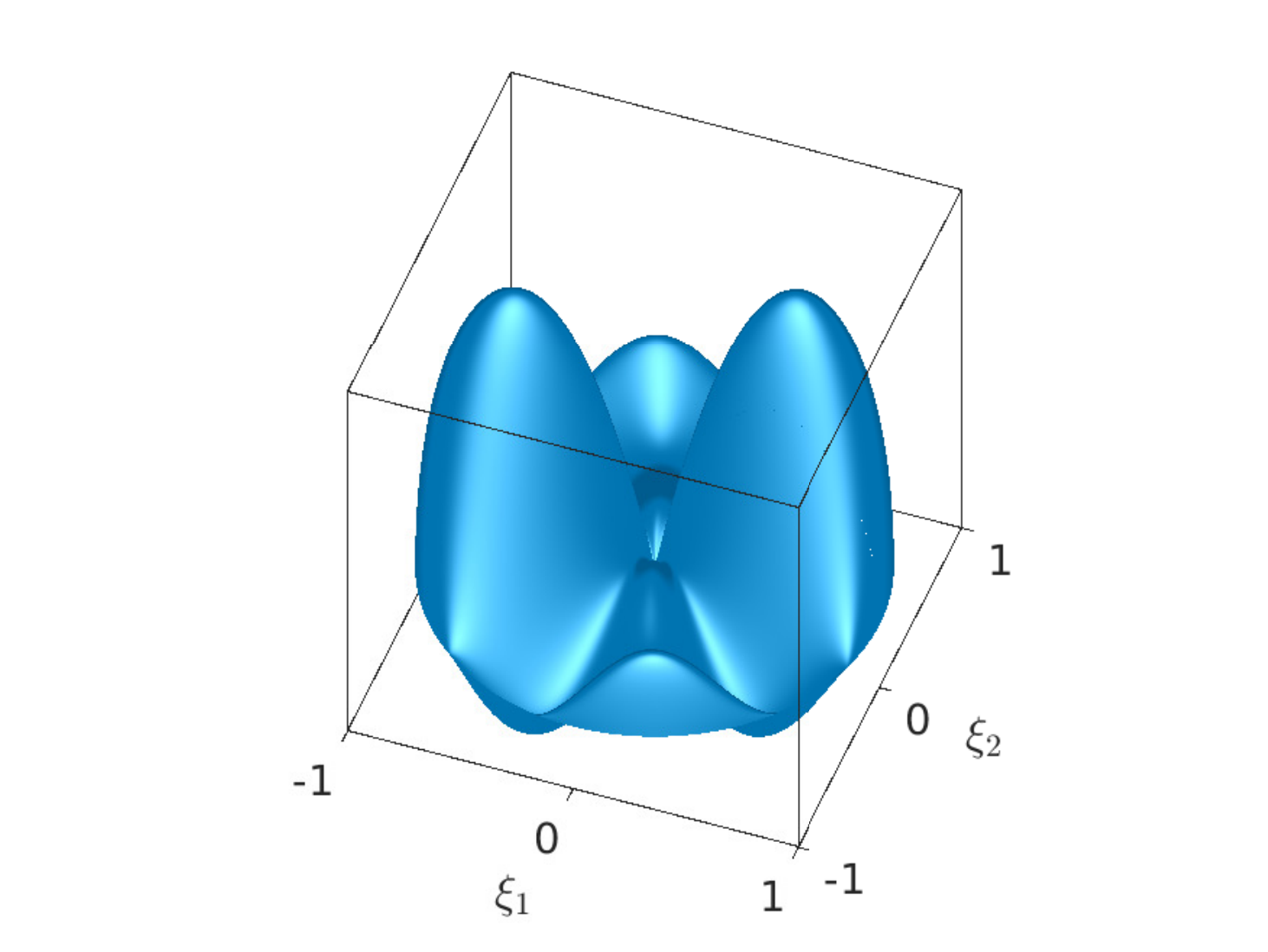} &
\includegraphics[scale=0.45,trim={3cm 0 3cm 0.5cm},clip]{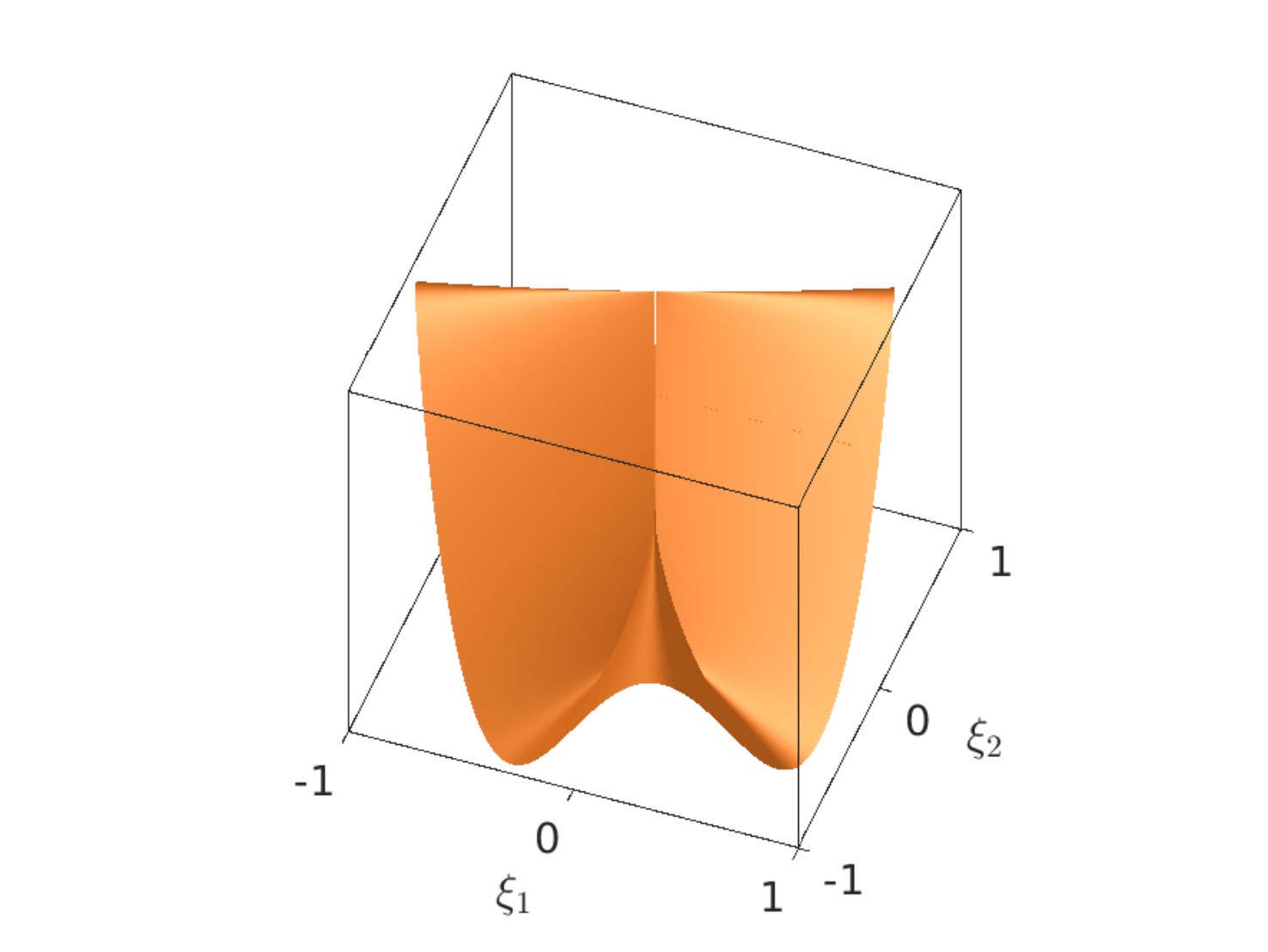}
\end{tabular} & \begin{tabular}{cc}
\includegraphics[scale=0.45,trim={3cm 0 3cm 0.5cm},clip]{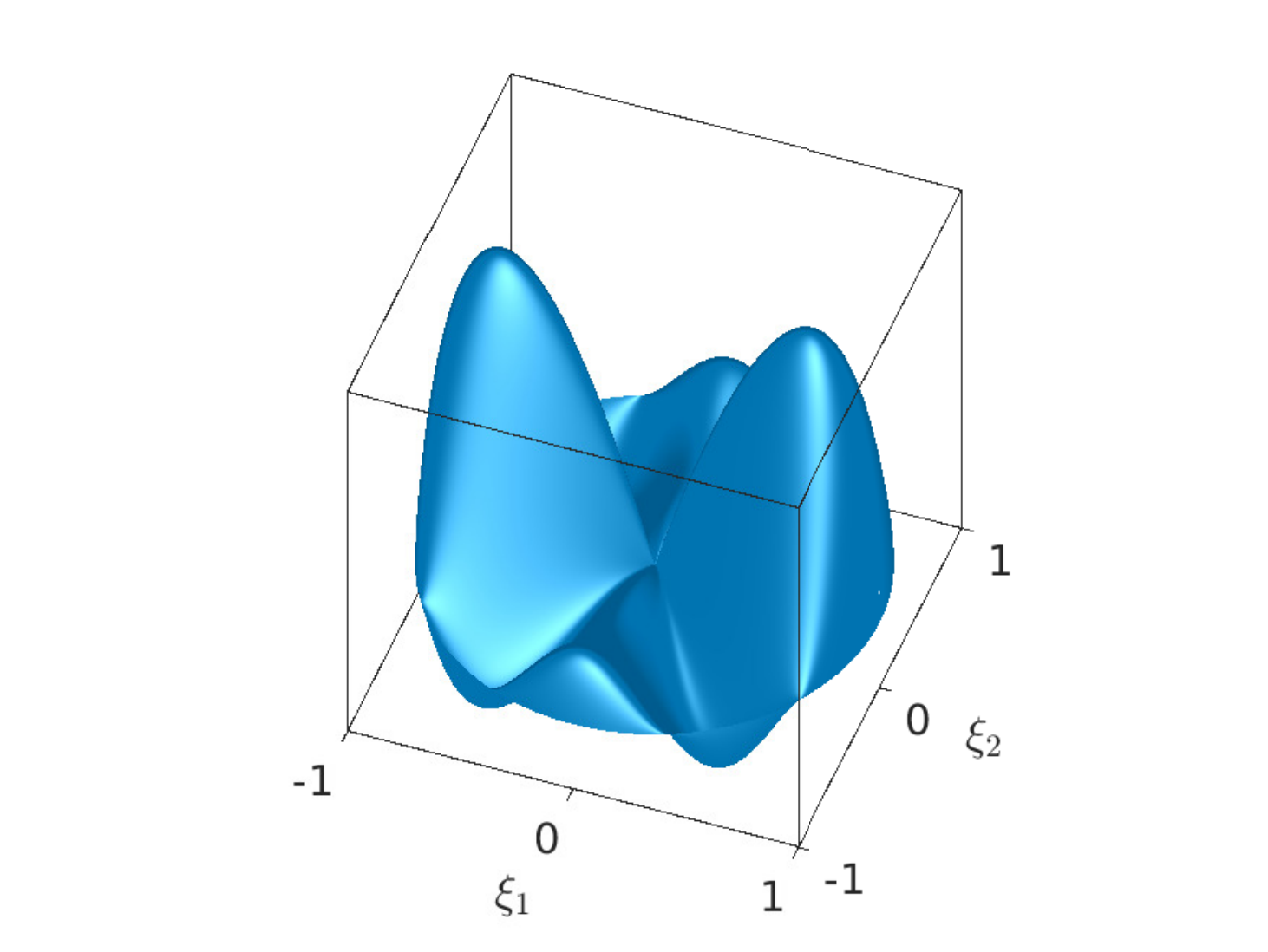} &
\includegraphics[scale=0.45,trim={3cm 0 3cm 0.5cm},clip]{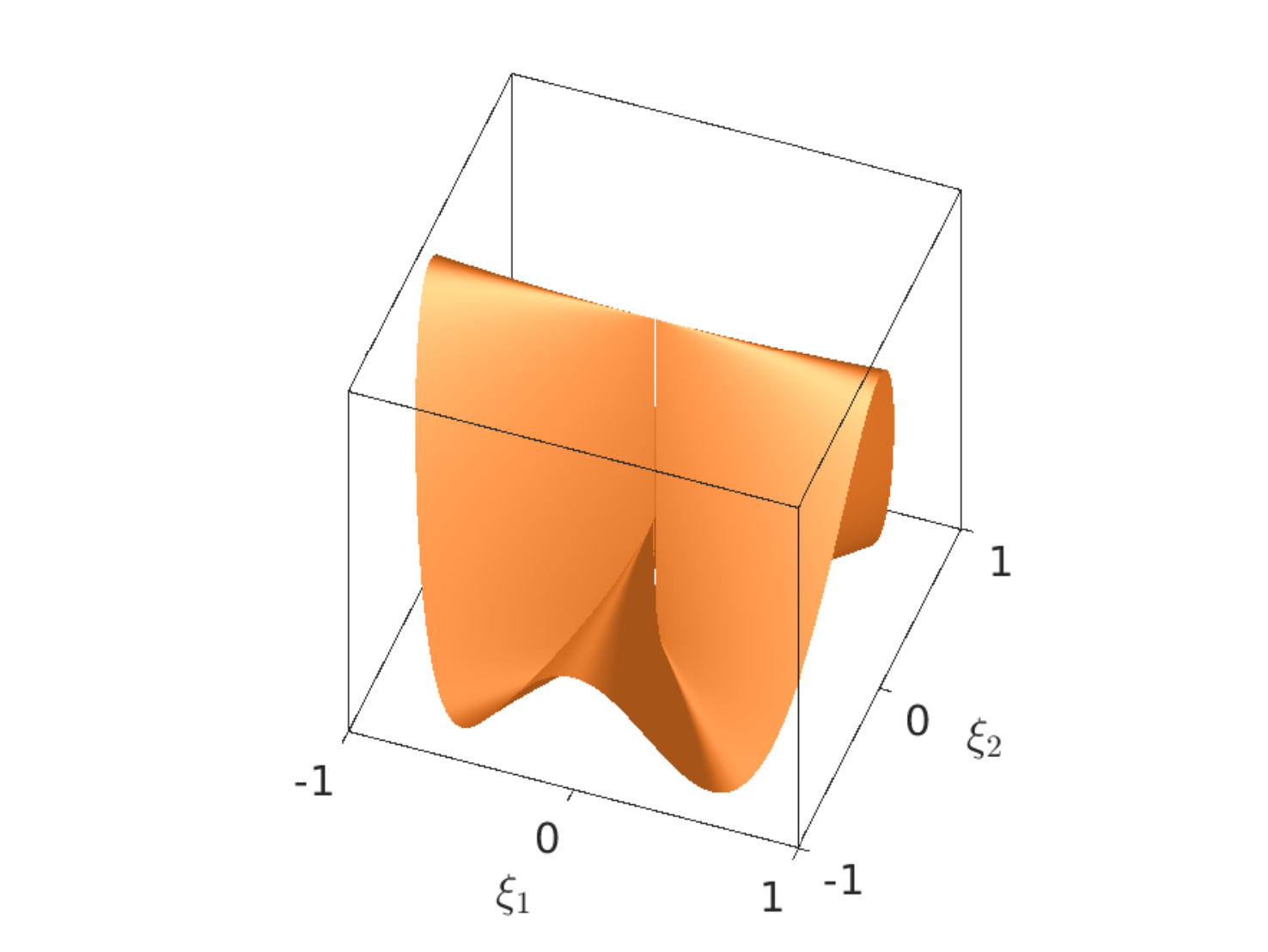}
\end{tabular} \\
Monoclinic ($\text{CoTeO}_4$) & Orthotropic ($\text{Te}_2\text{W}$) \\
%%%%%%%%%%%%%%%%%%%%%%%%%%%%%%%%%%%%%%%%%%%%%%%%%%%%%
\begin{tabular}{cc}
\includegraphics[scale=0.45,trim={3cm 0 3cm 0.5cm},clip]{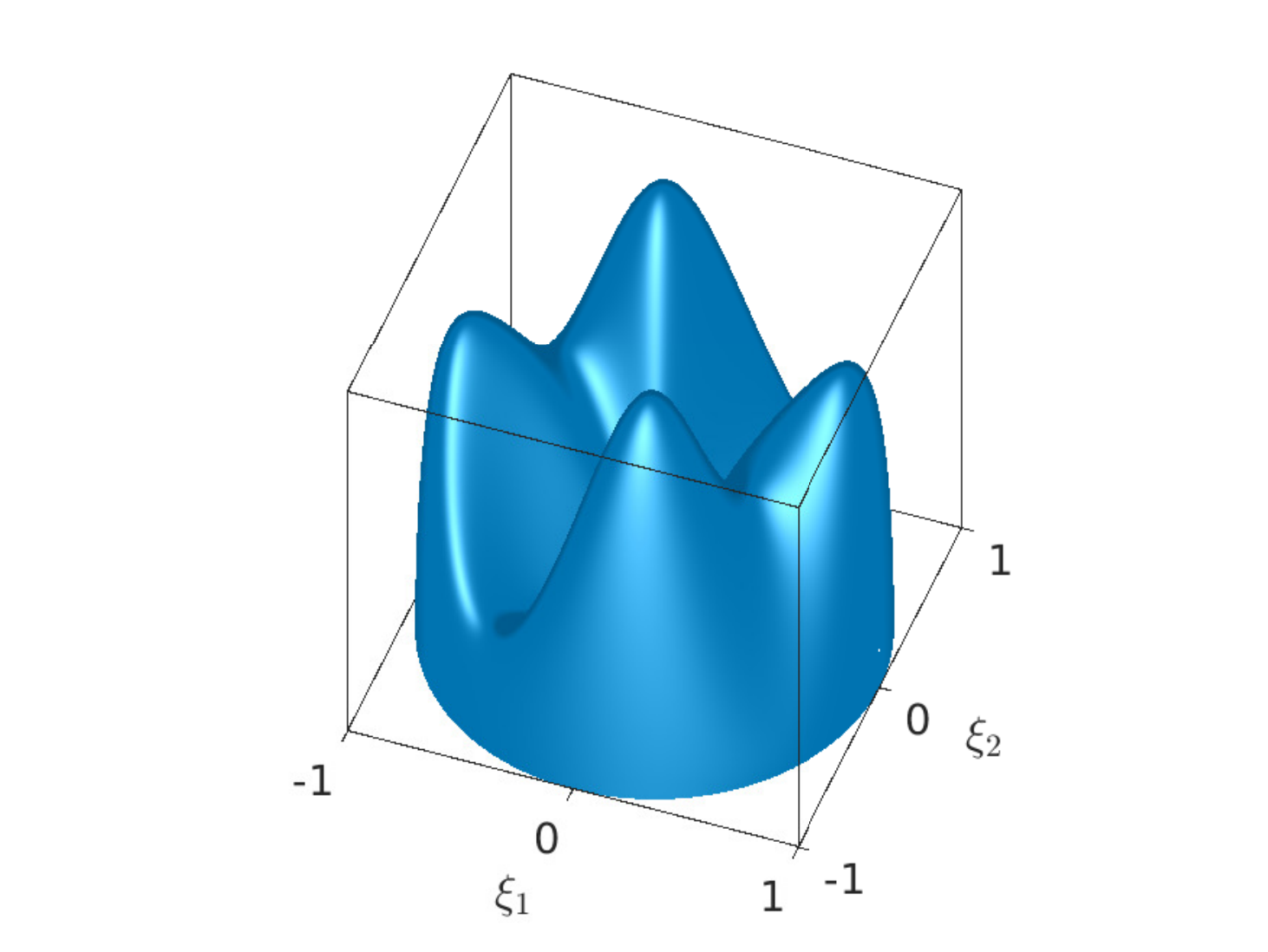} &
\includegraphics[scale=0.45,trim={3cm 0 3cm 0.5cm},clip]{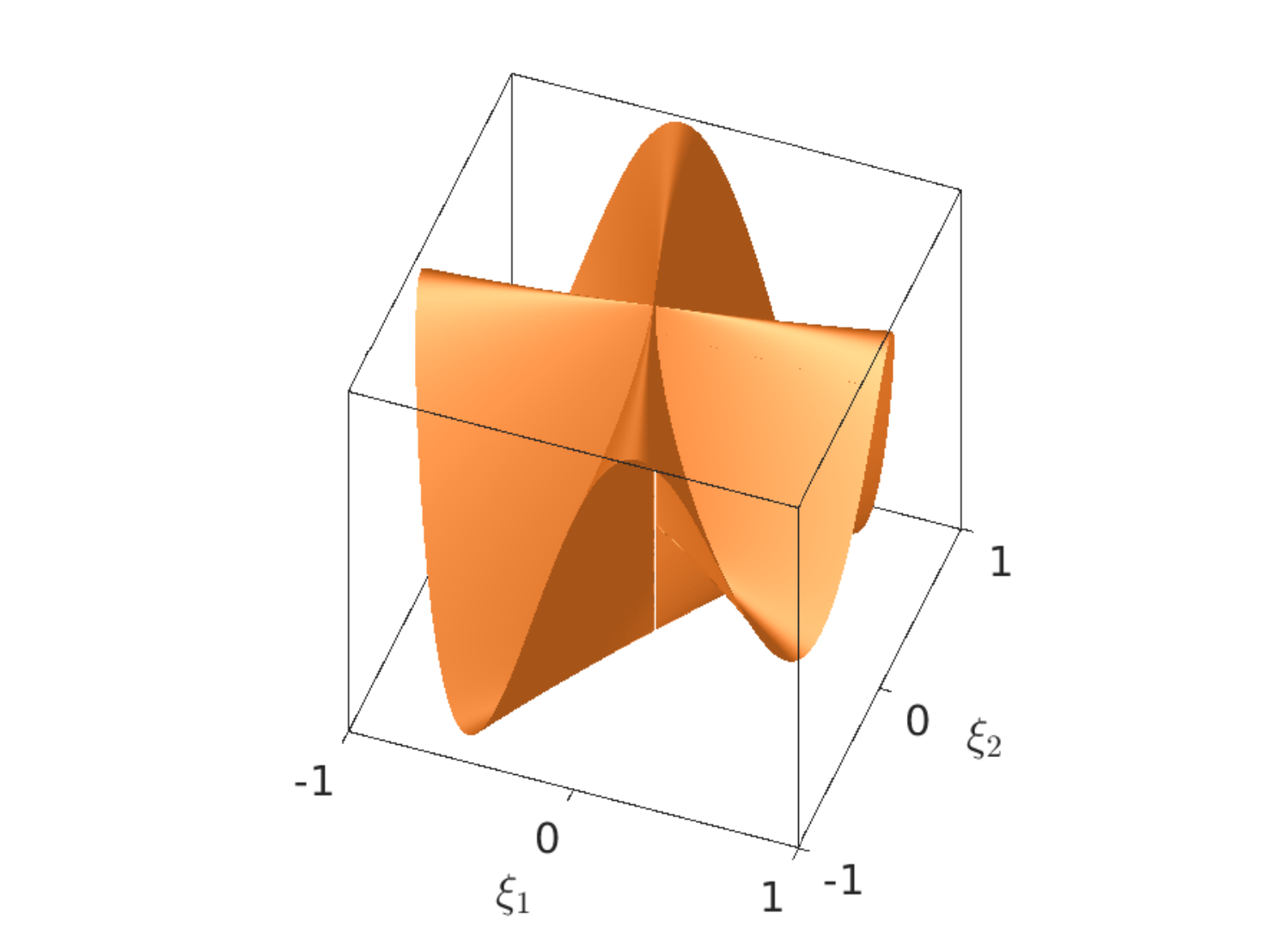}
\end{tabular} & \begin{tabular}{cc}
\includegraphics[scale=0.45,trim={3cm 0 3cm 0.5cm},clip]{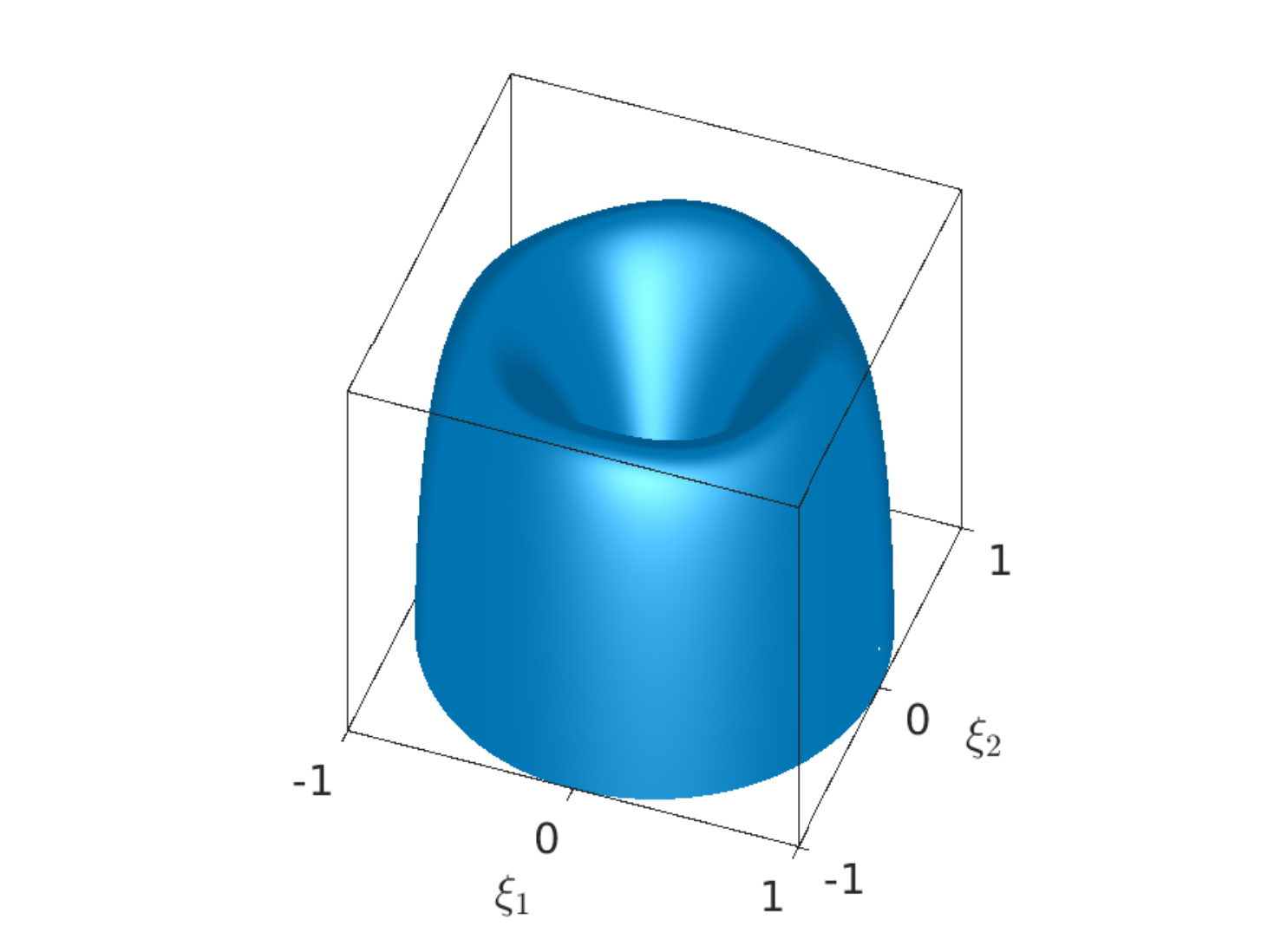} &
\includegraphics[scale=0.45,trim={3cm 0 3cm 0.5cm},clip]{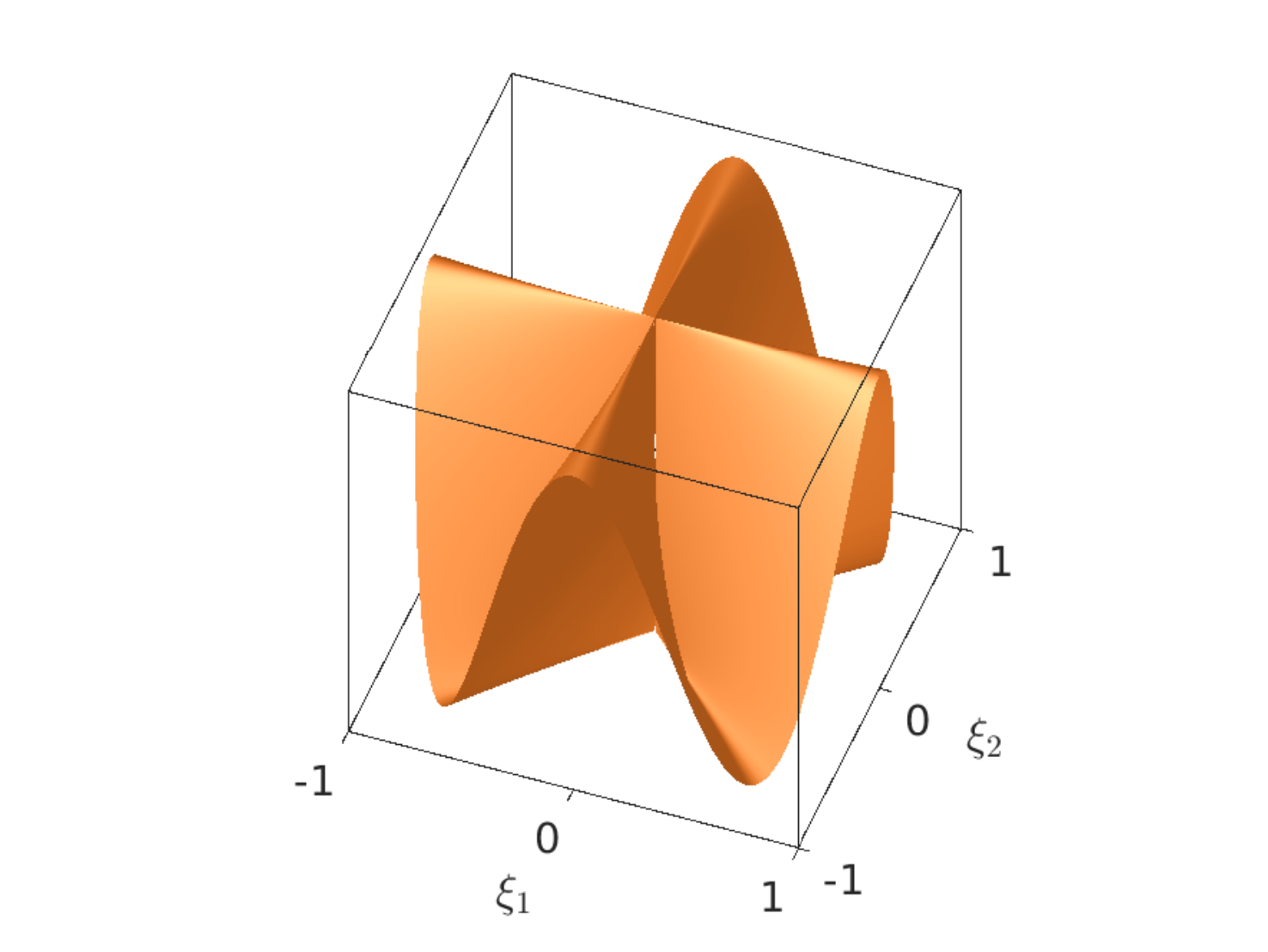}
\end{tabular} \\
Trigonal ($\text{Ta}_2\text{C}$) & Tetragonal (Si) \\
%%%%%%%%%%%%%%%%%%%%%%%%%%%%%%%%%%%%%%%%%%%%%%%%%%%%%
\begin{tabular}{cc}
\includegraphics[scale=0.45,trim={3cm 0 3cm 0.5cm},clip]{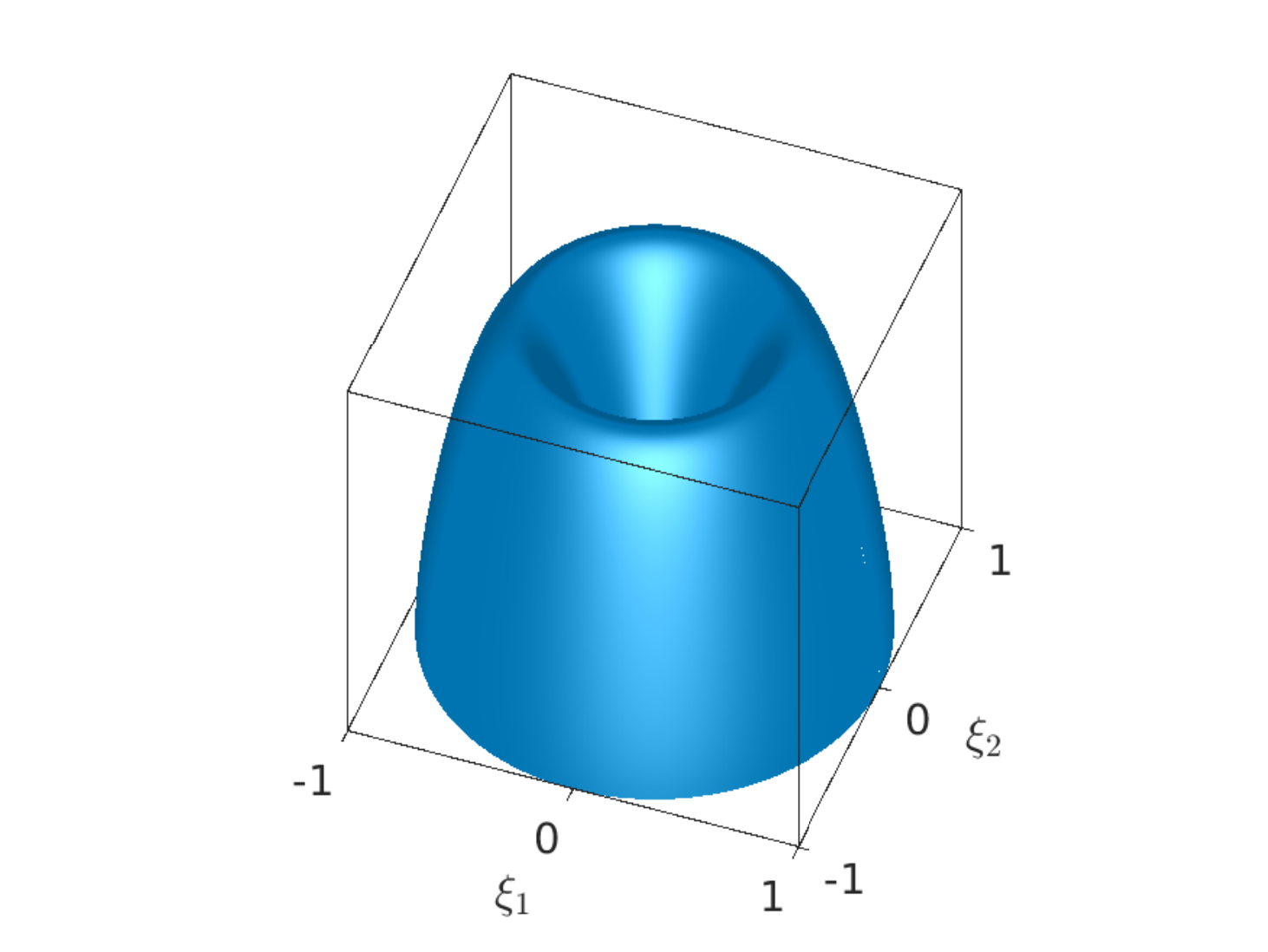} &
\includegraphics[scale=0.45,trim={3cm 0 3cm 0.5cm},clip]{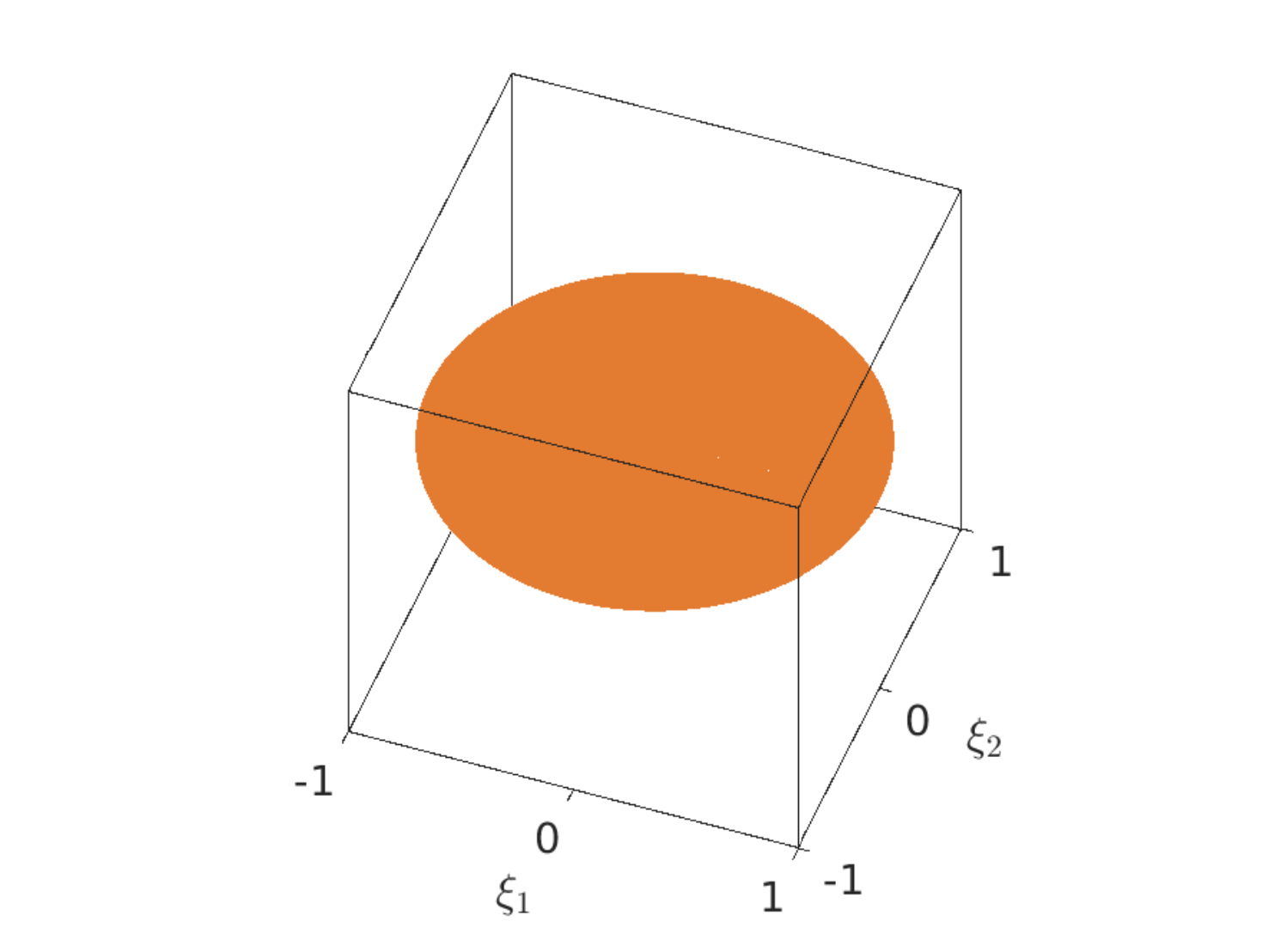}
\end{tabular} & \begin{tabular}{cc}
\includegraphics[scale=0.45,trim={3cm 0 3cm 0.5cm},clip]{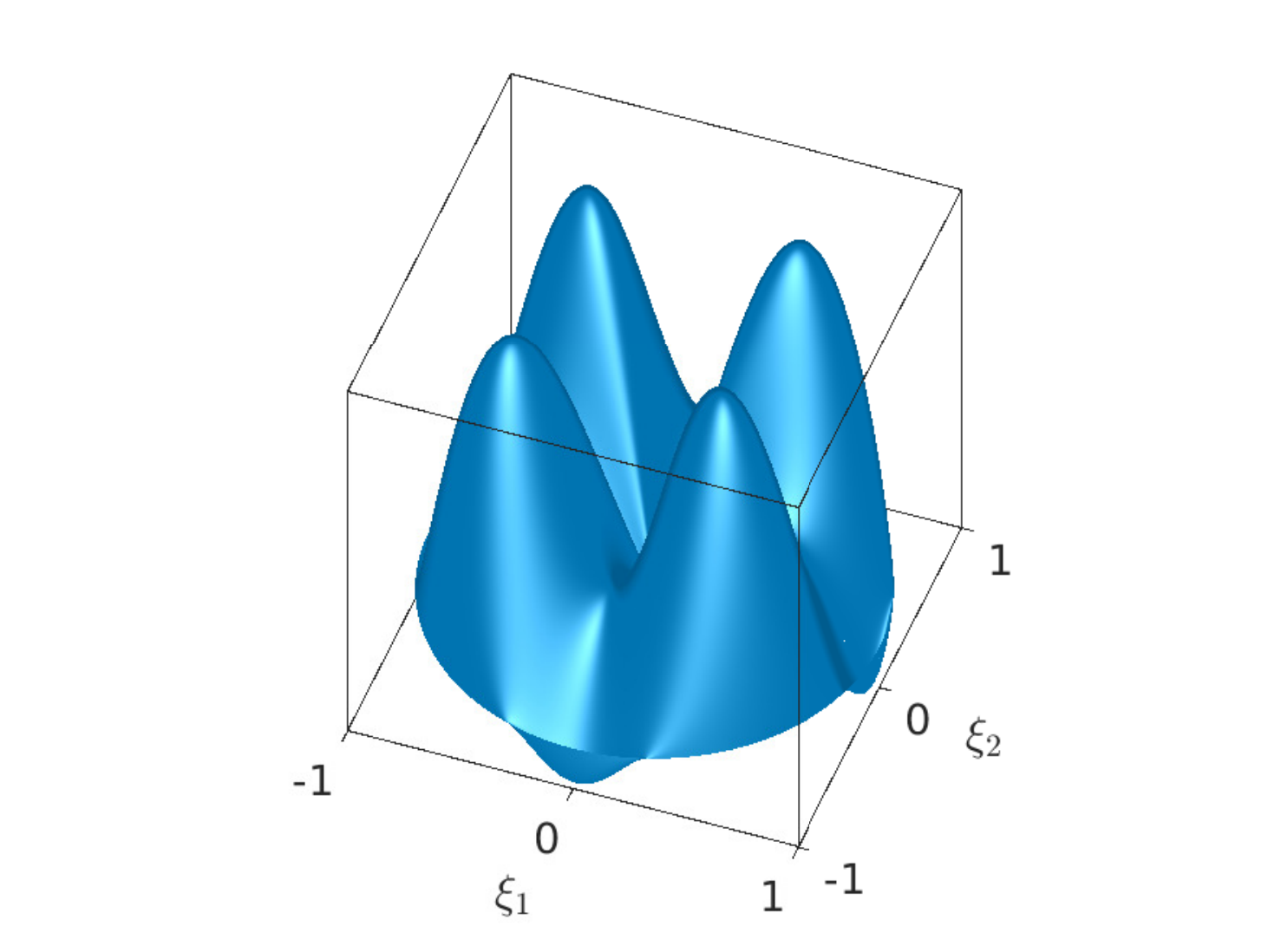} &
\includegraphics[scale=0.45,trim={3cm 0 3cm 0.5cm},clip]{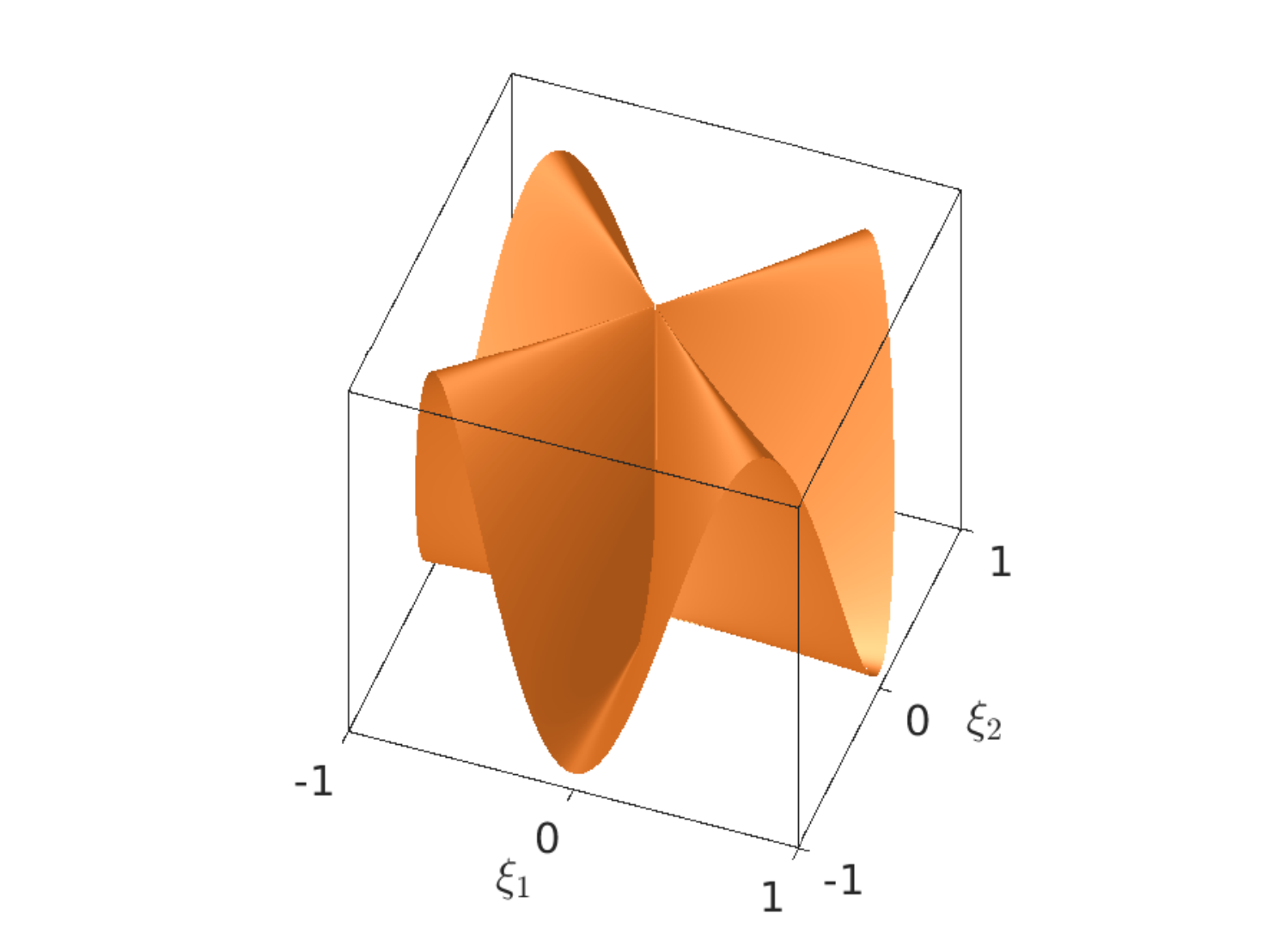}
\end{tabular} \\
Transversely Isotropic (MoN) & Cubic ($\text{MgAl}_2\text{O}_4$) \\
%%%%%%%%%%%%%%%%%%%%%%%%%%%%%%%%%%%%%%%%%%%%%%%%%%%%%
\multicolumn{2}{c}{
\begin{tabular}{cc}
\includegraphics[scale=0.45,trim={3cm 0 3cm 0.5cm},clip]{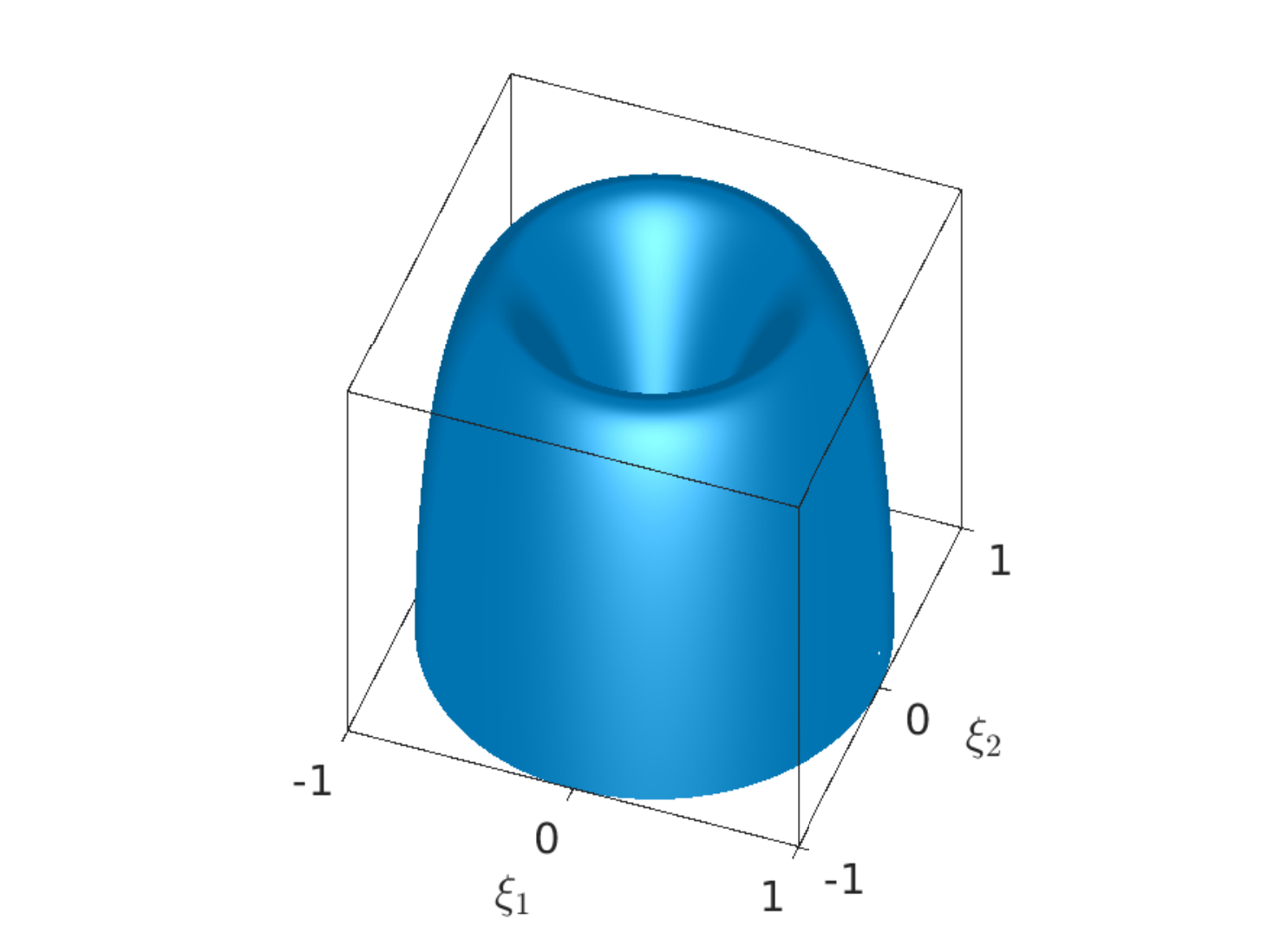} &
\includegraphics[scale=0.45,trim={3cm 0 3cm 0.5cm},clip]{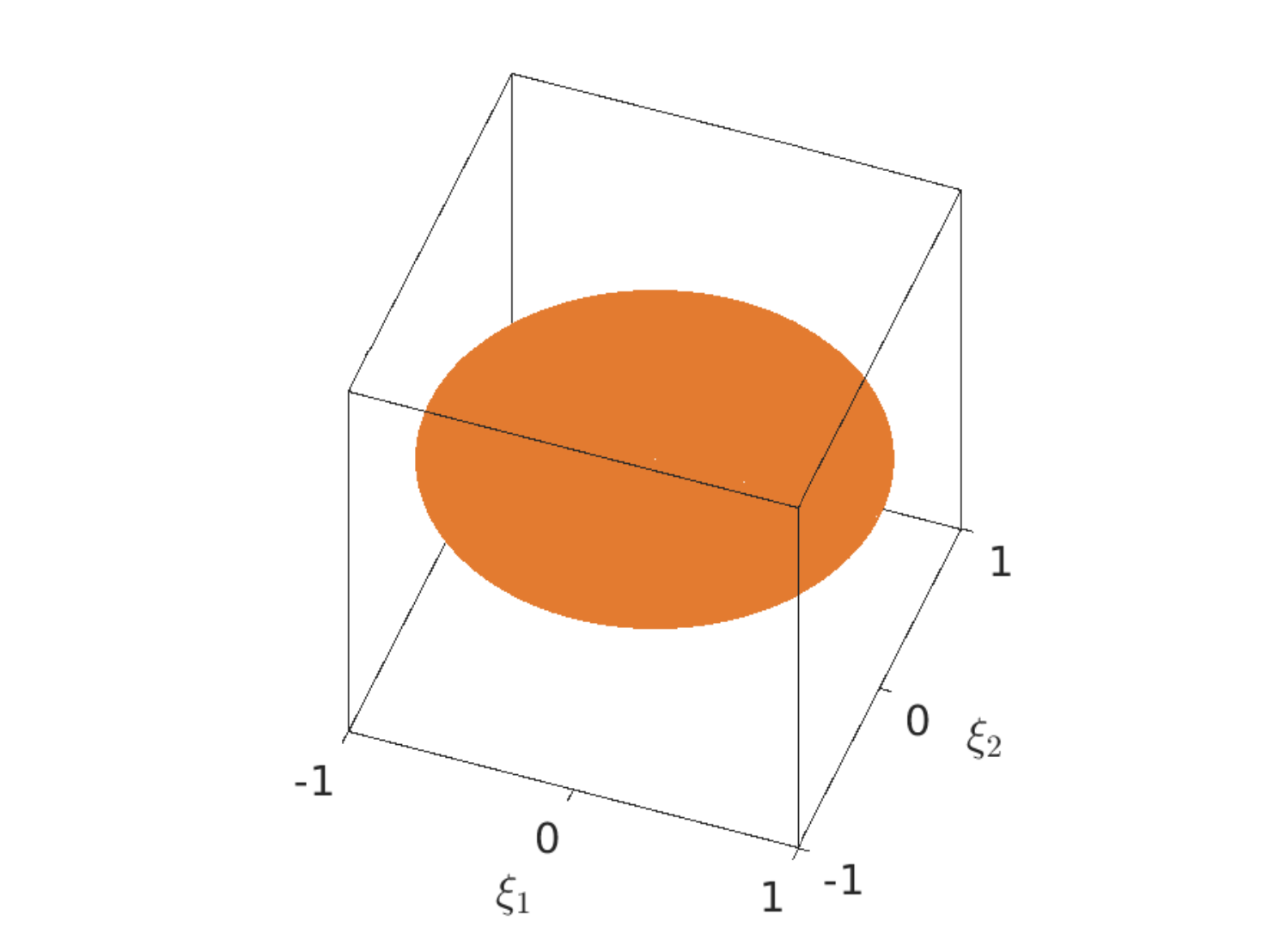} 
\end{tabular} }  \\
\multicolumn{2}{c}{Isotropic (Pyroceram 9608)} \\
\end{tabular}
\caption{Plots of $r^2 \lambda_0(\xi_1,\xi_2)$ ({\em cf.}~\eqref{eqn:PlaneStrainLambda_iForw=1overxia}) with $\omega(\|\bfxi \|) = 1$ (in blue) and $r^2 \lambda(\xi_1,\xi_2)$ ({\em cf.}~\eqref{eqn:lambdaobl}) with $\omega(r) = 1$ (in orange).}
\label{fig:planestrainmicromodulione}
\end{figure}

\textbf{Connections between the peridynamic and classical plane strain equations}: 

We now show that the peridynamic plane strain model reduces to the classical plane strain model under suitable assumptions.

\begin{proposition}\label{prop:periplanestraintaylor}
Suppose the micromodulus function $\lambda(\bfxi)$ is related to the elasticity tensor $\mathbb{C}$ through \eqref{eqn:Cperiexpression}. Given a smooth deformation, under a second-order Taylor expansion of the displacement field, the peridynamic generalized plane strain model~\eqref{eqn:genplanestrainmodel} reduces to the classical generalized plane strain model~\eqref{eqn:classeqnmot} with Cauchy's relations imposed. Moreover, under a similar Taylor expansion, the decoupled peridynamic plane strain in-plane equations of motion~\eqref{eqn:planestrainmodelinplane} and out-of-plane equation of motion~\eqref{eqn:planestrainmodeloutofplane} reduce to the classical plane strain in-plane equations of motion~\eqref{eqn:classicalmonoeqnmotion} and out-of-plane equation of motion~\eqref{eqn:classicalmonoeqnmotion3}, respectively, with Cauchy's relations imposed.
\end{proposition}

\begin{proof}
Recalling~\eqref{eqn:xi3integrationterms}, one may write 
\begin{equation*}
\int_{B^{2D}_{\delta}(\bf0)} \lambda_n(\xi_1,\xi_2) \xi_1^l \xi_2^m d \xi_1 d \xi_2 = \int_{B^{3D}_{\delta}(\bf0)} \lambda(\bfxi) \xi_1^l \xi_2^m \xi_3^n d \bfxi,
\end{equation*}
and thus the peridynamic generalized plane strain equation of motion \eqref{eqn:genplanestrainmodel} may be expressed as 
\begin{equation}\label{eqn:planestrainpretaylorexpasion}
    \rho(x,y) \ddot{u}_i(x,y,t) = \int_{B_\delta^{3D}(\bf0)} \lambda(\bfxi) \xi_i \xi_j (u_j(x+\xi_1,y+\xi_2,t) - u_j(x,y,t)) d \bfxi + b_i(x,y,t).
\end{equation}
Performing a Taylor expansion of $u_j(x + \xi_1, y + \xi_2,t)$ in \eqref{eqn:planestrainpretaylorexpasion} about $(\bfx,t) = (x,y,t)$, we obtain
\begin{equation}\label{eqn:planestrainposttaylorexpasion}
\begin{split}
    \rho(x,y) \ddot{u}_i(x,y,t) = \int_{B_\delta^{3D}(\bf0)} \lambda(\bfxi) \xi_i \xi_j  &\left[   \xi_1 \frac{\partial u_j}{\partial x}(x,y,t) +  \xi_2  \frac{\partial u_j}{\partial y}(x,y,t)  \right. \\
%%%%%%%%%%%%%%%%%%%%%%%%%%%%%%%%%%%%%%
&\left. + \frac{\xi_1^2}{2} \frac{\partial^2 u_j}{\partial x^2}(x,y,t) + \xi_1 \xi_2 \frac{\partial^2 u_j}{\partial x \partial y}(x,y,t) + \frac{\xi_2^2}{2}  \frac{\partial^2 u_j}{\partial y^2}(x,y,t)  + \cdots \right] d \bfxi + b_i(x,y,t).
\end{split}
\end{equation}
Assuming higher-order terms (beyond second-order) are negligible and utilizing antisymmetry under the transformation $\bfxi \rightarrow - \bfxi$ (recall $\lambda(-\bfxi) = \lambda(\bfxi)$ by~\eqref{eqn:lambdasymm}) to nullify the first-order terms in \eqref{eqn:planestrainposttaylorexpasion}, we obtain
\begin{equation}\label{eqn:planestrainposttaylorexpasion2}
\begin{split}
    \rho(x,y) \ddot{u}_i(x,y,t) = \frac{1}{2}\int_{B_\delta^{3D}(\bf0)} \lambda(\bfxi) \xi_i \xi_j  &\left[\xi_1^2 \frac{\partial^2 u_j}{\partial x^2}(x,y,t) + 2\xi_1 \xi_2 \frac{\partial^2 u_j}{\partial x \partial y}(x,y,t) + \xi_2^2  \frac{\partial^2 u_j}{\partial y^2}(x,y,t)  \right] d \bfxi + b_i(x,y,t).
\end{split}
\end{equation}
Employing the relationship between $\lambda(\bfxi)$ and $\mathbb{C}$ from \eqref{eqn:Cperiexpression}, we rewrite \eqref{eqn:planestrainposttaylorexpasion2} as 
\begin{equation}\label{eqn:planestrainposttaylorexpasion3}
\begin{split}
    \rho(x,y) \ddot{u}_i(x,y,t) = C_{11ij} \frac{\partial^2 u_j}{\partial x^2}(x,y,t) + 2 C_{ij12} \frac{\partial^2 u_j}{\partial x \partial y}(x,y,t) + C_{22ij}  \frac{\partial^2 u_j}{\partial y^2}(x,y,t) + b_i(x,y,t).
\end{split}
\end{equation}
Writing \eqref{eqn:planestrainposttaylorexpasion3} out, we obtain
\begin{subequations}\label{eqn:planestraineqnmotioncauchy}
\begin{align}
\begin{split}
\rho \ddot{u}_1 ={}& C_{1111} \frac{\partial^2 u_1}{\partial x^2} + 2 C_{1112} \frac{\partial^2 u_1}{\partial x \partial y} + C_{1122} \frac{\partial^2 u_1}{\partial y^2} + C_{1112}\frac{\partial^2 u_2}{\partial x^2} + 2C_{1122}\frac{\partial^2 u_2}{\partial x \partial y}   \\
&+ C_{2212} \frac{\partial^2 u_2}{\partial y^2}  + C_{1113} \frac{\partial^2 u_3}{\partial x^2}  + 2C_{1123} \frac{\partial^2 u_3}{\partial x \partial y}   + C_{2213} \frac{\partial^2 u_3}{\partial y^2} + b_1, \label{eqn:planestraineqnmotioncauchya}
\end{split}\\
%%%%%%%%%%%%%%%%%%%%%%%%%%%%%%%%%%%%%%%%%%%%%%%%%%%%%%%%%%%%%%
\begin{split}
\rho \ddot{u}_2 ={}& C_{1112} \frac{\partial^2 u_1}{\partial x^2}+ 2 C_{1122} \frac{\partial^2 u_1}{\partial x \partial y} + C_{2212} \frac{\partial^2 u_1}{\partial y^2}  + C_{1122} \frac{\partial^2 u_2}{\partial x^2} + 2 C_{2212} \frac{\partial^2 u_2}{\partial x \partial y}   \\
&+ C_{2222} \frac{\partial^2 u_2}{\partial y^2} + C_{1123} \frac{\partial^2 u_3}{\partial x^2} + 2 C_{2213} \frac{\partial^2 u_3}{\partial x \partial y}  + C_{2223} \frac{\partial^2 u_3}{\partial y^2} + b_2, \label{eqn:planestraineqnmotioncauchyb}
\end{split}\\
%%%%%%%%%%%%%%%%%%%%%%%%%%%%%%%%%%%%%%%%%%%%%%%%%%%%%%%%%%%%%%
\begin{split}
\rho \ddot{u}_3 ={}& C_{1113} \frac{\partial^2 u_1}{\partial x^2} + 2C_{1123} \frac{\partial^2 u_1}{\partial x \partial y} +  C_{2213} \frac{\partial^2 u_1}{\partial y^2} + C_{1123} \frac{\partial^2 u_2}{\partial x^2} + 2 C_{2213} \frac{\partial^2 u_2}{\partial x \partial y} \\
&+ C_{2223} \frac{\partial^2 u_2}{\partial y^2} + C_{1133} \frac{\partial^2 u_3}{\partial x^2}  + 2 C_{3312} \frac{\partial^2 u_3}{\partial x \partial y} + C_{2233} \frac{\partial^2 u_3}{\partial y^2} + b_3, \label{eqn:planestraineqnmotioncauchyc}
\end{split}
\end{align}
\end{subequations}
where we have omitted the arguments $x,y,$ and $t$ for brevity and a more direct comparison to the classical model. Noticing that \eqref{eqn:planestraineqnmotioncauchy} is exactly the classical generalized plane strain equation of motion \eqref{eqn:classeqnmot} with Cauchy's relations ({\em cf.}~\eqref{eqn:3DCauchyRelations}) imposed completes the proof of the first portion of Proposition \ref{prop:periplanestraintaylor}. 

To prove the peridynamic plane strain model converges to the classical plane strain model, first recall the in-plane and peridynamic out-of-plane plane strain equations of motion~\eqref{eqn:planestrainmodelinplane} and~\eqref{eqn:planestrainmodeloutofplane}, respectively, are obtained by imposing (P$\varepsilon$\ref{assmp:periplanestrain3full}) on the peridynamic generalized plane strain model~\eqref{eqn:genplanestrainmodel}. We thus follow the same derivation as for the generalized plane strain case until~\eqref{eqn:planestrainposttaylorexpasion2}. As before, we then employ the relationship between $\lambda(\bfxi)$ and $\mathbb{C}$ from~\eqref{eqn:Cperiexpression}. However, the additional assumption~(P$\varepsilon$\ref{assmp:periplanestrain3full}) results in the monoclinic symmetry restrictions~\eqref{eqn:monoelastrelations} being imposed on the elasticity tensor. Consequently, the peridynamic in-plane equations of motion reduce to ({\em cf}. \eqref{eqn:planestraineqnmotioncauchya} and \eqref{eqn:planestraineqnmotioncauchyb})
\begin{subequations}\label{eqn:planestraineqnmotioncauchynogena}
\begin{align}
\begin{split}
\rho \ddot{u}_1 ={}& C_{1111} \frac{\partial^2 u_1}{\partial x^2} + 2 C_{1112} \frac{\partial^2 u_1}{\partial x \partial y} + C_{1122} \frac{\partial^2 u_1}{\partial y^2} + C_{1112}\frac{\partial^2 u_2}{\partial x^2} + 2C_{1122}\frac{\partial^2 u_2}{\partial x \partial y}   \\
&+ C_{2212} \frac{\partial^2 u_2}{\partial y^2} + b_1, \label{eqn:planestraineqnmotioncauchyanogen}
\end{split}\\
%%%%%%%%%%%%%%%%%%%%%%%%%%%%%%%%%%%%%%%%%%%%%%%%%%%%%%%%%%%%%%
\begin{split}
\rho \ddot{u}_2 ={}& C_{1112} \frac{\partial^2 u_1}{\partial x^2}+ 2 C_{1122} \frac{\partial^2 u_1}{\partial x \partial y} + C_{2212} \frac{\partial^2 u_1}{\partial y^2}  + C_{1122} \frac{\partial^2 u_2}{\partial x^2} + 2 C_{2212} \frac{\partial^2 u_2}{\partial x \partial y}   \\
&+ C_{2222} \frac{\partial^2 u_2}{\partial y^2} + b_2, \label{eqn:planestraineqnmotioncauchybnogen}
\end{split}
\end{align}
\end{subequations} 
while the peridynamic out-of-plane equation of motion reduces to ({\em cf.} \eqref{eqn:planestraineqnmotioncauchyc})

\begin{align}\label{eqn:planestraineqnmotioncauchynogen}
\rho \ddot{u}_3 = C_{1133} \frac{\partial^2 u_3}{\partial x^2}  + 2 C_{3312} \frac{\partial^2 u_3}{\partial x \partial y} + C_{2233} \frac{\partial^2 u_3}{\partial y^2} + b_3.
\end{align}
Comparing \eqref{eqn:planestraineqnmotioncauchyanogen} and \eqref{eqn:planestraineqnmotioncauchybnogen} to the classical in-plane equations of motion \eqref{eqn:classicalmonoeqnmotion1} and \eqref{eqn:classicalmonoeqnmotion2}, with Cauchy's relations imposed, and comparing \eqref{eqn:planestraineqnmotioncauchynogen} to the classical out-of-plane equation of motion \eqref{eqn:classicalmonoeqnmotion3}, with Cauchy's relations imposed, completes the proof. \qed
\end{proof}
\begin{remark}\label{remark:planestrainCauchy}
Notice in the in-plane equations of motion for classical generalized plane strain, \eqref{eqn:classeqnmota} and \eqref{eqn:classeqnmotb}, only three Cauchy's relations are relevant:
\begin{equation*}
C_{1212} = C_{1122}, \quad C_{1312} = C_{1123}, \text{ and } \; C_{2312} = C_{2213},
\end{equation*}
while in the out-of-plane equation of motion for classical generalized plane strain, \eqref{eqn:classeqnmotc}, five Cauchy's relations are relevant:
\begin{equation*}
C_{1313} = C_{1133}, \quad C_{2323} = C_{2233}, \quad C_{1312} = C_{1123}, \quad C_{2312} = C_{2213}, \text{ and } \; C_{2313} = C_{3312}.
\end{equation*}
In contrast, in the in-plane equations of motion for classical plane strain, \eqref{eqn:classicalmonoeqnmotion1} and \eqref{eqn:classicalmonoeqnmotion2}, only one Cauchy's relation is relevant:
\begin{equation*}
C_{1212} = C_{1122},
\end{equation*}
while in the out-of-plane equation of motion for classical plane strain, \eqref{eqn:classicalmonoeqnmotion3}, three Cauchy's relations are relevant:
\begin{equation*}
C_{1313} = C_{1133}, \quad C_{2323} = C_{2233}, \text{ and } \; C_{2313} = C_{3312}.
\end{equation*}
The above difference between the classical generalized plane strain equations of motion and classical plane strain equations of motion is entirely due to the elimination of two Cauchy's relations by the monoclinic symmetry assumption. It is interesting to note that the in-plane equations of motion for classical plane strain, \eqref{eqn:classicalmonoeqnmotion1} and \eqref{eqn:classicalmonoeqnmotion2}, have exactly the same relevant Cauchy's relation as the classical two-dimensional equation of motion, \eqref{eqn:eqnmotionobl}.
\end{remark}

\subsubsection{Peridynamic plane stress}\label{sec:PeridynamicPlaneStress}
 
We now focus on the peridynamic analogue of classical plane stress ({\em cf.}~Section~\ref{sec:ClassicalPlaneStress}). Similarly to the peridynamic plane strain formulation in Section~\ref{sec:PeridynamicPlaneStrain}, we begin with the anisotropic three-dimensional bond-based linear peridynamic model \eqref{eqn:linearperieqn}  and impose peridynamic analogues of the classical generalized plane stress assumptions. 

One key issue with developing a peridynamic formulation for generalized plane stress is that it is unrealistic to restrict the analysis to material points in the bulk of a body for thin plates. In particular, we must consider position-dependent neighborhoods of material points, see, e.g., Figure \ref{fig:PlaneStressVariousNeighborhoods}. Thus, we must deal with a position-aware micromodulus function $\lambda(\bfx',\bfx)$, and the three-dimensional bond-based linear peridynamic model is given by \eqref{eqn:linearperieqn}. We suppose a homogeneous material response and only consider material points far from the boundaries of the plate in the first and second dimensions. This allows us to assume only position awareness in the third dimension, i.e.,
\begin{equation}\label{eqn:zposaware}
\lambda(\bfx',\bfx) =~\lambda(x'-x,y'-y,z',z).
\end{equation} 
Furthermore, to avoid unintentionally imposing boundary conditions on \eqref{eqn:linearperieqn}, we additionally suppose $\mathcal{H}_\bfx = B_{\delta}^{\text{3D}}(\bfx)$ ({\em cf.} Footnote \ref{footnote:imposingboundary}).

\begin{center}
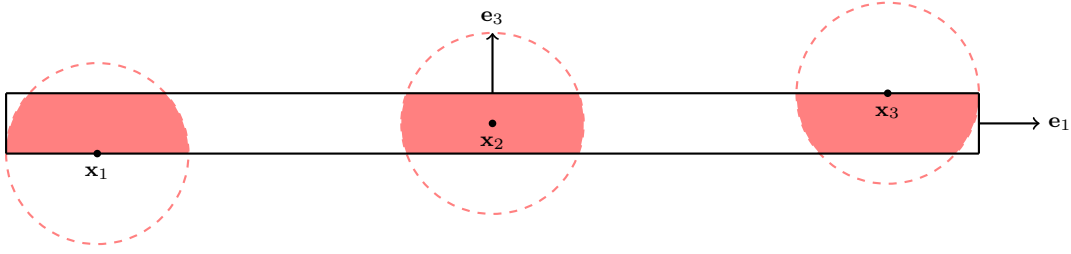
\begin{figure}
\begin{tikzpicture}[scale=0.8]

%dotted circles
\draw[dashed,red!50,thick] (-8,-1/2) arc (180:180+360:1.5);
\draw[dashed,red!50,thick] (-1.5,0) arc (180:180+360:1.5);
\draw[dashed,red!50,thick] (5,1/2) arc (180:180+360:1.5);          
	
%filled in regions for neighborhoods
\begin{scope}    
\path[clip] (-8,1/2)--(8,1/2)--(8,-1/2)--(-8,-1/2)--cycle;  
\path[fill=red!50] (0,0) circle(1.5);
\path[fill=red!50] (-6.5,-1/2) circle(1.5);      
\path[fill=red!50] (6.5,1/2) circle(1.5);    
\end{scope}

%plate cross-section
\draw[-,thick] (-8,1/2)--(8,1/2);
\draw[-,thick] (-8,-1/2)--(8,-1/2);
\draw[-,thick] (-8,1/2)--(-8,-1/2);
\draw[-,thick] (8,1/2)--(8,-1/2);    
    
%three points x,x_2,x_3    
\node (x1) at (-6.5,-1/2)[circle,fill=black,thick,inner sep=1pt,minimum size=0.1cm,label=below:$\bfx_1$]{}; 
\node (x2) at (0,0)[circle,fill=black,thick,inner sep=1pt,minimum size=0.1cm,label=below:$\bfx_2$]{}; 
\node (x3) at (6.5,1/2)[circle,fill=black,thick,inner sep=1pt,minimum size=0.1cm,label=below:$\bfx_3$]{}; 

%axes      
\draw[->,thick] (0,1/2)--(0,1.5) node[above]{$\bfe_3$};
\draw[->,thick] (8,0)--(9,0) node[right]{$\bfe_1$};      
            
\end{tikzpicture}
\caption{Various neighborhoods (in red) of material points $\bfx_1,\bfx_2,$ and $\bfx_3$ in an $xz$-cross-section of a thin plate.} \label{fig:PlaneStressVariousNeighborhoods}
\end{figure}
\end{center}

The peridynamic generalized plane stress assumptions are given as follows:
\begin{enumerate}[(P$\sigma$1)]
\item The body is a thin plate of thickness $2h \leqslant \delta$ occupying the region $-h \leqslant z \leqslant h$. \label{assump:PSs1Peri}
\item The density is constant in the third dimension: $\rho = \rho(x,y)$. \label{assump:PSs5Peri}
\item The body is subjected to a loading symmetric and parallel to the plane $z = 0$: \label{assump:PSs2Peri}
\begin{equation}
b_3(\bfx,t) = 0 \quad \text{and} \quad \bfb(x,y,z,t) = \bfb(x,y,-z,t).
\end{equation}
\item The first and second components of the initial and boundary conditions are symmetric while their third component is antisymmetric relative to the plane $z=0$. \label{assump:PSsSymmPeri} 
\item The micromodulus function $\lambda(\bfx',\bfx)$ is null when $\bfx'$ and $\bfx$ are not material points of the plate. \label{assump:PSs3Peri}
\item The average peridynamic traction $\overline{\tau_{3}(\bfx,t,\bfe_3)}$ ({\em cf.}~\eqref{eqn:peritraction}) is zero throughout the body. \label{assump:PSs4Peri}
\item The material has at least monoclinic symmetry with a plane of reflection corresponding to the plane $z=0$, i.e., the micromodulus function has the following symmetries\footnote{
Invariance of $\lambda$ with respect to $\mathbf{Ref}(\mathbf{e}_3)$ ensures
\begin{equation*}
\lambda(\xi_1,\xi_2,z',z) = \lambda(\xi_1,\xi_2,-z',-z).
\end{equation*}
Requiring invariance of $\lambda$ with respect to $-\bfI$ and then recalling 
\eqref{eqn:lambdasymm}, we find
\begin{equation*}
\lambda(\xi_1,\xi_2,-z',-z) = \lambda(-\xi_1,-\xi_2,z',z) = \lambda(\xi_1,\xi_2,z,z').
\end{equation*} 
} ({\em cf}. Definition \ref{def:generalsymm} and \eqref{eqn:generatorsmonoclinic}) : \label{assump:PSs7Peri}
\begin{equation}\label{assump:MonoclinicPlaneStress}
\lambda(\xi_1,\xi_2,z',z) = \lambda(\xi_1,\xi_2,-z',-z)  = \lambda(\xi_1,\xi_2,z,z').
%\lambda(x',y',z',x,y,z) = \lambda(x',y',z,x,y,z') = \lambda(x',y',-z',x,y,-z)  .
\end{equation}

%\item The horizon $\delta$ is at least as large as the thickness of the plate $2h$. \label{assump:PSs9Peri}
\item The displacement $\bfu(\bfx,t)$ is smooth in $z$ and its third component, $u_3(\bfx,t)$, is smooth in $\bfx$. \label{assump:PSs6Peri}
\end{enumerate} 

There are a variety of similarities between the peridynamic generalized plane stress assumptions and the classical generalized plane stress assumptions. Notice Assumption (P$\sigma$\ref{assump:PSs1Peri}) is Assumption (C$\sigma$\ref{assump:PSs1}) with an additional restriction on the thickness $2h$ that depends on the peridynamic horizon $\delta$. This assumption is reasonable for thin plates, since the computational expense associated with the discretization required to accommodate a horizon $\delta < 2h$ is likely impractical. This restriction could likely be removed, but it simplifies the analysis considerably.  Assumptions (P$\sigma$\ref{assump:PSs5Peri}), (P$\sigma$\ref{assump:PSs2Peri}), and (P$\sigma$\ref{assump:PSsSymmPeri}) are identical to Assumptions (C$\sigma$\ref{assump:PSs5}), (C$\sigma$\ref{assump:PSs2}), and (C$\sigma$\ref{assump:PSsSymm}), respectively. Assumption (P$\sigma$\ref{assump:PSs3Peri}) is a nonlocal analogue of Assumption (C$\sigma$\ref{assump:PSs3}) in that interactions through the top and bottom surfaces of the plate are nullified\footnote{For the problems considered here, i.e., thin plates with free surfaces, it is natural to assume all nonlocal interactions through the top and bottom surfaces of the plate are zero. In a more general setting, one could postulate a nonlocal analogue of (C$\sigma$\ref{assump:PSs3}), i.e., $\sigma_{13} = \sigma_{23} = \sigma_{33} = 0$ on the top and bottom surfaces of the plate, by employing a peridynamic stress tensor $\boldsymbol{\sigma}^{\text{peri}}$. Such an approach would involve nullifying the net forces on each material point on the top and bottom surfaces of the plate rather than nullifying each pairwise interaction.}. Assumption (P$\sigma$\ref{assump:PSs4Peri}) is a peridynamic analogue of Assumption (C$\sigma$\ref{assump:PSs4}). To see this, following \cite{SILLING2000}, under a homogeneous deformation, one may introduce a peridynamic stress tensor $\boldsymbol{\sigma}^{\text{peri}}$ so that the areal force density or peridynamic traction $\boldsymbol{\tau}(\bfx,t,\bfn) = \boldsymbol{\sigma}^{\text{peri}}(\bfx,t) \bfn$, and thus $\overline{\tau_{3}(\bfx,t,\bfe_3)} = \overline{\sigma_{33}^{\text{peri}}(\bfx,t)}$\footnote{The dependence of $\boldsymbol{\sigma}$ on $\bfx$ is introduced due to possible surface effects for points not in the bulk of the body.}. Assumptions~(P$\sigma$\ref{assump:PSs7Peri}) and (C$\sigma$\ref{assump:PSs6}) both impose monoclinic symmetry on their respective models. Finally, Assumption~(P$\sigma$\ref{assump:PSs6Peri}) imposes regularity conditions on the displacement field. These regularity  conditions are, however, significantly weaker than those in the classical theory, where $\bfu(\bfx,t)$ is assumed to be smooth in $\bfx$. 

A system satisfying all of (P$\sigma$\ref{assump:PSs1Peri})--(P$\sigma$\ref{assump:PSs6Peri}) is said to be in a state of peridynamic generalized plane stress. See Figure \ref{fig:planestress} for an illustration of a body in a state of plane stress. Similarly to the classical theory, we start by summarizing in Lemma~\ref{lem:symmplanestress} (peridynamic analogue of Lemma~\ref{lem:planestressymm}) the symmetries imposed on the displacement field $\bfu$ in a system in a state of peridynamic generalized plane stress.

\begin{lemma}\label{lem:symmplanestress}
Under Assumptions (P$\sigma$\ref{assump:PSs1Peri})--(P$\sigma$\ref{assump:PSs3Peri}) and (P$\sigma$\ref{assump:PSs7Peri}), the displacement field $\bfu$ of \eqref{eqn:linearperieqn} satisfies
\begin{align}\label{eqn:planestressmonoclinicdisplacementPeri}
u_i(x,y,-z,t) = \left\{ \begin{array}{ll}
u_i(x,y,z,t),   & i = 1,2  \\
-u_i(x,y,z,t), & i =3
\end{array}, \right.
\end{align}
i.e., the in-plane displacements, $u_1$ and $u_2$, are symmetric while the out-of-plane displacement, $u_3$, is antisymmetric relative to the plane $z = 0$.  
\end{lemma}

\begin{proof}
See Appendix \ref{appendix:symmplanestress}. \qed 
\end{proof}

In order to derive the peridynamic generalized plane stress model, we perform the following steps, which are analogous to those presented in the classical theory ({\em cf}. Remark \ref{rmk:planestressteps}):
\begin{enumerate}[Step 1:]
\item Take the average of the equation of motion \eqref{eqn:linearperieqn} over the thickness of the plate.
\item Utilize Assumption (P$\sigma$\ref{assump:PSs3Peri}) to eliminate interactions through the top and bottom surfaces of the plate.
\item Employ Assumption (P$\sigma$\ref{assump:PSs4Peri}) to replace expressions in $u_3$ by expressions in $u_1$ and $u_2$.
\item Replace expressions in $u_1$ and $u_2$ with expressions in $\overline{u}_1$ and $\overline{u}_2$.
\item Integrate in $z'$ and $z$ to remove the third dimension dependence from the equation of motion. 
\end{enumerate}
In the derivation, Lemma \ref{lem:symmplanestress} is utilized to eliminate various terms.

We now follow steps 1--5 in order to derive the peridynamic generalized plane stress model. We utilize the shorthand notation introduced in \eqref{eqn:averagedefinition}. The following steps involve various lemmas whose proofs have been moved to Appendix~\ref{Sec:AppPeriPlaneStressLemmas} in order to provide clarity to the derivation of the peridynamic generalized plane stress model.

%%%%%%%%%%%%%%%%%%%%%%%%%%%%%%%%%%%%%%%%%%%%%%%%%%%%%%%
%%%%%%%%%%%%%%%%%%%%%%%%%%%%%%%%%%%%%%%%%%%%%%%%%%%%%%%

\textbf{Step 1}: We start by taking the average of the peridynamic equation of motion~\eqref{eqn:linearperieqn} over the thickness of the plate (recall $\rho$ is constant in $z$ by (P$\sigma$\ref{assump:PSs5Peri})):
\begin{equation}\label{eqn:perimotionpretraction}
\begin{split}
\rho(x,y) \ddot{\overline{u}}_i(x,y,t) =& \frac{1}{2h} \int_{-h}^h \int_{B_{\delta}^{3D}(\bfx)} \lambda(\bfx',\bfx) \xi_i \xi_j (u_j(\bfx',t) - u_j(\bfx,t)) d \bfx' d z + \overline{b}_i(x,y,t).
\end{split}
\end{equation}

%%%%%%%%%%%%%%%%%%%%%%%%%%%%%%%%%%%%%%%%%%%%%%%%%%%%%%%
%%%%%%%%%%%%%%%%%%%%%%%%%%%%%%%%%%%%%%%%%%%%%%%%%%%%%%%

\textbf{Step 2}: We impose Assumptions  (P$\sigma$\ref{assump:PSs1Peri}) and (P$\sigma$\ref{assump:PSs3Peri}) on \eqref{eqn:perimotionpretraction}. Since $\lambda(\bfx,\bfx') = 0$ when $\bfx'$ is not a material point of the plate and $B_\delta^{3D}(\bfx)$ intersects the top and bottom surfaces of the plate for any material point $\bfx$ of the plate, we may restrict the region of integration for $z'$ to $[-h,h]$:
\begin{equation}\label{eqn:PeriPlaneStressStep1Eqnpre}
\begin{split}
\rho(x,y) \ddot{\overline{u}}_i(x,y,t) =& \frac{1}{2h} \int_{-h}^h \int_{-h}^h \int_{B^{2D}_r(x,y)} \lambda(\bfx',\bfx) \xi_i \xi_j (u_j(\bfx',t) - u_j(\bfx,t)) d \bfx' d z + \overline{b}_i (x,y,t),
\end{split}
\end{equation}
where $r = \sqrt{\delta^2 - (z'-z)^2}$. Furthermore, due to \eqref{def:Hx} we may suppose (potentially by multiplying by a characteristic function)
\begin{equation}\label{eqn:lambda=0outsideball}
\lambda(\bfx',\bfx) = 0, \quad \| \bfx' - \bfx \| \geqslant \delta,
\end{equation}
and extend the region of integration so that 
\begin{equation}\label{eqn:PeriPlaneStressStep1Eqn}
\begin{split}
\rho(x,y) \ddot{\overline{u}}_i(x,y,t) =& \frac{1}{2h} \int_{-h}^h \int_{-h}^h \int_{B^{2D}_\delta(x,y)} \lambda(\bfx',\bfx) \xi_i \xi_j (u_j(\bfx',t) - u_j(\bfx,t)) d \bfx' d z + \overline{b}_i (x,y,t).
\end{split}
\end{equation}

%%%%%%%%%%%%%%%%%%%%%%%%%%%%%%%%%%%%%%%%%%%%%%%%%%%%%%%
%%%%%%%%%%%%%%%%%%%%%%%%%%%%%%%%%%%%%%%%%%%%%%%%%%%%%%%

\textbf{Step 3}: The goal of this step is to replace the term
\begin{equation}\label{eqn:planestressstep3eqn1}
\begin{split}
\frac{1}{2h} \int_{-h}^h \int_{-h}^h \int_{B^{2D}_\delta(x,y)} \lambda(\bfx',\bfx) \xi_i \xi_3 (u_3(\bfx',t) - u_3(\bfx,t)) d \bfx' d z
\end{split}
\end{equation}
in the right-hand side of \eqref{eqn:PeriPlaneStressStep1Eqn} by expressions in $u_1$ and $u_2$. This is facilitated by Assumption~(P$\sigma$\ref{assump:PSs4Peri}). In order to employ (P$\sigma$\ref{assump:PSs4Peri}), following \cite{SILLING2000}, 
we first define the peridynamic traction $\boldsymbol{\tau}$ at a material point~$\bfx$ in the direction of $\bfe_3$ as  

\begin{align}\label{eqn:peritraction}
\boldsymbol{\tau}(\bfx,t,\bfe_3) &:= \int_0^\delta \int_{\mathcal{B}_\bfx^+(\bfx-s \bfe_3)} \bff(\bfu(\bfx',t)-\bfu(\bfx-s \bfe_3,t),\bfx',(\bfx-s \bfe_3)) d \bfx' ds, %\label{eqn:peritractiona}
%\boldsymbol{\tau}(\bfx,-\bfe_3) &:= \int_0^\delta \int_{\mathcal{B}_\bfx^-(\bfx+s \bfe_3)} \bff (\bfu(\bfx')-\bfu(\bfx + s \bfe_3),\bfx'-(\bfx + s \bfe_3)) d \bfx' ds. \label{eqn:peritractionb}
\end{align}
where
\begin{align}
&\mathcal{B}_\bfx^+(\bfx-s\bfe_3) := \left\{\bfx' \in B_{\delta}(\bfx-s \bfe_3) : z'>z \right\}
%&\mathcal{B}^-(\bfx+s\bfe_3) := \left\{\bfx' \in B_{\delta}(\bfx+s \bfe_3) : z' < z \right\},
\end{align}
and $\bff$ is the pairwise force function. In this work, we are only concerned with linear bond-based peridynamic models, and therefore $\bff$ is given by \eqref{eqn:generalBBpairwiseforceformfinal}. In Figure~\ref{fig:PeridynamicStress}, we present an illustration of the region $\mathcal{B}_\bfx^+(\bfx-s \bfe_3)$.% and $\mathcal{B}^-(\bfx+s \bfe_3)$.

\begin{figure}
\begin{center}
\begin{tikzpicture}[scale=1.8]
\begin{scope}
\path[clip] (2,2.5) -- (-2,2.5) -- (-1.902113,0.618034) arc (180:360:3.80422606518/2 and 0.4) --cycle;  
  \shade[ball color = gray!80, opacity = 0.5] (0,0) circle (2cm);
  \draw (0,0) circle (2cm);
\end{scope}
\draw (-1.902113,0.618034) arc (180:360:3.80422606518/2 and 0.4);
\draw[dashed] (1.902113,0.618034) arc (0:180:3.80422606518/2 and 0.4);

\node (bfx) at (0,0.65)[circle,fill=black,thick,inner sep=1pt,minimum size=0.1cm,label=right:$\bfx$]{};
\node (bfxsbfe) at (0,-0.5)[circle,fill=black,thick,inner sep=1pt,minimum size=0.1cm,label=below:$\bfx-s \bfe_3$]{};
\node (delta) at (0.5,-0.25)[label=$\delta$]{};
\node (R+) at (0,2)[label=above:$\mathcal{B}_\bfx^+(\bfx-s \bfe_3)$]{};
\draw[dashed] (0,-0.5) -- (0,0.65);
\draw[dashed] (0,-0.5) -- (1.902113,0.618034);
\end{tikzpicture}
\end{center}

\caption{Illustration of region $\mathcal{B}_\bfx^+(\bfx-s\bfe_3)$.} \label{fig:PeridynamicStress}
\end{figure}
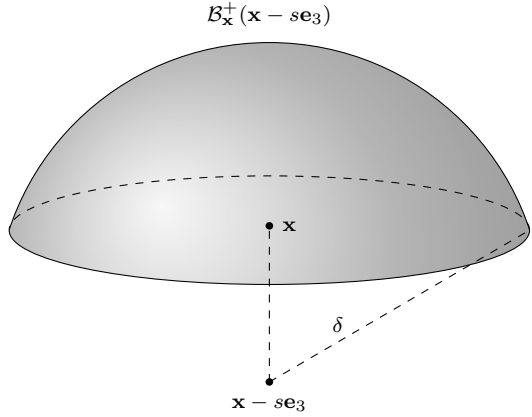

\begin{remark}
As explained above, under a homogeneous deformation, one may introduce a peridynamic stress tensor $\boldsymbol{\sigma}^{\text{peri}}$ such that $\bftau = \boldsymbol{\sigma}^{\text{peri}} \bfn$. Then $\sigma_{i3}^{\text{peri}} = \tau_i(\bfx,t,\bfe_3)$. By Assumption (P$\sigma$\ref{assump:PSs3Peri}), we immediately find $\lambda(\bfx',\bfx-s\bfe_3) = 0$ for $z=\pm h$ and consequently $\sigma_{13}^{\text{peri}} = \sigma_{23}^{\text{peri}} = \sigma_{33}^{\text{peri}} = 0$ for $z=\pm h$, just as in the classical assumption (C$\sigma$\ref{assump:PSs3}). 
\end{remark}

In order to replace the term in \eqref{eqn:planestressstep3eqn1}, we introduce two lemmas. The first lemma, Lemma \ref{lemma:zeroavgtraction}, expresses the average peridynamic traction $\overline{\tau_3(\bfx,t,\bfe_3)}$ in a form more readily comparable to~\eqref{eqn:planestressstep3eqn1}, a necessary step to prove the second lemma, Lemma \ref{lem:PlaneStrainu3SubEstimates}. Lemma \ref{lem:PlaneStrainu3SubEstimates} provides an approximation which allows us to replace \eqref{eqn:planestressstep3eqn1} with an expression in $u_1$ and $u_2$, thus decoupling the in-plane displacements from the out-of-plane displacement in \eqref{eqn:PeriPlaneStressStep1Eqn}. 

\begin{lemma}\label{lemma:zeroavgtraction}
Under Assumptions (P$\sigma$\ref{assump:PSs1Peri})--(P$\sigma$\ref{assump:PSs3Peri}) and (P$\sigma$\ref{assump:PSs7Peri}), the average of the peridynamic traction $\tau_3(\bfx,t,\bfe_3)$, with the pairwise force function $\bff$ given by~\eqref{eqn:generalBBpairwiseforceformfinal}, over the thickness of the plate satisfies
\begin{equation}\label{eqn:zeroavgtraction}
\begin{split}
\overline{\tau_3(\bfx,t,\bfe_3)} = \frac{1}{4h} \int_{-h}^h \int_{-h}^h \int_{B^{2D}_{\delta}(x,y)} \lambda(\bfx',\bfx) \xi_3^2 \xi_j (u_j(\bfx',t)-u_j(\bfx,t)) d \bfx' d z.
\end{split}
\end{equation}
\end{lemma}

\begin{proof}
See Appendix \ref{lemma:zeroavgtractionapp}. \qed
\end{proof}

\begin{lemma}\label{lem:PlaneStrainu3SubEstimates}
For $i=1$ or $2$, under Assumptions (P$\sigma$\ref{assump:PSs1Peri})--(P$\sigma$\ref{assump:PSs6Peri}), the following approximation holds for second-order Taylor expansions of the displacement $u_3$ about $(x,y,0,t)$:

\begin{equation}\label{eqn:PlaneStrainu3Subeqn2}
\int_{-h}^h \int_{-h}^h \int_{B^{2D}_\delta(x,y)} \lambda(\bfx',\bfx) \xi_i \xi_3 (u_3(\bfx',t) - u_3(\bfx,t)) d \bfx' d z \approx -\frac{1}{2} \int_{-h}^h \int_{-h}^h \int_{B^{2D}_\delta(x,y)} \lambda(\bfx',\bfx) \xi_i \xi_3^2 A(x',y',t) d \bfx' d z
\end{equation}
where
\begin{equation}\label{eqn:PlaneStressAtermPostSub}
A(x',y',t) := \frac{ \int_{-h}^h \int_{-h}^h \int_{B^{2D}_{\delta}(x',y')} \lambda(\bfx'',\bfx') \zeta^2_3 \left[ \zeta_1 (u_1(\bfx'',t)- u_1(\bfx',t)) + \zeta_2 (u_2(\bfx'',t)- u_2(\bfx',t)) \right] d \bfx'' d z'}{\int_{-h}^h \int_{-h}^h \int_{B^{2D}_{\delta}(x',y')} \lambda(\bfx'',\bfx') \zeta^4_3 d \bfx'' dz'}.
\end{equation}
Here, we defined $\bfzeta := \bfx'' - \bfx'$ to avoid confusion with the terms in~\eqref{eqn:PlaneStrainu3Subeqn2}. 

\end{lemma}

\begin{proof}
See Appendix \ref{lem:PlaneStrainu3SubEstimatesapp}. \qed
\end{proof}
Utilizing Lemma \ref{lem:PlaneStrainu3SubEstimates}, we may substitute \eqref{eqn:PlaneStrainu3Subeqn2} into \eqref{eqn:PeriPlaneStressStep1Eqn} to obtain our decoupled model (for $i=1,2$):

\begin{equation}\label{eqn:PlaneStressSubu3PreSub3}
\begin{split}
\rho(x,y) \ddot{\overline{u}}_i(x,y,t) \approx& \frac{1}{2h} \int_{-h}^h \int_{-h}^h \int_{B^{2D}_\delta(x,y)} \lambda(\bfx',\bfx) \xi_i \left[ \xi_1 (u_1(\bfx',t) - u_1(\bfx,t)) + \xi_2 (u_2(\bfx',t) - u_2(\bfx,t)) \right. \\
&\hspace*{1.9in} \left.- \frac{1}{2} \xi_3^2 A(x',y',t)\right] d \bfx' d z + \overline{b}_i(x,y,t).  
\end{split}
\end{equation}

%%%%%%%%%%%%%%%%%%%%%%%%%%%%%%%%%%%%%%%%%%%%%%%%%%%%%%%
%%%%%%%%%%%%%%%%%%%%%%%%%%%%%%%%%%%%%%%%%%%%%%%%%%%%%%%

\textbf{Step 4}:
In this step, we remove the dependence of $u_1$ and $u_2$ in~\eqref{eqn:PlaneStressSubu3PreSub3} on the third dimension in order to obtain expressions in $\overline{u}_1$ and $\overline{u}_2$. In classical generalized plane stress, $u_1$ and $u_2$ are immediately replaced by expressions in $\overline{u}_1$ and $\overline{u}_2$ by taking the average of the equation of motion ({\em cf}. Step 1 from Remark~\ref{rmk:planestressteps}). However, in peridynamics we cannot simply integrate over the third dimension on the right-hand side of \eqref{eqn:PlaneStressSubu3PreSub3} to directly obtain expressions in $\overline{u}_1$ and $\overline{u}_2$. This is due to the presence of the micromodulus function $\lambda$ and the third component of the bond, $\xi_3$, in the integrand in~\eqref{eqn:PlaneStressSubu3PreSub3}, which would result in weighted averages of $u_1$ and $u_2$ instead. To overcome this obstacle, we introduce Lemma \ref{lemma:getubar} below.

\begin{lemma}\label{lemma:getubar}
Under Assumptions (P$\sigma$\ref{assump:PSs1Peri})--(P$\sigma$\ref{assump:PSs3Peri}) and (P$\sigma$\ref{assump:PSs7Peri})--(P$\sigma$\ref{assump:PSs6Peri}), we have for $i=1$ or $2$:

\begin{equation}\label{eqn:getu}
u_i(\bfx,t) = \overline{u}_i(x,y,t) +O(h^2).
\end{equation}

\end{lemma}
\begin{proof}
See Appendix \ref{lemma:getubarapp}. \qed
\end{proof}

Utilizing Lemma \ref{lemma:getubar} and recalling terms of order $O(h^2)$ are small by Assumption~(P$\sigma$\ref{assump:PSs1Peri}), we may suppose $u_i(\bfx,t) \approx \overline{u}_i(x,y,t)$. Substituting this approximation into  \eqref{eqn:PlaneStressAtermPostSub} and  \eqref{eqn:PlaneStressSubu3PreSub3}, results in 
\begin{equation}\label{eqn:PlaneStressSubu3PreSub3avg}
\begin{split}
\rho(x,y) \ddot{\overline{u}}_i(x,y,t) \approx& \frac{1}{2h} \int_{-h}^h \int_{-h}^h \int_{B^{2D}_\delta(x,y)} \lambda(\bfx',\bfx) \xi_i \left[ \xi_1 (\overline{u}_1(\bfx',t) - \overline{u}_1(\bfx,t)) + \xi_2 (\overline{u}_2(\bfx',t) - \overline{u}_2(\bfx,t)) \right. \\
&\hspace*{1.9in} \left.- \frac{1}{2} \xi_3^2 A(x',y',t)\right] d \bfx' d z + \overline{b}_i(x,y,t),
\end{split}
\end{equation}
where
\begin{equation}\label{eqn:PlaneStressAtermPostSubavg}
A(x',y',t) \approx \frac{ \int_{-h}^h \int_{-h}^h \int_{B^{2D}_{\delta}(x',y')} \lambda(\bfx'',\bfx') \zeta^2_3 \left[ \zeta_1 (\overline{u}_1(x'',y'',t)- \overline{u}_1(x',y',t)) + \zeta_2 (\overline{u}_2(x'',y'',t) - \overline{u}_2(x',y',t)) \right] d \bfx'' d z'}{\int_{-h}^h \int_{-h}^h \int_{B^{2D}_{\delta}(x',y')} \lambda(\bfx'',\bfx') \zeta^4_3 d \bfx'' d z'}.
\end{equation}

%%%%%%%%%%%%%%%%%%%%%%%%%%%%%%%%%%%%%%%%%%%%%%%%%%%%%%%
%%%%%%%%%%%%%%%%%%%%%%%%%%%%%%%%%%%%%%%%%%%%%%%%%%%%%%%

\textbf{Step 5}:
The final step is to integrate over $z'$ and $z$ in order to remove the third dimension dependence from \eqref{eqn:PlaneStressSubu3PreSub3avg} and \eqref{eqn:PlaneStressAtermPostSubavg}. For convenience, we introduce the shorthand notation (recall \eqref{eqn:zposaware})
\begin{equation}\label{eqn:defforlambiplanestress}
\lambda_i(\xi_1,\xi_2) := \int_{-h}^h \int_{-h}^h \lambda(\bfx',\bfx) \xi_3^i dz' dz. 
\end{equation}
Since the limits of integration in  \eqref{eqn:PlaneStressSubu3PreSub3avg} and \eqref{eqn:PlaneStressAtermPostSubavg} are independent of each other, we may change the order of integration without altering the limits so that we may apply \eqref{eqn:defforlambiplanestress}. In addition, we perform the change of variables $(x',y') \rightarrow (x+\xi_1,y+\xi_2)$ in \eqref{eqn:PlaneStressSubu3PreSub3avg} and $(x'',y'') \rightarrow (x'+\zeta_1,y'+\zeta_2)$ in \eqref{eqn:PlaneStressAtermPostSubavg}. Up to the approximations in Steps 1--5,  the peridynamic generalized plane stress equation of motion is given by (analogue of \eqref{eqn:planestressequationsofmotionclassicala} and \eqref{eqn:planestressequationsofmotionclassicalb}):

\begin{subequations}\label{eqn:PlaneStressFinal}
\begin{align}
\rho(x,y) \ddot{\overline{u}}_1(x,y,t) ={}& \frac{1}{2h}  \int_{B^{2D}_\delta(\bf0)} \lambda_0(\xi_1,\xi_2) \xi_1 \left[ \xi_1 (\overline{u}_1(x+\xi_1,y+\xi_2,t) - \overline{u}_1(x,y,t)) \right. \nonumber \\
&\left. \hspace*{1.4in} + \xi_2 (\overline{u}_2(x+\xi_1,y+\xi_2,t) - \overline{u}_2(x,y,t)) \right] \nonumber \\
&\hspace*{0.65in}- \frac{1}{2} \lambda_2(\xi_1,\xi_2) A(x+\xi_1,y+\xi_2,t) \xi_1  d \xi_1 d \xi_2 + \overline{b}_1(x,y,t),   \\
\rho(x,y) \ddot{\overline{u}}_2(x,y,t) ={}& \frac{1}{2h}  \int_{B^{2D}_\delta(\bf0)} \lambda_0(\xi_1,\xi_2) \xi_2 \left[ \xi_2 (\overline{u}_1(x+\xi_1,y+\xi_2,t) - \overline{u}_1(x,y,t))  \right.  \nonumber \\
&\hspace*{1.4in} \left. + \xi_2 (\overline{u}_2(x+\xi_1,y+\xi_2,t) - \overline{u}_2(x,y,t)) \right] \nonumber  \\
&\hspace*{0.65in} - \frac{1}{2} \lambda_2(\xi_1,\xi_2) A(x+\xi_1,y+\xi_2,t) \xi_2  d \xi_1 d \xi_2 + \overline{b}_2(x,y,t), 
\end{align}
\end{subequations}
where
\begin{equation}\label{eqn:PlaneStressAtermFinal}
A(x',y',t) = \frac{ \int_{B^{2D}_{\delta}(\bf0)} \lambda_2(\zeta_1,\zeta_2) \left( \zeta_1 (\overline{u}_1(x'+\zeta_1,y'+\zeta_2,t)- \overline{u}_1(x',y',t)) + \zeta_2 (\overline{u}_2(x'+\zeta_1,y'+\zeta_2,t) - \overline{u}_2(x',y',t)) \right) d \zeta_1 d \zeta_2}{\int_{B^{2D}_{\delta}(\bf0)} \lambda_4(\zeta_1,\zeta_2) d \zeta_1 d \zeta_2 }.
\end{equation}

\begin{remark}
Due to the presence of $A(x',y',t)$ in the equation of motion \eqref{eqn:PlaneStressFinal}, the resulting peridynamic generalized plane stress model is not a bond-based peridynamic model but rather a state-based peridynamic model~\cite{Silling2007}. Analogously to Remark~\ref{rmk:planestrainvector}, we can express the corresponding equation of motion for~\eqref{eqn:PlaneStressFinal} in vector form by letting $\bfx =(x,y)$, $\bfxi =(\xi_1,\xi_2)$, $\bfzeta=(\zeta_1,\zeta_2)$, $\overline{\bfu} =(\overline{u}_1,\overline{u}_2)$, $\mathcal{H} = B_\delta^{2D}(\bf0)$, and $\overline{\bfb}=(\overline{b}_1,\overline{b}_2)$. In this case, using the fact that $\lambda_i(x',y',x,y) = \lambda_i(x,y,x',y')$ by \eqref{eqn:lambdasymm} and Assumption (P$\sigma$\ref{assump:PSs7Peri}), the peridynamic generalized plane stress equation of motion can be formulated as
\begin{equation*}
\rho(\bfx) \ddot{\overline{\bfu}}(\bfx,t) = \int_{\mathcal{H}} \left\{ \underline{\bfT}[\bfx ,t]\langle \bfxi \rangle - \underline{\bfT}[\bfx + \bfxi,t]\langle -\bfxi \rangle \right\} d \bfxi + \overline{\bfb}(\bfx,t),
\end{equation*}
where
\begin{equation*}
\underline{\bfT}\left[ \bfx, t \right] \langle \bfxi \rangle = \frac{1}{4h} \left[  \lambda_0(\bfxi) \bfxi \otimes \bfxi (\overline{\bfu}(\bfx+\bfxi,t) - \overline{\bfu}(\bfx,t)) -\lambda_2(\bfxi)A(\bfx,t) \bfxi \right] 
\end{equation*}
and
\begin{equation*}
A(\bfx,t) = \frac{1}{ \int_{\mathcal{H}} \lambda_4(\bfzeta) d \bfzeta} \int_{\mathcal{H}} \lambda_2(\bfzeta) \bfzeta \cdot \left( \overline{\bfu}(\bfx + \bfzeta,t) - \overline{\bfu}(\bfx,t) \right) d \bfzeta.
\end{equation*}

\end{remark}

%%%%%%%%%%%%%%%%%%%%%%%%%%%%%%%%%%%%%%%%%%%%%%%%%%%%%%%
%%%%%%%%%%%%%%%%%%%%%%%%%%%%%%%%%%%%%%%%%%%%%%%%%%%%%%%

\underline{Peridynamic plane stress micromodulus functions}:

In the peridynamic generalized plane stress model presented in this work, the requirements placed on the micromodulus function $\lambda(\bfx',\bfx)$ have been fairly minimal up to this point. In order to investigate the behavior of the peridynamic plane stress micromodulus functions $\lambda_0(\xi_1,\xi_2)$ and $\lambda_2(\xi_1,\xi_2)$ appearing in \eqref{eqn:PlaneStressFinal} and inform them with the classical elasticity tensor $\mathbb{C}$, we necessarily must provide a prototype micromodulus $\lambda(\bfx',\bfx)$. Unfortunately, the surface effects endemic in peridynamic formulations of generalized plane stress require $\lambda(\bfx',\bfx)$ to be position aware in order to satisfy \eqref{eqn:lambdacondfullgen}. Rather than introducing a new micromodulus function to accommodate surface effects, we posit a plausible alternative strategy: relax the requirements  \eqref{eqn:lambdacondfullgen} on $\lambda(\bfx',\bfx)$ to hold only in the average over the thickness of the plate\footnote{This concept parallels the theory for classical generalized plane stress which is developed for quantities averaged over the thickness of the plate ({\em cf}. Section~\ref{sec:ClassicalPlaneStress}).}. Specifically, 
\begin{subequations}\label{eqn:CijklRelationAverage}
\begin{align}
0 ={}& \int_{-h}^h \int_{\mathcal{H}_\bfx} \lambda(\bfx',\bfx) \xi_i \xi_j \xi_k d \bfx' d z \label{eqn:CijklRelationAveragea} \\
C_{ijkl} ={}& \frac{1}{4h} \int_{-h}^h \int_{\mathcal{H}_\bfx} \lambda(\bfx',\bfx) \xi_i \xi_j \xi_k \xi_l d \bfx' d z. \label{eqn:CijklRelationAverageb} 
\end{align}
\end{subequations}
As can be seen in Proposition \ref{prop:relationplanestresstoclassical}, if we impose the weaker conditions \eqref{eqn:CijklRelationAverage} (in relation to \eqref{eqn:lambdacondfullgen}) on $\lambda(\bfx',\bfx)$, the peridynamic generalized plane stress model \eqref{eqn:PlaneStressFinal}  reduces to the classical generalized plane stress model \eqref{eqn:planestressequationsofmotionclassical} under a second-order Taylor expansion. Moreover, if we suppose the micromodulus is a function only of the bond, such as in the case of \eqref{def:lambdagenform}, then \eqref{eqn:CijklRelationAveragea} is trivially satisfied. Consequently, when averaging over the plate, the surface effects negate each other for such a micromodulus.

With the relaxed formulation \eqref{eqn:CijklRelationAverage}, we may consider a far larger class of functions. In particular, one can inform the constants of the peridynamic tensor $\mathbbm{\Lambda}$ so that \eqref{def:lambdagenform} satisfies \eqref{eqn:CijklRelationAverage}. This is precisely the approach we take in this section in order to investigate the behavior of the peridynamic plane stress micromodulus functions $\lambda_0(\xi_1,\xi_2)$ and $\lambda_2(\xi_1,\xi_2)$ appearing in \eqref{eqn:PlaneStressFinal}. 
%In this section we utilize \eqref{def:lambdagenform} as our prototype micromodulus in order to investigate the behavior of the peridynamic plane stress micromodulus functions $\lambda_0(\xi_1,\xi_2)$ and $\lambda_2(\xi_1,\xi_2)$ appearing in \eqref{eqn:PlaneStressFinal}.
Since no material point is assumed to be in the bulk of the material in our peridynamic plane stress formulation, we multiply \eqref{def:lambdagenform} by 
\begin{equation}
\chi_{B_{\delta}(\bf0)} := \left\{ \begin{array}{ll}
1, & \| \bfxi \| < 0 \\
0, & \text{else}
\end{array} \right. 
\end{equation}
in order to enforce \eqref{eqn:lambda=0outsideball}. Then, given the micromodulus function described by~\eqref{def:lambdagenform} multiplied with $\chi_{B_\delta(\mathbf{\bfx})}$, the micromodulus functions \eqref{eqn:defforlambiplanestress} are formulated as:
\begin{equation}\label{eqn:lambdaiplanestressminmaxdef}
\begin{split}
\lambda_i(\xi_1,\xi_2) ={}& \int_{-h}^h \int_{-h}^{h} \lambda(\bfx'-\bfx) \xi_3^i \chi_{B_{\delta}(\bfx)} d z' d z \\
={}& \int_{-h}^h \int_{-h-z}^{h-z} \lambda(\bfxi) \xi_3^i \chi_{B_{\delta}(\bf0)} d \xi_3 d z \\
={}& \int_{-h}^h \int_{\max \left\{ -h-z,-\sqrt{\delta^2-r^2} \right\}}^{ \min \left\{ h-z,\sqrt{\delta^2-r^2} \right\} } \lambda(\bfxi) \xi_3^i d \xi_3 dz,
\end{split}
\end{equation}
where $r = \sqrt{\xi_1^2+\xi_2^2}$. In order to remove the piecewise nature of the limits of integration in \eqref{eqn:lambdaiplanestressminmaxdef}, we consider two regions: $r^2 < \delta^2 - 4h^2$ and $\delta^2 - 4h^2 \leqslant r^2 < \delta^2$. 

\textbf{Region $\mathbf{r^2 < \boldsymbol{\delta}^2-4h^2}$}: In this region, it follows that $2h < \sqrt{\delta^2-r^2}$. Since $|z| \leqslant h$, in this region we have $h-z \leqslant 2h < \sqrt{\delta^2-r^2}$ and $-\sqrt{\delta^2-r^2} < -2h \leqslant -h-z$. Consequently,  \eqref{eqn:lambdaiplanestressminmaxdef} simplfies to
\begin{equation}\label{eqn:planestresslambifirstregion}
\lambda_i(\xi_1,\xi_2) = \int_{-h}^h \int_{ -h-z}^{ h-z} \lambda(\bfxi) \xi_3^i d \xi_3 dz.
\end{equation}

\textbf{Region $\mathbf{ \boldsymbol{\delta}^2 -4h^2 \leqslant r^2 \leqslant \boldsymbol{\delta}^2}$}: In this region, it follows that $\sqrt{\delta^2 - r^2} \leqslant 2h$. We first split the integral in \eqref{eqn:lambdaiplanestressminmaxdef} to find
\begin{equation}\label{eqn:kernelfirstsplit}
\begin{split}
\lambda_i(\xi_1,\xi_2) ={}&  \int_{-h}^{h-\sqrt{\delta^2-r^2}} \int_{\max \left\{ -h-z,-\sqrt{\delta^2-r^2} \right\}}^{ \min \left\{ h-z,\sqrt{\delta^2-r^2} \right\} } \lambda(\bfxi) \xi_3^i d \xi_3 dz + \int_{h-\sqrt{\delta^2-r^2}}^h \int_{\max \left\{ -h-z,-\sqrt{\delta^2-r^2} \right\}}^{ \min \left\{ h-z,\sqrt{\delta^2-r^2} \right\} } \lambda(\bfxi) \xi_3^i d \xi_3 dz \\
%%%%%%%%%%%%%%%%%%%%%%%%%%%%%%%%%%%%%%%%
={}& \int_{-h}^{h-\sqrt{\delta^2-r^2}} \int_{\max \left\{ -h-z,-\sqrt{\delta^2-r^2} \right\}}^{ \sqrt{\delta^2-r^2}  } \lambda(\bfxi) \xi_3^i d \xi_3 dz + \int_{h-\sqrt{\delta^2-r^2}}^h \int_{\max \left\{ -h-z,-\sqrt{\delta^2-r^2} \right\}}^{ h-z } \lambda(\bfxi) \xi_3^i d \xi_3 dz . \\
 \end{split} 
\end{equation}
For the second equality in \eqref{eqn:kernelfirstsplit}, we utilized the fact that $h-z \geqslant \sqrt{\delta^2 - r^2}$ for $z \in [-h, h - \sqrt{\delta^2 -r^2}]$ and $h-z \leqslant \sqrt{\delta^2 -r^2}$ for $z \in [h-\sqrt{\delta^2-r^2},h]$. We then split the integrals in \eqref{eqn:kernelfirstsplit} to find
\begin{equation}\label{eqn:kernelsecondsplit}
\begin{split}
\lambda_i(\xi_1,\xi_2) ={}& \int_{-h}^{-h+\sqrt{\delta^2-r^2}} \int_{\max \left\{ -h-z,-\sqrt{\delta^2-r^2} \right\}}^{ \sqrt{\delta^2-r^2}  } \lambda(\bfxi) \xi_3^i d \xi_3 dz + \int_{-h+\sqrt{\delta^2-r^2}}^{h-\sqrt{\delta^2-r^2}} \int_{\max \left\{ -h-z,-\sqrt{\delta^2-r^2} \right\}}^{ \sqrt{\delta^2-r^2}  } \lambda(\bfxi) \xi_3^i d \xi_3 dz \\
&+ \int_{h-\sqrt{\delta^2-r^2}}^{-h+\sqrt{\delta^2-r^2}} \int_{\max \left\{ -h-z,-\sqrt{\delta^2-r^2} \right\}}^{ h-z } \lambda(\bfxi) \xi_3^i d \xi_3 dz + \int_{-h+\sqrt{\delta^2-r^2}}^h \int_{\max \left\{ -h-z,-\sqrt{\delta^2-r^2} \right\}}^{ h-z } \lambda(\bfxi) \xi_3^i d \xi_3 dz \\
%%%%%%%%%%%%%%%%%%%%%%%%%%%%%%%%%%%%%%%%%%%%
={}& \int_{-h}^{-h+\sqrt{\delta^2-r^2}} \int_{-h-z}^{ \sqrt{\delta^2-r^2}  } \lambda(\bfxi) \xi_3^i d \xi_3 dz + \int_{-h+\sqrt{\delta^2-r^2}}^{h-\sqrt{\delta^2-r^2}} \int_{\max \left\{ -h-z,-\sqrt{\delta^2-r^2} \right\}}^{ \sqrt{\delta^2-r^2}  } \lambda(\bfxi) \xi_3^i d \xi_3 dz \\
&+ \int_{h-\sqrt{\delta^2-r^2}}^{-h+\sqrt{\delta^2-r^2}} \int_{\max \left\{ -h-z,-\sqrt{\delta^2-r^2} \right\}}^{ h-z } \lambda(\bfxi) \xi_3^i d \xi_3 dz + \int_{-h+\sqrt{\delta^2-r^2}}^h \int_{ -\sqrt{\delta^2-r^2} }^{ h-z } \lambda(\bfxi) \xi_3^i d \xi_3 dz \\
%%%%%%%%%%%%%%%%%%%%%%%%%%%%%%%%%%%%%%%%%%
={}& \int_{-h}^{-h+\sqrt{\delta^2-r^2}} \int_{-h-z}^{ \sqrt{\delta^2-r^2}  } \lambda(\bfxi) \xi_3^i d \xi_3 dz + \int_{-h+\sqrt{\delta^2-r^2}}^{h-\sqrt{\delta^2-r^2}} \int^{ \sqrt{\delta^2 - r^2} }_{ h-z } \lambda(\bfxi) \xi_3^i d \xi_3 dz \\
&+ \int_{-h+\sqrt{\delta^2-r^2}}^h \int_{ -\sqrt{\delta^2-r^2} }^{ h-z } \lambda(\bfxi) \xi_3^i d \xi_3 dz .  \\
\end{split} 
\end{equation}
In the second equality of \eqref{eqn:kernelsecondsplit}, we utilized the fact that $-h-z > - \sqrt{\delta^2-r^2}$ for $z \in [-h,-h+\sqrt{\delta^2-r^2}]$ and $-h-z < -\sqrt{\delta^2-r^2}$ for $z \in [-h+\sqrt{\delta^2-r^2},h]$. In the third equality of \eqref{eqn:kernelsecondsplit}, we changed the order of the limits of integration in both integrals in the third term and then combined it with the second term. With the more amenable limits of integration in \eqref{eqn:planestresslambifirstregion} and \eqref{eqn:kernelsecondsplit}, closed-form expressions for $\lambda_i(\xi_1,\xi_2)$ can be deduced. To accomplish this, we utilize the shorthand notation introduced in \eqref{eqn:lambformwithA} and \eqref{eqn:PeriPlaneStrainAList} and drop terms not satisfying monoclinic symmetry to get
\begin{equation}\label{eqn:lambformwithAplanestress}
\begin{split}
\lambda(\bfxi) ={}&  \frac{\omega(\|\bfxi\|)}{\|\bfxi \|^2} \frac{A_0(\xi_1,\xi_2) + A_2(\xi_1,\xi_2)\xi_3^2 + A_4(\xi_1,\xi_2)\xi_3^4 }{\|\bfxi \|^4},
\end{split} 
\end{equation}
where
\begin{subequations}\label{eqn:PeriPlaneStressAList}
\begin{align}
&A_0(\xi_1,\xi_2) = \Lambda_{1111} \xi_1^4 + 4\Lambda_{1112}\xi_1^3 \xi_2 +6 \Lambda_{1122} \xi_1^2 \xi_2^2 + 4 \Lambda_{2212} \xi_1 \xi_2^3 + \Lambda_{2222} \xi_2^4, \\
&A_2(\xi_1,\xi_2) = 6 \left( \Lambda_{1133}\xi_1^2 + \Lambda_{2233}\xi_2^2 + 2 \Lambda_{3312} \xi_1 \xi_2  \right), \\
&A_4(\xi_1,\xi_2) = \Lambda_{3333},
\end{align}
\end{subequations}
and the $\Lambda_{ijkl}$ are to be determined later.

Depending on the choice of influence function $\omega$, it may be possible to produce closed-form expressions for the integrals in \eqref{eqn:lambdaiplanestressminmaxdef}. Two commonly utilized influence functions in peridynamics are $\omega(\| \bfxi\|) = \frac{1}{\| \bfxi \|}$ and $\omega(\| \bfxi \|) = 1$. With either choice of influence function, the micromodulus functions \eqref{eqn:lambdaiplanestressminmaxdef} are given by

\begin{subequations}\label{eqn:PlaneStressLambda_i}
\begin{align}
\begin{split}
\lambda_0(\xi_1,\xi_2) ={}& \omega(r)  \left[ \frac{A_0(\xi_1,\xi_2)}{ r^4 } M_0\left( \frac{h}{r} , \frac{r}{\delta} \right)   +  \frac{A_2(\xi_1,\xi_2)}{r^2 } M_1 \left( \frac{h}{r}, \frac{r}{\delta} \right) + A_4(\xi_1,\xi_2)  M_2\left( \frac{h}{r}, \frac{r}{\delta} \right)  \right], \label{eqn:PlanestressLambda_0}
\end{split} \\
\begin{split}
\lambda_2(\xi_2,\xi_2) ={}& r^2 \omega(r) \left[  \frac{A_0(\xi_1,\xi_2)}{r^4  }  M_1\left( \frac{h}{r}, \frac{r}{\delta} \right)  + \frac{A_2(\xi_1,\xi_2)}{r^2}  M_2\left( \frac{h}{r}, \frac{r}{\delta} \right) + A_4(\xi_1,\xi_2) M_3\left( \frac{h}{r}, \frac{r}{\delta} \right) \right], \label{eqn:PlaneStressLambda_2}
\end{split}
\end{align}
\end{subequations}
where for $\omega(\| \bfxi \|) = 1$ we have\footnote{The regions $\frac{h^2}{\delta^2-4h^2} < x^2$ and $\frac{h^2}{\delta^2} \leqslant x^2 \leqslant \frac{h^2}{\delta^2-4h^2}$ reduce to the regions $r^2 < \delta^2 - 4h^2$ and $\delta^2 - 4h^2 \leqslant r^2 \leqslant \delta^2$, respectively, when $x = \frac{h}{r}$. These are precisely the regions we considered above in order to remove the piecewise function limits in \eqref{eqn:lambdaiplanestressminmaxdef}.}
\begin{equation*}
\begin{split}
    &M_0(x,y) := \left\{ \begin{array}{ll} \frac{3x}{2}  \arctan(2x) + \frac{x^2}{4x^2+1}, &  \frac{h^2}{\delta^2-4h^2} < x^2 \\[0.1in]
    \frac{3x}{2} \arctan\left( \sqrt{y^{-2}-1} \right) + \frac{1}{2} \left( y^4 - 1 \right) + \frac{x}{2} \left( 3 y + 2 y^3 \right) \sqrt{1-y^2}, & \frac{h^2}{\delta^2} \leqslant x^2 \leqslant \frac{h^2}{\delta^2-4h^2}
    \end{array} \right. \\
    %%%%%%%%%%%%%%%%%%%%%%%%%%%%%%%%%%%%%
    %%%%%%%%%%%%%%%%%%%%%%%%%%%%%%%%%%%%%
    &M_1(x,y) :=  \left\{
    \begin{array}{ll}
    \frac{x}{2}  \arctan(2x) - \frac{x^2}{4x^2+1}, & \frac{h^2}{\delta^2-4h^2} < x^2 \\[0.1in]
    \frac{x}{2} \arctan\left( \sqrt{y^{-2}-1} \right) - \frac{1}{2} \left(y^2-1 \right)^2+\frac{x}{2} \left( y - 2y^3 \right) \sqrt{1-y^2}, & \frac{h^2}{\delta^2} \leqslant x^2 \leqslant \frac{h^2}{\delta^2-4h^2}
\end{array}
\right. \\
%%%%%%%%%%%%%%%%%%%%%%%%%%%%%%%%%%%%%%%%
%%%%%%%%%%%%%%%%%%%%%%%%%%%%%%%%%%%%%%%%
    &M_2(x,y) := \left\{ \begin{array}{ll} \frac{3x}{2}  \arctan(2x) + \frac{x^2}{4x^2+1}+\ln \left( \frac{1}{4x^2+1} \right), & \frac{h^2}{\delta^2-4h^2} < x^2 \\[0.1in]
\frac{3x}{2} \arctan \left( \sqrt{y^{-2}-1} \right) + \frac{1}{2}(y^2-1)(y^2-3) - \frac{x}{2} \sqrt{1-y^2} (5y-2y^3) + 2\ln(y), & \frac{h^2}{\delta^2} \leqslant x^2 \leqslant \frac{h^2}{\delta^2-4h^2}
    \end{array}
    \right. \\
    %%%%%%%%%%%%%%%%%%%%%%%%%%%%%%%%%%%%%%%%%%
    %%%%%%%%%%%%%%%%%%%%%%%%%%%%%%%%%%%%%%%%%%
    &M_3(x,y) := \left\{ \begin{array}{ll} -\frac{15x}{2} \arctan(2x) + \frac{16x^4+3x^2}{4x^2+1}-3 \ln \left( \frac{1}{4x^2+1} \right), & \frac{h^2}{\delta^2-4h^2} < x^2 \\[0.1in]
    \begin{array}{l} 
    -\frac{15x}{2} \arctan\left(\sqrt{y^{-2}-1} \right)-\frac{1}{2y^2}\left(y^2-1 \right) \left(y^4-5y^2-2 \right) \\
    + \frac{x}{2} \sqrt{1-y^2} \left( 8y^{-1} + 9y -2y^3 \right) -6 \ln(y),
    \end{array} & \frac{h^2}{\delta^2} \leqslant x^2 \leqslant \frac{h^2}{\delta^2-4h^2}
    \end{array}
    \right.
\end{split}
\end{equation*}
and for $\omega(\| \bfxi \|) = \frac{1}{\| \bfxi \|}$ we have
\begin{equation*}
\begin{split}
    &M_0(x,y) := \left\{ \begin{array}{ll} -\frac{2}{5} + \frac{2}{15} \left( 128x^4+48x^2+3\right) \left( 4x^2+1 \right)^{-\frac{3}{2}}, & \frac{h^2}{\delta^2-4h^2} < x^2 \\[0.1in]
    \frac{4x}{15} \left( 3 y^4+4 y^2 + 8 \right) \sqrt{1-y^2} + \frac{6}{15} \left(y^5 -1 \right), & \frac{h^2}{\delta^2} \leqslant x^2 \leqslant \frac{h^2}{\delta^2-4h^2}
    \end{array}
    \right. \\
    %%%%%%%%%%%%%%%%%%%%%%%%%%%%%%
    %%%%%%%%%%%%%%%%%%%%%%%%%%%%%%
    &M_1(x,y) := \left\{ \begin{array}{ll} -\frac{4}{15} + \frac{4}{15} \left(16x^4+6x^2+1 \right) \left( 4x^2 + 1 \right)^{-\frac{3}{2}}, & \frac{h^2}{\delta^2-4h^2} < x^2 \\[0.1in]
    \frac{4x}{15}\left(2+y^2-3y^4 \right) \sqrt{1-y^2} - \frac{1}{15} \left( 6y^5-10y^3+4 \right), & \frac{h^2}{\delta^2} \leqslant x^2 \leqslant \frac{h^2}{\delta^2-4h^2}
    \end{array}
    \right. \\
    %%%%%%%%%%%%%%%%%%%%%%%%%%%%%%
    %%%%%%%%%%%%%%%%%%%%%%%%%%%%%%	
    &M_2(x,y) := \left\{ \begin{array}{ll} -\frac{16}{15} + \frac{16}{15} \left( 6x^4+6x^2+1 \right) \left( 4x^2+1 \right)^{-\frac{3}{2}}, & \frac{h^2}{\delta^2-4h^2} < x^2 \\[0.1in]
    \frac{4x}{5} \left(1-y^2 \right)^{\frac{5}{2}} + \frac{2}{15} \left( 3 y^2+9y+8\right) \left(y-1 \right)^3,& \frac{h^2}{\delta^2} \leqslant x^2 \leqslant \frac{h^2}{\delta^2-4h^2}
    \end{array}
    \right. \\
    %%%%%%%%%%%%%%%%%%%%%%%%%%%%%%
    %%%%%%%%%%%%%%%%%%%%%%%%%%%%%%    
    &M_3(x,y) := \left\{ \begin{array}{ll} \frac{32}{5} - \frac{8}{15} \left( 152x^4+87x^2+12 \right) \left( 4x^2+1 \right)^{-\frac{3}{2}} +4x \, \text{arsinh}(2x), & \frac{h^2}{\delta^2-4h^2} < x^2  \\[0.1in]
    -\frac{4x}{15} \left(3y^4-11y^2+23 \right) \sqrt{1-y^2} - \frac{6}{15} \left( 5y^{-1} + 4 + y \right) \left( y-1 \right)^4 + 4x \, \text{arsinh}\left( \sqrt{y^{-2}-1} \right),& \frac{h^2}{\delta^2} \leqslant x^2 \leqslant \frac{h^2}{\delta^2-4h^2}
    \end{array}
    \right. .
\end{split}
\end{equation*}

For plane strain, we were able to determine a relationship between $\mathbbm{\Lambda}$ and $\mathbb{C}$ independent of the influence function ({\em cf}. \eqref{eqn:SijklToCijklRelations3D}). While this may be possible for peridynamic plane stress, the expressions are rather impractical and thus it is far more convenient to provide expressions for our specific influence functions. We inform the peridynamic tensor $\mathbbm{\Lambda}$, and consequently \eqref{eqn:PlaneStressLambda_i} (through \eqref{eqn:PeriPlaneStressAList}), with the elasticity tensor $\mathbb{C}$. This is accomplished by employing \eqref{eqn:CijklRelationAverage} to find
\begin{equation}\label{eqn:CtoLambRelationPlaneStress}
\left[
\begin{array}{c}
C_{1111} \\
C_{1122} \\
C_{1133} \\
C_{2222} \\
C_{2233} \\
C_{3333} \\
C_{1112} \\
C_{2212} \\
C_{3312} 
\end{array}
\right]
=
\frac{\pi}{64}
\left[
\begin{array}{ccccccccc}
35 \alpha_1  & 5 \alpha_1   & 40 \alpha_2  & 3 \alpha_1& 8 \alpha_2   & 48 \alpha_3 & 0 & 0  & 0 \\
5 \alpha_1    & 3 \alpha_1   & 8 \alpha_2    & 5 \alpha_1& 8 \alpha_2   & 16 \alpha_3  & 0 & 0 & 0 \\ 
40 \alpha_2 & 8 \alpha_2   & 48 \alpha_3 & 8 \alpha_2& 16 \alpha_3  & 64 \alpha_4 & 0 & 0 & 0 \\
3 \alpha_1    & 5\alpha_1    & 8 \alpha_2    & 35 \alpha_1& 40 \alpha_2 & 48 \alpha_3 & 0 & 0 & 0 \\
8 \alpha_2   & 8 \alpha_2   & 16\alpha_3   & 40 \alpha_2 & 48 \alpha_3 & 64 \alpha_4 & 0 & 0 & 0\\
48 \alpha_3 & 16 \alpha_3 & 64 \alpha_4 & 48 \alpha_3 & 64 \alpha_4 & 128 \alpha_5 & 0 & 0 & 0\\
0                  & 0                   & 0                  & 0                   & 0                  & 0 & 5 \alpha_1 & 3 \alpha_1 & 8 \alpha_2 \\
0                  & 0                   & 0                  & 0                   & 0                  & 0 & 3 \alpha_1 & 5 \alpha_1  & 8 \alpha_2 \\
0                  & 0                   & 0                  & 0                   & 0                  & 0 & 8 \alpha_2 & 8 \alpha_2 & 16 \alpha_3
\end{array}
\right]
\left[
\begin{array}{c}
\Lambda_{1111} \\
6\Lambda_{1122} \\
6\Lambda_{1133} \\
\Lambda_{2222} \\
6\Lambda_{2233} \\
\Lambda_{3333} \\
4\Lambda_{1112} \\
4\Lambda_{2212} \\
12\Lambda_{3312}
\end{array}
\right],
\end{equation}
where for $\omega( \| \bfxi \|) = 1$, (with $p = \frac{h}{\delta}$)
\begin{equation*}
\begin{split}
\alpha_1:={}&  \delta^5 \left(  - \frac{64}{45} p^9 +\frac{32}{7} p^7 +\frac{176}{75} p^5 - \frac{4}{3} p^3 + \frac{1}{4} p - \frac{32}{5} p^5 \ln\left( 2p \right) \right), \\
\alpha_2:={}&  \delta^5 \left(  \frac{64}{45} p^9 - \frac{24}{7} p^7 -\frac{28}{75} p^5 + \frac{1}{3} p^3 + \frac{16}{5} p^5 \ln \left( 2p \right) \right), \\
\alpha_3:={}& \delta^5 \left( - \frac{64}{45} p^9 + \frac{16}{7} p^7 - \frac{92}{225} p^5 -\frac{16}{15} p^5 \ln \left( 2 p \right)  \right), \\
\alpha_4:={}& \delta^5 \left( \frac{64}{45} p^9 - \frac{8}{7} p^7 + \frac{4}{15} p^5 \right), \\
\alpha_5:={}&  \delta^5 \left( -\frac{64}{45} p^9 + \frac{4}{15} p^5 \right),
\end{split}
\end{equation*}
and for $\omega( \| \bfxi \|) = \frac{1}{\| \bfxi \|}$, (with $p = \frac{h}{\delta}$)
\begin{equation*}
\begin{split}
\alpha_1:={}&  \delta^4 \left( -\frac{256}{225} p^9 + \frac{64}{21} p^7 -\frac{32}{5} p^5 + \frac{512}{75} p^4 -\frac{8}{3} p^3 + \frac{1}{3} p \right), \\
\alpha_2:={}&   \delta^4 \left( \frac{256}{225} p^9 - \frac{16}{7}p^7 + \frac{16}{5} p^5 - \frac{64}{25} p^4 + \frac{2}{3} p^3 \right), \\
\alpha_3:={}&   \delta^4 \left( -\frac{256}{225} p^9 + \frac{32}{21} p^7 - \frac{16}{15} p^5 + \frac{32}{75} p^4 \right), \\
\alpha_4:={}&  \delta^4 \left( \frac{256}{225} p^9 - \frac{16}{21} p^7 + \frac{8}{75} p^4 \right), \\
\alpha_5:={}&   \delta^4 \left( -\frac{256}{225} p^9 + \frac{4}{25} p^4 \right).
\end{split}
\end{equation*}

\begin{remark}\label{rem:angulardpendenceplanestress}
Similarly to plane strain ({\em cf.} Remark \ref{rem:angulardpendenceplanestrain}), for $i \in \left\{0,\ldots,4 \right\}$, note that $\frac{A_i(\xi_1,\xi_2)}{r^{4-i}}$ is radially independent ({\em cf.}~\eqref{eqn:PeriPlaneStressAList}) and, consequently, only contributes to the angular portion of \eqref{eqn:PlaneStressLambda_i}. Moreover, as the other terms in \eqref{eqn:PlaneStressLambda_i} are radial functions, the $\frac{A_i(\xi_1,\xi_2)}{r^{4-i}}$ terms make up the entirety of the angular dependence of the micromodulus functions.
\end{remark}

Next, we take a closer look at the micromodulus functions $\left\{ \lambda_i(\xi_1,\xi_2) \right\}$ ({\em cf.}~\eqref{eqn:lambdaiplanestressminmaxdef}) of the peridynamic plane stress model \eqref{eqn:PlaneStressFinal} for various symmetry classes. Following Remark \ref{rem:angulardpendenceplanestress}, the choice of symmetry class only has an effect on the $\left\{ A_i(\xi_1,\xi_2) \right\}$, while the general form of the micromodulus functions~\eqref{eqn:PlaneStressLambda_i} remains unchanged. Therefore, we only present the $\left\{ A_i(\xi_1,\xi_2) \right\}$ for each symmetry class. %appearing in the peridynamic plane strain model \eqref{eqn:planestrainmodelinplane} and \eqref{eqn:planestrainmodeloutofplane}. 
As explained earlier, for many of the symmetry classes, there are multiple planes of reflection symmetry to choose from in order to satisfy Assumption~(P$\sigma$\ref{assump:PSs7Peri}). However, here we only consider a specific example from each symmetry class. More specifically, we choose the orientations corresponding to the elasticity tensors presented in Section \ref{sec:threedimclassicalelasticity}. We do not consider triclinic symmetry as this case is excluded by Assumption~(P$\sigma$\ref{assump:PSs7Peri}). 

\underline{Monoclinic}: We substitute \eqref{eqn:monoelastrelations}, with Cauchy's relations imposed, into \eqref{eqn:CtoLambRelationPlaneStress} to find ({\em cf.}~\eqref{eqn:PeriPlaneStressAList}):
\begin{subequations}\label{eqn:AtermsMonoStress}
\begin{align}
&A_0(\xi_1,\xi_2) := \Lambda_{1111} \xi_1^4 + 4\Lambda_{1112}\xi_1^3 \xi_2 +6 \Lambda_{1122} \xi_1^2 \xi_2^2 + 4 \Lambda_{2212} \xi_1 \xi_2^3 + \Lambda_{2222} \xi_2^4, \\
&A_2(\xi_1,\xi_2) := 6 \left( \Lambda_{1133}\xi_1^2 + \Lambda_{2233}\xi_2^2 + 2 \Lambda_{3312} \xi_1 \xi_2  \right), \\
&A_4(\xi_1,\xi_2) := \Lambda_{3333},
\end{align}
\end{subequations}
where
\begin{equation}\label{eqn:CtoLambRelationPlaneStressMono}
\left[
\begin{array}{c}
\Lambda_{1111} \\
6\Lambda_{1122} \\
6\Lambda_{1133} \\
\Lambda_{2222} \\
6\Lambda_{2233} \\
\Lambda_{3333} \\
4\Lambda_{1112} \\
4\Lambda_{2212} \\
12\Lambda_{3312}
\end{array}
\right]
=
\frac{64}{\pi}
\left[
\begin{array}{ccccccccc}
35 \alpha_1  & 5 \alpha_1   & 40 \alpha_2  & 3 \alpha_1& 8 \alpha_2   & 48 \alpha_3 & 0 & 0  & 0 \\
5 \alpha_1    & 3 \alpha_1   & 8 \alpha_2    & 5 \alpha_1& 8 \alpha_2   & 16 \alpha_3  & 0 & 0 & 0 \\ 
40 \alpha_2 & 8 \alpha_2   & 48 \alpha_3 & 8 \alpha_2& 16 \alpha_3  & 64 \alpha_4 & 0 & 0 & 0 \\
3 \alpha_1    & 5\alpha_1    & 8 \alpha_2    & 35 \alpha_1& 40 \alpha_2 & 48 \alpha_3 & 0 & 0 & 0 \\
8 \alpha_2   & 8 \alpha_2   & 16\alpha_3   & 40 \alpha_2 & 48 \alpha_3 & 64 \alpha_4 & 0 & 0 & 0\\
48 \alpha_3 & 16 \alpha_3 & 64 \alpha_4 & 48 \alpha_3 & 64 \alpha_4 & 128 \alpha_5 & 0 & 0 & 0\\
0                  & 0                   & 0                  & 0                   & 0                  & 0 & 5 \alpha_1 & 3 \alpha_1 & 8 \alpha_2 \\
0                  & 0                   & 0                  & 0                   & 0                  & 0 & 3 \alpha_1 & 5 \alpha_1  & 8 \alpha_2 \\
0                  & 0                   & 0                  & 0                   & 0                  & 0 & 8 \alpha_2 & 8 \alpha_2 & 16 \alpha_3
\end{array}
\right]^{-1}
\left[
\begin{array}{c}
C_{1111} \\
C_{1122} \\
C_{1133} \\
C_{2222} \\
C_{2233} \\
C_{3333} \\
C_{1112} \\
C_{2212} \\
C_{3312} 
\end{array}
\right].
\end{equation}

\underline{Orthotropic}: We substitute \eqref{eqn:orthoelastrelations}, with Cauchy's relations imposed, into \eqref{eqn:CtoLambRelationPlaneStress}  to find ({\em cf.}~\eqref{eqn:PeriPlaneStressAList})

\begin{subequations}\label{eqn:AtermsOrthoStress}
\begin{align}
&A_0(\xi_1,\xi_2) = \Lambda_{1111} \xi_1^4 +6 \Lambda_{1122} \xi_1^2 \xi_2^2 + \Lambda_{2222} \xi_2^4, \\
&A_2(\xi_1,\xi_2) = 6 \left( \Lambda_{1133}\xi_1^2 + \Lambda_{2233}\xi_2^2  \right), \\
&A_4(\xi_1,\xi_2) = \Lambda_{3333},
\end{align}
\end{subequations}
where
\begin{equation}\label{eqn:CtoLambRelationPlaneStressOrtho}
\left[
\begin{array}{c}
\Lambda_{1111} \\
6\Lambda_{1122} \\
6\Lambda_{1133} \\
\Lambda_{2222} \\
6\Lambda_{2233} \\
\Lambda_{3333} \\
\end{array}
\right]
=
\frac{64}{\pi}
\left[
\begin{array}{ccccccccc}
35 \alpha_1  & 5 \alpha_1   & 40 \alpha_2  & 3 \alpha_1& 8 \alpha_2   & 48 \alpha_3 \\
5 \alpha_1    & 3 \alpha_1   & 8 \alpha_2    & 5 \alpha_1& 8 \alpha_2   & 16 \alpha_3 \\ 
40 \alpha_2 & 8 \alpha_2   & 48 \alpha_3 & 8 \alpha_2& 16 \alpha_3  & 64 \alpha_4 \\
3 \alpha_1    & 5\alpha_1    & 8 \alpha_2    & 35 \alpha_1& 40 \alpha_2 & 48 \alpha_3 \\
8 \alpha_2   & 8 \alpha_2   & 16\alpha_3   & 40 \alpha_2 & 48 \alpha_3 & 64 \alpha_4 \\
48 \alpha_3 & 16 \alpha_3 & 64 \alpha_4 & 48 \alpha_3 & 64 \alpha_4 & 128 \alpha_5 \\
\end{array}
\right]^{-1}
\left[
\begin{array}{c}
C_{1111} \\
C_{1122} \\
C_{1133} \\
C_{2222} \\
C_{2233} \\
C_{3333} \\
\end{array}
\right].
\end{equation}

\underline{Trigonal}: We substitute \eqref{eqn:trigonalelastrelations}, with Cauchy's relations imposed, into \eqref{eqn:CtoLambRelationPlaneStress}  to find ({\em cf.}~\eqref{eqn:PeriPlaneStressAList})
\begin{subequations}\label{eqn:AtermsTrigonalStress}
\begin{align}
&A_0(\xi_1,\xi_2) := \Lambda_{1111} \xi_1^4 + 4\Lambda_{1112}\xi_1^3 \xi_2 +6 \Lambda_{1122} \xi_1^2 \xi_2^2 + 4 \Lambda_{2212} \xi_1 \xi_2^3 + \Lambda_{2222} \xi_2^4, \\
&A_2(\xi_1,\xi_2) := 6 \left( \Lambda_{1133}\xi_1^2 + \Lambda_{2233}\xi_2^2 + 2 \Lambda_{3312} \xi_1 \xi_2  \right), \\
&A_4(\xi_1,\xi_2) := \Lambda_{3333},
\end{align}
\end{subequations}
where \footnote{While there appear to be nine independent constants in \eqref{eqn:AtermsTrigonalStress}, the constants are not actually independent. Since the trigonal elasticity tensor $\mathbb{C}$ with Cauchy's relations imposed has only four independent constants, there are actually only four independent constants in \eqref{eqn:AtermsTrigonalStress}.}
\begin{equation}\label{eqn:CtoLambRelationPlaneStressTrigonal}
\left[
\begin{array}{c}
\Lambda_{1111} \\
6\Lambda_{1122} \\
6\Lambda_{1133} \\
\Lambda_{2222} \\
6\Lambda_{2233} \\
\Lambda_{3333} \\
4\Lambda_{1112} \\
4\Lambda_{2212} \\
12\Lambda_{3312}
\end{array}
\right]
=
\frac{64}{\pi}
\left[
\begin{array}{ccccccccc}
35 \alpha_1  & 5 \alpha_1   & 40 \alpha_2  & 3 \alpha_1& 8 \alpha_2   & 48 \alpha_3 & 0 & 0  & 0 \\
5 \alpha_1    & 3 \alpha_1   & 8 \alpha_2    & 5 \alpha_1& 8 \alpha_2   & 16 \alpha_3  & 0 & 0 & 0 \\ 
40 \alpha_2 & 8 \alpha_2   & 48 \alpha_3 & 8 \alpha_2& 16 \alpha_3  & 64 \alpha_4 & 0 & 0 & 0 \\
3 \alpha_1    & 5\alpha_1    & 8 \alpha_2    & 35 \alpha_1& 40 \alpha_2 & 48 \alpha_3 & 0 & 0 & 0 \\
8 \alpha_2   & 8 \alpha_2   & 16\alpha_3   & 40 \alpha_2 & 48 \alpha_3 & 64 \alpha_4 & 0 & 0 & 0\\
48 \alpha_3 & 16 \alpha_3 & 64 \alpha_4 & 48 \alpha_3 & 64 \alpha_4 & 128 \alpha_5 & 0 & 0 & 0\\
0                  & 0                   & 0                  & 0                   & 0                  & 0 & 5 \alpha_1 & 3 \alpha_1 & 8 \alpha_2 \\
0                  & 0                   & 0                  & 0                   & 0                  & 0 & 3 \alpha_1 & 5 \alpha_1  & 8 \alpha_2 \\
0                  & 0                   & 0                  & 0                   & 0                  & 0 & 8 \alpha_2 & 8 \alpha_2 & 16 \alpha_3
\end{array}
\right]^{-1}
\left[
\begin{array}{c}
C_{1111} \\
C_{1122} \\
C_{1122} \\
C_{2222} \\
\frac{1}{3} C_{2222} \\
C_{2222} \\
0 \\
C_{2212} \\
-C_{2212} 
\end{array}
\right].
\end{equation}

\underline{Tetragonal}: We substitute \eqref{eqn:tetraelastrelations}, with Cauchy's relations imposed, into \eqref{eqn:CtoLambRelationPlaneStress}  to find ({\em cf.}~\eqref{eqn:PeriPlaneStressAList})

\begin{subequations}\label{eqn:AtermsTetraStress}
\begin{align}
&A_0(\xi_1,\xi_2) = \Lambda_{1111} (\xi_1^4 + \xi_2^4) + 6 \Lambda_{1122} \xi_1^2 \xi_2^2, \\
&A_2(\xi_1,\xi_2) = 6 \Lambda_{1133}r^2, \\
&A_4(\xi_1,\xi_2) = \Lambda_{3333},
\end{align}
\end{subequations}
where
\begin{equation}\label{eqn:CtoLambRelationPlaneStressTetra}
\left[
\begin{array}{c}
\Lambda_{1111} \\
6\Lambda_{1122} \\
6\Lambda_{1133} \\
\Lambda_{3333} \\
\end{array}
\right]
=
\frac{64}{\pi}
\left[
\begin{array}{cccccccc}
38 \alpha_1  & 5 \alpha_1   & 48 \alpha_2  &   48 \alpha_3 \\
10 \alpha_1    & 3 \alpha_1   & 16 \alpha_2    &   16 \alpha_3 \\ 
48 \alpha_2 & 8 \alpha_2   & 64 \alpha_3 &   64 \alpha_4 \\
96 \alpha_3 & 16 \alpha_3 & 128 \alpha_4 &   128 \alpha_5 \\
\end{array}
\right]^{-1}
\left[
\begin{array}{c}
C_{1111} \\
C_{1122} \\
C_{1133} \\
C_{3333} \\
\end{array}
\right].
\end{equation}

\underline{Transversely Isotropic}: We substitute \eqref{eqn:tisoelastrelations}, with Cauchy's relations imposed, into \eqref{eqn:CtoLambRelationPlaneStress}  to find ({\em cf.}~\eqref{eqn:PeriPlaneStressAList})

\begin{subequations}\label{eqn:AtermsTisoStress}
\begin{align}
&A_0(\xi_1,\xi_2) = \Lambda_{1111} r^4, \\
&A_2(\xi_1,\xi_2) = 6 \Lambda_{1133}r^2, \\
&A_4(\xi_1,\xi_2) = \Lambda_{3333},
\end{align}
\end{subequations}
where
\begin{equation}\label{eqn:CtoLambRelationPlaneStressTiso}
\left[
\begin{array}{c}
\Lambda_{1111} \\
6\Lambda_{1133} \\
\Lambda_{3333} \\
\end{array}
\right]
=
\frac{64}{\pi}
\left[
\begin{array}{cccccccc}
48 \alpha_1   & 48 \alpha_2  &   48 \alpha_3 \\
64 \alpha_2  & 64 \alpha_3 &   64 \alpha_4 \\
128 \alpha_3 & 128 \alpha_4 &   128 \alpha_5 \\
\end{array}
\right]^{-1}
\left[
\begin{array}{c}
C_{1111} \\
C_{1133} \\
C_{3333} \\
\end{array}
\right].
\end{equation}

\underline{Cubic}: We substitute \eqref{eqn:cubicelastrelations}, with Cauchy's relations imposed, into \eqref{eqn:CtoLambRelationPlaneStress}  to find ({\em cf.}~\eqref{eqn:PeriPlaneStressAList})

\begin{subequations}\label{eqn:AtermsCubicStress}
\begin{align}
&A_0(\xi_1,\xi_2) = \Lambda_{1111} (\xi_1^4 + \xi_2^4) + 6 \Lambda_{1122} \xi_1^2 \xi_2^2, \\
&A_2(\xi_1,\xi_2) = 6 \Lambda_{1133}r^2, \\
&A_4(\xi_1,\xi_2) = \Lambda_{3333},
\end{align}
\end{subequations}
where\footnote{While there appear to be four independent constants in \eqref{eqn:AtermsCubicStress}, the constants are not actually independent. Since the cubic elasticity tensor $\mathbb{C}$ with Cauchy's relations imposed has only four independent constants, there are actually only two independent constants in \eqref{eqn:AtermsCubicStress}.}
\begin{equation}\label{eqn:CtoLambRelationPlaneStressCubic}
\left[
\begin{array}{c}
\Lambda_{1111} \\
6\Lambda_{1122} \\
6\Lambda_{1133} \\
\Lambda_{3333} \\
\end{array}
\right]
=
\frac{64}{\pi}
\left[
\begin{array}{cccccccc}
38 \alpha_1  & 5 \alpha_1   & 48 \alpha_2  &   48 \alpha_3 \\
10 \alpha_1    & 3 \alpha_1   & 16 \alpha_2    &   16 \alpha_3 \\ 
48 \alpha_2 & 8 \alpha_2   & 64 \alpha_3 &   64 \alpha_4 \\
96 \alpha_3 & 16 \alpha_3 & 128 \alpha_4 &   128 \alpha_5 \\
\end{array}
\right]^{-1}
\left[
\begin{array}{c}
C_{1111} \\
C_{1122} \\
C_{1122} \\
C_{1111} \\
\end{array}
\right].
\end{equation}

\underline{Isotropic}: We substitute \eqref{eqn:iso3Delastrelations}, with Cauchy's relations imposed, into \eqref{eqn:CtoLambRelationPlaneStress}  to find ({\em cf.}~\eqref{eqn:PeriPlaneStressAList})

\begin{subequations}\label{eqn:AtermsIsoStress}
\begin{align}
&A_0(\xi_1,\xi_2) = \Lambda_{1111} r^4, \\
&A_2(\xi_1,\xi_2) = 6 \Lambda_{1133}r^2, \\
&A_4(\xi_1,\xi_2) = \Lambda_{3333},
\end{align}
\end{subequations}
where\footnote{While there appear to be three independent constants in \eqref{eqn:AtermsIsoStress}, the constants are not actually independent. Since the isotropic elasticity tensor $\mathbb{C}$ with Cauchy's relations imposed has only one independent constant, there is actually only one independent constant in \eqref{eqn:AtermsIsoStress}.}
\begin{equation}\label{eqn:CtoLambRelationPlaneStressIso}
\left[
\begin{array}{c}
\Lambda_{1111} \\
6\Lambda_{1133} \\
\Lambda_{3333} \\
\end{array}
\right]
=
\frac{64}{\pi}
\left[
\begin{array}{cccccccc}
48 \alpha_1   & 48 \alpha_2  &   48 \alpha_3 \\
64 \alpha_2  & 64 \alpha_3 &   64 \alpha_4 \\
128 \alpha_3 & 128 \alpha_4 &   128 \alpha_5 \\
\end{array}
\right]^{-1}
\left[
\begin{array}{c}
C_{1111} \\
\frac{1}{3} C_{1111} \\
C_{1111} \\
\end{array}
\right].
\end{equation}

\begin{remark}
In classical linear elasticity, the plane stress model for each symmetry class reduces identically to a corresponding two-dimensional model ({\em cf}.~Table~\ref{tab:modelequivalenceplanestrainstressto2D}). A similar situation occurs with the peridynamic plane stress model. While the plane stress micromodulus functions $\lambda_0$ and $\lambda_2$, given by~\eqref{eqn:PlaneStressLambda_i}, are not equivalent to the two-dimensional micromodulus function \eqref{eqn:lambdaobl}, they do possess one of the four symmetries of two-dimensional classical linear elasticity ({\em cf}. Theorem \ref{thm:symmetryclasses}) for each symmetry class ({\em cf}.~Table~\ref{tab:micromodulusequivalenceplanestressto2D}). This can be observed by considering the symmetries of $\left\{ A_i(\xi_1,\xi_2) \right\}$ for each symmetry class. Unlike their two-dimensional counterparts, the plane stress micromodulus functions incorporate out-of-plane information. It is also interesting to note that the same correspondence between each three-dimensional symmetry class and the corresponding two-dimensional symmetry class in Table~\ref{tab:modelequivalenceplanestrainstressto2D} occurs for the plane stress micromodulus functions, which is summarized in Table \ref{tab:micromodulusequivalenceplanestressto2D}. However, while different plane stress micromodulus functions may possess the same two-dimensional symmetry, e.g., tetragonal and cubic micromodulus functions both have square symmetry, the resulting plane stress micromodulus functions are unique for each three-dimensional symmetry class. Due to this fact, the resulting plane stress equations of motion are unique for each three-dimensional symmetry class. 
\end{remark}

\begin{table}
\begin{center}
\begin{tabular}{|c|c|}
\hline
\textbf{Pure Two-Dimensional Peridynamic Model} & \textbf{Peridynamic Plane Stress Model} \\
\hline
Oblique & Monoclinic and Trigonal \\
\hline
Rectangular & Orthotropic \\
\hline
Square & Tetragonal and Cubic \\
\hline
Isotropic & Transversely Isotropic and Isotropic \\
\hline
\end{tabular}
\caption{Symmetry equivalence between the pure two-dimensional peridynamic models and the peridynamic plane stress models (provided the micromodulus functions $\lambda_0$ and $\lambda_2$ are informed by three-dimensional elasticity tensors having the form of those in Section~\ref{sec:threedimclassicalelasticity}).}\label{tab:micromodulusequivalenceplanestressto2D}
\end{center}
\end{table}

A common approach for modeling plane stress in the peridynamic literature simply employs a two-dimensional peridynamic model and matches the model constants to those in the corresponding classical plane stress model (see, e.g., \cite{Gerstle2005,Le2014,Sarego2016}). In bond-based peridynamics, this produces a two-dimensional bond-based model with a single micromodulus function. Unlike in the case for peridynamic plane strain ({\em cf}. Section \ref{sec:PeridynamicPlaneStrain}), the peridynamic plane stress model \eqref{eqn:PlaneStressFinal} developed in this work is a state-based model with two micromodulus functions $\lambda_0$ and $\lambda_2$, and thus a direct comparison considering only micromodulus functions is insufficient. Instead, we explore the behavior of the plane stress micromodulus functions \eqref{eqn:lambdaiplanestressminmaxdef} for various symmetry classes.

\textbf{Micromodulus function visualization for peridynamic generalized plane stress}

In this section, we provide visualizations of the behavior of the plane stress micromodulus functions \eqref{eqn:lambdaiplanestressminmaxdef} for various symmetry classes. We observe that in the peridynamic plane stress equation of motion \eqref{eqn:PlaneStressFinal}, the micromodulus function $\lambda_0(\xi_1,\xi_2)$ is multiplied by $\xi_i \xi_j$ and the micromodulus function $\lambda_2(\xi_1,\xi_2)$ is multiplied by $\xi_i$. Thus, a singularity of the form $\frac{1}{r^2}$ for $\lambda_0(\xi_1,\xi_2)$ and $\frac{1}{r}$ for $\lambda_2(\xi_1,\xi_2)$, represent removable singularities of the model. To better visualize the anisotropic behavior of the micromodulus functions, we remove such singularities by plotting $r^2 \lambda_0(\xi_1,\xi_2)$ and $r \lambda_2(\xi_1,\xi_2)$. 

Our visualization of the micromodulus functions covers materials in every symmetry class except triclinic, since the plane stress Assumption (P$\varepsilon$\ref{assump:PSs7Peri}) excludes this class. Since we are dealing with bond-based peridynamic models, we must consider materials satisfying (at least approximately) Cauchy's relations. For the sake of consistency, we consider the same materials as were utilized to illustrate anisotropy in the plane strain micromodulus functions. In particular, we inform the peridynamic plane stress micromodulus functions \eqref{eqn:PlaneStressLambda_i} with the material properties summarized in Table \ref{tab:MaterialProp}.

In Figure \ref{fig:planestresslamb0} we plot $r^2 \lambda_0(\xi_1,\xi_2)$ with $\omega(\| \bfxi \|) = \frac{1}{\| \bfxi \|}$ (in green) and $\omega(\| \bfxi \|) = 1$ (in blue).  Similarly, in Figure \ref{fig:planestresslamb2} we plot $r \lambda_2(\xi_2,\xi_2)$ with $\omega(\| \bfxi \|) = \frac{1}{\| \bfxi \|}$ (in green) and $\omega(\| \bfxi \|) = 1$ (in blue). For each symmetry class, these functions are found by substituting the corresponding elasticity tensor $\tilde{\bfC}$ from Table \ref{tab:MaterialProp} into~\eqref{eqn:PlaneStressLambda_i} (with \eqref{eqn:PeriPlaneStressAList} informed by inverting the system \eqref{eqn:CtoLambRelationPlaneStress}). For all of the plots, we took $3h = \delta = 1$. 

\begin{figure}
\begin{tabular}{cc}
\begin{tabular}{cc}
\includegraphics[scale=0.45,trim={3cm 0 3cm 0.5cm},clip]{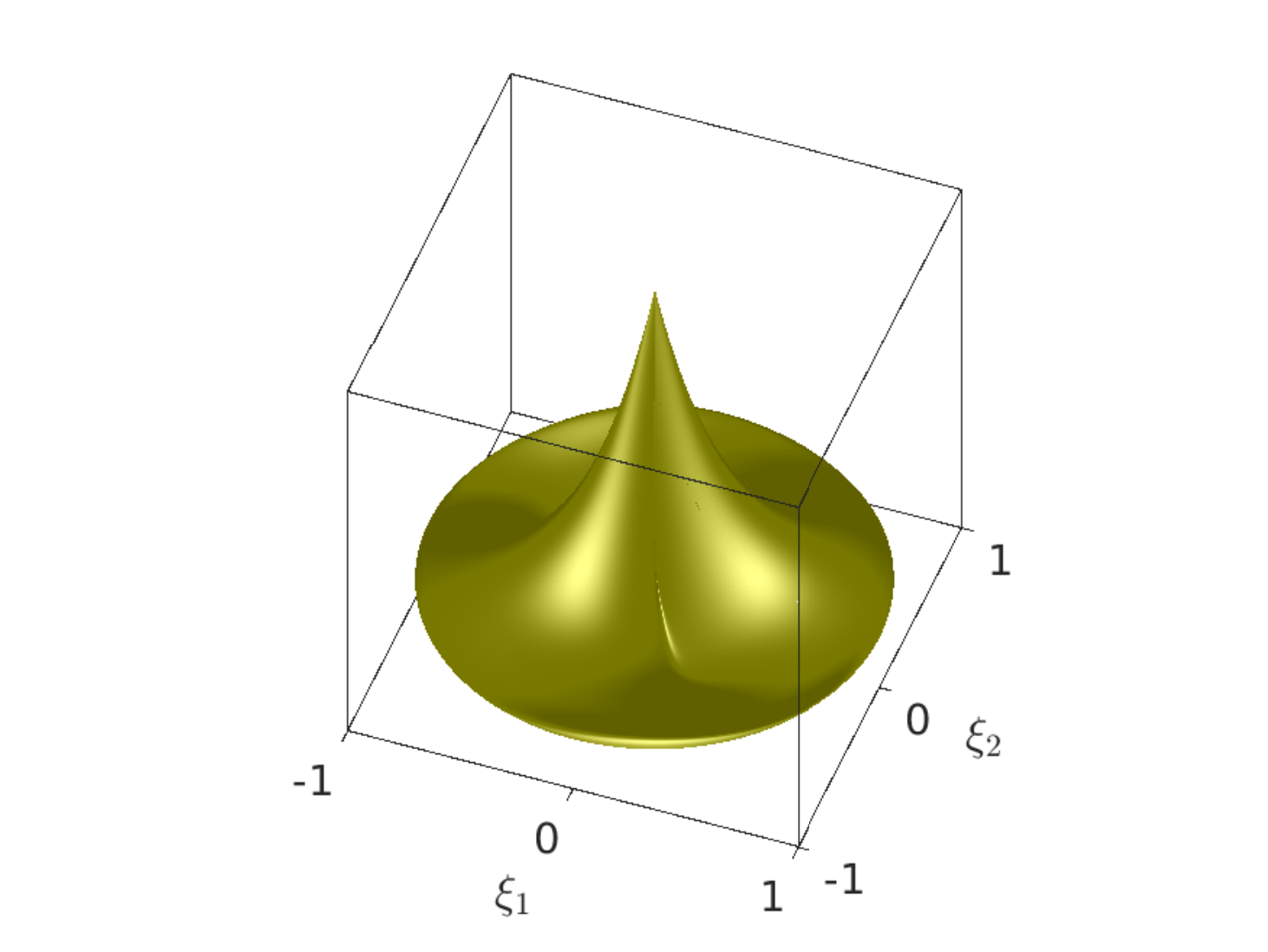} &
\includegraphics[scale=0.45,trim={3cm 0 3cm 0.5cm},clip]{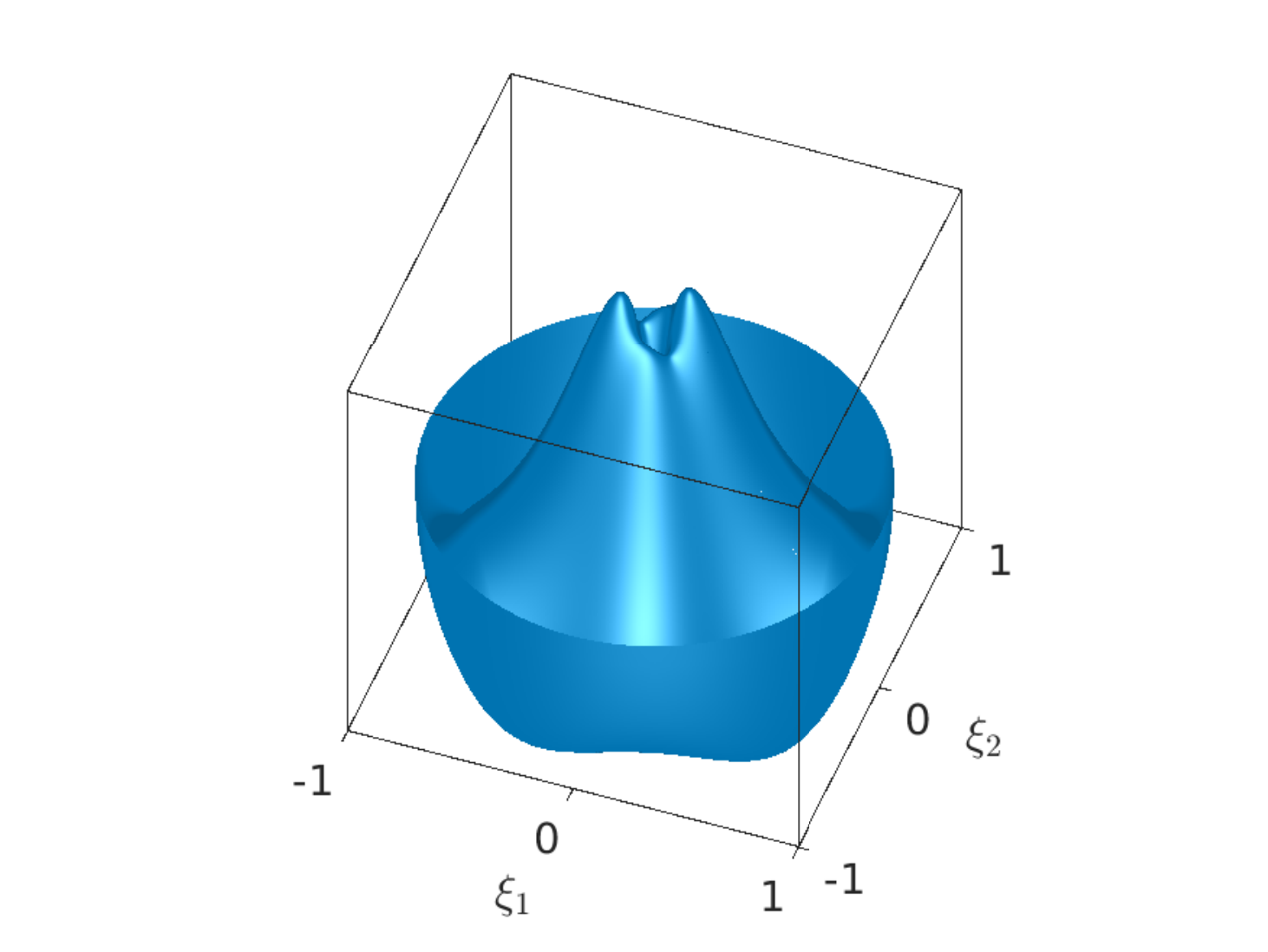}
\end{tabular} & \begin{tabular}{cc}
\includegraphics[scale=0.45,trim={3cm 0 3cm 0.5cm},clip]{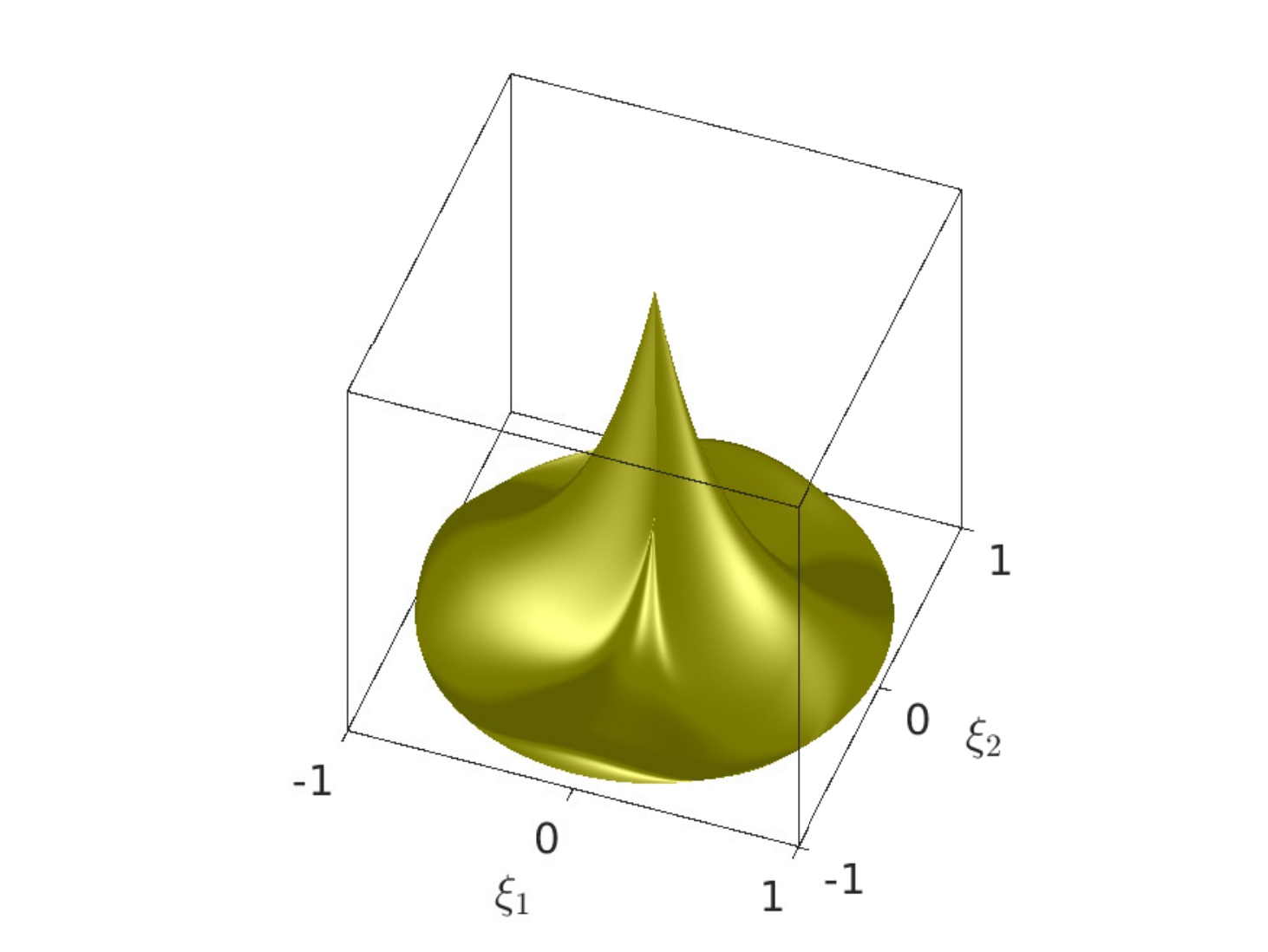} &
\includegraphics[scale=0.45,trim={3cm 0 3cm 0.5cm},clip]{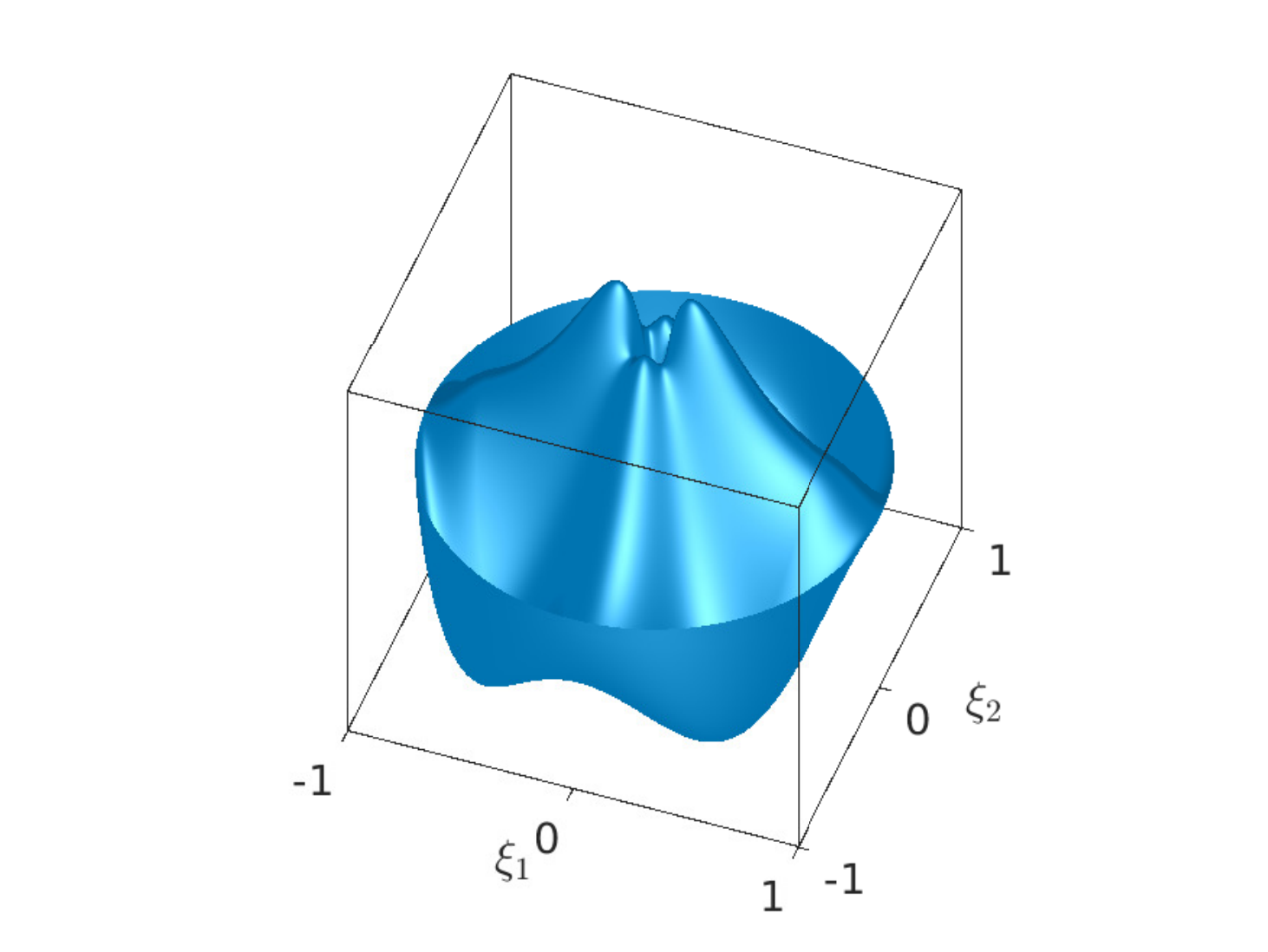}
\end{tabular} \\
Monoclinic ($\text{CoTeO}_4$) & Orthotropic ($\text{Te}_2\text{W}$) \\
%%%%%%%%%%%%%%%%%%%%%%%%%%%%%%%%%%%%%%%%%%%%%%%%%%%%%
\begin{tabular}{cc}
\includegraphics[scale=0.45,trim={3cm 0 3cm 0.5cm},clip]{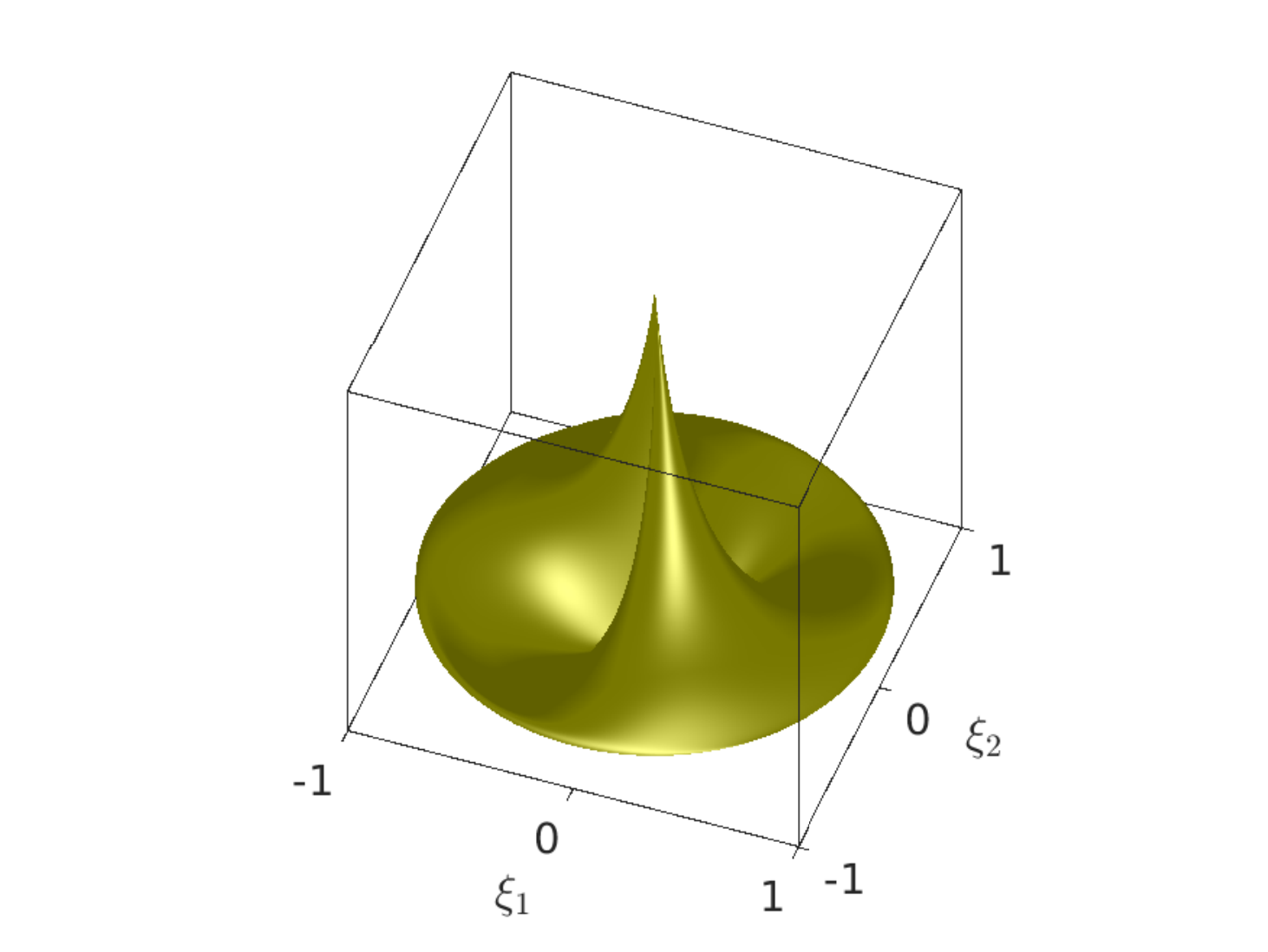} &
\includegraphics[scale=0.45,trim={3cm 0 3cm 0.5cm},clip]{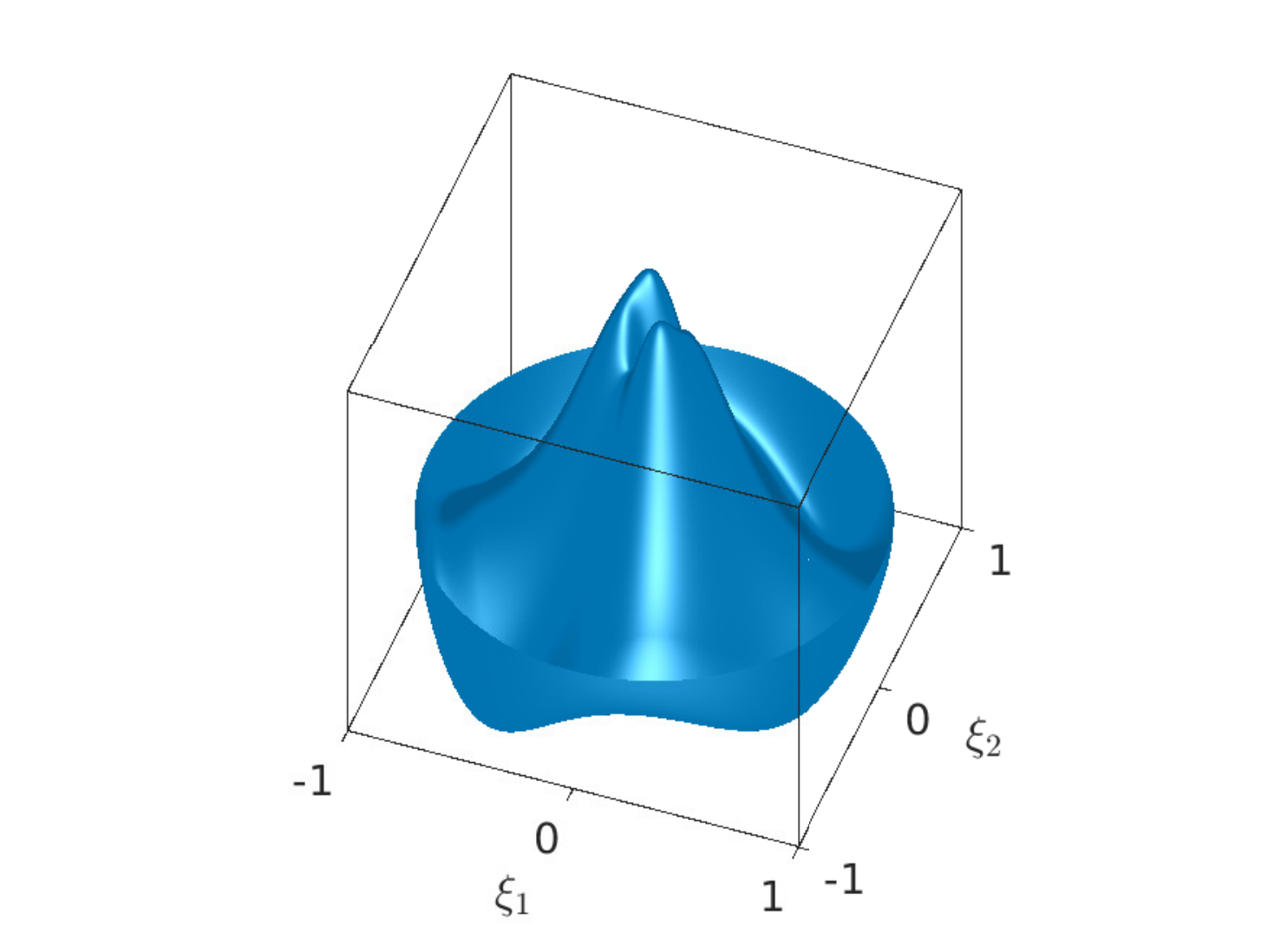}
\end{tabular} & \begin{tabular}{cc}
\includegraphics[scale=0.45,trim={3cm 0 3cm 0.5cm},clip]{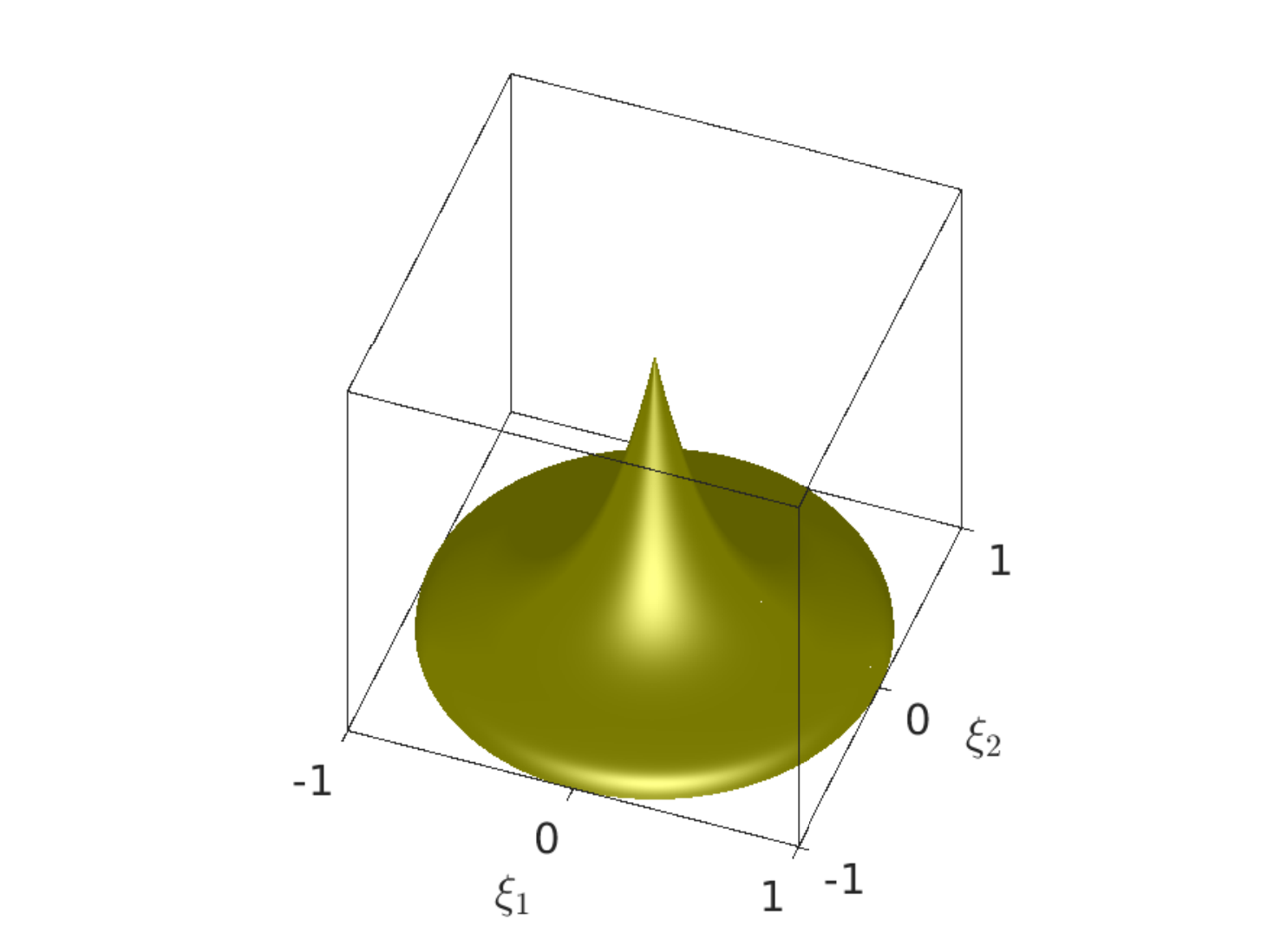} &
\includegraphics[scale=0.45,trim={3cm 0 3cm 0.5cm},clip]{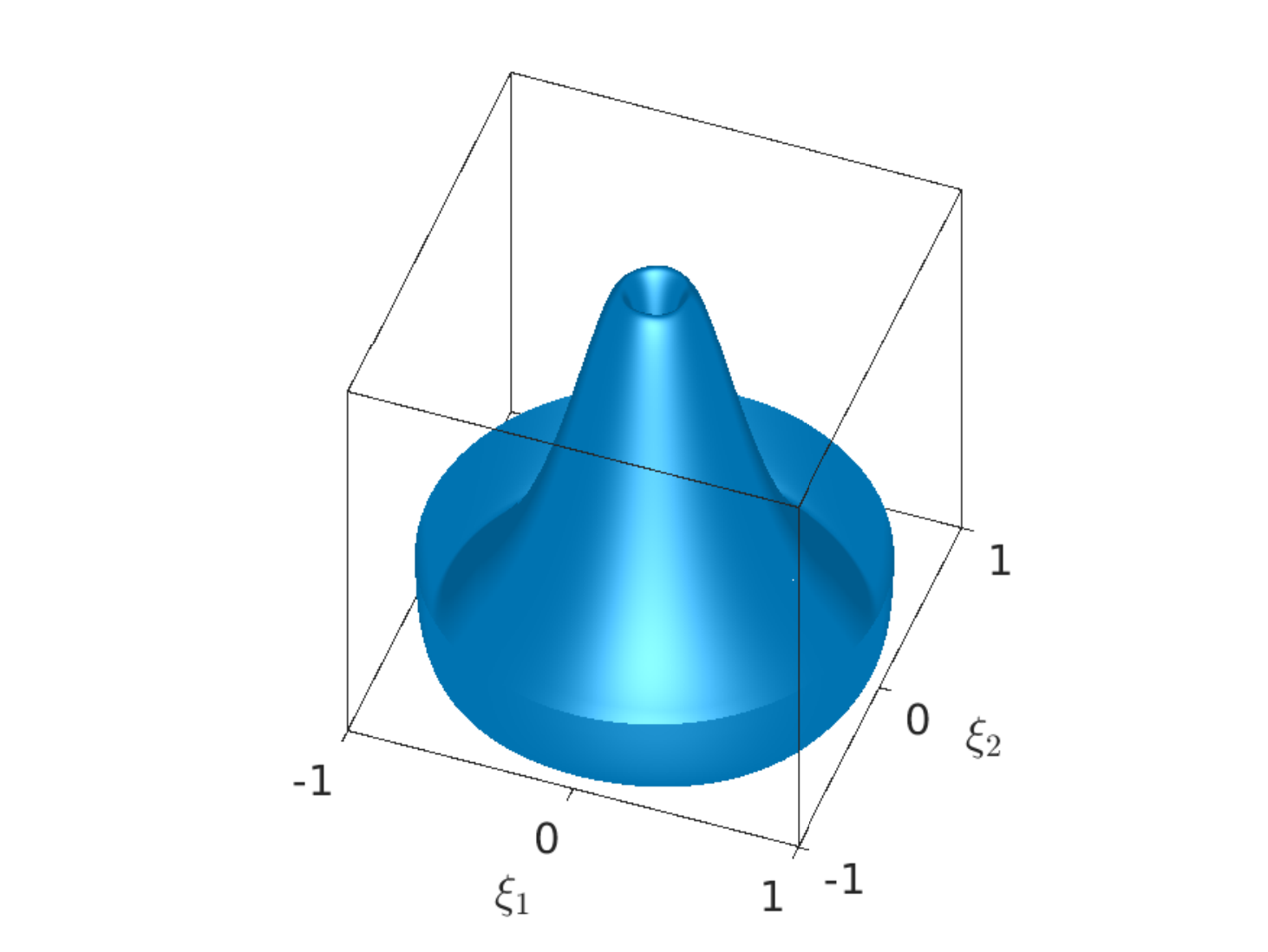}
\end{tabular} \\
Trigonal ($\text{Ta}_2\text{C}$) & Tetragonal (Si) \\
%%%%%%%%%%%%%%%%%%%%%%%%%%%%%%%%%%%%%%%%%%%%%%%%%%%%%
\begin{tabular}{cc}
\includegraphics[scale=0.45,trim={3cm 0 3cm 0.5cm},clip]{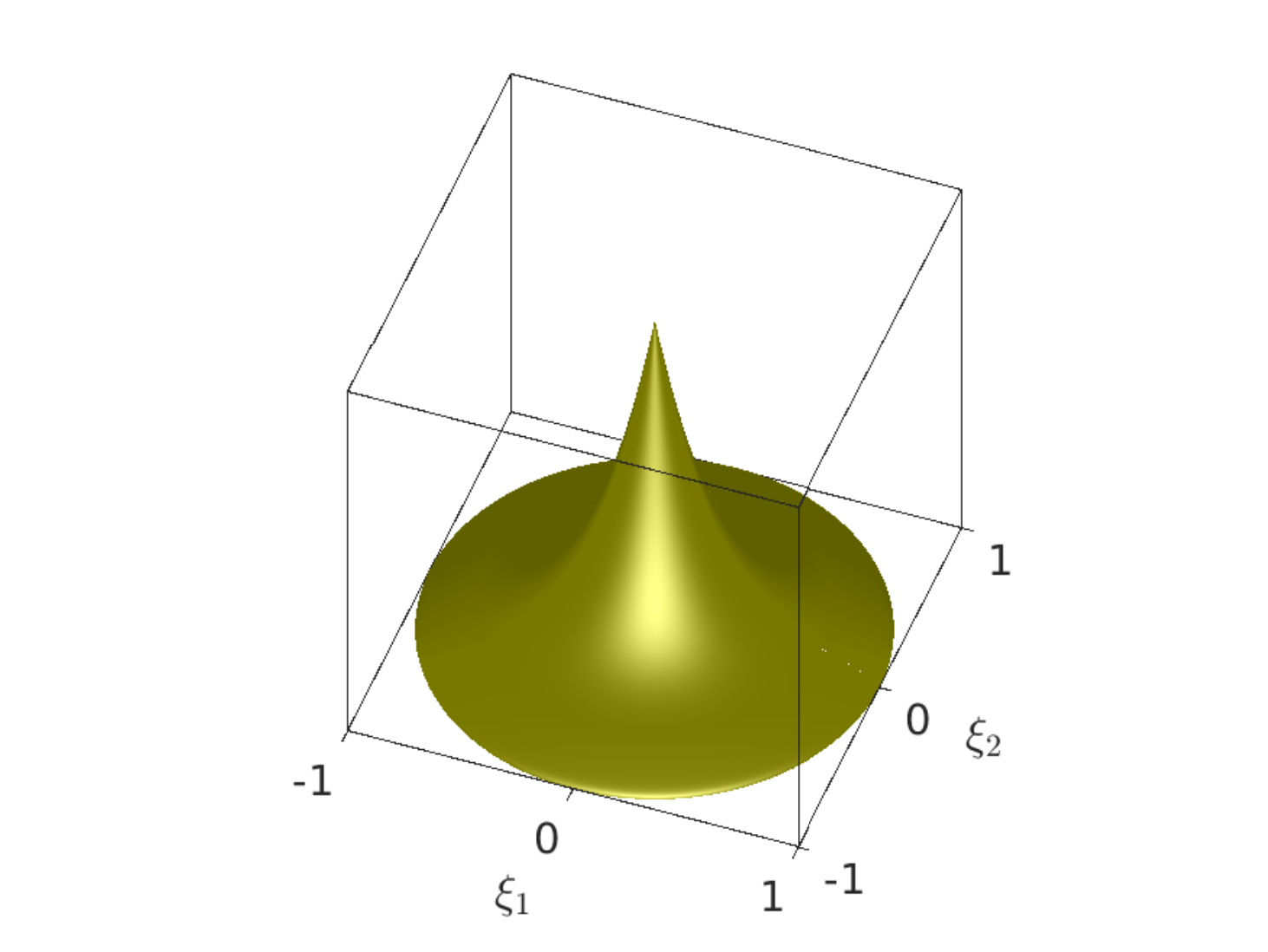} &
\includegraphics[scale=0.45,trim={3cm 0 3cm 0.5cm},clip]{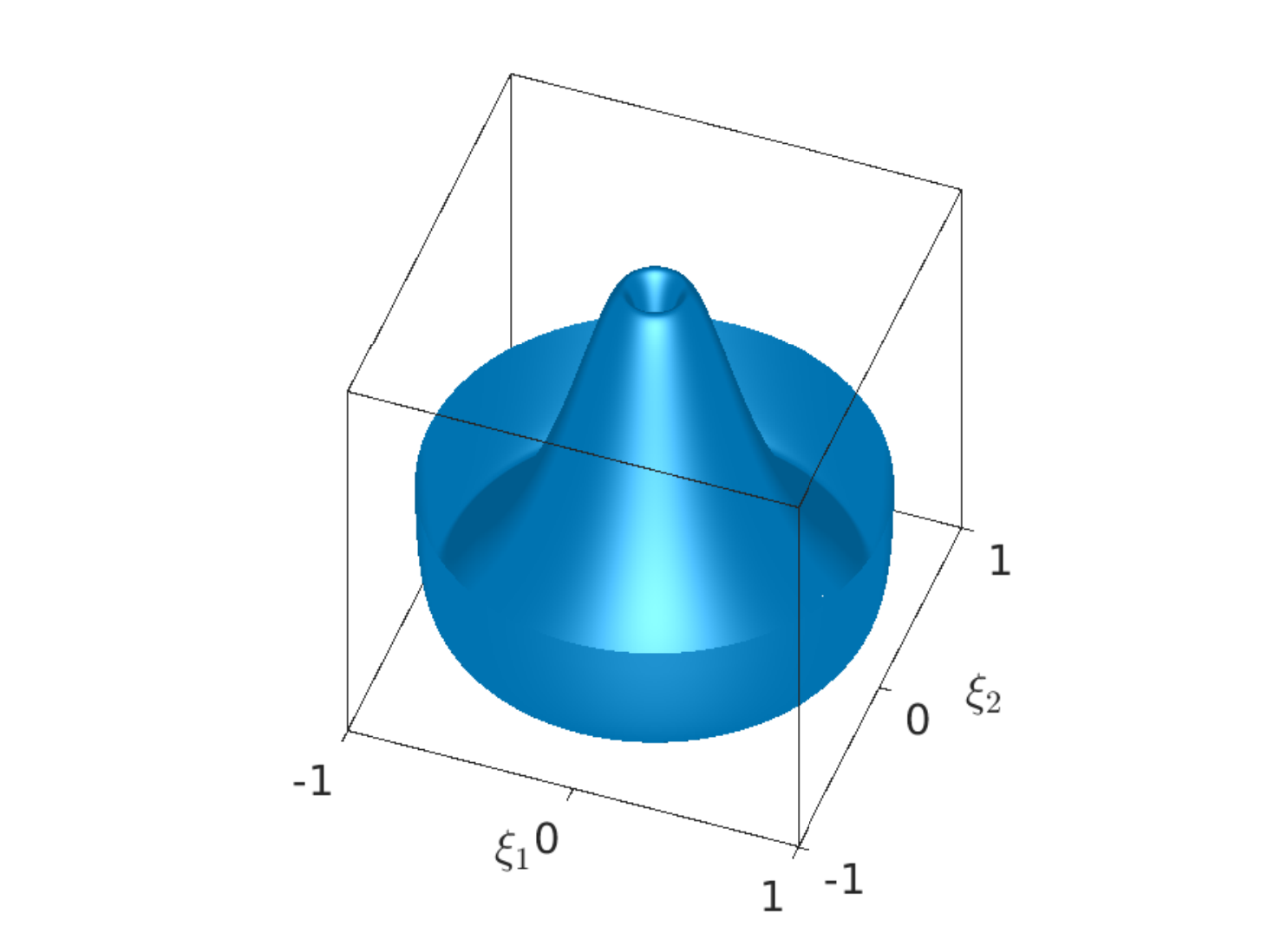}
\end{tabular} & \begin{tabular}{cc}
\includegraphics[scale=0.45,trim={3cm 0 3cm 0.5cm},clip]{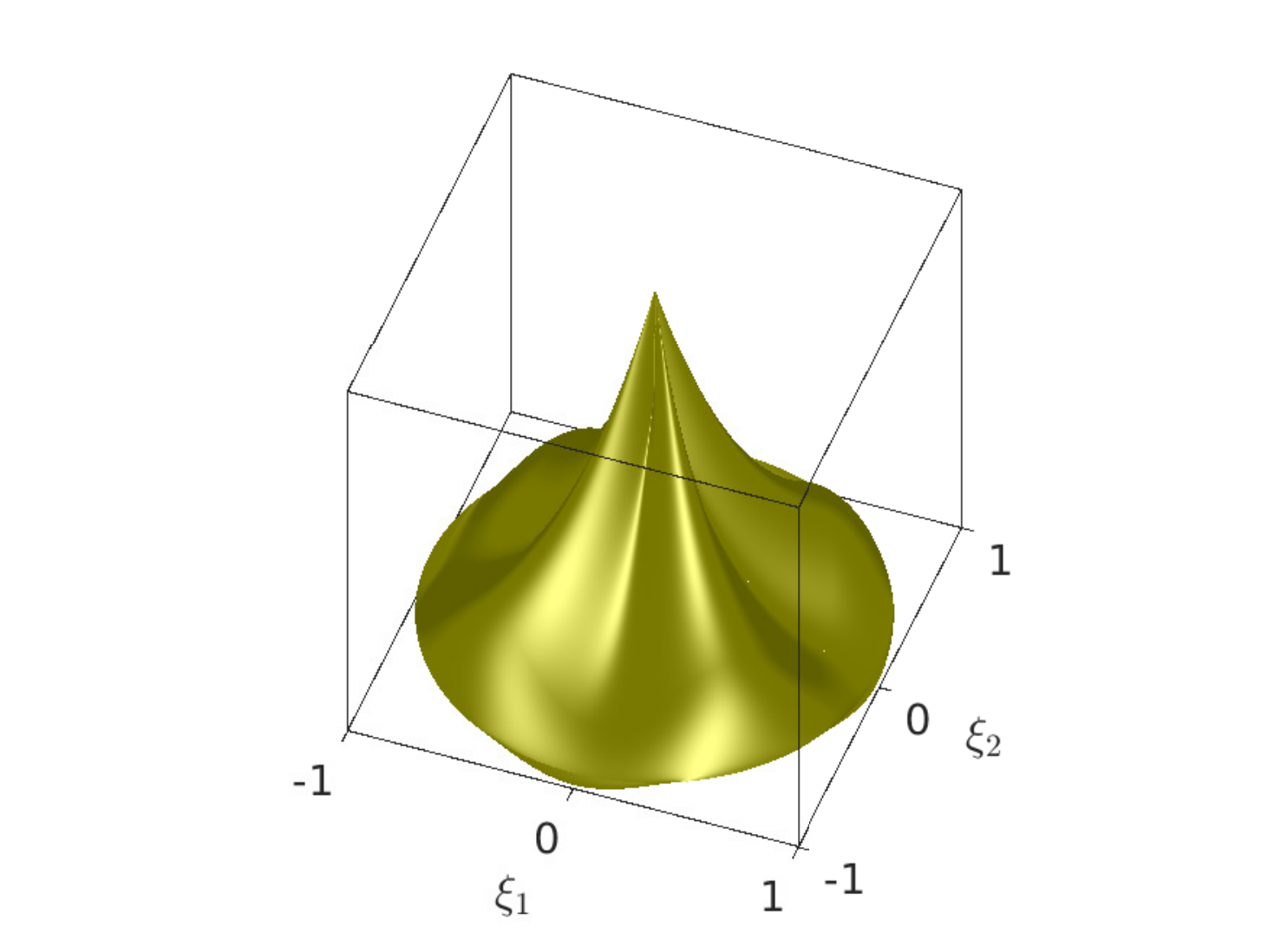} &
\includegraphics[scale=0.45,trim={3cm 0 3cm 0.5cm},clip]{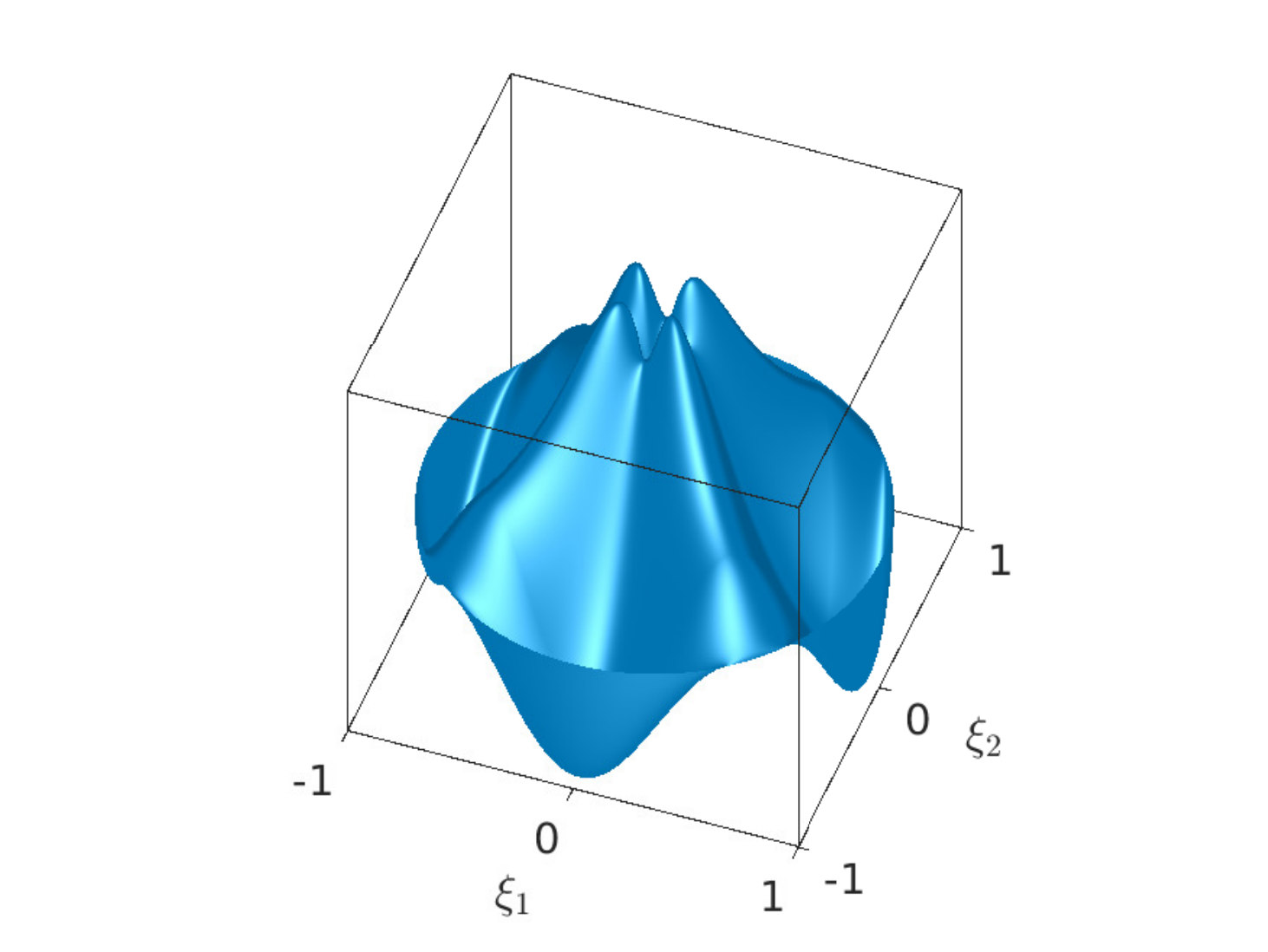}
\end{tabular} \\
Transversely Isotropic (MoN) & Cubic ($\text{MgAl}_2\text{O}_4$) \\
%%%%%%%%%%%%%%%%%%%%%%%%%%%%%%%%%%%%%%%%%%%%%%%%%%%%%
\multicolumn{2}{c}{
\begin{tabular}{cc}
\includegraphics[scale=0.45,trim={3cm 0 3cm 0.5cm},clip]{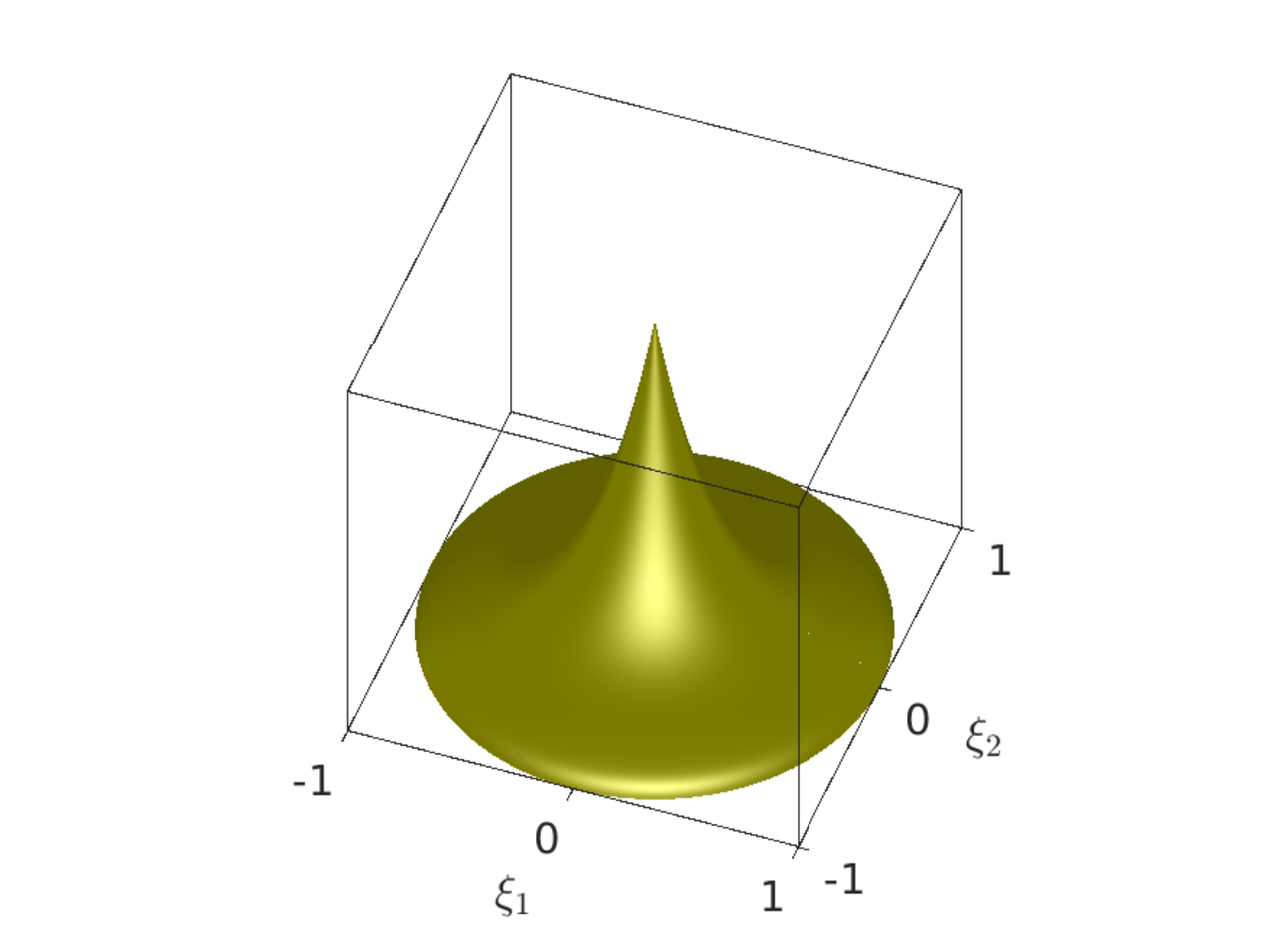} &
\includegraphics[scale=0.45,trim={3cm 0 3cm 0.5cm},clip]{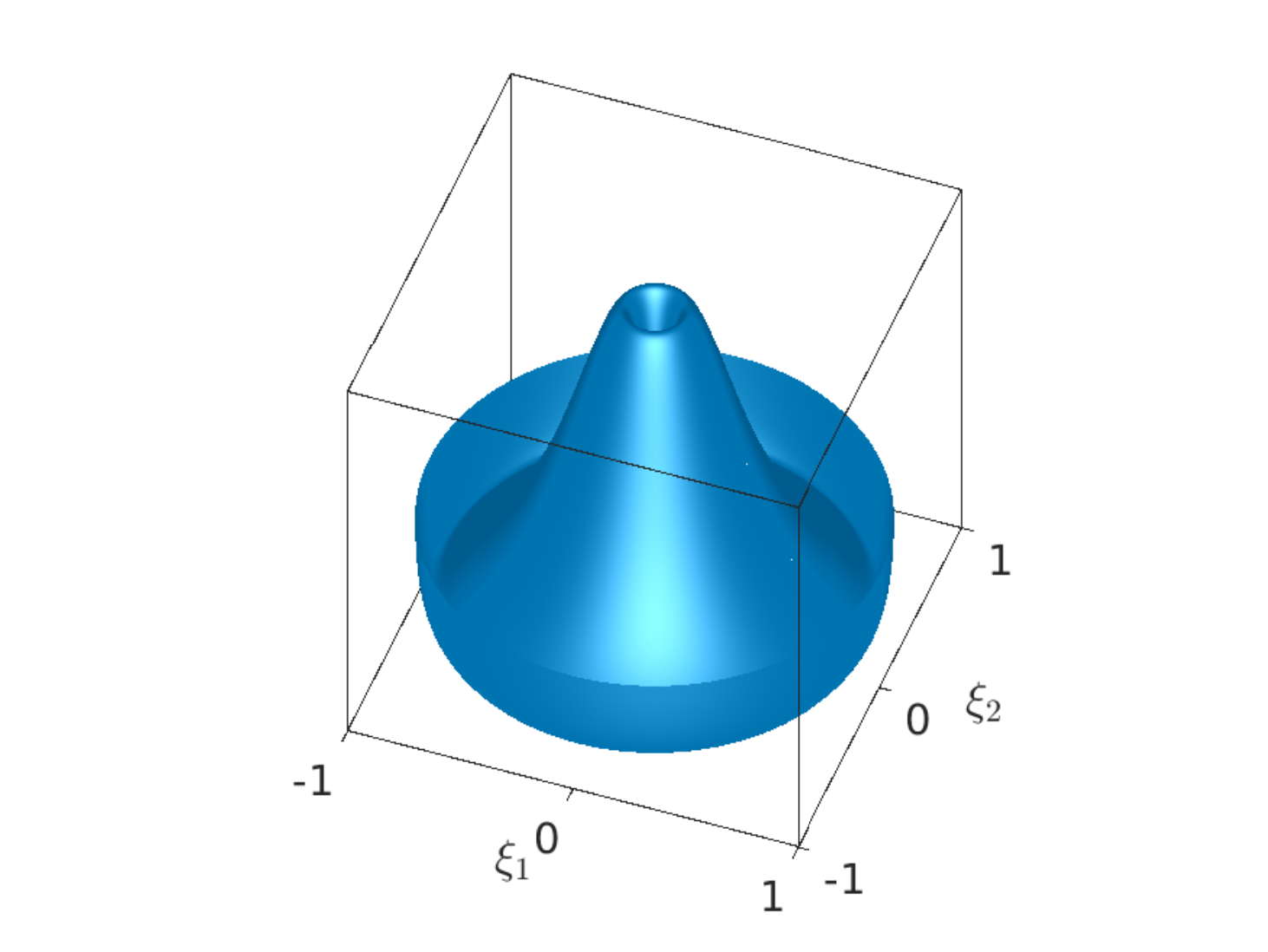}
\end{tabular} } \\
\multicolumn{2}{c}{Isotropic (Pyroceram 9608)}
\end{tabular}
\caption{Plots of $r^2 \lambda_{0} (\xi_1,\xi_2)$ ({\em cf.} \eqref{eqn:PlanestressLambda_0}) with $\omega( \|\bfxi \|) = \frac{1}{\| \bfxi \|}$ (in green) and $\omega( \| \bfxi \| ) = 1$ (in blue).}
\label{fig:planestresslamb0}
\end{figure}

%%%%%%%%%%%%%%%%%%%%%%%%%%%%%%%%%%%%%%%

\begin{figure}
\begin{tabular}{cc}
\begin{tabular}{cc}
\includegraphics[scale=0.45,trim={3cm 0 3cm 0.5cm},clip]{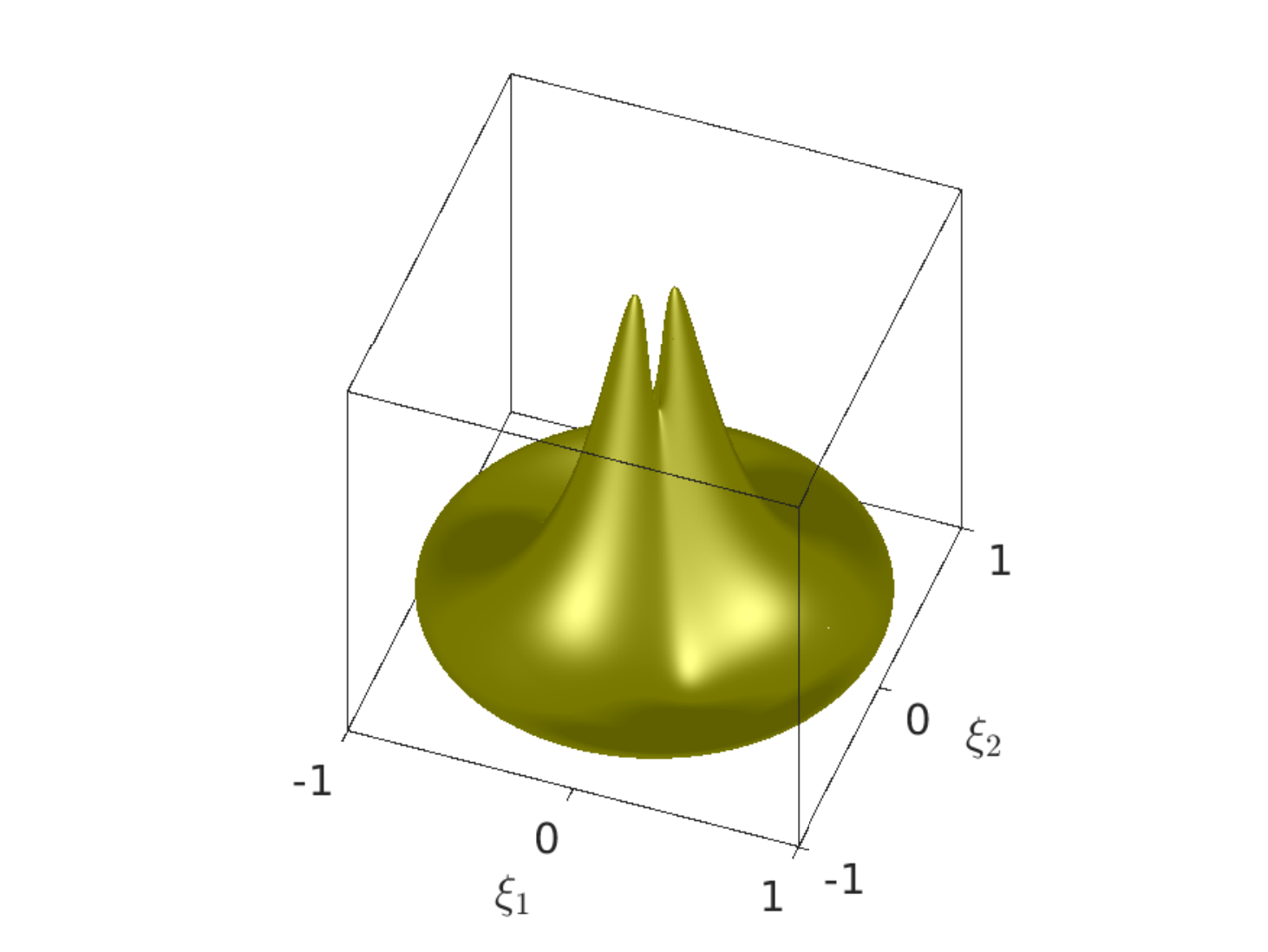} &
\includegraphics[scale=0.45,trim={3cm 0 3cm 0.5cm},clip]{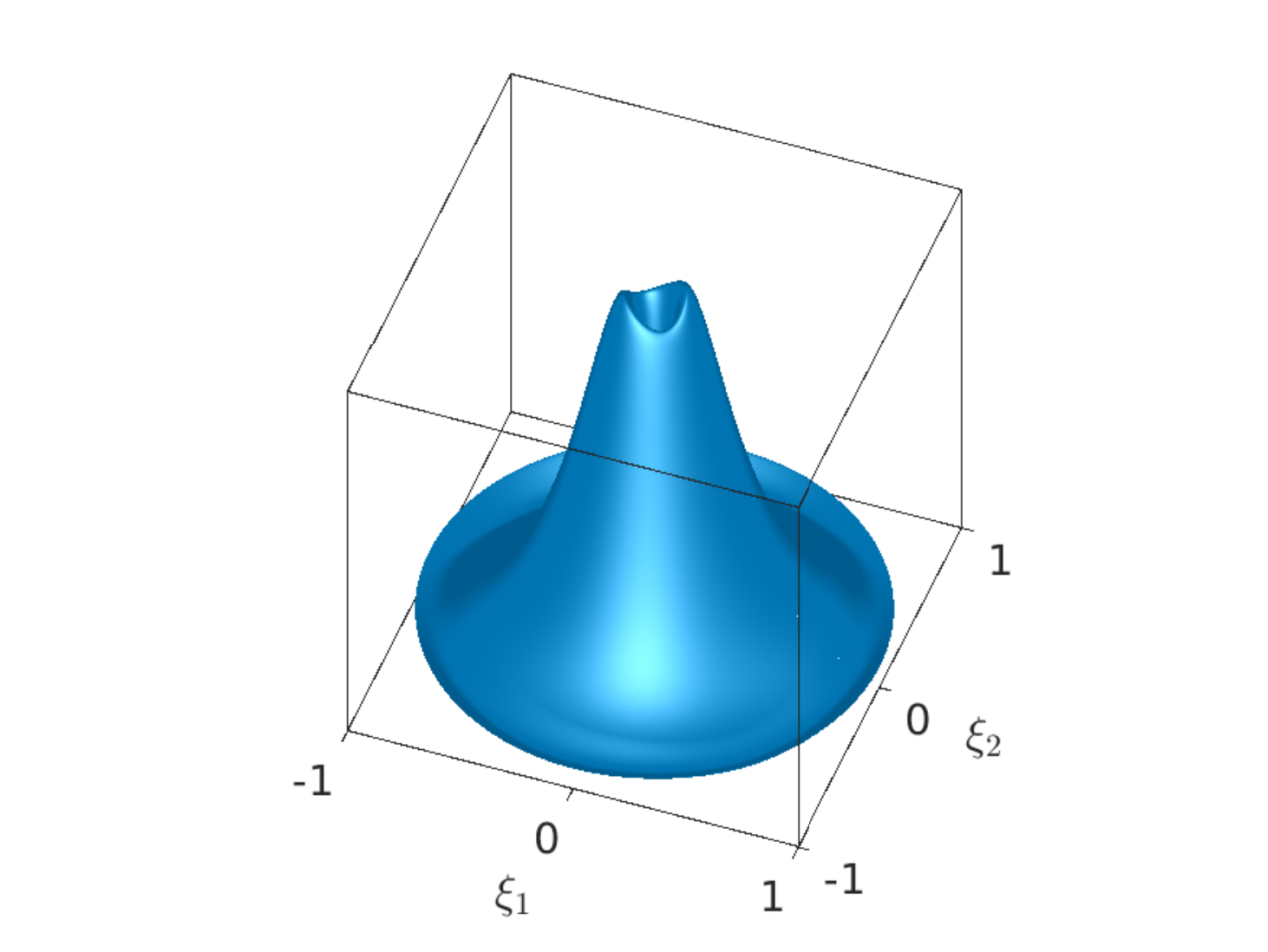}
\end{tabular} & \begin{tabular}{cc}
\includegraphics[scale=0.45,trim={3cm 0 3cm 0.5cm},clip]{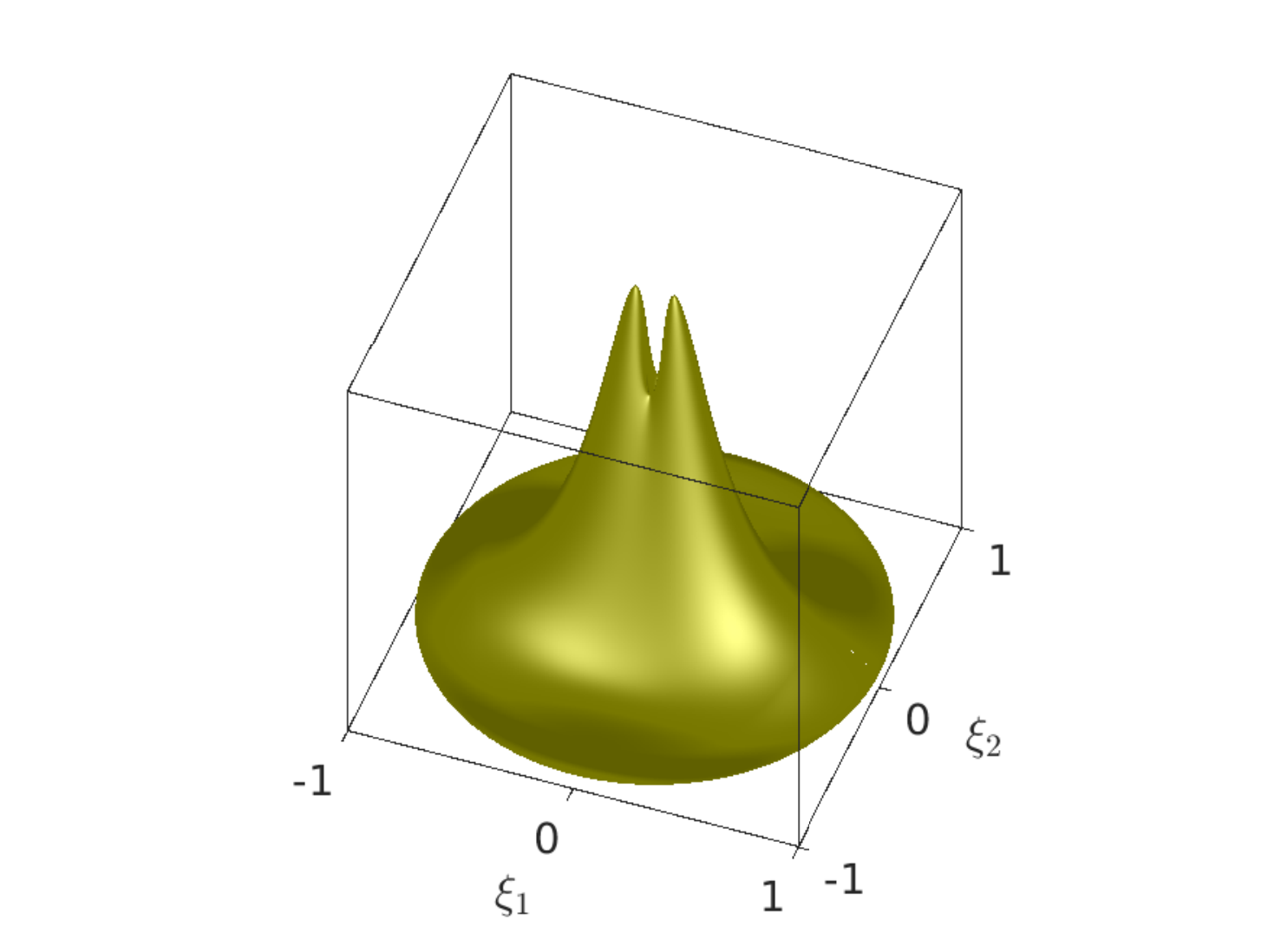} &
\includegraphics[scale=0.45,trim={3cm 0 3cm 0.5cm},clip]{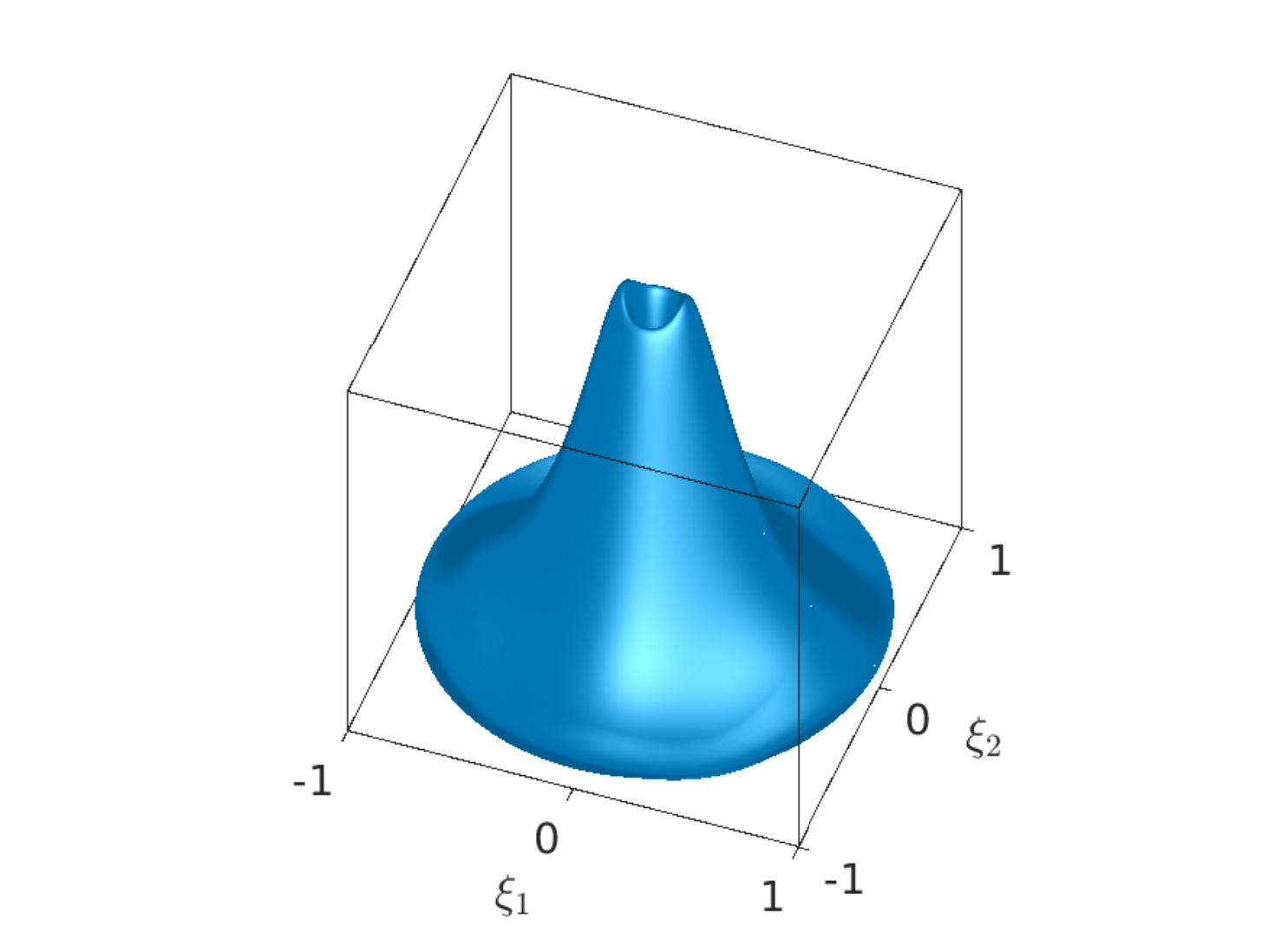}
\end{tabular} \\
Monoclinic ($\text{CoTeO}_4$) & Orthotropic ($\text{Te}_2\text{W}$) \\
%%%%%%%%%%%%%%%%%%%%%%%%%%%%%%%%%%%%%%%%%%%%%%%%%%%%%
\begin{tabular}{cc}
\includegraphics[scale=0.45,trim={3cm 0 3cm 0.5cm},clip]{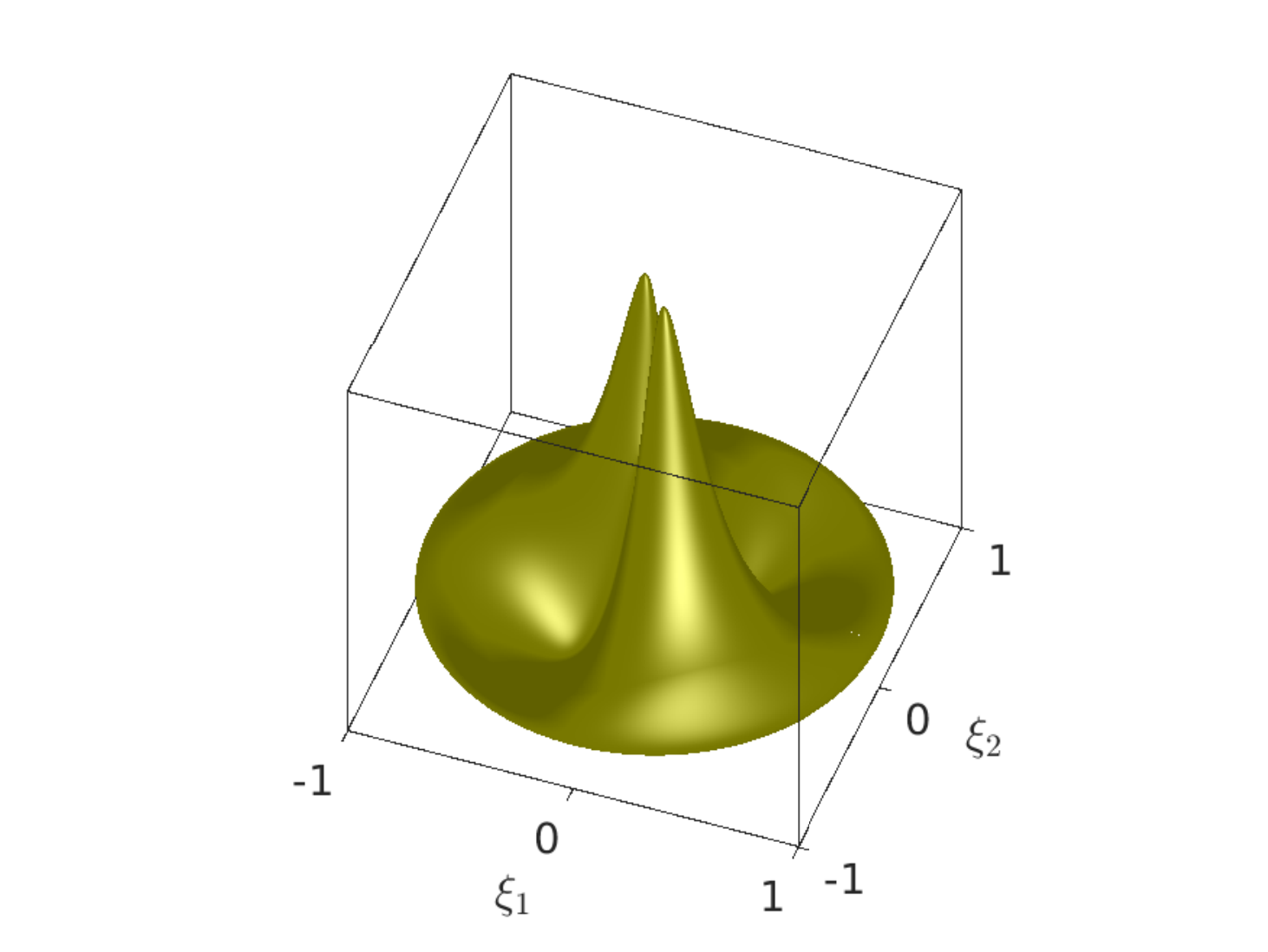} &
\includegraphics[scale=0.45,trim={3cm 0 3cm 0.5cm},clip]{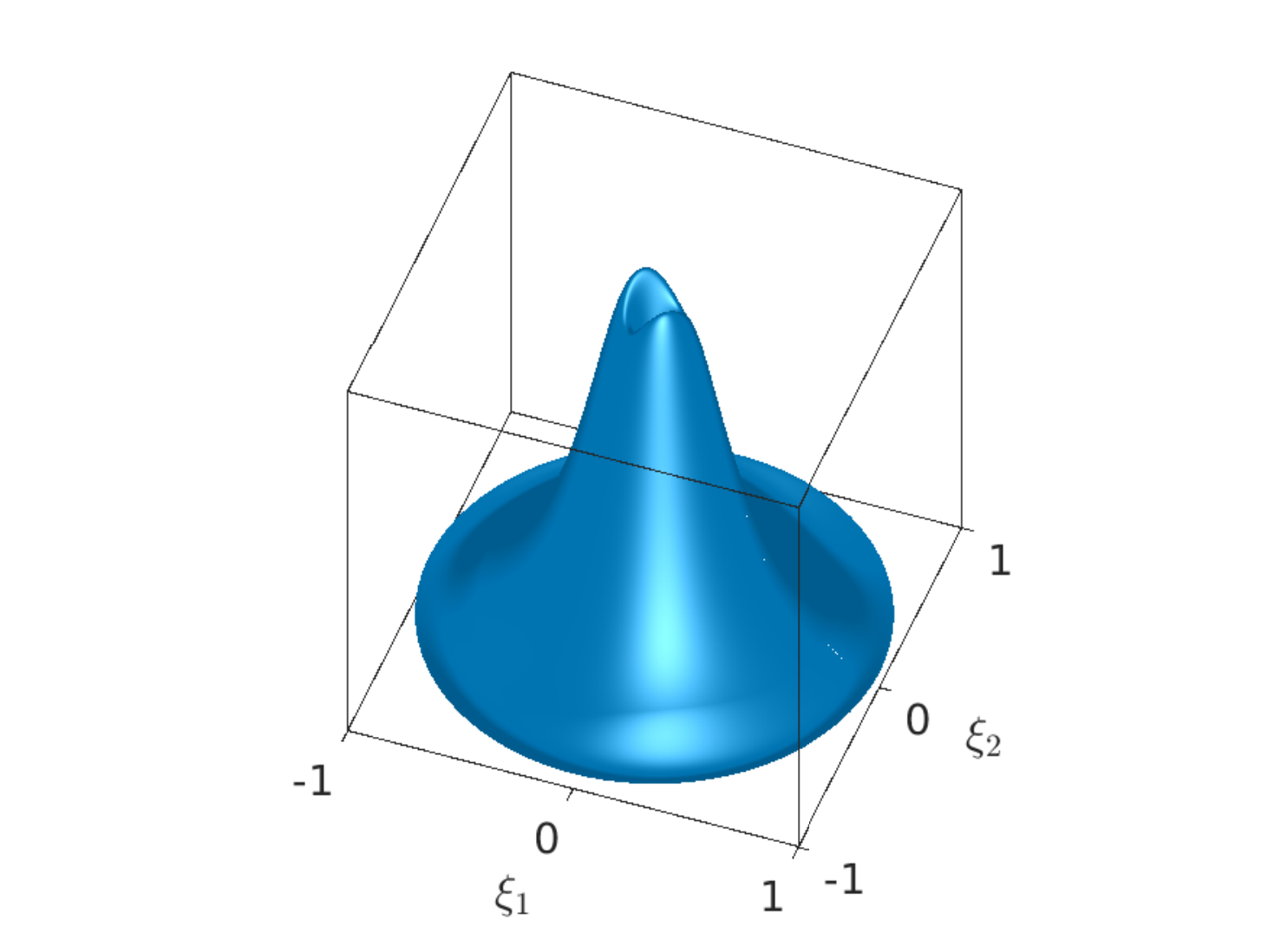}
\end{tabular} & \begin{tabular}{cc}
\includegraphics[scale=0.45,trim={3cm 0 3cm 0.5cm},clip]{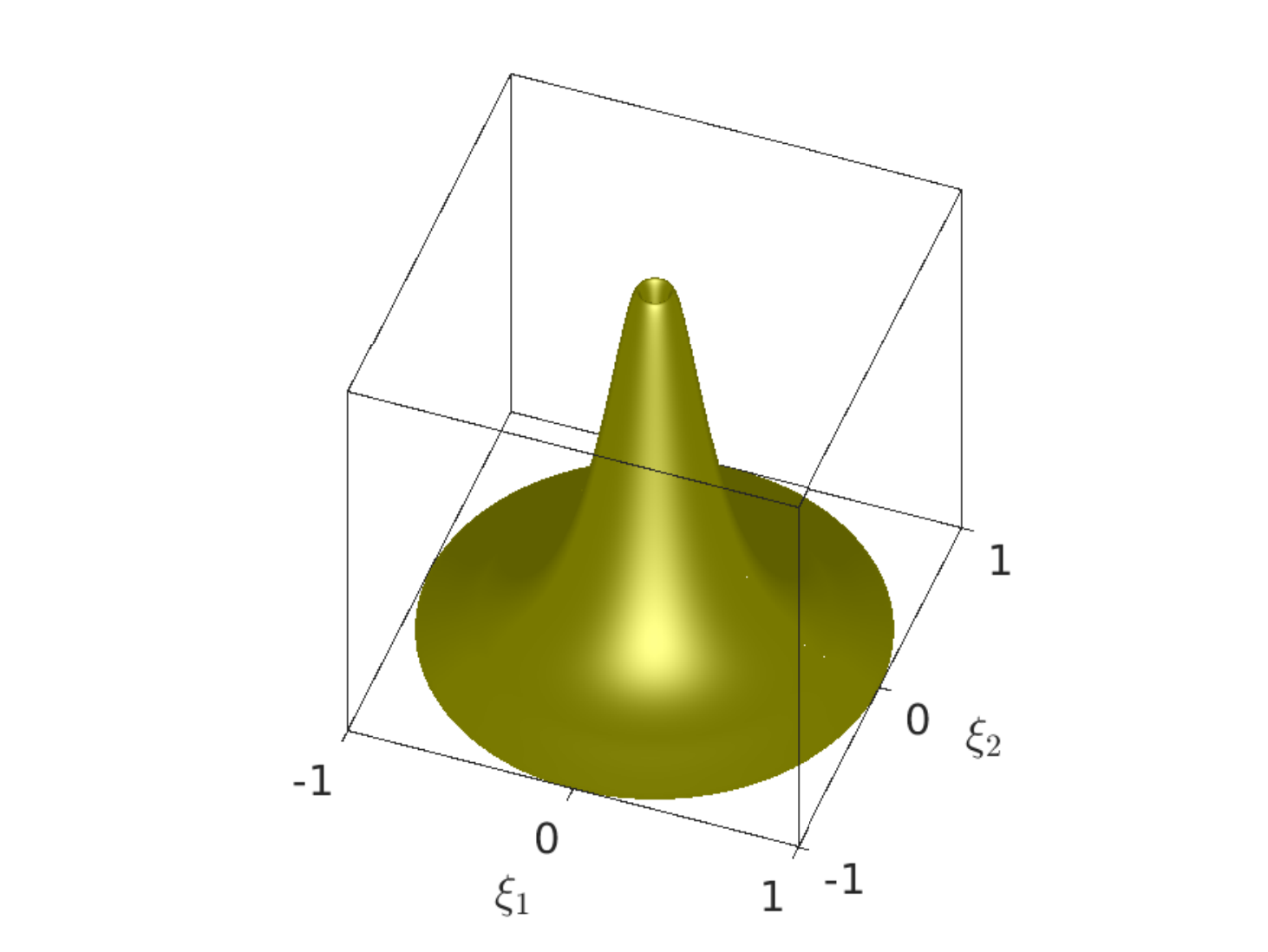} &
\includegraphics[scale=0.45,trim={3cm 0 3cm 0.5cm},clip]{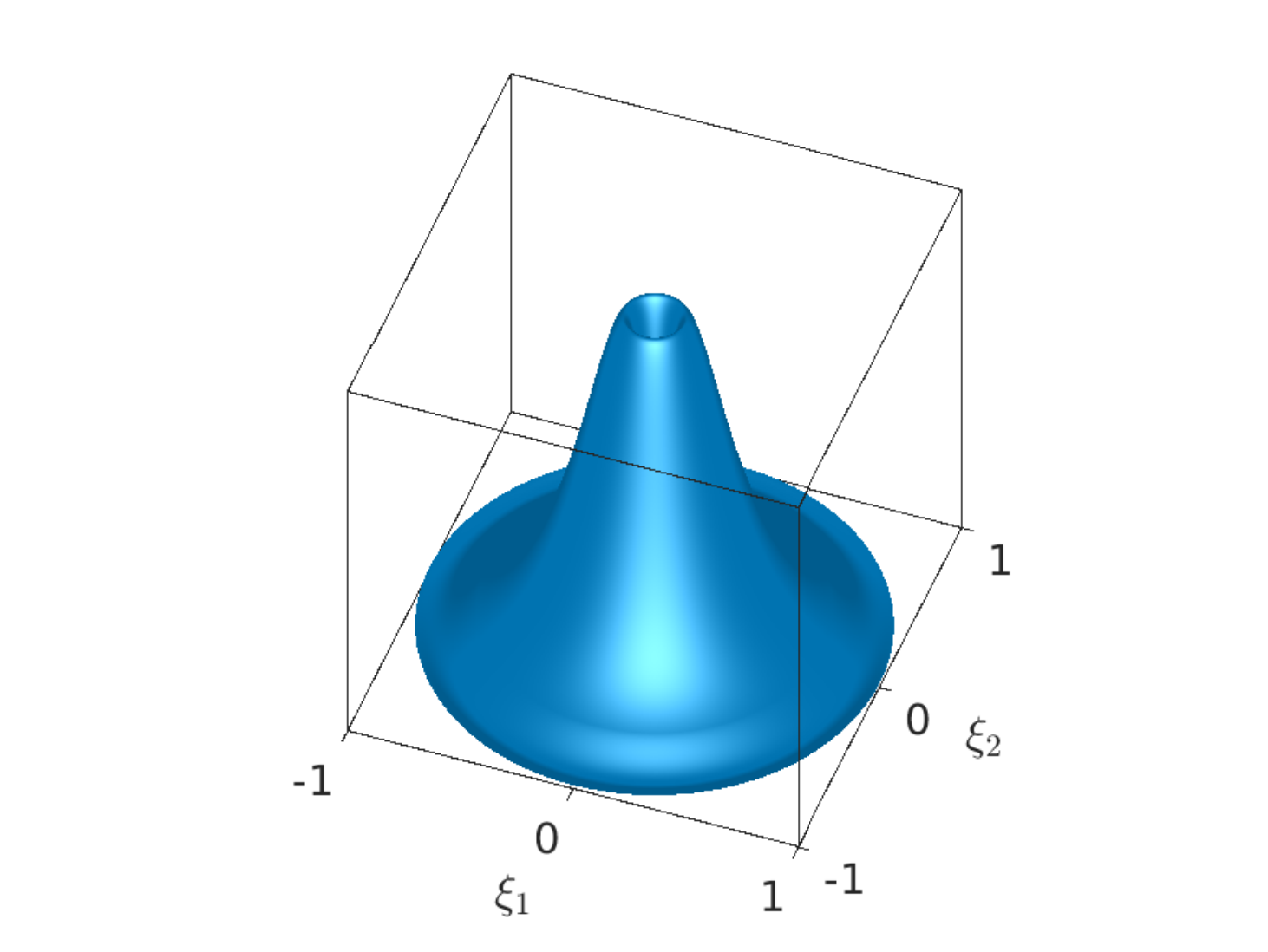}
\end{tabular} \\
Trigonal ($\text{Ta}_2\text{C}$) & Tetragonal (Si) \\
%%%%%%%%%%%%%%%%%%%%%%%%%%%%%%%%%%%%%%%%%%%%%%%%%%%%%
\begin{tabular}{cc}
\includegraphics[scale=0.45,trim={3cm 0 3cm 0.5cm},clip]{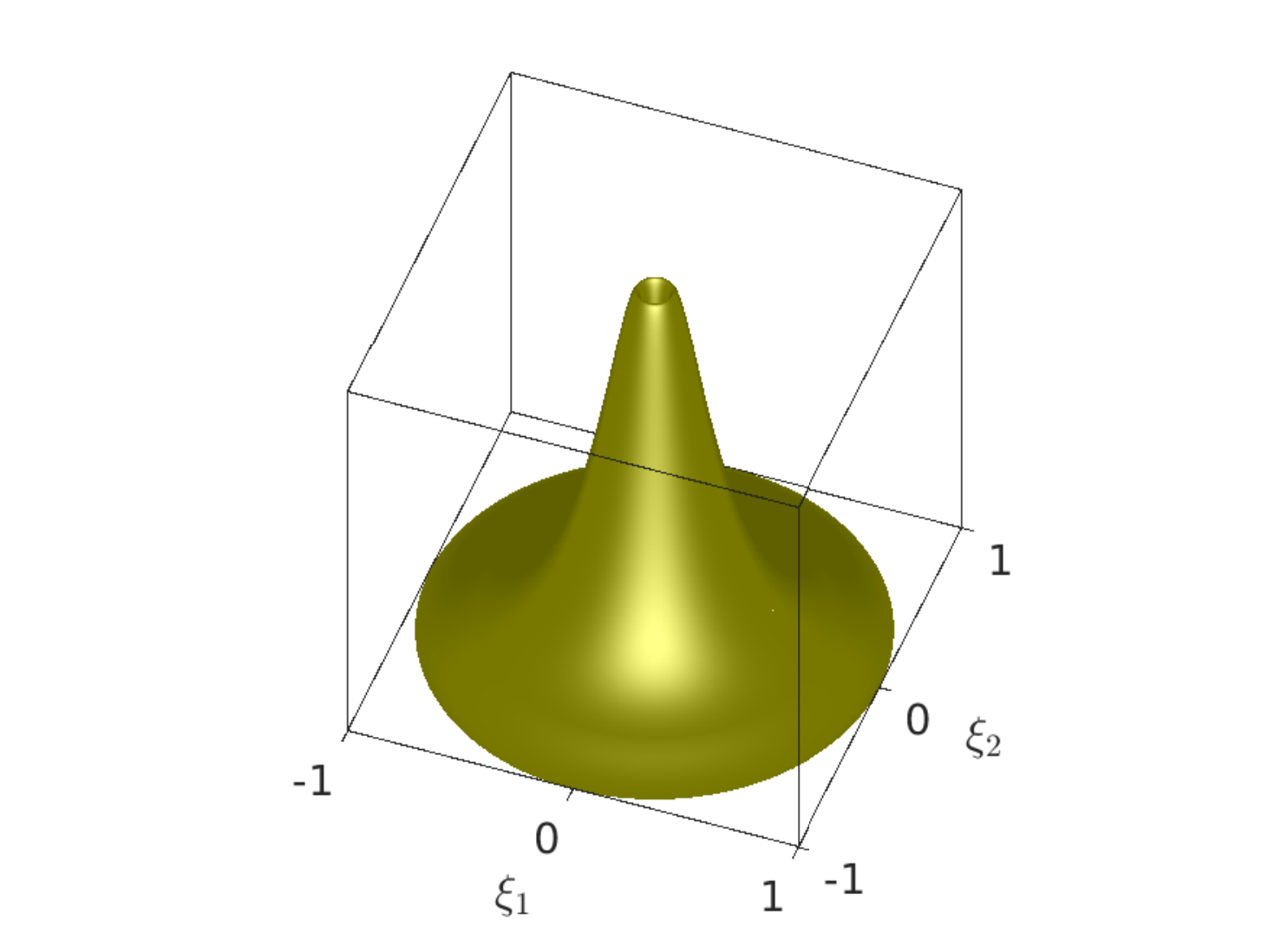} &
\includegraphics[scale=0.45,trim={3cm 0 3cm 0.5cm},clip]{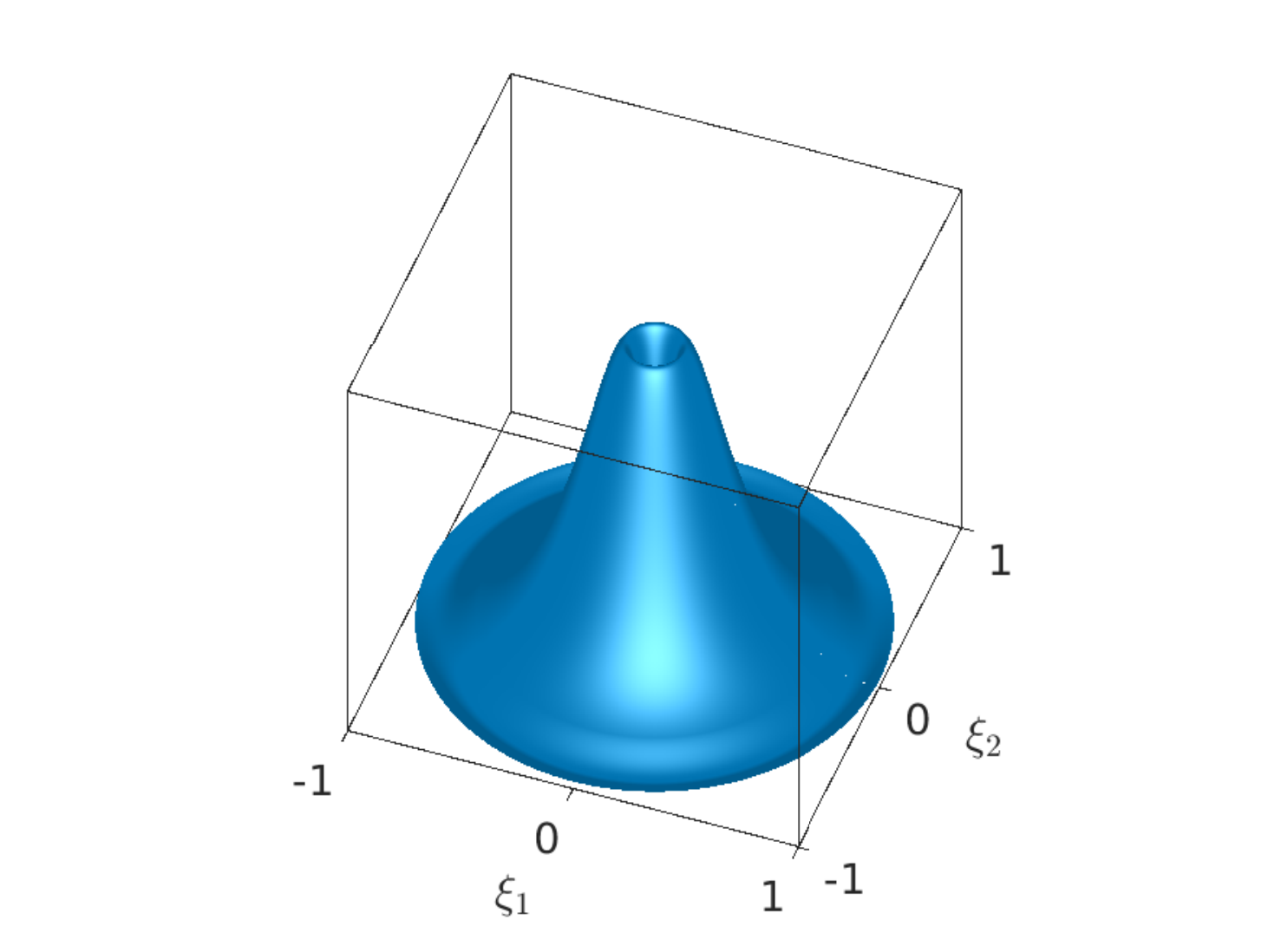}
\end{tabular} & \begin{tabular}{cc}
\includegraphics[scale=0.45,trim={3cm 0 3cm 0.5cm},clip]{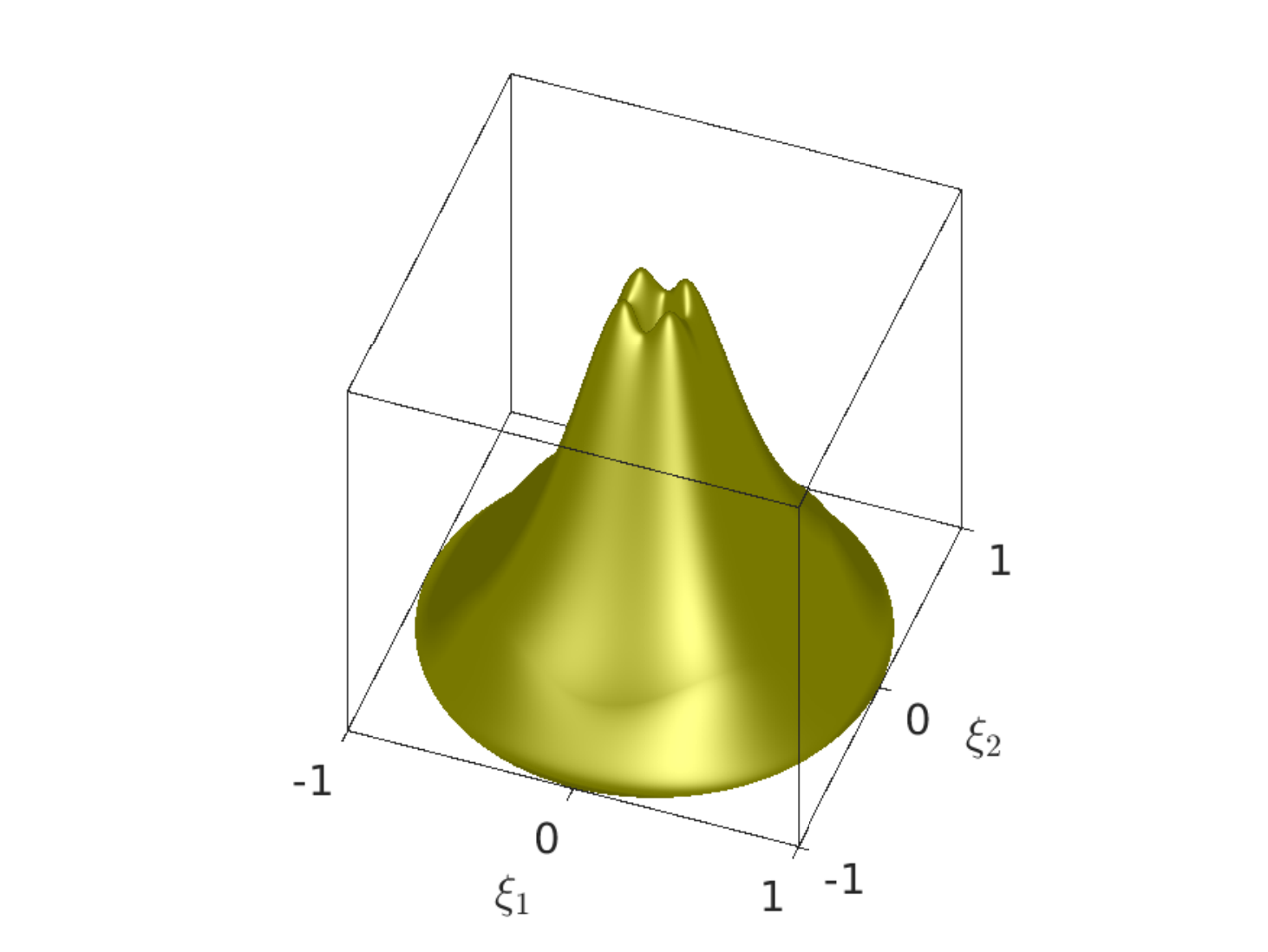} &
\includegraphics[scale=0.45,trim={3cm 0 3cm 0.5cm},clip]{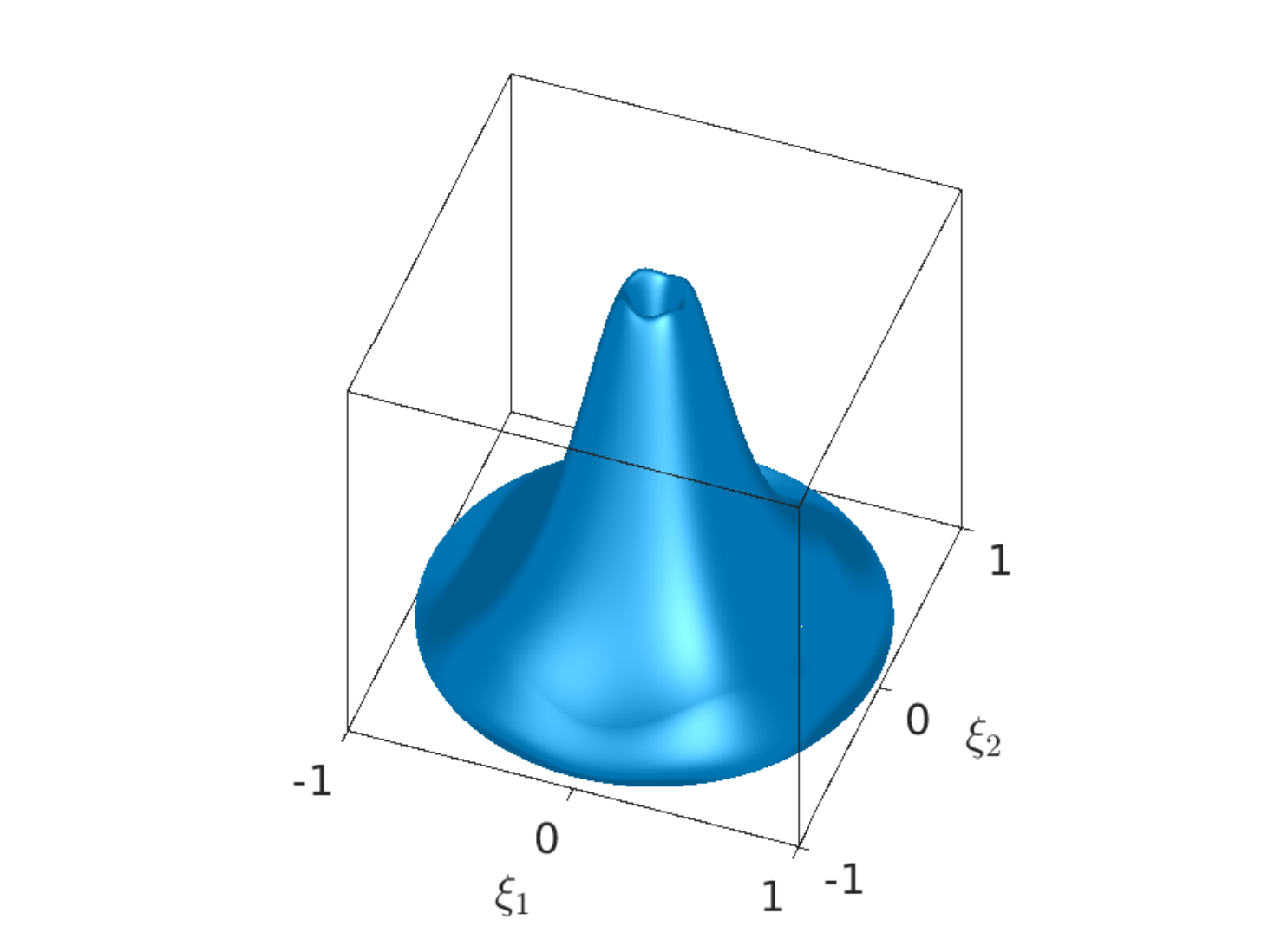}
\end{tabular} \\
Transversely Isotropic (MoN) & Cubic ($\text{MgAl}_2\text{O}_4$) \\
%%%%%%%%%%%%%%%%%%%%%%%%%%%%%%%%%%%%%%%%%%%%%%%%%%%%%
\multicolumn{2}{c}{
\begin{tabular}{cc}
\includegraphics[scale=0.45,trim={3cm 0 3cm 0.5cm},clip]{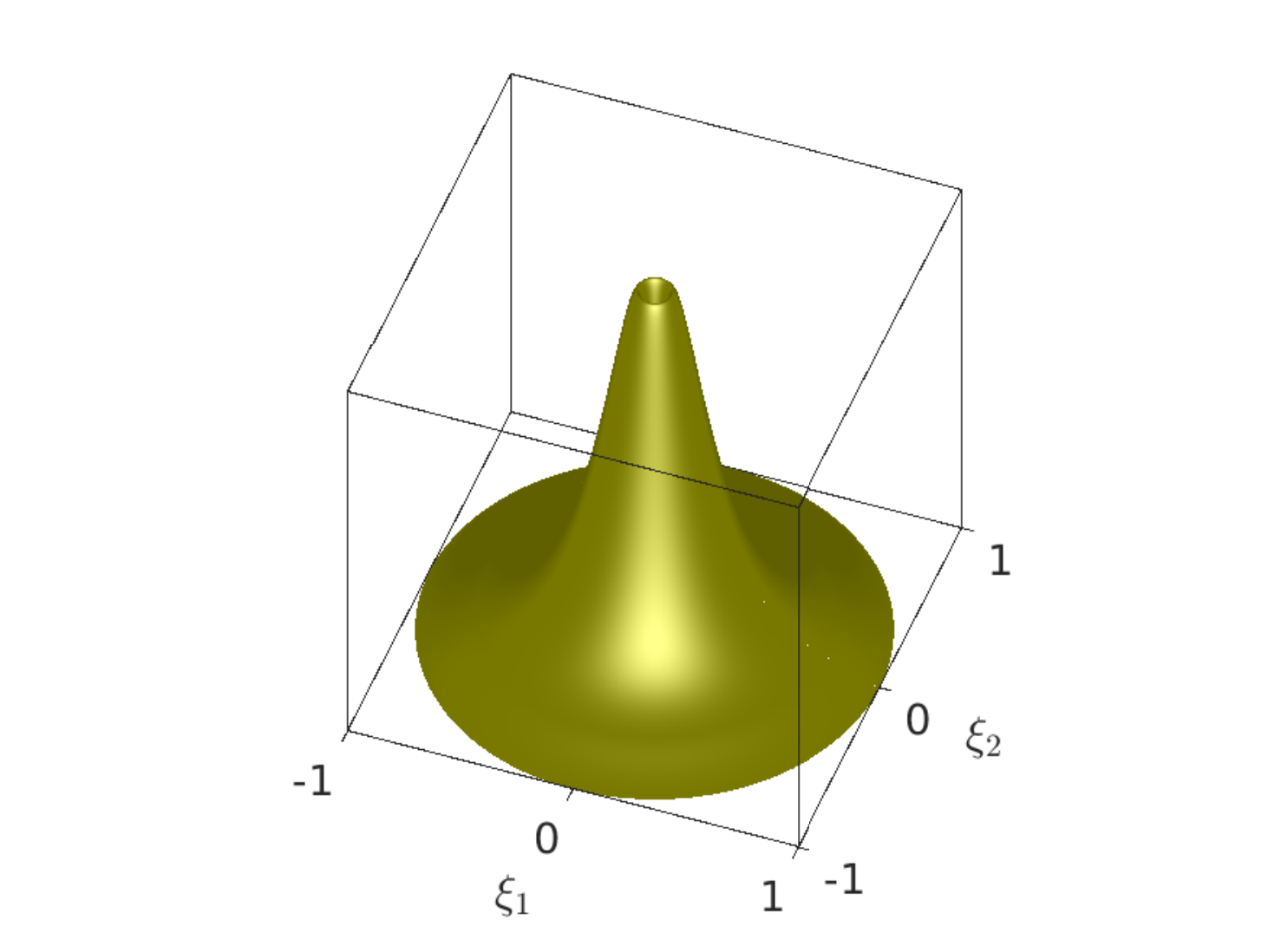} &
\includegraphics[scale=0.45,trim={3cm 0 3cm 0.5cm},clip]{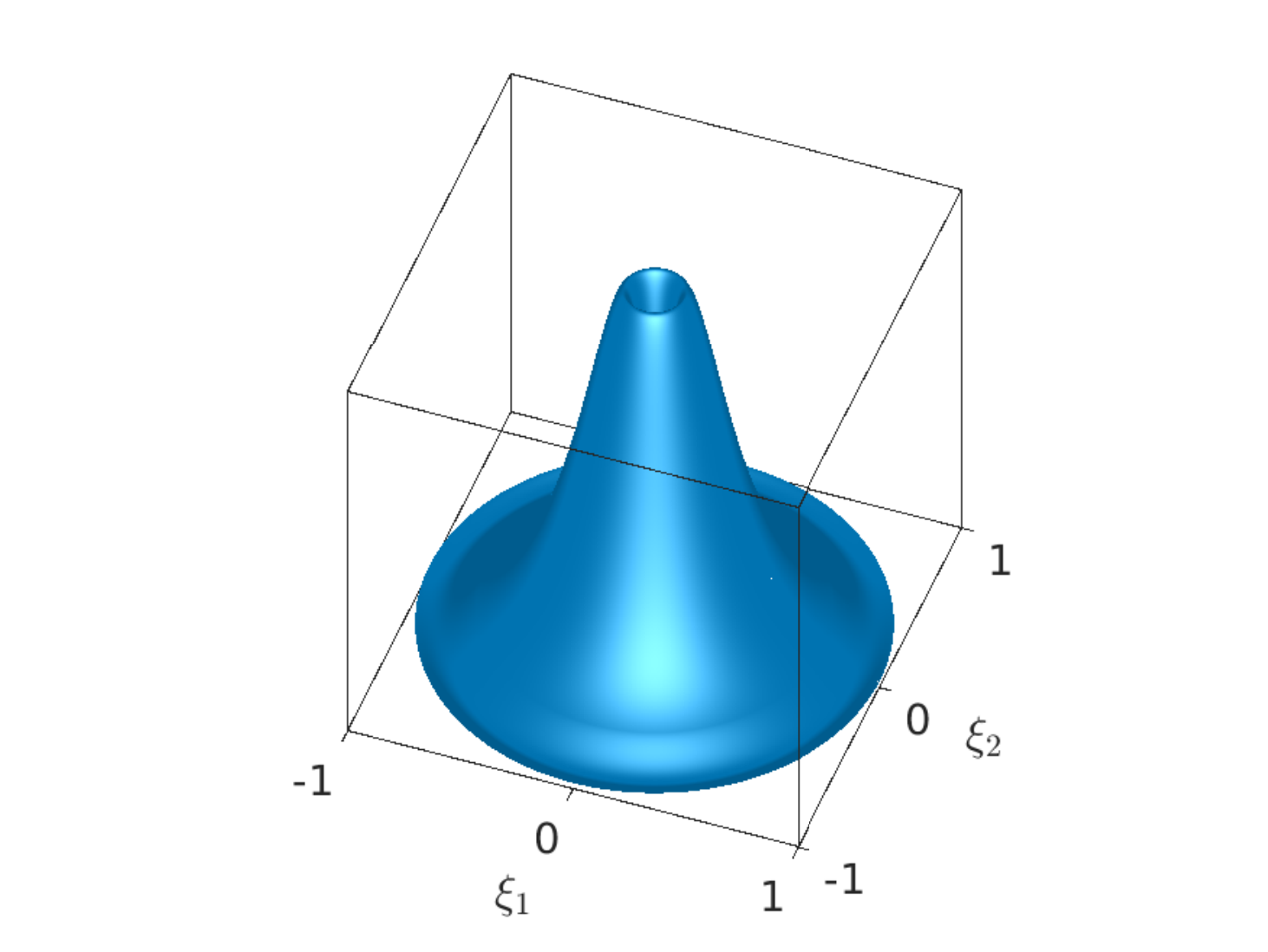}
\end{tabular} } \\
\multicolumn{2}{c}{Isotropic (Pyroceram 9608)}
\end{tabular}
\caption{Plots of $r \lambda_{2} (\xi_1,\xi_2)$ ({\em cf.} \eqref{eqn:PlanestressLambda_0}) with $\omega( \|\bfxi \|) = \frac{1}{\| \bfxi \|}$ (in green) and $\omega( \| \bfxi \| ) = 1$ (in blue).}
\label{fig:planestresslamb2}
\end{figure}

%%%%%%%%%%%%%%%%%%%%%%%%%%%%%%%%%%%%%%%%%%%%%%%%%%%%%%%
%%%%%%%%%%%%%%%%%%%%%%%%%%%%%%%%%%%%%%%%%%%%%%%%%%%%%%%

\textbf{Connections between the peridynamic and classical generalized plane stress equations}:

We now show that the peridynamic generalized plane stress model reduces to the classical generalized plane stress model under suitable assumptions.

\begin{proposition}\label{prop:relationplanestresstoclassical}
Suppose the micromodulus function $\lambda(\bfx',\bfx)$ is related to the elasticity tensor $\mathbb{C}$ through \eqref{eqn:CijklRelationAverageb}\footnote{Note that \eqref{eqn:CijklRelationAverageb} is a consequence of \eqref{eqn:lambdacondfullgenb} and therefore this result also holds under~\eqref{eqn:lambdacondfullgenb}.   }. Given a smooth deformation, under a second-order Taylor expansion of the displacement field, the peridynamic generalized plane stress model \eqref{eqn:PlaneStressFinal} reduces to the classical generalized plane stress model \eqref{eqn:planestressequationsofmotionclassical} with Cauchy's relations imposed.
\end{proposition}

\begin{proof}
First, we simplify the denominator of $A(x',y',t)$ ({\em cf}.~\eqref{eqn:PlaneStressAtermFinal}) by noticing from \eqref{eqn:defforlambiplanestress} and \eqref{eqn:CijklRelationAverageb} that (recall \eqref{eqn:lambda=0outsideball})
\begin{equation*}
\int_{B_{\delta}^{2D} (\bf0)} \lambda_4(\zeta_1,\zeta_2) d \zeta_1 d \zeta_2 = \int_{-h}^h \int_{-h}^h \int_{B_{\delta}^{2D} (x,y)} \lambda(\bfx'',\bfx') \zeta_3^4 d x'' d y'' d z'' d z' = 4hC_{3333}.
\end{equation*}
Applying a Taylor expansion of $\overline{u}_j(x'+\zeta_1,y'+\zeta_2,t)$ about $(x',y',t)$ in the numerator of $A(x',y',t)$, we obtain
\begin{equation}\label{eqn:expansionAtermpre}
\begin{split}
A(x',y',t) ={}& \frac{ \int_{B^{2D}_{\delta}(\bf0)} \lambda_2(\zeta_1,\zeta_2)  \zeta_j \left[ \overline{u}_j(x'+\zeta_1,y'+\zeta_2,t)- \overline{u}_j(x',y',t)\right] d \zeta_1 d \zeta_2}{\int_{B^{2D}_{\delta}(\bf0)} \lambda_4(\zeta_1,\zeta_2) d \zeta_1 d \zeta_2 } \\
={}& \frac{1}{4h C_{3333}} \int_{B_{\delta}^{2D}(\bf0)} \lambda_2(\zeta_1,\zeta_2) \zeta_j \left[ \zeta_1 \frac{\partial \overline{u}_j}{\partial x}(x',y',t)  + \zeta_2 \frac{\partial \overline{u}_j}{\partial y}(x',y',t)  \right. \\
& \hspace*{1.8in} \left. + \frac{\zeta_1^2}{2} \frac{\partial^2 \overline{u}_j}{\partial x^2}(x',y',t) + \zeta_1\zeta_2 \frac{\partial^2 \overline{u}_j}{\partial x \partial y}(x',y',t) + \frac{\zeta_2^2}{2} \frac{\partial^2 \overline{u}_j}{\partial y^2}(x',y',t)  + \cdots \right] d \zeta_1 d \zeta_2 \\
={}& \frac{1}{4h C_{3333}} \int_{B_{\delta}^{2D}(\bf0)} \lambda_2(\zeta_1,\zeta_2) \zeta_j \left[ \zeta_1 \frac{\partial \overline{u}_j}{\partial x}(x',y',t)  + \zeta_2 \frac{\partial \overline{u}_j}{\partial y}(x',y',t)  \right] d \zeta_1 d \zeta_2,
\end{split}
\end{equation}
where the summation in $j$ is over $1$ and $2$. In the last line of \eqref{eqn:expansionAtermpre} we nullified terms with even-order derivatives due to antisymmetry in the transformation $(\zeta_1,\zeta_2) \rightarrow (-\zeta_1,-\zeta_2)$, using the fact that $\lambda_2(-\zeta_1,-\zeta_2) = \lambda_2(\zeta_1,\zeta_2)$ by Assumption (P$\sigma$\ref{assump:PSs7Peri}); we additionally supposed the remaining higher-order terms (above order two) are negligible. Recalling \eqref{eqn:defforlambiplanestress} and \eqref{eqn:CijklRelationAverageb}, we simplify \eqref{eqn:expansionAtermpre}:
\begin{equation}\label{eqn:expansionAterm}
\begin{split}
A(x',y',t) ={}&  \frac{C_{33j1}}{C_{3333}} \frac{\partial \overline{u}_j}{\partial x}(x',y',t)  + \frac{C_{33j2}}{C_{3333}} \frac{\partial \overline{u}_j}{\partial y}(x',y',t).
%A(x',y') \approx& \frac{1}{C_{3333}} \left[ C_{33j1} \frac{\partial \overline{u}_j}{\partial x}(x',y')  + C_{33j2} \frac{\partial \overline{u}_j}{\partial y}(x',y') \right].
\end{split}
\end{equation}

Substituting \eqref{eqn:expansionAterm} into \eqref{eqn:PlaneStressFinal}, we find for $i=1,2$ that
\begin{equation}\label{eqn:PlaneStressConvergence1}
\begin{split}
\rho(x,y) \ddot{\overline{u}}_i(x,y,t) ={}& \frac{1}{2h}  \int_{B^{2D}_\delta(\bf0)} \lambda_0(\xi_1,\xi_2) \xi_i \xi_j (\overline{u}_j(x+\xi_1,y+\xi_2,t) - \overline{u}_j(x,y,t)) d \xi_1 d \xi_2  \\
&- \frac{1}{4h} \int_{B^{2D}_\delta (\bf0)}\lambda_2(\xi_1,\xi_2) \xi_i \left[ \frac{C_{33j1}}{C_{3333}} \frac{\partial \overline{u}_j}{\partial x}(x+\xi_1,y+\xi_2,t)  + \frac{C_{33j2}}{C_{3333}} \frac{\partial \overline{u}_j}{\partial y}(x+\xi_1,y+\xi_2,t) \right]   d \xi_1 d \xi_2 \\
&+ \overline{b}_i(x,y,t), 
\end{split}
\end{equation}
where the summation in $j$ is over $1$ and $2$.
Employing a Taylor expansion about $(x,y,t)$ of $\overline{u}_j(x',y',t)$ and its derivatives in \eqref{eqn:PlaneStressConvergence1}, eliminating terms antisymmetric with respect to the transformation $(\xi_1,\xi_2) \rightarrow (-\xi_1,-\xi_2)$, and supposing higher-order terms are negligible, we obtain
\begin{equation}\label{eqn:PlaneStressConvergence2}
\begin{split}
\rho(x,y) \ddot{\overline{u}}_i(x,y,t) ={}& \frac{1}{2h}  \int_{B^{2D}_\delta(\bf0)} \lambda_0(\xi_1,\xi_2) \xi_i \xi_j \left[ \frac{\xi_1^2}{2} \frac{\partial^2 \overline{u}_j}{\partial x^2}(x,y,t) + \xi_1\xi_2 \frac{\partial^2 \overline{u}_j}{\partial x \partial y}(x,y,t) + \frac{\xi_2^2}{2} \frac{\partial^2 \overline{u}_j}{\partial y^2}(x,y,t) \right] d \xi_1 d \xi_2  \\
&- \frac{1}{4h} \int_{B^{2D}_\delta (\bf0)}\lambda_2(\xi_1,\xi_2) \xi_i  \left[ \frac{C_{33j1}}{C_{3333}} \left( \xi_1 \frac{\partial^2 \overline{u}_j}{\partial x^2}(x,y,t) + \xi_2 \frac{\partial^2 \overline{u}_j}{\partial x \partial y}(x,y,t) \right)  \right. \\
&\hspace*{1.75in} \left. + \frac{C_{33j2}}{C_{3333}} \left( \xi_1 \frac{\partial^2 \overline{u}_j}{\partial x \partial y}(x,y,t) + \xi_2 \frac{\partial^2 \overline{u}_j}{\partial y^2}(x,y,t) \right) \right]   d \xi_1 d \xi_2 + \overline{b}_i(x,y,t). 
\end{split}
\end{equation}

Applying \eqref{eqn:defforlambiplanestress} and \eqref{eqn:CijklRelationAverageb} to \eqref{eqn:PlaneStressConvergence2}, we find
\begin{equation}\label{eqn:PlaneStressConvergence3}
\begin{split}
\rho(x,y) \ddot{\overline{u}}_i(x,y,t) ={}& C_{ij11} \frac{\partial^2 \overline{u}_j}{\partial x^2}(x,y,t) + 2C_{ij12} \frac{\partial^2 \overline{u}_j}{\partial x \partial y}(x,y,t)  + C_{ij22} \frac{\partial^2 \overline{u}_j}{\partial y^2}(x,y,t)  \\
&- \left[ \frac{C_{33i1}  C_{33j1}}{C_{3333}} \frac{\partial^2 \overline{u}_j}{\partial x^2}(x,y,t) + \frac{C_{33i2}  C_{33j1}}{C_{3333}} \frac{\partial^2 \overline{u}_j}{\partial x \partial y}(x,y,t) \right] \\
&- \left[ \frac{C_{33i1} C_{33j2}}{C_{3333}}  \frac{\partial^2 \overline{u}_j}{\partial x \partial y}(x,y,t) + \frac{C_{33i2} C_{33j2}}{C_{3333}} \frac{\partial^2 \overline{u}_j}{\partial y^2}(x,y,t)\right] + \overline{b}_i(x,y,t). 
\end{split} 
\end{equation}
Writing \eqref{eqn:PlaneStressConvergence3} out, we obtain
\begin{subequations}\label{eqn:planestresseqnmotioncauchy}
\begin{align}
\begin{split}
\rho \ddot{\overline{u}}_1 ={}&\left(C_{1111} - \frac{C_{1133}^2}{C_{3333}} \right)\frac{\partial^2 \overline{u}_1}{\partial x^2} + 2 \left( C_{1112} - \frac{C_{1133}C_{3312}}{C_{3333}} \right) \frac{\partial^2 \overline{u}_1}{\partial x \partial y}  +  \left( C_{1122}- \frac{C_{3312}^2}{C_{3333}}\right) \frac{\partial^2 \overline{u}_1}{\partial y^2}  \\
&+ \left( C_{1112} - \frac{C_{1133}C_{3312}}{C_{3333}} \right) \frac{\partial^2 \overline{u}_2}{\partial x^2} + \left( 2C_{1122} - \frac{C_{1133}C_{2233}}{C_{3333}} - \frac{C_{3312}^2}{C_{3333}} \right) \frac{\partial^2 \overline{u}_2}{\partial x \partial y} \\
&+\left(C_{2212}- \frac{C_{2233}C_{3312}}{C_{3333}}\right) \frac{\partial^2 \overline{u}_2}{\partial y^2} + \overline{b}_1, 
\end{split} \label{eqn:planestresseqnmotioncauchya} \\
%%%%%%%%%%%%%%%%%%%%%%%%%%%%%%%%%%%%
\begin{split}
\rho \ddot{\overline{u}}_2 ={}& \left(C_{1112} - \frac{C_{1133} C_{3312} }{C_{3333}}\right) \frac{\partial^2 \overline{u}_1}{\partial x^2} + \left( 2C_{1122} - \frac{C_{1133}C_{2233}}{C_{3333}} - \frac{C_{3312}^2}{C_{3333}} \right) \frac{\partial^2 \overline{u}_1}{\partial x \partial y}    \\
&+ \left( C_{2212} - \frac{C_{2233}C_{3312}}{C_{3333}}\right) \frac{\partial^2 \overline{u}_1}{\partial y^2} + \left( C_{1122}- \frac{C_{3312}^2}{C_{3333}}\right) \frac{\partial^2 \overline{u}_2}{\partial x^2}    \\
&+ 2 \left( C_{2212} - \frac{C_{2233}C_{3312}}{C_{3333}}\right) \frac{\partial^2 \overline{u}_2}{\partial x \partial y} + \left(C_{2222}- \frac{C_{2233}^2}{C_{3333}}\right) \frac{\partial^2 \overline{u}_2}{\partial y^2} + \overline{b}_2, \label{eqn:planestresseqnmotioncauchyb}
\end{split}
\end{align}
\end{subequations}
where we have omitted the $(x,y,t)$ dependence for brevity and a more direct comparison to the classical model. Comparing \eqref{eqn:planestresseqnmotioncauchy} to \eqref{eqn:planestressequationsofmotionclassical} with Cauchy's relations ({\em cf.}~\eqref{eqn:3DCauchyRelations}) imposed completes the proof. \qed
\end{proof}

\begin{remark}
Notice in the equation for classical generalized plane stress, \eqref{eqn:planestressequationsofmotionclassicala} and \eqref{eqn:planestressequationsofmotionclassicalb},  only four Cauchy's relations are relevant:
\begin{equation*}
C_{1212} = C_{1122}, \quad C_{1313} = C_{1133}, \quad C_{2323} = C_{2233}, \text{ and } \; C_{2313} = C_{3312}.
\end{equation*}
These four Cauchy's relations are exactly the same ones that are relevant to the classical plane strain model \eqref{eqn:classicalmonoeqnmotion} and \eqref{eqn:classicalmonoeqnmotion3}. This is not surprising since the remaining Cauchy's relations are trivially satisfied by monoclinic symmetry, an assumption of both classical plane strain and classical generalized plane stress. In the in-plane equations of motion for plane strain, \eqref{eqn:classicalmonoeqnmotion1} and \eqref{eqn:classicalmonoeqnmotion2}, the only relevant Cauchy's relation is $C_{1212} = C_{1122}$ ({\em cf}.~Remark \ref{remark:planestrainCauchy}). The entirety of the discrepancy between which Cauchy's relations are relevant in classical generalized plane stress and the in-plane equations for classical plane strain is due to the fact that $u_3$ may be a function of $z$ in the case of classical generalized plane stress, whereas $u_3$ is independent of $z$ in classical plane strain. In fact, one can easily see in \eqref{eqn:AvgStressRelation} and \eqref{eqn:[u3]solnplanestress} that the contribution from the substitution for $[u_3]$ in the derivation of the classical generalized plane stress model is the sole reason for the relevance of the additional Cauchy's relations, $C_{1313} = C_{1133}, C_{2323} = C_{2233},$ and $C_{2313} = C_{3312}$.
\end{remark}

\section{Conclusions}\label{sec:Conclusions}

This paper has a twofold objective. First, it reviews pure two-dimensional, plane strain, and plane stress anisotropic models in classical linear elasticity. Second, it introduces novel formulations analogous to the classical models within the bond-based peridynamic theory of solid mechanics. 
%This paper reviews two-dimensional, plane strain, and plane stress anisotropic models in classical linear elasticity and introduces analogous formulations within the bond-based peridynamic theory of solid mechanics. 

Our review of classical linear elasticity began with a new elementary and self-contained proof that there are exactly four material symmetry classes in classical linear elasticity in two dimensions: oblique, rectangular, square, and isotropic. Then, we presented pure anisotropic two-dimensional classical linear elasticity models for each of those symmetry classes. We further discussed planar approximations of three-dimensional anisotropic models in classical linear elasticity, specifically plane strain and plan stress formulations. The former normally concerns thick structures, whereas the latter often applies to thin plates. In three dimensions, there are eight material symmetry classes: triclinic, monoclinic, trigonal, orthotropic, tetragonal, transversely isotropic, cubic, and isotropic. Under certain assumptions, planar approximations reduce three-dimensional models to two-dimensional formulations. We reviewed plane strain and plane stress formulations in classical linear elasticity and specialized those formulations to the various symmetry classes. These two-dimensional approximations are based on decoupling in-plane and out-of-plane deformations which is achieved by assuming the material has a plane of reflection symmetry, i.e., the material symmetry is at least monoclinic. To discuss connections between the two-dimensional planar approximation models and the pure two-dimensional models in anisotropic classical linear elasticity, we reviewed engineering constants. It turns out that classical generalized plane stress models are equivalent, in terms of engineering constants, to pure two-dimensional models. In contrast, even though the classical plane strain models look identical to the pure two-dimensional models when expressed in terms of elasticity constants, they differ from the pure two-dimensional models, because the meaning of the elasticity constants varies upon the dimension. We finalized our review of classical linear elasticity with a discussion of Cauchy's relations in order to connect classical models to bond-based peridynamic models. While there is a single Cauchy's relation in two dimensions, there are six Cauchy's relations in three dimensions. Overall, Cauchy's relations reduce the number of independent constants in classical linear  elasticity from 6 to 5, for general oblique models, and from 21 to 15, for general triclinic models. 
% Why not classical plane stress?

In the context of the bond-based peridynamic theory, we began by introducing new anisotropic models, which can accommodate all four material symmetry classes in two dimensions and discuss related micromodulus functions with corresponding visualizations. We then derived novel peridynamic plane strain and plane stress formulations. As opposed to common approaches for planar approximations in peridynamics, which are based on simply matching constants of two-dimensional peridynamic models to corresponding constants appearing in planar approximations in classical linear elasticity, our models directly apply peridynamic analogues of classical planar assumptions to reduce three-dimensional models to two-dimensional formulations. For this purpose, we resorted to the three-dimensional anisotropic peridynamic models developed in \cite{STG2019}. We discussed the resulting plane strain and plane stress micromodulus functions with corresponding visualizations, and we proved the convergence of our peridynamic plane strain and plane stress models to their counterparts in classical linear elasticity with imposed Cauchy's relations. 

It is interesting to observe that, as opposed to peridynamic plane strain, which due to the thickness of the structure can simply consider points in the bulk of a body, peridynamic plane stress deals with thin plates and thus requires examination of surface effects. Furthermore, the peridynamic plane stress approximations result in a state-based peridynamic formulation.

The work presented in this paper offers a framework for simulation of two-dimensional problems based on the bond-based peridynamic theory, concerning all material symmetry classes found in classical linear elasticity. The newly introduced peridynamic plane strain and plane stress models provide means of reducing fully anisotropic three-dimensional bond-based peridynamic problems to two-dimensional formulations, resulting in significant computational savings, while retaining the dynamics of the original three-dimensional problems under proper assumptions.

%%%%%%%%%%%%%%%%%%%%%%%%%%%%%%%%%%%%
%%%%%%%%%%%%%%%%%%%%%%%%%%%%%%%%%%%%

\appendix

\section{Poisson's ratio restriction in two-dimensional bond-based peridynamics}\label{appendix:poissonratio}

In this section, we utilize a peridynamic traction in order to compute the engineering constants for isotropic homogeneous materials in the sense of the two-dimensional classical theory of linear elasticity. We suppose the material undergoes a (static) homogeneous deformation given by $u_1(\bfx) = \epsilon_{11} x$ and $u_2(\bfx) \equiv 0$. Then, the pairwise force function ({\em cf}.~\eqref{eqn:generalBBpairwiseforceformfinal}) has components given by
\begin{equation}\label{eqn:pairwiseforcefor2Dtractionproof}
    \begin{split}
        &f_1(\bfu(\bfx')-\bfu(\bfx),\bfx'-\bfx) = \epsilon_{11} \lambda(\|\bfx'-\bfx \|) (x'-x)^3, \\
        &f_2(\bfu(\bfx')-\bfu(\bfx),\bfx'-\bfx) = \epsilon_{11} \lambda(\|\bfx'-\bfx \|) (x'-x)^2 (y'-y).
    \end{split}
\end{equation}
Analogously to the three-dimensional formulation in \cite{SILLING2000}, given a body $\mathcal{B}\subset \mathbb{R}^2$, we define the two-dimensional peridynamic traction $\boldsymbol{\tau}$ at a material point $\bfx \in \mathcal{B}$ in the direction of the unit vector $\bfn$ as
\begin{equation*}
\bftau(\bfx,t,\bfn) := \int_0^\delta \int_{\mathcal{B}_\bfx^+(\bfx-s \bfn)} \bff(\bfu(\bfx',t)-\bfu(\bfx - s \bfn,t), \bfx' - (\bfx-s \bfn)) d \bfx' ds,
\end{equation*}
where
\begin{equation*}
\mathcal{B}_{\bfx}^+(\bfx-s \bfn) = \left\{ \bfx' \in B_\delta^{\text{2D}}(\bfx-s \bfn):  (\bfx'-\bfx) \cdot \bfn \geqslant 0 \right\}
\end{equation*}
and $B_{\delta}^{2D}(\bfx)$ is the two-dimensional ball of radius $\delta$ centered at $\bfx$.
Under a (static) homogeneous deformation, we may introduce a peridynamic stress tensor $\boldsymbol{\sigma}^{\text{peri}}$, independent of $\bfx$, such that $\bftau(\bfx,\bfn) = \boldsymbol{\sigma}^{\text{peri}} \bfn$. In order to compute the engineering constants, we calculate $\sigma^{\text{peri}}_{ij}$ for a material point $\bfx$ in the bulk of the body and relate it to the corresponding component of the classical stress tensor, $\sigma_{ij}$. In the directions of $\bfe_1 = \left[ \begin{array}{c} 1 \\ 0 \end{array} \right]$ and $\bfe_2 = \left[ \begin{array}{c} 0 \\ 1 \end{array} \right]$, we find
\begin{equation}\label{eqn:2Dtraction}
\begin{split}
&\left[\begin{array}{c}\sigma^{\text{peri}}_{11} \\ \sigma^{\text{peri}}_{12} \end{array} \right] = \boldsymbol{\tau}(\bfx,\bfe_1) = \int_0^\delta \int_{\mathcal{B}_\bfx^+(\bfx-s \bfe_1)} \bff(\bfu(\bfx') - \bfu(\bfx- s \bfe_1), \bfx' - (\bfx-s \bfe_1)) d \bfx' ds, \\
&\left[\begin{array}{c}\sigma^{\text{peri}}_{21} \\ \sigma^{\text{peri}}_{22} \end{array} \right] = \boldsymbol{\tau}(\bfx,\bfe_2) = \int_0^\delta \int_{\mathcal{B}_\bfx^+(\bfx-s \bfe_2)} \bff(\bfu(\bfx') - \bfu(\bfx- s \bfe_2), \bfx' - (\bfx-s \bfe_2)) d \bfx' ds.
\end{split}
\end{equation}
In Figure \ref{fig:twodtractionregion}, we present an illustration of the regions $\mathcal{B}_\bfx^+(\bfx-s \bfe_1)$ and $\mathcal{B}_\bfx^+(\bfx-s \bfe_2)$. \\

\begin{figure}
\centering
\begin{tikzpicture}[scale=0.85]
\begin{scope}    
\path[clip] (0,4)--(0,-4)--(4,-4)--(4,4)--cycle;  
\path[fill=black!20] (-3/2,0) circle(4);
\draw (-3/2,0) circle(4);
\end{scope}

\node (bfxsbfe) at (-3/2,0)[circle,fill=black,thick,inner sep=1pt,minimum size=0.1cm,label=below:$\bfx-s \bfe_1$]{};
\node (bfx) at (0,0)[circle,fill=black,thick,inner sep=1pt,minimum size=0.1cm,label=below right:$\bfx$]{};
\draw[dashed] (-3/2,0) -- (0,0);
\draw[dashed] (-3/2,0) -- (0,3.7080992);
\draw (0,-3.7080992) -- (0,3.7080992);

\node (delta) at (-3/4-0.3,1.35)[label=$\delta$]{};
\node (R+) at (5/2,0)[label=right:$\mathcal{B}_{\bfx}^+(\bfx-s \bfe_1)$]{};

%%%%%%%%%%%%%%%%%%%%%%%%%%%%%%%%

\begin{scope}    
\path[clip] (5,4)--(13,4)--(13,0)--(5,0)--cycle;  
\path[fill=black!20] (9,-3/2) circle(4);
\draw (9,-3/2) circle(4);
\end{scope}

\node (bfxsbfe) at (9,-3/2)[circle,fill=black,thick,inner sep=1pt,minimum size=0.1cm,label=below:$\bfx-s \bfe_2$]{};
\node (bfx) at (9,0)[circle,fill=black,thick,inner sep=1pt,minimum size=0.1cm,label=below right:$\bfx$]{};
\draw[dashed] (9,-3/2) -- (9,0);
\draw[dashed] (9,-3/2) -- (12.7080992,0);
\draw (5.2919008,0) -- (12.7080992,0);

\node (delta) at (10.35,-3/4-0.3)[label=$\delta$]{};
\node (R+) at (9,5/2)[label=above:$\mathcal{B}_{\bfx}^+(\bfx-s \bfe_2)$]{};
\end{tikzpicture}
\caption{Illustration of regions of integration for the two-dimensional peridynamic traction.} \label{fig:twodtractionregion}
\end{figure}
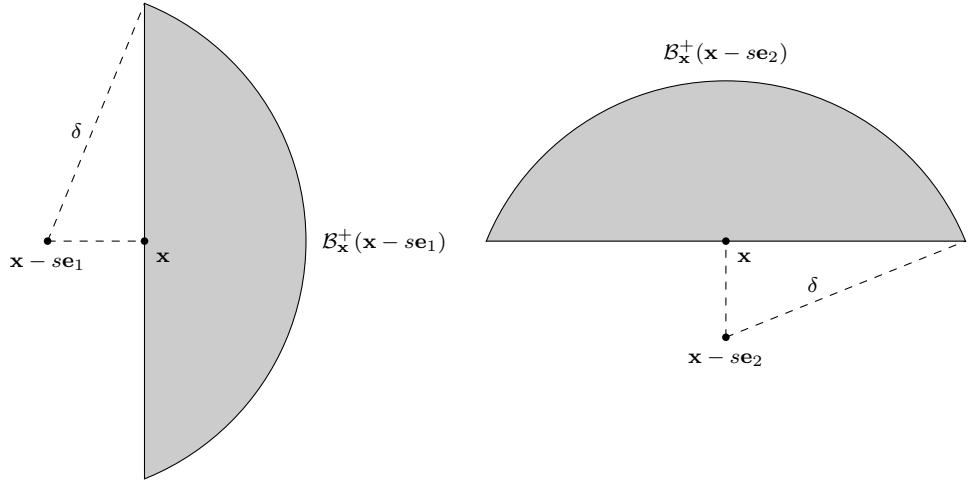

Now, we compute the stress components $\sigma_{ij}^{\text{peri}} = \tau_{i}(\bfx,\bfe_j)$ in \eqref{eqn:2Dtraction}. Combining~\eqref{eqn:pairwiseforcefor2Dtractionproof} and~\eqref{eqn:2Dtraction}, we find
\begin{equation}\label{eqn:tractionproof1}
    \begin{split}
        \tau_1(\bfx,\bfe_1) ={}& \epsilon_{11} \int_0^\delta \int_{\mathcal{B}_\bfx^+(\bfx-s \bfe_1)} \lambda(\|\bfx' - \hat{\bfx}\|) (x'-\hat{x})^3 d \bfx' ds\\
%        ={}& \epsilon_{11} \int_0^\delta \int_x^{x-s+\delta} \int_{y-\sqrt{\delta^2 - (x'-(x-s))^2}}^{y+\sqrt{\delta^2 - (x'-(x-s))^2}} \lambda(\|\bfx' - \hat{\bfx} \|) (x'-\hat{x})^3 d y' dx' ds\\
        ={}& \epsilon_{11} \int_0^\delta \int_x^{\hat{x}+\delta} \int_{y-\sqrt{\delta^2 - (x'-\hat{x})^2}}^{y+\sqrt{\delta^2 - (x'-\hat{x})^2}} \lambda(\|\bfx' - \hat{\bfx} \|) (x'-\hat{x})^3 d y' dx' ds, \\
        %%%%%%%%%%%%%%%%%%%%%%%%%%%%%%%%%%%%%
        \tau_1(\bfx,\bfe_2) ={}& \epsilon_{11} \int_0^\delta \int_{\mathcal{B}_\bfx^+(\bfx-s \bfe_2)} \lambda(\|\bfx' - \hat{\bfx} \|) (x'-\hat{x})^3 d \bfx' ds\\
%        ={}& \epsilon_{11} \int_0^\delta \int_y^{y-s+\delta} \int_{x-\sqrt{\delta^2 - (y'-(y-s))^2}}^{x+\sqrt{\delta^2 - (y'-(y-s))^2}} \lambda(\|\bfx' - \hat{\bfx} \|) (x'-\hat{x})^3 d x' d y' ds\\
        ={}& \epsilon_{11} \int_0^\delta \int_y^{\hat{y}+\delta} \int_{x-\sqrt{\delta^2 - (y'-\hat{y})^2}}^{x+\sqrt{\delta^2 - (y'-\hat{y})^2}} \lambda(\|\bfx' - \hat{\bfx} \|) (x'-\hat{x})^3 d x' d y' ds, \\
        %%%%%%%%%%%%%%%%%%%%%%%%%%%%%%%%%%%%%
        \tau_2(\bfx,\bfe_1) ={}& \epsilon_{11} \int_0^\delta \int_{\mathcal{B}_\bfx^+(\bfx-s \bfe_1)} \lambda(\|\bfx' - \hat{\bfx} \|) (x'-\hat{x})^2 (y'-\hat{y}) d \bfx' ds\\
%        ={}& \epsilon_{11} \int_0^\delta \int_x^{x-s+\delta} \int_{y-\sqrt{\delta^2 - (x'-(x-s))^2}}^{y+\sqrt{\delta^2 - (x'-(x-s))^2}} \lambda(\|\bfx' - \hat{\bfx} \|) (x'-\hat{x})^2 (y'-\hat{y}) d y' d x' ds\\
        ={}& \epsilon_{11} \int_0^\delta \int_x^{\hat{x}+\delta} \int_{y-\sqrt{\delta^2 - (x'-\hat{x})^2}}^{y+\sqrt{\delta^2 - (x'-\hat{x})^2}} \lambda(\|\bfx' - \hat{\bfx} \|) (x'-\hat{x})^2 (y'-\hat{y}) d y' d x' ds, \\
        %%%%%%%%%%%%%%%%%%%%%%%%%%%%%%%%%%%%%
        \tau_2(\bfx,\bfe_2) ={}& \epsilon_{11} \int_0^\delta \int_{\mathcal{B}_\bfx^+(\bfx-s \bfe_2)} \lambda(\|\bfx' - \hat{\bfx} \|) (x'-\hat{x})^2 (y'-\hat{y}) d \bfx' ds \\
%        ={}& \epsilon_{11} \int_0^\delta \int_y^{y-s+\delta} \int_{x-\sqrt{\delta^2 - (y'-(y-s))^2}}^{x+\sqrt{\delta^2 - (y'-(y-s))^2}} \lambda(\|\bfx' - \hat{\bfx} \|) (x'-\hat{x})^2 (y'-\hat{y}) d x' dy' ds,
        ={}& \epsilon_{11} \int_0^\delta \int_y^{\hat{y}+\delta} \int_{x-\sqrt{\delta^2 - (y'-\hat{y})^2}}^{x+\sqrt{\delta^2 - (y'-\hat{y})^2}} \lambda(\|\bfx' - \hat{\bfx} \|) (x'-\hat{x})^2 (y'-\hat{y}) d x' dy' ds,
        %%%%%%%%%%%%%%%%%%%%%%%%%%%%%%%%%%%%%
    \end{split} 
\end{equation}
where in $\tau_i(\bfx,\bfe_j)$, we defined $\hat{\bfx} = \left[\begin{array}{c} \hat{x} \\ \hat{y} \end{array} \right]  := \bfx - s \bfe_j$.

To simplify \eqref{eqn:tractionproof1}, we apply the changes of variables $x' \rightarrow x' - s$ to the expressions for $\tau_i(\bfx,\bfe_1)$ and $y' \rightarrow y'-s$ to the expressions for $\tau_i(\bfx,\bfe_2)$. We follow this by applying the change of variables $\bfx' \rightarrow \bfx + \bfxi$ to all the resulting equations. This process yields the following:
\begin{equation}\label{eqn:tractionproof2}
    \begin{split}
        \tau_1(\bfx,\bfe_1) ={}& \epsilon_{11} \int_0^\delta \int_{x+s}^{x+\delta} \int_{y-\sqrt{\delta^2 - (x'-x)^2}}^{y+\sqrt{\delta^2 - (x'-x)^2}} \lambda(\|\bfx' - \bfx \|) (x'-x)^3 d y' dx' ds\\
        ={}& \epsilon_{11} \int_0^\delta \int_{s}^{\delta} \int_{-\sqrt{\delta^2 - \xi_1^2}}^{\sqrt{\delta^2 - \xi_1^2}} \lambda(\|\bfxi \|) \xi_1^3 d \xi_2 d \xi_1 ds, \\
        %%%%%%%%%%%%%%%%%%%%%%%%%%%%%%%%%%%%%
        \tau_1(\bfx,\bfe_2) ={}& \epsilon_{11} \int_0^\delta \int_{y+s}^{y+\delta} \int_{x-\sqrt{\delta^2 - (y'-y)^2}}^{x+\sqrt{\delta^2 - (y'-y)^2}} \lambda(\|\bfx' - \bfx \|) (x'-x)^3 d x' d y' ds\\
        ={}& \epsilon_{11} \int_0^\delta \int_{s}^{\delta} \int_{-\sqrt{\delta^2 - \xi_2^2}}^{\sqrt{\delta^2 - \xi_2^2}} \lambda(\|\bfxi \|) \xi_1^3 d \xi_1 d\xi_2 ds = 0, \\
        %%%%%%%%%%%%%%%%%%%%%%%%%%%%%%%%%%%%%
        \tau_2(\bfx,\bfe_1) ={}& \epsilon_{11} \int_0^\delta \int_{x+s}^{x+\delta} \int_{y-\sqrt{\delta^2 - (x'-x)^2}}^{y+\sqrt{\delta^2 - (x'-x)^2}} \lambda(\| \bfx' - \bfx \|) (x'-x)^2 (y'-y) d y' d x' ds\\
        ={}& \epsilon_{11} \int_0^\delta \int_{s}^{\delta} \int_{-\sqrt{\delta^2 - \xi_1^2}}^{\sqrt{\delta^2 - \xi_1^2}} \lambda(\| \bfxi \|) \xi_1^2 \xi_2 d \xi_2 d \xi_1 ds = 0, \\
        %%%%%%%%%%%%%%%%%%%%%%%%%%%%%%%%%%%%%
        \tau_2(\bfx,\bfe_2) ={}& \epsilon_{11} \int_0^\delta \int_{y+s}^{y+\delta} \int_{x-\sqrt{\delta^2 - (y'-y)^2}}^{x+\sqrt{\delta^2 - (y'-y)^2}} \lambda(\|\bfx' - \bfx \|) (x'-x)^2 (y'-y) d x' dy' ds \\
        ={}& \epsilon_{11} \int_0^\delta \int_{s}^{\delta} \int_{-\sqrt{\delta^2 - \xi_2^2}}^{\sqrt{\delta^2 - \xi_2}} \lambda(\| \bfxi \|) \xi_1^2 \xi_2 d \xi_1 d \xi_2 ds.
        %%%%%%%%%%%%%%%%%%%%%%%%%%%%%%%%%%%%%
    \end{split} 
\end{equation}
Above we found $\tau_1(\bfx,\bfe_2) = 0$ and $\tau_2(\bfx,\bfe_1) = 0$ by antisymmetry in $\xi_1$ and $\xi_2$, respectively. Applying the polar coordinate change of variables $(\xi_1,\xi_2) \rightarrow (r \cos(\theta),r \sin(\theta))$ to \eqref{eqn:tractionproof2}, we find
\begin{equation}\label{eqn:tractionproof3}
    \begin{split}
        \tau_1(\bfx,\bfe_1) ={}& \epsilon_{11} \int_0^\delta \int_{s}^{\delta} \int_{-\arccos\left( \frac{s}{r} \right)}^{\arccos\left( \frac{s}{r} \right)} \lambda(r) r^4 \cos^3(\theta) d \theta dr ds, \\
        %%%%%%%%%%%%%%%%%%%%%%%%%%%%%%%%%%%%%
%        \tau_1(\bfx,\bfe_2) ={}& \epsilon_{11} \int_0^\delta \int_{s}^{\delta} \int_{\frac{\pi}{2}-\arccos\left( \frac{s}{r} \right)}^{\frac{\pi}{2} + \arccos\left( \frac{s}{r} \right)} \lambda(r) r^4 \cos^3(\theta) d \theta dr ds\\
%	\tau_1(\bfx,\bfe_2) ={}& 0, \\
        %%%%%%%%%%%%%%%%%%%%%%%%%%%%%%%%%%%%%
%        \tau_2(\bfx,\bfe_1) ={}& \epsilon_{11} \int_0^\delta \int_{s}^{\delta} \int_{-\arccos\left( \frac{s}{r} \right)}^{\arccos\left( \frac{s}{r} \right)} \lambda(r) r^4 \cos^2(\theta) \sin(\theta) d \theta dr ds\\
%		\tau_2(\bfx,\bfe_1) ={}& 0, \\
        %%%%%%%%%%%%%%%%%%%%%%%%%%%%%%%%%%%%%
        \tau_2(\bfx,\bfe_2) ={}& \epsilon_{11} \int_0^\delta \int_{s}^{\delta} \int_{\frac{\pi}{2}-\arccos\left( \frac{s}{r} \right)}^{\frac{\pi}{2} + \arccos\left( \frac{s}{r} \right)} \lambda(r) r^4 \cos^2(\theta) \sin(\theta) d \theta dr ds.
        %%%%%%%%%%%%%%%%%%%%%%%%%%%%%%%%%%%%%
    \end{split} 
\end{equation}
Changing the order of integration between $r$ and $s$, and then integrating in $\theta$ and $s$, we find
\begin{equation}\label{eqn:tractionproof4}
    \begin{split}
        \tau_1(\bfx,\bfe_1) ={}& \epsilon_{11} \int_0^\delta \int_{0}^{r} \int_{-\arccos\left( \frac{s}{r} \right)}^{\arccos\left( \frac{s}{r} \right)} \lambda(r) r^4 \cos^3(\theta) d \theta ds dr  = \frac{3 \pi \epsilon_{11}}{8} \int_0^\delta \lambda(r) r^5 dr, \\
        %%%%%%%%%%%%%%%%%%%%%%%%%%%%%%%%%%%%%
%	\tau_1(\bfx,\bfe_2) ={}& 0, \\
        %%%%%%%%%%%%%%%%%%%%%%%%%%%%%%%%%%%%%
%		\tau_2(\bfx,\bfe_1) ={}& 0, \\
        %%%%%%%%%%%%%%%%%%%%%%%%%%%%%%%%%%%%%
        \tau_2(\bfx,\bfe_2) ={}& \epsilon_{11} \int_0^\delta \int_{0}^{r} \int_{\frac{\pi}{2}-\arccos\left( \frac{s}{r} \right)}^{\frac{\pi}{2} + \arccos\left( \frac{s}{r} \right)} \lambda(r) r^4 \cos^2(\theta) \sin(\theta) d \theta ds dr = \frac{\pi \epsilon_{11}}{8} \int_0^\delta \lambda(r) r^5 dr.
        %%%%%%%%%%%%%%%%%%%%%%%%%%%%%%%%%%%%%
    \end{split} 
\end{equation}

Thus, the peridynamic stress tensor is given by
\begin{equation*}
    \boldsymbol{\sigma}^{\text{peri}} = \left[
    \begin{array}{cc}
    3 A & 0 \\
    0 & A
    \end{array}
    \right],
\end{equation*}
where
\begin{equation*}
    A = \frac{\pi \epsilon_{11}}{8} \int_0^\delta \lambda(r) r^5 dr.
\end{equation*}
Under the identical (static) homogeneous deformation, $\varepsilon_{11} = \epsilon_{11}, \varepsilon_{22} = 0$, and $\varepsilon_{12} = 0$, the classical stress tensor ({\em cf.} \eqref{eqn:2DElastTensor}) for an isotropic material ({\em cf.} \eqref{eqn:isotensor}) is given by 
\begin{equation*}
    \boldsymbol{\sigma} = \epsilon_{11} \left[
    \begin{array}{cc}
    C_{1111} & 0 \\
    0 & C_{1122}
    \end{array}
    \right].
\end{equation*}
Equating the peridynamic and classical stress tensors, we obtain the restriction $C_{1122} = \frac{1}{3} C_{1111}$, which is exactly Cauchy's relation for an isotropic elasticity tensor in two dimensions ({\em cf.} \eqref{eqn:IsoTensorCauchy}). From Table \ref{tab:CauchyRelationsTechnicalConstants}, we know that for two-dimensional isotropic materials satisfying Cauchy's relation, the Poisson's ratio $\nu = \frac{1}{3}$. 

While our derivation above relies on the introduction of the peridynamic traction, in Section \ref{sec:LinearElastPeri} we present a derivation, independent of the definition of peridynamic traction, that Cauchy's relations are imposed on bond-based peridynamic models. Consequently, $\nu = \frac{1}{3}$ is an intrinsic restriction on isotropic two-dimensional bond-based peridynamic models.  

%%%%%%%%%%%%%%%%%%%%%%%%%%%%%%%%%%%%%%%%%%%%%%%%%%%%%%
%%%%%%%%%%%%%%%%%%%%%%%%%%%%%%%%%%%%%%%%%%%%%%%%%%%%%%
%%%%%%%%%%%%%%%%%%%%%%%%%%%%%%%%%%%%%%%%%%%%%%%%%%%%%%

\section{Proof of lemma for classical plane stress}

\subsection{Proof of Lemma \ref{lem:planestressymm}}\label{Appendix:classplanestresssymm}
 
\begin{proof}

Let $\bfu(\bfx,t)$ satisfy Assumption (C$\sigma$\ref{assump:PSsSymm}). Then,
\begin{equation}\label{eqn:usymmclassicalplanestress}
\bfv(\bfx,t) := \left[ \begin{array}{c} u_1(x,y,-z,t) \\ u_2(x,y,-z,t) \\ -u_3(x,y,-z,t) \end{array} \right]
\end{equation} 
and $\bfu(\bfx,t)$ satisfy the same initial and boundary conditions. We show $\bfv$ satisfies \eqref{eqn:eqnofmotionclassical} whenever $\bfu$ satisfies  \eqref{eqn:eqnofmotionclassical}. Since the solution to \eqref{eqn:eqnofmotionclassical} with appropriately defined initial and boundary conditions is unique, we may conclude $\bfv(\bfx,t) = \bfu(\bfx,t)$, and the result of the lemma follows.
Imposing (C$\sigma$\ref{assump:PSs6}) on the equation of motion \eqref{eqn:eqnofmotionclassical}, we obtain (for brevity, we drop the dependence on $x,y,$ and $t$)
\begin{equation}\label{eqn:usymmclassicalplanestressfulleqns}
\begin{split}
\rho(z) \ddot{u}_1(z) =& C_{1111}\frac{\partial^2 u_1}{\partial x^2} (z) + 2 C_{1112}\frac{\partial^2 u_1}{\partial x \partial y} (z) + C_{1212}\frac{\partial^2 u_1}{\partial y^2} (z) + C_{1313}\frac{\partial^2 u_1}{\partial z^2} (z) \\
&+ C_{1112}\frac{\partial^2 u_2}{\partial x^2} (z) +  \left(C_{1122} + C_{1212} \right)\frac{\partial^2 u_2}{\partial x \partial y} (z) + C_{2212}\frac{\partial^2 u_2}{\partial y^2} (z) + C_{2313}\frac{\partial^2 u_2}{\partial z^2} (z) \\
&+ \left(C_{1133} + C_{1313} \right) \frac{\partial^2 u_3}{\partial x \partial z} (z) + \left(C_{3312} + C_{2313} \right) \frac{\partial^2 u_3}{\partial y \partial z} (z) + b_1(z),
\\
%%%%%%%%%%%%%%%%%%%%%%%%%%%%%%%
\rho(z) \ddot{u}_2(z) ={}& C_{1112}\frac{\partial^2 u_1}{\partial x^2} (z) +  \left(C_{1122}+C_{1212} \right) \frac{\partial^2 u_1}{\partial x \partial y} (z) + C_{2212}\frac{\partial^2 u_1}{\partial y^2} (z) + C_{2313}\frac{\partial^2 u_1}{\partial z^2} (z) \\
&+ C_{1212}\frac{\partial^2 u_2}{\partial x^2} (z) + 2 C_{2212} \frac{\partial^2 u_2}{\partial x \partial y} (z) + C_{2222}\frac{\partial^2 u_2}{\partial y^2} (z) + C_{2323}\frac{\partial^2 u_2}{\partial z^2} (z) \\
&+ \left(C_{3312} + C_{2313} \right) \frac{\partial^2 u_3}{\partial x \partial z} (z) + \left(C_{2233} + C_{2323} \right) \frac{\partial^2 u_3}{\partial y \partial z} (z) + b_2(z),
\\
%%%%%%%%%%%%%%%%%%%%%%%%%%%%%%%
\rho(z) \ddot{u}_3(z) ={}& \left(C_{1133}+C_{1313} \right) \frac{\partial^2 u_1}{\partial x \partial z} (z) +  \left(C_{3312}+C_{2313} \right) \frac{\partial^2 u_1}{\partial y \partial z} (z)  \\
&+ \left( C_{3312} + C_{2313} \right) \frac{\partial^2 u_2}{\partial x \partial z} (z) + \left(C_{2233} + C_{2323} \right) \frac{\partial^2 u_2}{\partial y \partial z} (z) \\
&+C_{1313}\frac{\partial^2 u_3}{\partial x^2} (z) +  2C_{2313} \frac{\partial^2 u_3}{\partial x \partial y} (z) + C_{2323}\frac{\partial^2 u_3}{\partial y^2} (z) + C_{3333}\frac{\partial^2 u_3}{\partial z^2} (z) + b_3(z).
\end{split}
\end{equation}

By Assumption (C$\sigma$\ref{assump:PSs1}), we may perform the substitution $z \rightarrow -z$ in \eqref{eqn:usymmclassicalplanestressfulleqns}. In addition, we impose Assumptions (C$\sigma$\ref{assump:PSs5}) and (C$\sigma$\ref{assump:PSs2}) on \eqref{eqn:usymmclassicalplanestressfulleqns}, and then we multiply the equation for $u_3$ by negative one, to obtain

\begin{equation}\label{eqn:usymmclassicalplanestressfulleqns2}
\begin{split}
\rho \ddot{u}_1(-z) ={}& C_{1111}\frac{\partial^2 u_1}{\partial x^2} (-z) + 2 C_{1112}\frac{\partial^2 u_1}{\partial x \partial y} (-z) + C_{1212}\frac{\partial^2 u_1}{\partial y^2} (-z) + C_{1313}\frac{\partial^2 u_1}{\partial z^2} (-z) \\
&+ C_{1112}\frac{\partial^2 u_2}{\partial x^2} (-z) +  \left(C_{1122} + C_{1212} \right)\frac{\partial^2 u_2}{\partial x \partial y} (-z) + C_{2212}\frac{\partial^2 u_2}{\partial y^2} (-z) + C_{2313}\frac{\partial^2 u_2}{\partial z^2} (-z) \\
&+ \left(C_{1133} + C_{1313} \right) \frac{\partial^2 u_3}{\partial x \partial z} (-z) + \left(C_{3312} + C_{2313} \right) \frac{\partial^2 u_3}{\partial y \partial z} (-z) + b_1(z),
\\
%%%%%%%%%%%%%%%%%%%%%%%%%%
\rho \ddot{u}_2(-z) ={}& C_{1112}\frac{\partial^2 u_1}{\partial x^2} (-z) +  \left(C_{1122}+C_{1212} \right) \frac{\partial^2 u_1}{\partial x \partial y} (-z) + C_{2212}\frac{\partial^2 u_1}{\partial y^2} (-z) + C_{2313}\frac{\partial^2 u_1}{\partial z^2} (-z) \\
&+ C_{1212}\frac{\partial^2 u_2}{\partial x^2} (-z) + 2 C_{2212} \frac{\partial^2 u_2}{\partial x \partial y} (-z) + C_{2222}\frac{\partial^2 u_2}{\partial y^2} (-z) + C_{2323}\frac{\partial^2 u_2}{\partial z^2} (-z) \\
&+ \left(C_{3312} + C_{2313} \right) \frac{\partial^2 u_3}{\partial x \partial z} (-z) + \left(C_{2233} + C_{2323} \right) \frac{\partial^2 u_3}{\partial y \partial z} (-z) + b_2(z),
\\
%%%%%%%%%%%%%%%%%%%%%%%%%%
-\rho \ddot{u}_3(-z) ={}& -\left(C_{1133}+C_{1313} \right) \frac{\partial^2 u_1}{\partial x \partial z} (-z) -\left(C_{3312}+C_{2313} \right) \frac{\partial^2 u_1}{\partial y \partial z} (-z)  \\
& -\left( C_{3312} + C_{2313} \right) \frac{\partial^2 u_2}{\partial x \partial z} (-z) -\left(C_{2233} + C_{2323} \right) \frac{\partial^2 u_2}{\partial y \partial z} (-z) \\
&-C_{1313}\frac{\partial^2 u_3}{\partial x^2} (-z) -  2C_{2313} \frac{\partial^2 u_3}{\partial x \partial y} (-z) - C_{2323}\frac{\partial^2 u_3}{\partial y^2} (-z) - C_{3333}\frac{\partial^2 u_3}{\partial z^2} (-z).
\end{split}
\end{equation}
By noticing
\begin{equation*}
\frac{\partial^2}{\partial x_j \partial x_l} u_i(x,y,-z,t) = \left\{
\begin{array}{ll}
\frac{\partial^2}{\partial x_j \partial x_l} u_i(x,y,-z,t), & j,l \neq 3 \text{ or } j=l=3 \\
-\frac{\partial^2}{\partial x_j \partial x_l} u_i(x,y,-z,t), & \text{else}
\end{array}
\right. ,
\end{equation*}
and recalling \eqref{eqn:usymmclassicalplanestress}, we may rewrite \eqref{eqn:usymmclassicalplanestressfulleqns2} as
\begin{equation}\label{eqn:usymmclassicalplanestressfulleqns4}
\begin{split}
\rho \ddot{v}_1(z) ={}& C_{1111}\frac{\partial^2 v_1}{\partial x^2} (z) + 2 C_{1112}\frac{\partial^2 v_1}{\partial x \partial y} (z) + C_{1212}\frac{\partial^2 v_1}{\partial y^2} (z) + C_{1313}\frac{\partial^2 v_1}{\partial z^2} (z) \\
&+ C_{1112}\frac{\partial^2 v_2}{\partial x^2} (z) +  \left(C_{1122} + C_{1212} \right)\frac{\partial^2 v_2}{\partial x \partial y} (z) + C_{2212}\frac{\partial^2 v_2}{\partial y^2} (z) + C_{2313}\frac{\partial^2 v_2}{\partial z^2} (z) \\
&+ \left(C_{1133} + C_{1313} \right) \frac{\partial^2 v_3}{\partial x \partial z} (z) + \left(C_{3312} + C_{2313} \right) \frac{\partial^2 v_3}{\partial y \partial z} (z) + b_1(z), \\
%%%%%%%%%%%%%%%%%%%%%%%%%%%%
\rho \ddot{v}_2(z) ={}& C_{1112}\frac{\partial^2 v_1}{\partial x^2} (z) +  \left(C_{1122}+C_{1212} \right) \frac{\partial^2 v_1}{\partial x \partial y} (z) + C_{2212}\frac{\partial^2 v_1}{\partial y^2} (z) + C_{2313}\frac{\partial^2 v_1}{\partial z^2} (z) \\
&+ C_{1212}\frac{\partial^2 v_2}{\partial x^2} (z) + 2 C_{2212} \frac{\partial^2 v_2}{\partial x \partial y} (z) + C_{2222}\frac{\partial^2 v_2}{\partial y^2} (z) + C_{2323}\frac{\partial^2 v_2}{\partial z^2} (z) \\
&+ \left(C_{3312} + C_{2313} \right) \frac{\partial^2 v_3}{\partial x \partial z} (z) + \left(C_{2233} + C_{2323} \right) \frac{\partial^2 v_3}{\partial y \partial z} (z) + b_2(z), \\
%%%%%%%%%%%%%%%%%%%%%%%%%%%%
\rho \ddot{v}_3(z) ={}& \left(C_{1133}+C_{1313} \right) \frac{\partial^2 v_1}{\partial x \partial z} (z) +\left(C_{3312}+C_{2313} \right) \frac{\partial^2 v_1}{\partial y \partial z} (z)  \\
&+\left( C_{3312} + C_{2313} \right) \frac{\partial^2 v_2}{\partial x \partial z} (z) +\left(C_{2233} + C_{2323} \right) \frac{\partial^2 v_2}{\partial y \partial z} (z) \\
&+C_{1313}\frac{\partial^2 v_3}{\partial x^2} (z) +  2C_{2313} \frac{\partial^2 v_3}{\partial x \partial y} (z) + C_{2323}\frac{\partial^2 v_3}{\partial y^2} (z) + C_{3333}\frac{\partial^2 u_3}{\partial z^2} (z). 
\end{split}
\end{equation}

Comparing \eqref{eqn:usymmclassicalplanestressfulleqns} and \eqref{eqn:usymmclassicalplanestressfulleqns4}, we see under the assumptions of Lemma \ref{lem:planestressymm}, $\bfv$ satisfies \eqref{eqn:eqnofmotionclassical} whenever $\bfu$ does. \qed
\end{proof}

%%%%%%%%%%%%%%%%%%%%%%%%%%%%%%%%%%%%%%%%%%%%%%%
%%%%%%%%%%%%%%%%%%%%%%%%%%%%%%%%%%%%%%%%%%%%%%%
%%%%%%%%%%%%%%%%%%%%%%%%%%%%%%%%%%%%%%%%%%%%%%%

\section{Proofs of lemmas for peridynamic plane stress}\label{Sec:AppPeriPlaneStressLemmas}

\subsection{Proof of Lemma~\ref{lem:symmplanestress}}\label{appendix:symmplanestress}

\begin{proof}
Let $\bfu(\bfx,t)$ satisfies Assumption (P$\sigma$\ref{assump:PSsSymmPeri}). Then,
\begin{equation}\label{eqn:usymmperiplanestress}
\bfv(\bfx,t) := \left[ \begin{array}{c} u_1(x,y,-z,t) \\ u_2(x,y,-z,t) \\ -u_3(x,y,-z,t) \end{array} \right]
\end{equation} 
and $\bfu(\bfx,t)$ satisfy the same boundary and initial conditions. We show $\bfv$ satisfies \eqref{eqn:linearperieqn} whenever $\bfu$ satisfies  \eqref{eqn:linearperieqn}. Since the solution to \eqref{eqn:linearperieqn} with appropriately defined boundary and initial conditions is unique, we may conclude $\bfv(\bfx,t) = \bfu(\bfx,t)$ and the result of the lemma follows.

Imposing Assumptions (P$\sigma$\ref{assump:PSs1Peri}) and (P$\sigma$\ref{assump:PSs3Peri}) on \eqref{eqn:linearperieqn}, we obtain
\begin{equation}\label{eqn:planestressinitialeqnmotionlemmasymmetry}
\begin{split}
\rho(\bfx) \ddot{u}_i(\bfx,t) ={}& \int_{\mathcal{H}_\bfx} \lambda(\bfx',\bfx) (x_i'-x_i)(x_j'-x_j) (u_j(\bfx',t)-u_j(\bfx,t)) d \bfx' + b_i(\bfx,t) \\
={}& \int_{-h}^h \int_{B^{2D}_{r}(x,y)} \lambda(\bfx',\bfx) (x_i'-x_i)(x_j'-x_j) (u_j(\bfx',t)-u_j(\bfx,t)) d \bfx' + b_i(\bfx,t),
\end{split}
\end{equation} 
where $r = \sqrt{\delta^2 - (z'-z)^2}$ and $B_{r}^{2D}(\bfx)$ is the two-dimensional ball of radius $r$ centered at $\bfx$. For brevity, we omit the functional dependence of $\rho$, $\bfu$, $\lambda$, and $\bfb$ on $x,y,x',y',$ and~$t$. By Assumption (P$\sigma$\ref{assump:PSs1Peri}), \eqref{eqn:planestressinitialeqnmotionlemmasymmetry} also holds for $-z$ (since it holds all $z \in [-h,h]$). Substitute $z \rightarrow -z$ and then apply the change of variables $z' \rightarrow -z'$ to \eqref{eqn:planestressinitialeqnmotionlemmasymmetry} to obtain (note $r$ is unchanged):
\begin{equation}\label{eqn:planestressinitialeqnmotionlemmasymmetryterm1}
\begin{split}
\rho(-z) \ddot{u}_1(-z) ={}& \int_{-h}^h \int_{B^{2D}_{r}(x,y)} \lambda(-z',-z) (x'-x)^2 (u_1(-z')-u_1(-z)) d \bfx' \\
&+ \int_{-h}^h \int_{B^{2D}_{r}(x,y)} \lambda(-z',-z) (x'-x)(y'-y) (u_2(-z')-u_2(-z)) d \bfx' \\
&- \int_{-h}^h \int_{B^{2D}_{r}(x,y)} \lambda(-z',-z) (x'-x)(z'-z) (u_3(-z')-u_3(-z)) d \bfx' + b_1(-z), \\
\rho(-z) \ddot{u}_2(-z) ={}& \int_{-h}^h \int_{B^{2D}_{r}(x,y)} \lambda(-z',-z) (y'-y)(x'-x) (u_1(-z')-u_1(-z)) d \bfx' \\
&+ \int_{-h}^h \int_{B^{2D}_{r}(x,y)} \lambda(-z',-z) (y'-y)^2 (u_2(-z')-u_2(-z)) d \bfx' \\
&- \int_{-h}^h \int_{B^{2D}_{r}(x,y)} \lambda(-z',-z) (y'-y)(z'-z) (u_3(-z')-u_3(-z)) d \bfx' + b_2(-z), \\
\rho(-z) \ddot{u}_3(-z) ={}& -\int_{-h}^h \int_{B^{2D}_{r}(x,y)} \lambda(-z',-z) (z'-z)(x'-x) (u_1(-z')-u_1(-z)) d \bfx' \\
&- \int_{-h}^h \int_{B^{2D}_{r}(x,y)} \lambda(-z',-z) (z'-z)(y'-y) (u_2(-z')-u_2(-z)) d \bfx' \\
&+ \int_{-h}^h \int_{B^{2D}_{r}(x,y)} \lambda(-z',-z) (z'-z)^2 (u_3(-z')-u_3(-z)) d \bfx' + b_3(-z). \\
\end{split}
\end{equation}
Next, we apply Assumptions (P$\sigma$\ref{assump:PSs5Peri}), (P$\sigma$\ref{assump:PSs2Peri}), and (P$\sigma$\ref{assump:PSs7Peri}) to \eqref{eqn:planestressinitialeqnmotionlemmasymmetryterm1} to obtain
\begin{equation}\label{eqn:planestressinitialeqnmotionlemmasymmetryterm2}
\begin{split}
\rho \ddot{u}_1(-z) ={}& \int_{-h}^h \int_{B^{2D}_{r}(x,y)} \lambda(z',z) (x'-x)^2 (u_1(-z')-u_1(-z)) d \bfx' \\
&+ \int_{-h}^h \int_{B^{2D}_{r}(x,y)} \lambda(z',z) (x'-x)(y'-y) (u_2(-z')-u_2(-z)) d \bfx' \\
&- \int_{-h}^h \int_{B^{2D}_{r}(x,y)} \lambda(z',z) (x'-x)(z'-z) (u_3(-z')-u_3(-z)) d \bfx' + b_1(z), \\
\rho \ddot{u}_2(-z) ={}& \int_{-h}^h \int_{B^{2D}_{r}(x,y)} \lambda(z',z) (y'-y)(x'-x) (u_1(-z')-u_1(-z)) d \bfx' \\
&+ \int_{-h}^h \int_{B^{2D}_{r}(x,y)} \lambda(z',z) (y'-y)^2 (u_2(-z')-u_2(-z)) d \bfx' \\
&- \int_{-h}^h \int_{B^{2D}_{r}(x,y)} \lambda(z',z) (y'-y)(z'-z) (u_3(-z')-u_3(-z)) d \bfx' + b_2(z), \\
\rho \ddot{u}_3(-z) ={}& - \int_{-h}^h \int_{B^{2D}_{r}(x,y)} \lambda(z',z) (z'-z)(x'-x) (u_1(-z')-u_1(-z)) d \bfx' \\
&- \int_{-h}^h \int_{B^{2D}_{r}(x,y)} \lambda(z',z) (z'-z)(y'-y) (u_2(-z')-u_2(-z)) d \bfx' \\
&+ \int_{-h}^h \int_{B^{2D}_{r}(x,y)} \lambda(z',z) (z'-z)^2 (u_3(-z')-u_3(-z)) d \bfx'. \\
\end{split}
\end{equation}
Recalling \eqref{eqn:usymmperiplanestress}, we may rewrite \eqref{eqn:planestressinitialeqnmotionlemmasymmetryterm2} as 
\begin{equation}\label{eqn:planestressinitialeqnmotionlemmasymmetryterm3}
\begin{split}
\rho \ddot{v}_1(z) ={}& \int_{-h}^h \int_{B^{2D}_{r}(x,y)} \lambda(z',z) (x'-x)^2 (v_1(z')-v_1(z)) d \bfx' \\
&+ \int_{-h}^h \int_{B^{2D}_{r}(x,y)} \lambda(z',z) (x'-x)(y'-y) (v_2(z')-v_2(z)) d \bfx' \\
&+ \int_{-h}^h \int_{B^{2D}_{r}(x,y)} \lambda(z',z) (x'-x)(z'-z) (v_3(z')-v_3(z)) d \bfx' + b_1(z), \\
\rho \ddot{v}_2(z) ={}& \int_{-h}^h \int_{B^{2D}_{r}(x,y)} \lambda(z',z) (y'-y)(x'-x) (v_1(z')-u_1(z)) d \bfx' \\
&+ \int_{-h}^h \int_{B^{2D}_{r}(x,y)} \lambda(z',z) (y'-y)^2 (v_2(z')-v_2(z)) d \bfx' \\
&+ \int_{-h}^h \int_{B^{2D}_{r}(x,y)} \lambda(z',z) (y'-y)(z'-z) (v_3(z')-v_3(z)) d \bfx' + b_2(z), \\
\rho \ddot{v}_3(z) ={}& \int_{-h}^h \int_{B^{2D}_{r}(x,y)} \lambda(z',z) (z'-z)(x'-x) (v_1(z')-v_1(z)) d \bfx' \\
&+ \int_{-h}^h \int_{B^{2D}_{r}(x,y)} \lambda(z',z) (z'-z)(y'-y) (v_2(z')-v_2(z)) d \bfx' \\
&+ \int_{-h}^h \int_{B^{2D}_{r}(x,y)} \lambda(z',z) (z'-z)^2 (v_3(z')-v_3(z)) d \bfx' . \\
\end{split}
\end{equation}
Comparing \eqref{eqn:planestressinitialeqnmotionlemmasymmetryterm3} with \eqref{eqn:planestressinitialeqnmotionlemmasymmetry}, we see under the assumptions of Lemma \ref{lem:symmplanestress}, $\bfv$ satisfies \eqref{eqn:linearperieqn} whenever $\bfu$ does. \qed

\end{proof}

%%%%%%%%%%%%%%%%%%%%%%%%%%%%%%%%%%%%%%%%%%%%%%%%%%%%%%%
%%%%%%%%%%%%%%%%%%%%%%%%%%%%%%%%%%%%%%%%%%%%%%%%%%%%%%%
%%%%%%%%%%%%%%%%%%%%%%%%%%%%%%%%%%%%%%%%%%%%%%

\subsection{Proof of Lemma~\ref{lemma:zeroavgtraction}}\label{lemma:zeroavgtractionapp}

\begin{proof}

We begin by writing out the limits of integration for the third component of the peridynamic traction \eqref{eqn:peritraction}:
\begin{equation*}
\tau_3(\bfx,t,\bfe_3) = \int_0^\delta \int_{z}^{z-s+\delta}  \int_{B^{\text{2D}}_{\hat{r}}(x,y)}  f_3 (\bfu(\bfx',t)-\bfu(\bfx-s \bfe_3,t),\bfx',(\bfx-s \bfe_3)) d \bfx' ds,
\end{equation*}
where $\hat{r} = \sqrt{\delta^2-(z'-(z-s))^2}$. Next, we apply the change of variable $s \rightarrow z- \hat{z}$ to obtain
\begin{equation*}
\tau_3 (\bfx,t,\bfe_3) = \int_{z-\delta}^{z} \int_{z}^{\hat{z}+\delta}  \int_{B^{\text{2D}}_{\hat{r}}(x,y)}  f_3 (\bfu(\bfx',t)-\bfu(\hat{\bfx},t),\bfx', \hat{\bfx} ) d \bfx' d\hat{z},
\end{equation*}
where $\hat{\bfx} = \left[ \begin{array}{c} x \\ y \\ \hat{z} \end{array} \right]$ and $\hat{r} = \sqrt{\delta^2 - (z'-\hat{z})^2}$. Taking the average of the stress $\tau_3(\bfx,t,\bfe_3)$ over the thickness of the plate, we find
\begin{equation}\label{eqn:avgstressinitial}
\overline{\tau_3 (\bfx,t,\bfe_3)} = \frac{1}{2h} \int_{-h}^h \int_{z-\delta}^{z} \int_{z}^{\hat{z}+\delta}  \int_{B^{\text{2D}}_{\hat{r}}(x,y)}  f_3(\bfu(\bfx',t)-\bfu(\hat{\bfx},t),\bfx',\hat{\bfx} ) d \bfx' d\hat{z} dz.
\end{equation}
Appealing to Assumptions (P$\sigma$\ref{assump:PSs1Peri}) and (P$\sigma$\ref{assump:PSs3Peri}), we see the integrand of \eqref{eqn:avgstressinitial} is null unless $|\hat{z}| \leqslant h$ and $|z| \leqslant h$. Thus, we may rewrite \eqref{eqn:avgstressinitial} as (noting $h < \hat{z}+\delta$ and $-h > z-\delta$)
\begin{equation}\label{eqn:avgstressnoexteral}
\overline{\tau_3(\bfx,t,\bfe_3)} = \frac{1}{2h} \int_{-h}^h \int_{-h}^{z} \int_{z}^{h}  \int_{B^{\text{2D}}_{\hat{r}}(x,y)}  f_3(\bfu(\bfx',t)-\bfu(\hat{\bfx},t),\bfx', \hat{\bfx} ) d \bfx' d\hat{z} dz.
\end{equation}
Next, we change the order of integration so that the integral in $z$ comes first. For clarity, we do this in steps by first changing the order of integration between $\hat{z}$ and $z$. Viewing Figure~\ref{Fig:FirstChangeIntegrationAverageStress}, we deduce the new limits of integration after changing the order of integration between $\hat{z}$ and~$z$: 
\begin{equation}\label{eqn:avgstressfirstchange}
\overline{\tau_3(\bfx,t,\bfe_3)} = \frac{1}{2h} \int_{-h}^h \int_{\hat{z}}^{h} \int_{z}^{h}  \int_{B^{\text{2D}}_{\hat{r}}(x,y)}  f_3 (\bfu(\bfx',t)-\bfu(\hat{\bfx},t),\bfx', \hat{\bfx} ) d \bfx' dz d\hat{z}. 
\end{equation}

\begin{figure}
\begin{center}
\begin{tikzpicture}[scale=0.8]
\draw[help lines, color=gray!30, dashed] (-4.9,-4.9) grid (4.9,4.9);
\draw[->,ultra thick] (-5,0)--(5,0) node[right]{$z$};
\draw[->,ultra thick] (0,-5)--(0,5) node[above]{$\hat{z}$};
\draw[-,ultra thick] (-0.2,4)--(0.2,4) node[right]{$h$};
\draw[-,ultra thick] (-0.2,-4)--(0.2,-4) node[right]{$-h$};
\draw[-,ultra thick] (4,-0.2)--(4,0.2) node[above]{$h$};
\draw[-,ultra thick] (-4,-0.2)--(-4,0.2) node[above]{$-h$};

\path[fill=blue!70,opacity = 0.3] (-4,-4) -- (4,-4) -- (4,4) --cycle;
\end{tikzpicture}
\caption{The $z$ and $\hat{z}$ region of integration for $\overline{\tau_3(\bfx,t,\bfe_3)}$ in \eqref{eqn:avgstressnoexteral}.} \label{Fig:FirstChangeIntegrationAverageStress}
\end{center}
\end{figure}
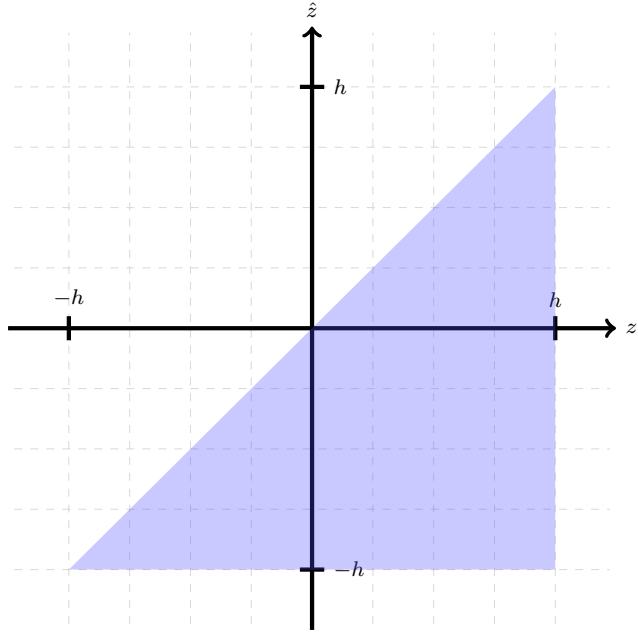

Viewing Figure \ref{Fig:SecondChangeIntegrationAverageStress}, we deduce the new limits of integration after changing the order of integration between $z'$ and $z$ in \eqref{eqn:avgstressfirstchange}:  
\begin{equation}\label{eqn:avgstresssecondchange}
\overline{\tau_3(\bfx,t,\bfe_3)} = \frac{1}{2h} \int_{-h}^h \int_{\hat{z}}^{h} \int_{\hat{z}}^{z'}  \int_{B^{\text{2D}}_{\hat{r}}(x,y)} f_3(\bfu(\bfx',t)-\bfu(\hat{\bfx},t),\bfx', \hat{\bfx} ) dx' dy' dz dz' d\hat{z}. 
\end{equation}

\begin{figure}
\begin{center}
\begin{tikzpicture}[scale=0.8]
\draw[help lines, color=gray!30, dashed] (-4.9,-4.9) grid (4.9,4.9);
\draw[->,ultra thick] (-5,0)--(5,0) node[right]{$z$};
\draw[->,ultra thick] (0,-5)--(0,5) node[above]{$z'$};
\draw[-,ultra thick] (-0.2,4)--(0.2,4) node[right]{$h$};
\draw[-,ultra thick] (-0.2,-4)--(0.2,-4) node[right]{$-h$};
\draw[-,ultra thick] (4,-0.2)--(4,0.2) node[above]{$h$};
\draw[-,ultra thick] (-4,-0.2)--(-4,0.2) node[above]{$-h$};
\draw[-,ultra thick] (-0.2,-2)--(0.2,-2) node[right]{$\hat{z}$};
\draw[-,ultra thick] (-2,-0.2)--(-2,0.2) node[above]{$\hat{z}$};

\path[fill=blue!70,opacity = 0.3] (-2,-2) -- (-2,4) -- (4,4) --cycle;
\end{tikzpicture}
\caption{The $z$ and $z'$ region of integration for $\overline{\tau_3(\bfx,t,\bfe_3)}$ in \eqref{eqn:avgstressfirstchange}.} \label{Fig:SecondChangeIntegrationAverageStress}
\end{center}
\end{figure}
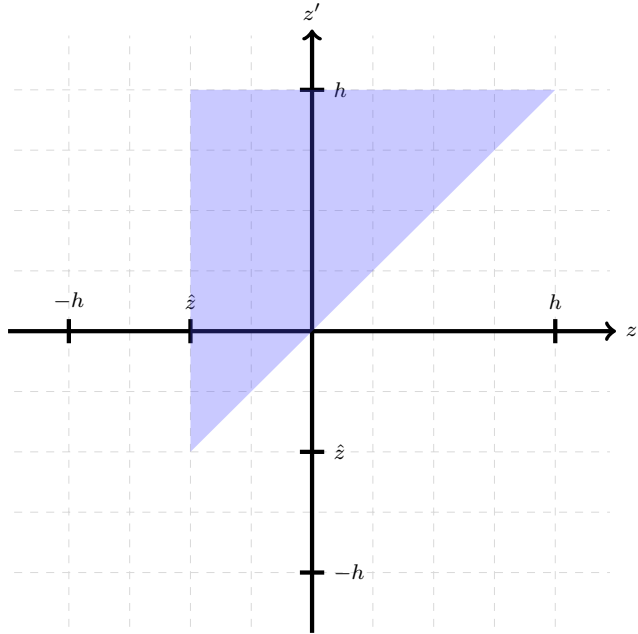

Since $B^{\text{2D}}_{\hat{r}}(x,y)$ is independent of $z$ we may integrate over $z$ first to obtain
\begin{equation}\label{eqn:avgplanestressprechangofvariablestop}
\begin{split}
\overline{\tau_3(\bfx,t,\bfe_3)} ={}& \frac{1}{2h} \int_{-h}^h \int_{\hat{z}}^{h}  \int_{B^{\text{2D}}_{\hat{r}}(x,y)} \int_{\hat{z}}^{z'} f_3 (\bfu(\bfx',t)-\bfu(\hat{\bfx},t),\bfx',\hat{\bfx} ) dz dx' dy' dz' d\hat{z}\\
={}& \frac{1}{2h} \int_{-h}^h \int_{\hat{z}}^{h}  \int_{B^{\text{2D}}_{\hat{r}}(x,y)} (z'-\hat{z}) f_3 (\bfu(\bfx',t)-\bfu(\hat{\bfx},t),\bfx', \hat{\bfx} ) dx' dy' dz' d\hat{z}.
 \end{split} 
\end{equation}
Note that 
\begin{equation*}
    (z'-\hat{z}) f_3(\bfu(\bfx',t)-\bfu(\hat{\bfx},t),\bfx',\hat{\bfx}) = \lambda(\bfx',\hat{\bfx}) (z'-\hat{z})^2 (x_j' -\hat{x}_j) (u_j(\bfx',t)-u_j(\bfx,t))
\end{equation*} 
is invariant under the transformation $(z',\hat{z}) \rightarrow (-z',-\hat{z})$ by Lemma \ref{lem:symmplanestress} and \eqref{assump:MonoclinicPlaneStress}. Thus, applying the change of variables $z' \rightarrow -z'$ and $\hat{z} \rightarrow -\hat{z}$ to \eqref{eqn:avgplanestressprechangofvariablestop} results in 
\begin{equation}\label{eqn:avgplanestressprechangofvariablesbot}
\begin{split}
\overline{\tau_3(\bfx,t,\bfe_3)} ={}& \frac{1}{2 h} \int_{-h}^h \int_{-h}^{\hat{z}}  \int_{B^{\text{2D}}_{\hat{r}}(x,y)} (z'-\hat{z}) f_3 (\bfu(\bfx',t)-\bfu(\hat{\bfx},t),\bfx', \hat{\bfx} ) dx' dy' dz' d\hat{z}.
 \end{split} 
\end{equation}
We sum \eqref{eqn:avgplanestressprechangofvariablestop} and \eqref{eqn:avgplanestressprechangofvariablesbot} to obtain
\begin{equation}\label{eqn:avgplanestressprechangofvariablesfinal}
\begin{split}
2\overline{\tau_3(\bfx,t,\bfe_3)} ={}& \frac{1}{2 h} \int_{-h}^h \int_{-h}^{h}  \int_{B^{\text{2D}}_{\hat{r}}(x,y)} (z'-\hat{z}) f_3 (\bfu(\bfx',t)-\bfu(\hat{\bfx},t),\bfx', \hat{\bfx} ) dx' dy' dz' d\hat{z}.
 \end{split} 
\end{equation}
Since $\lambda(\bfx',\hat{\bfx}) = 0$ when $\| \bfx' - \hat{\bfx} \| \geqslant \delta$, we may extend the region of integration so that 
\begin{equation}\label{eqn:avgplanestressprechangofvariablesfinalfinal}
\begin{split}
2\overline{\tau_3(\bfx,t,\bfe_3)} ={}& \frac{1}{2 h} \int_{-h}^h \int_{-h}^{h}  \int_{B^{\text{2D}}_{\delta}(x,y)} (z'-\hat{z}) f_3 (\bfu(\bfx',t)-\bfu(\hat{\bfx},t),\bfx', \hat{\bfx} ) dx' dy' dz' d\hat{z}.
 \end{split} 
\end{equation}

Dividing both sides of \eqref{eqn:avgplanestressprechangofvariablesfinalfinal} by $2$ and then relabeling $\hat{z}$ with $z$ completes the proof. \qed

\end{proof}

%%%%%%%%%%%%%%%%%%%%%%%%%%%%%%%%%%%%%%%%%%
%%%%%%%%%%%%%%%%%%%%%%%%%%%%%%%%%%%%%%%%%%
%%%%%%%%%%%%%%%%%%%%%%%%%%%%%%%%%%%%%%%%%%

\subsection{Proof of Lemma~\ref{lem:PlaneStrainu3SubEstimates}}\label{lem:PlaneStrainu3SubEstimatesapp}

\begin{proof}
We start by developing an estimate of $\frac{\partial u_3}{\partial z} (x,y,0,t)$. By Assumption (P$\sigma$\ref{assump:PSs4Peri}), $\overline{\tau_3(\bfx,t,\bfe_3)}~=~0$ and, appealing to Lemma \ref{lemma:zeroavgtraction}, we obtain
\begin{equation}\label{eqn:zerotractioncond}
\begin{split}
&0 = \frac{1}{4h} \int_{-h}^h \int_{-h}^h \int_{B^{2D}_{\delta}(x,y)} \lambda(\bfx',\bfx) \xi^2_3 \xi_j (u_j(\bfx',t)- u_j(\bfx,t)) d \bfx' d z.
\end{split}
\end{equation}
By solving for the integral containing $u_3$ in \eqref{eqn:zerotractioncond}, we find
\begin{equation}\label{eqn:avgperistresssolvedu3}
\begin{split}
&\int_{-h}^h \int_{-h}^h \int_{B^{2D}_{\delta}(x,y)} \lambda(\bfx',\bfx) \xi^3_3 (u_3(\bfx',t)- u_3(\bfx,t)) d \bfx' d z \\ &= -\int_{-h}^h \int_{-h}^h \int_{B^{2D}_{\delta}(x,y)} \lambda(\bfx',\bfx) \xi^2_3 \left[ \xi_1 (u_1(\bfx',t)- u_1(\bfx,t))  + \xi_2 (u_2(\bfx',t)- u_2(\bfx,t)) \right]  d \bfx' d z.
\end{split}
\end{equation}
Recalling Assumption (P$\sigma$\ref{assump:PSs6Peri}), we perform a second-order Taylor expansion in $u_3(\bfx',t)$ and $u_3(\bfx,t)$ about the point $(x,y,0,t)$ on the left-hand side of \eqref{eqn:avgperistresssolvedu3} to get
\begin{equation}\label{eqn:taylorexpandu3subterm1}
\begin{split}
&\int_{-h}^h \int_{-h}^h \int_{B^{2D}_{\delta}(x,y)} \lambda(\bfx',\bfx) \xi^3_3 (u_3(\bfx',t)- u_3(\bfx,t)) d \bfx' d z \\
&\approx \int_{-h}^h \int_{-h}^h \int_{B^{2D}_{\delta}(x,y)} \lambda(\bfx',\bfx) \xi^3_3  \left[ u_3 + \frac{\partial u_3}{\partial x} \xi_1 + \frac{\partial u_3}{\partial y} \xi_2 + \frac{\partial u_3}{\partial z} z'  \right. \\
&\hspace*{1.7in} + \frac{1}{2} \frac{\partial^2 u_3}{\partial x^2} \xi_1^2 + \frac{1}{2} \frac{\partial^2 u_3}{\partial y^2} \xi_2^2 + \frac{1}{2} \frac{\partial^2 u_3}{\partial z^2} z'^2 + \frac{\partial^2 u_3}{\partial x \partial y} \xi_1 \xi_2 + \frac{\partial^2 u_3}{\partial x \partial z} \xi_1 z' + \frac{\partial^2 u_3}{\partial y \partial z} \xi_2 z'  \\
&\hspace*{1.7in} \left. - \left( u_3 + \frac{\partial u_3}{\partial z} z + \frac{1}{2} \frac{\partial^2 u_3}{\partial z^2} z^2 \right) \right] d \bfx' d z,
\end{split}
\end{equation}
where each occurrence of $u_3$ or its derivatives is evaluated at $(x,y,0,t)$. Note that many of the terms in \eqref{eqn:taylorexpandu3subterm1} are antisymmetric in $(z,z')$, i.e., $g(z,z') = -g(-z,-z')$, by Assumption~(P$\sigma$\ref{assump:PSs7Peri}). Since we are integrating over the symmetric domain $[-h,h] \times [-h,h]$, those terms are negated under integration and \eqref{eqn:taylorexpandu3subterm1} simplifies to
\begin{equation}\label{eqn:taylorexpandu3subterm2}
\begin{split}
&\int_{-h}^h \int_{-h}^h \int_{B^{2D}_{\delta}(x,y)} \lambda(\bfx',\bfx) \xi^3_3 (u_3(\bfx',t)- u_3(\bfx,t)) d \bfx' d z \\
&\approx \int_{-h}^h \int_{-h}^h \int_{B^{2D}_{\delta}(x,y)} \lambda(\bfx',\bfx) \xi^3_3 \left( \frac{\partial u_3}{\partial z} (z' -z) + \frac{\partial^2 u_3}{\partial x \partial z} \xi_1 z' + \frac{\partial^2 u_3}{\partial y \partial z} \xi_2 z' \right) d \bfx' d z. \\
\end{split}
\end{equation}
Perform the change of variables $(x',y') \rightarrow (x + \xi_1,y+\xi_2)$ and recall \eqref{eqn:zposaware} to obtain
\begin{equation}\label{eqn:taylorexpandu3subterm3}
\begin{split}
&\int_{-h}^h \int_{-h}^h \int_{B^{2D}_{\delta}(x,y)} \lambda(\bfx',\bfx) \xi^3_3 (u_3(\bfx',t)- u_3(\bfx,t)) d \bfx' d z \\
&\approx \int_{-h}^h \int_{-h}^h \int_{B^{2D}_{\delta}(\bf0)} \lambda(\xi_1,\xi_2,z',z) \xi^3_3 \left( \frac{\partial u_3}{\partial z} \xi_3 +  \frac{\partial^2 u_3}{\partial x \partial z} \xi_1 z' + \frac{\partial^2 u_3}{\partial y \partial z} \xi_2 z' \right) d \xi_1 d \xi_2 dz' d z. \\
\end{split}
\end{equation}

From \eqref{eqn:lambdasymm} and \eqref{assump:MonoclinicPlaneStress} we have
\begin{equation}\label{eqn:lambdaeveninxi1xi2}
\lambda(\xi_1,\xi_2,z',z) = \lambda(-\xi_1,-\xi_2,z,z') = \lambda(-\xi_1,-\xi_2,z',z),
\end{equation}
i.e., $\lambda(\xi_1,\xi_2,z',z)$ is invariant under the transformation $(\xi_1,\xi_2) \rightarrow (-\xi_1,-\xi_2)$. Since the limits of integration in $(\xi_1,\xi_2)$ are over $B_\delta^{2D}(\bf0)$, we may employ antisymmetry to further reduce \eqref{eqn:taylorexpandu3subterm3} to

\begin{equation}\label{eqn:taylorexpandu3subterm4}
\begin{split}
\int_{-h}^h \int_{-h}^h \int_{B^{2D}_{\delta}(x,y)} \lambda(\bfx',\bfx) \xi^3_3 (u_3(\bfx',t)- u_3(\bfx,t)) d \bfx' d z \approx{}& \int_{-h}^h \int_{-h}^h \int_{B^{2D}_{\delta}(\bf0)} \lambda(\bfx',\bfx) \xi^4_3 d \xi_1 d \xi_2 d z' d z \frac{\partial u_3}{\partial z} (x,y,0,t).
\end{split}
\end{equation}
Solving for $\frac{\partial u_3}{\partial z}(x,y,0,t)$ in \eqref{eqn:taylorexpandu3subterm4} provides us with the estimate \footnote{Through the relationship \eqref{eqn:lambdacondfullgenb} or the relaxed formulation \eqref{eqn:CijklRelationAverage}, the denominator in~\eqref{eqn:taylorexpandu3subterm5} will be $4h C_{3333}$, which is positive for stable materials \cite{cowley1976acoustic}.}
\begin{equation}\label{eqn:taylorexpandu3subterm5}
\begin{split}
\frac{\partial u_3}{\partial z} (x,y,0,t) \approx \frac{ \int_{-h}^h \int_{-h}^h \int_{B^{2D}_{\delta}(x,y)} \lambda(\bfx',\bfx) \xi^3_3 (u_3(\bfx',t)- u_3(\bfx,t)) d \bfx' d z }{\int_{-h}^h \int_{-h}^h \int_{B^{2D}_{r}(\bf0)} \lambda(\bfx',\bfx) \xi^4_3 d \xi_1 d \xi_2 d z' d z} .
\end{split}
\end{equation}
Recalling \eqref{eqn:avgperistresssolvedu3}, we may replace the numerator of \eqref{eqn:taylorexpandu3subterm5} with an expression solely in $u_1$ and $u_2$:
\begin{equation}\label{eqn:planestressfirstestimate}
\begin{split}
\frac{\partial u_3}{\partial z} (x,y,0,t) \approx -\frac{\int_{-h}^h \int_{-h}^h \int_{B^{2D}_{\delta}(x,y)} \lambda(\bfx',\bfx) \xi^2_3 \left[ \xi_1 (u_1(\bfx',t)- u_1(\bfx,t))  + \xi_2 (u_2(\bfx',t)- u_2(\bfx,t)) \right]  d \bfx' d z}{\int_{-h}^h \int_{-h}^h \int_{B^{2D}_{r}(\bf0)} \lambda(\bfx',\bfx) \xi^4_3 d \xi_1 d \xi_2 d z' d z}.
\end{split}
\end{equation}

Now that we have the estimate \eqref{eqn:planestressfirstestimate}, the derivation of \eqref{eqn:PlaneStrainu3Subeqn2} is fairly straightforward. By again employing Assumption (P$\sigma$\ref{assump:PSs6Peri}), we apply a second-order Taylor expansion to the left-hand side of \eqref{eqn:PlaneStrainu3Subeqn2} about $(x',y',0,t)$ for $u_3(\bfx',t)$ and $(x,y,0,t)$ for $u_3(\bfx)$ to find (for $i = 1,2$)
\begin{equation}\label{eqn:eqn:taylorexpandu3subterm5}
\begin{split}
&\int_{-h}^h \int_{-h}^h \int_{B^{2D}_\delta(x,y)} \lambda(\bfx',\bfx) \xi_i \xi_3 (u_3(\bfx',t) - u_3(\bfx,t)) d \bfx' d z  \\
&\approx \int_{-h}^h \int_{-h}^h \int_{B^{2D}_\delta(x,y)} \lambda(\bfx',\bfx) \xi_i \xi_3 \left[ u_3(x',y',0,t) + z'\frac{\partial u_3}{\partial z} (x',y',0,t)  + \frac{z'^2}{2} \frac{ \partial^2 u_3}{\partial z'^2}(x',y',0,t) \right. \\
&\hspace*{1.8in} \left. -u_3(x,y,0,t) - z\frac{\partial u_3}{\partial z} (x,y,0,t)  - \frac{z^2}{2} \frac{ \partial^2 u_3}{\partial z^2}(x,y,0,t) \right] d \bfx' d z.
\end{split}
\end{equation}
We again exploit antisymmetry with respect to the transformation $(z',z) \rightarrow (-z',-z)$ to reduce \eqref{eqn:eqn:taylorexpandu3subterm5} to
\begin{equation}\label{eqn:approxbeforeantisymmremoval}
\begin{split}
&\int_{-h}^h \int_{-h}^h \int_{B^{2D}_\delta(x,y)} \lambda(\bfx',\bfx) \xi_i \xi_3 (u_3(\bfx',t) - u_3(\bfx,t)) d \bfx' d z \\
&\approx \int_{-h}^h \int_{-h}^h \int_{B^{2D}_\delta(x,y)} \lambda(\bfx',\bfx) \xi_i \xi_3 \left( z'\frac{\partial u_3}{\partial z} (x',y',0,t)  - z\frac{\partial u_3}{\partial z} (x,y,0,t)  \right) d \bfx' d z.
\end{split}
\end{equation}
Recall from \eqref{eqn:lambdaeveninxi1xi2} that $\lambda$ is symmetric with respect to $(\xi_1,\xi_2) \rightarrow (-\xi_1,-\xi_2)$. Thus, the second term in the integrand in \eqref{eqn:approxbeforeantisymmremoval} is antisymmetric in $(\xi_1,\xi_2)$ and therefore we may reduce \eqref{eqn:approxbeforeantisymmremoval} to  
\begin{equation}\label{eqn:approxafterantisymmremoval}
\begin{split}
&\int_{-h}^h \int_{-h}^h \int_{B^{2D}_\delta(x,y)} \lambda(\bfx',\bfx) \xi_i \xi_3 (u_3(\bfx',t) - u_3(\bfx,t)) d \bfx' d z \approx \int_{-h}^h \int_{-h}^h \int_{B^{2D}_r(x,y)} \lambda(\bfx',\bfx) \xi_i \xi_3  z'\frac{\partial u_3}{\partial z} (x',y',0,t)  d \bfx' d z.
\end{split}
\end{equation}
By interchanging $z$ and $z'$ and employing \eqref{assump:MonoclinicPlaneStress}, we have
\begin{equation}\label{eqn:getz'termalso}
\int_{-h}^h \int_{-h}^h \int_{B^{2D}_\delta(x,y)} \lambda(\bfx',\bfx) \xi_i \xi_3  z'\frac{\partial u_3}{\partial z} (x',y',0,t)  d \bfx' d z = - \int_{-h}^h \int_{-h}^h \int_{B^{2D}_\delta(x,y)} \lambda(\bfx',\bfx) \xi_i \xi_3  z \frac{\partial u_3}{\partial z} (x',y',0,t)  d \bfx' d z.
\end{equation}
Since the right-hand sides of \eqref{eqn:approxafterantisymmremoval} and \eqref{eqn:getz'termalso} are equivalent, we may sum them and divide by 2 in order to reformulate \eqref{eqn:approxafterantisymmremoval} as 
\begin{equation}\label{eqn:approxafterantisymmremovalfinal}
\begin{split}
\int_{-h}^h \int_{-h}^h \int_{B^{2D}_\delta(x,y)} \lambda(\bfx',\bfx) \xi_i \xi_3 (u_3(\bfx',t) - u_3(\bfx,t)) d \bfx' d z % \approx& \frac{1}{2} \int_{-h}^h \int_{-h}^h \int_{B^{2D}_\delta(x,y)} \lambda(\bfx',\bfx) \xi_i \xi_3  z'\frac{\partial u_3}{\partial z} (x',y',0,t)  d \bfx' d z \\
%&- \frac{1}{2} \int_{-h}^h \int_{-h}^h \int_{B^{2D}_\delta(x,y)} \lambda(\bfx',\bfx) \xi_i \xi_3  z\frac{\partial u_3}{\partial z} (x',y',0)  d \bfx' d z \\
\approx{}& \frac{1}{2} \int_{-h}^h \int_{-h}^h \int_{B^{2D}_\delta(x,y)} \lambda(\bfx',\bfx) \xi_i \xi_3^2 \frac{\partial u_3}{\partial z} (x',y',0,t)  d \bfx' d z.
\end{split}
\end{equation}
Replacing the term $\frac{\partial u_3}{\partial z}(x',y',0,t)$ in \eqref{eqn:approxafterantisymmremovalfinal} with the corresponding estimate from \eqref{eqn:planestressfirstestimate} completes the proof. \qed
\end{proof}

%%%%%%%%%%%%%%%%%%%%%%%%%%%%%%%%%%%%%%%%%%
%%%%%%%%%%%%%%%%%%%%%%%%%%%%%%%%%%%%%%%%%%
%%%%%%%%%%%%%%%%%%%%%%%%%%%%%%%%%%%%%%%%%%

\subsection{Proof of Lemma~\ref{lemma:getubar}}\label{lemma:getubarapp}

\begin{proof}
This proof has two steps. First, we relate $u_1(\bfx,t)$ and $u_2(\bfx,t)$ to their corresponding evaluations on the plane $z=0$, $u_1(x,y,0,t)$ and $u_2(x,y,0,t)$. This is followed by approximating $u_1(x,y,0,t)$ and $u_2(x,y,0,t)$ with their corresponding averages $\overline{u}_1(x,y,t)$ and $\overline{u}_2(x,y,t)$. Combining the two results completes the proof.

To relate $u_1(\bfx,t)$ and $u_2(\bfx,t)$ to their evaluations on the plane $z=0$, we appeal to Assumption (P$\sigma$\ref{assump:PSs6Peri}) to perform a Taylor expansion:
\begin{equation}\label{eqn:relateutouonplane}
\begin{split}
u_i(\bfx,t) ={}&  u_i(x,y,0,t) + \sum_{n=1}^\infty \frac{1}{n!} \frac{\partial^n u_i}{\partial z^n}(x,y,0,t)z^n \\
={}&  u_i(x,y,0,t) + \sum_{k=1}^\infty \frac{1}{(2k)!} \frac{\partial^{2k} u_i}{\partial z^{2k}}(x,y,0,t)z^{2k} \\
 ={}&  u_i(x,y,0,t) + O(h^2). \\
\end{split}
\end{equation}

In the second equality of \eqref{eqn:relateutouonplane}, we appealed to Lemma \ref{lem:symmplanestress} so that
\begin{equation}\label{eqn:oddderivatives}
\frac{\partial^n u_i}{\partial z^n} (x,y,z,t) = (-1)^n \frac{\partial^n u_i}{\partial z^n}(x,y,-z,t)  \Rightarrow \frac{\partial^n u_i}{\partial z^n}(x,y,0,t) = 0 \; \; \text{for } n \text{ odd}.
\end{equation}

Similarly, by appealing to Assumption (P$\sigma$\ref{assump:PSs6Peri}), we perform a Taylor expansion to find

\begin{equation}\label{eqn:almostlinearinuegivesavg}
\begin{split}
\overline{u}_i (x,y,t) =& \frac{1}{2h} \int_{-h}^h u_i (\bfx,t) d z \\
=& \frac{1}{2h} \int_{-h}^h u_i(x,y,0,t) + \sum_{n=1}^\infty \frac{1}{n!} \frac{\partial^n u_i}{\partial z^n} (x,y,0,t) z^n  d z  \\
=& \frac{1}{2h} \int_{-h}^h u_i(x,y,0,t) + \sum_{k=1}^\infty \frac{1}{(2k)!} \frac{\partial^{2k} u_i}{\partial z^{2k}} (x,y,0,t) z^{2k}  d z  \\
=& u_i(x,y,0,t) + \frac{1}{h} \sum_{k=1}^\infty \frac{1}{(2k)!} \frac{\partial^{2k} u_i}{\partial z^{2k}} (x,y,0,t) \frac{h^{2k+1}}{2k+1} \\
=& u_i(x,y,0,t) + O(h^2). \\
\end{split}
\end{equation}
In the third line of \eqref{eqn:almostlinearinuegivesavg}, we utilized \eqref{eqn:oddderivatives}. Combining \eqref{eqn:relateutouonplane} and \eqref{eqn:almostlinearinuegivesavg} completes the proof.

\hspace*{0in} \qed
\end{proof}
%%%%%%%%%%%%%%%%%%%%%%%%%%%%%%%%%%%%%%%%%%%%%%%%%%%%%%%
%%%%%%%%%%%%%%%%%%%%%%%%%%%%%%%%%%%%%%%%%%%%%%%%%%%%%%%
%%%%%%%%%%%%%%%%%%%%%%%%%%%%%%%%%%%%%%%%%%%%%%%%%%%%%%%

\begin{acknowledgements}
Research sponsored by the Laboratory Directed Research and Development Program of Oak Ridge National Laboratory, managed by UT-Battelle, LLC, for the U. S. Department of Energy.
\end{acknowledgements}

\bibliographystyle{spmpsci}      
\bibliography{JetBib1.bib}

\end{document}